\PassOptionsToPackage{unicode}{hyperref}
\PassOptionsToPackage{hyphens}{url}
\PassOptionsToPackage{dvipsnames,svgnames,x11names}{xcolor}
\documentclass[
  12pt]{article}

\usepackage{amsmath,amssymb}
\usepackage{iftex}
\ifPDFTeX
  \usepackage[T1]{fontenc}
  \usepackage[utf8]{inputenc}
  \usepackage{textcomp} 
\else 
  \usepackage{unicode-math}
  \defaultfontfeatures{Scale=MatchLowercase}
  \defaultfontfeatures[\rmfamily]{Ligatures=TeX,Scale=1}
\fi
\usepackage{lmodern}
\IfFileExists{upquote.sty}{\usepackage{upquote}}{}
\IfFileExists{microtype.sty}{
  \usepackage[]{microtype}
  \UseMicrotypeSet[protrusion]{basicmath} 
}{}
\makeatletter
\@ifundefined{KOMAClassName}{
  \IfFileExists{parskip.sty}{%
    \usepackage{parskip}
  }{
    \setlength{\parindent}{0pt}
    \setlength{\parskip}{6pt plus 2pt minus 1pt}}
}{
  \KOMAoptions{parskip=half}}
\makeatother
\usepackage{xcolor}
\usepackage{longtable,booktabs,array}
\usepackage{calc} 
\usepackage{etoolbox}
\makeatletter
\patchcmd\longtable{\par}{\if@noskipsec\mbox{}\fi\par}{}{}
\makeatother
\IfFileExists{footnotehyper.sty}{\usepackage{footnotehyper}}{\usepackage{footnote}}
\makesavenoteenv{longtable}
\usepackage{graphicx}
\makeatletter
\def\maxwidth{\ifdim\Gin@nat@width>\linewidth\linewidth\else\Gin@nat@width\fi}
\def\maxheight{\ifdim\Gin@nat@height>\textheight\textheight\else\Gin@nat@height\fi}
\makeatother
\makeatletter
\def\fps@figure{htbp}
\makeatother

\makeatletter
\@ifpackageloaded{caption}{}{\usepackage{caption}}
\@ifpackageloaded{float}{}{\usepackage{float}}
\@ifpackageloaded{subcaption}{}{\usepackage{subcaption}}
\makeatother

\AtBeginEnvironment{table}{\renewcommand{\baselinestretch}{1}\small\normalsize}
\AtBeginEnvironment{figure}{\renewcommand{\baselinestretch}{1}\small\normalsize}
\AtBeginEnvironment{sidewaystable}{\renewcommand{\baselinestretch}{1}\small\normalsize}
\AtBeginEnvironment{sidewaysfigure}{\renewcommand{\baselinestretch}{1}\small\normalsize}
\AtBeginEnvironment{lyxalgorithm}{\par\vspace{0.45\baselineskip}\renewcommand{\baselinestretch}{1}\small\normalsize}

\makeatletter
\newcommand*{\jbes@tabletype}{table}
\g@addto@macro\@floatboxreset{\ifx\@captype\jbes@tabletype\footnotesize\fi}
\makeatother

\usepackage[english]{babel}
\usepackage{mathrsfs}
\usepackage{multirow}
\usepackage[safe]{tipa}
\usepackage{amsthm}
\usepackage{stmaryrd}
\usepackage{stackrel}
\usepackage{rotfloat}
\usepackage{wasysym}

\makeatletter

\newcommand{\lyxmathsym}[1]{\ifmmode\begingroup\def\b@ld{bold}
  \text{\ifx\math@version\b@ld\bfseries\fi#1}\endgroup\else#1\fi}
\providecommand{\tabularnewline}{\\}
\theoremstyle{plain}
\newtheorem{assumption}{\protect\assumptionname}
\theoremstyle{plain}
\newtheorem{lyxalgorithm}{\protect\algorithmname}
\theoremstyle{plain}
\newtheorem{thm}{\protect\theoremname}
\theoremstyle{plain}
\newtheorem{prop}{\protect\propositionname}
\theoremstyle{plain}
\newtheorem{lem}{\protect\lemmaname}
\theoremstyle{plain}
\newtheorem{remark}{Remark}
\makeatother

\providecommand{\algorithmname}{Algorithm}
\providecommand{\assumptionname}{Assumption}
\providecommand{\lemmaname}{Lemma}
\providecommand{\propositionname}{Proposition}
\providecommand{\theoremname}{Theorem}

\usepackage[authoryear]{natbib}
\usepackage{xr-hyper}
\usepackage{bookmark}

\IfFileExists{xurl.sty}{\usepackage{xurl}}{} 
\hypersetup{
  pdftitle={Supervised Mixed-Frequency Learning for Macro-Financial Forecasting When Factors are Weak},
  pdfauthor={Ulrich Hounyo; Zhendong Li},
  pdfkeywords={Supervised scaled PCA, weak factors, factor-MIDAS, boosting, macro-financial forecasting},
  colorlinks=true,
  linkcolor={blue},
  filecolor={Maroon},
  citecolor={Blue},
  urlcolor={Blue},
  pdfcreator={LaTeX via pandoc}}

\newcommand{\anon}{1}

\begin{document}

\if0\anon\hypersetup{pdfauthor={}}\fi

\def\spacingset#1{\renewcommand{\baselinestretch}%
{#1}\small\normalsize} \spacingset{1}


\if1\anon
{
  \title{\bf Supervised Mixed-Frequency Learning for Macro-Financial
Forecasting When Factors are Weak}
  \author{Ulrich Hounyo\hspace{.2cm}\\
    Department of Economics, University at Albany - SUNY\\
    and \\
    Zhendong Li \\
    Department of Economics, University at Albany - SUNY}
  \maketitle
} \fi

\if0\anon
{
  \bigskip
  \bigskip
  \bigskip
  \begin{center}
    {\LARGE\bf Supervised Mixed-Frequency Learning for Macro-Financial
Forecasting When Factors are Weak}
\end{center}
  \medskip
} \fi


\bigskip
\begin{abstract}
Factor-MIDAS regressions forecast a low-frequency target by extracting
common factors from a large panel of high-frequency predictors via
principal component analysis (PCA). While PCA mitigates the curse of
dimensionality, it relies on factor pervasiveness, an assumption often
violated when factors are weak, as is common in macro-financial
forecasting. We propose SsPCA-MIDAS, which integrates supervised scaled
PCA (SsPCA) into the mixed-data sampling framework. We establish
consistency and asymptotic normality under weak factors, permitting
inference on the prediction target. Simulations show that SsPCA-MIDAS
outperforms competing PCA-based and supervised methods, especially when
weak factors are prevalent. Applying machine-learning techniques such as
boosting to the cleaner factors it extracts yields further gains. An
extensive application to U.S. macro-financial forecasting shows that
SsPCA-MIDAS selects economically meaningful predictors and improves
forecasts of GDP, inflation, unemployment, asset prices, and volatility.
\end{abstract}

\noindent%

\noindent%
{\it J.E.L. Codes:} C38, C53, C55, C58, E44.
\vfill

\newpage
\spacingset{1.8} 

\section{Introduction}

Mixed-Data Sampling (MIDAS) regressions are widely used in economic forecasting. Introduced by \cite{ghysels2004midas,ghysels2007midas}, they exploit high-frequency predictors through a parametric aggregation function to forecast a lower-frequency target, and have become a standard tool for nowcasting and prediction in finance and macroeconomics \citep{Clements2008}.

A notable extension is the ``factor-MIDAS'' (or factor-augmented MIDAS) regression, which incorporates common factors extracted from large panels of higher-frequency series. Exploiting the dimension-reduction properties of factor models, these regressions forecast effectively and are widely applied \citep{marcellino2010factor,kim2018methods,ferrara2019nowcasting}. \cite{koh2023inference} provides theoretical support by deriving the asymptotic distribution of the factor-MIDAS coefficient estimators with PCA-estimated factors, and \cite{beyhum2024factor} extend the literature with a factor-augmented sparse MIDAS regression combining sparse and dense dimension reduction for nowcasting.

PCA constructs orthogonal principal components capturing maximum variance. It is consistent under strong factors \citep{bai2002determining,bai2003inferential} and remains effective under weaker assumptions \citep{bai2023approximate}. But it is unsupervised, ignoring the target, which limits its ability to identify the most useful low-dimensional predictors for mixed-frequency forecasting: when the signal-to-noise ratio is low, the factor space spanned by the principal components becomes inconsistent or even nearly orthogonal to the true factor space \citep{johnstone2009consistency}. We call such factors \textsl{weak}. Weak factors have been defined in various ways, including large idiosyncratic variances, small nonzero loadings, and sparsity with many zero loadings \citep{lettau2020estimating,uematsu2022estimation,freyaldenhoven2022factor}.

Weak factors are hard to capture and can bias predictions substantially, the \textsl{weak factor problem}. First identified by \cite{kan1999gmm,kan1999two} for useless factors and generalized by \cite{kleibergen2009tests}, who shows weak factors distort inference throughout the model, it was further shown by \cite{giglio2023prediction,giglio2021test} to bias the estimation of all factors, including strong ones, impairing asset pricing and forecasting.

Importantly, the weak factor problem is not merely a statistical artifact; it is tied to the structure of macro-financial data. A factor is weak when its loadings are pervasive only across a limited subset of predictors, so it explains a small share of the panel's cross-sectional variance. Because principal components are ordered by explained variance, such factors are relegated to lower-ranked components that are discarded or swamped by noise. Yet the variance a factor explains need not coincide with its predictive content: signals that move only a narrow set of predictors can still carry first-order information about future outcomes. Many economically important signals fit this profile. The term spread and corporate bond spreads load on a few interest-rate and credit series, yet are among the most reliable predictors of output growth and recessions \citep{Ang2006,ludvigson2009macro}; the global financial cycle operates through a limited set of cross-country channels, so its factor is weak in variance terms while remaining informative for domestic outcomes \citep{rey2015dilemma,miranda2020us}. 

To recover such weak but economically meaningful factors, a natural remedy is supervision, using the target to guide factor extraction. Supervised PCA originated in cancer diagnosis \citep{bair2006prediction} and entered macroeconomic forecasting via \cite{bai2008forecasting} through hard and soft thresholding.

\cite{huang2022scaled} propose Scaled PCA (sPCA), which scales each predictor by its predictive slope relative to the target before applying PCA, upweighting stronger predictors and downweighting weaker ones and so partially addressing weak factors. Since sPCA still uses the entire predictor set, however, irrelevant predictors are downweighted but not eliminated, leaving residual noise.

In parallel, \cite{giglio2023prediction,giglio2021test} propose Supervised PCA (SPCA), which uses supervised selection: guided by the target, it iteratively selects a subset of the most predictive predictors, extracts a latent factor by PCA, and projects it out before repeating, raising the signal-to-noise ratio. For complex predictor sets with diverse irrelevant factors, though, the selected subset may not remain sufficiently informative.

In a same-frequency setting, \cite{hounyo2023forecasting} propose Supervised Scaled PCA (SsPCA) to jointly address weak factors and irrelevant-factor noise, but without theoretical justification. SsPCA first scales each predictor by its predictive slope and then applies SPCA; the scaling amplifies relevant signals and shrinks irrelevant ones while the selection removes uninformative predictors entirely, the two steps being complementary. Table \ref{Table A.1} in Appendix \ref{Appendix A} summarizes how these PCA-based methods differ in their use of target information, required factor strength, and recovery and prediction properties, with details in Section \ref{Section 3.4}.

Our primary contribution is to establish the asymptotic theory of SsPCA, including consistency and inference on the prediction target, in a framework where the sample size and cross-sectional dimension may grow at different rates, and to show analytically why SsPCA outperforms existing supervised PCA methods.

We further extend this framework to factor-MIDAS, terming the procedure SsPCA-MIDAS (hereafter SsPCA unless the distinction matters), which is practically important given the growing availability of high-frequency data for mixed-frequency forecasting and nowcasting. Unlike the supervised dynamic PCA of \cite{gao2024supervised}, which scales and combines predictors and their lags in a same-frequency setting, MIDAS aggregates lagged information parsimoniously through a parametric weighting function \citep{ghysels2007midas}. Building on \cite{koh2023inference}, who develops inference for strong factor-MIDAS, we establish the corresponding theory for weak factors, accounting for supervised selection and the jointly estimated MIDAS weights.

In our setting, the weak factor problem arises from the factor loading
matrix, whose eigenvalues grow but potentially more slowly than the
cross-sectional dimension; the factors we examine are thus weaker than
in \cite{bai2023approximate}. PCA fails to recover such factors, leading
to biased predictions \citep{giglio2023prediction}, and we show that
sPCA is also biased here. When factors are extremely weak, with
eigenvalues of the same order as the idiosyncratic ones
\citep{onatski2009testing,onatski2010determining}, the factor space
cannot be consistently recovered \citep{onatski2012asymptotics}, and
SsPCA is no exception; the self-normalized score statistic of
\cite{chao2022selecting} could replace our correlation screening but
does not resolve this case. Like \cite{giglio2023prediction}, our results
do not require perfect recovery of the predictors correlated with the
factors and account for errors accumulated over the iterative process.

A growing literature applies machine learning to problems where
traditional econometric methods struggle \citep{kelly2023financial},
including within MIDAS regressions, to exploit modern datasets and the
sluggish update frequency of some indicators \citep{babii2022machine}.
In particular, \cite{lahiri2022boosting} use MIDAS and boosting to
forecast New York State tax revenues and find that boosting with PCA
factors performs best in their context. This motivates us to ask whether
the cleaner factors extracted by SsPCA can further enhance boosting for
mixed-frequency forecasting.

Through Monte Carlo simulations, we assess the finite-sample performance
of SsPCA against other PCA-based methods and PLS within the factor-MIDAS
framework, across three scenarios that vary factor strength from strong
to weak, with irrelevant factors adding noise in all cases. SsPCA
outperforms competing methods, especially when weak factors dominate,
and applying boosting to the refined factors it extracts yields further
gains in prediction accuracy.

As an extension, we show that the asymptotic consistency of factor estimation in boosting with PCA, established by \cite{bai2009boosting}, carries over to our setting. Specifically, boosting combined with SsPCA ensures consistent factor estimation for both component-wise and block-wise boosting. This suggests that SsPCA may benefit techniques using latent factors.

Finally, we conduct an extensive empirical application to U.S.
macro-financial forecasting, covering eight quarterly targets and over
1,500 monthly predictors from U.S. and international sources. SsPCA
consistently improves forecast accuracy relative to PCA-based methods and
other benchmarks across horizons from one quarter to two years, with
further gains from boosting. The selected predictors are economically
interpretable, with rates, spreads, leverage, and profitability measures
central and cross-country factors contributing, and the COVID-19 crisis
marks a clear structural break in predictor relevance, highlighting the
method's ability to adapt to evolving information sets.

The rest of the paper is organized as follows. Section \ref{Section 2}
introduces the factor-MIDAS model and develops the SsPCA procedure. Section \ref{Section 3} establishes the asymptotic theory of SsPCA. Simulation results are shown in Section
\ref{Section 4}, and Section \ref{Section 5} presents the empirical
application to U.S. macro-financial mixed-frequency forecasting. Section
\ref{Section 6} concludes the paper. The appendices collect additional details and material omitted from the main text.

\section{Methodology\label{Section 2}}

\subsection{Notation}

Throughout this paper, we use the following notations. For any low-frequency
time series of vectors $\left\{ z_{t}\right\} _{t=1}^{T}$, we use
the capital letter $Z$ to denote the matrix $\left(z_{1},z_{2},...,z_{T}\right)$,
$\overline{Z}$ for $\left(z_{1+h},z_{2+h},...,z_{T}\right)$, and
$\underline{Z}$ for $\left(z_{1},z_{2},...,z_{T-h}\right)$, for
some horizons $h$. For any time series of high-frequency vectors
$\ensuremath{\left\{ z_{t}\right\} _{t=1}^{T_{H}}}$, we use the bold
letter $\ensuremath{\bold{z}_{H}}$ to denote the matrix $\ensuremath{\left(z_{1},z_{2},...,z_{T_{H}}\right)}$
and $\ensuremath{\underline{\bold{z}}_{H}}$ for $\ensuremath{\left(z_{1},z_{2},...,z_{T_{H}-mh}\right)}$,
where $m$ is the multiple by which high-frequency data increases
when an additional unit of low-frequency data is added. We use $[N]$
to denote the set of integers: $\{1,2,...,N\}.$ For an index set
$I\subset[N]$, we use $|I|$ to denote its cardinality. We use $Z_{[I]}$
to denote a sub-matrix of $Z$ whose rows are indexed in $I$.

We write $z\apprle x$ if $z\leq Cx$ for some $C>0$, $z\apprle_{\textrm{P}}x$ if $z=O_{\textrm{P}}(x)$, $z\asymp x$ if $z\apprle x$ and $x\apprle z$, and $z\asymp_{\textrm{P}}x$ analogously. We denote by $\lambda_{\min}(Z)$, $\lambda_{\max}(Z)$, and $\lambda_{i}(Z)$ the minimum, maximum, and $i$th largest eigenvalues of $Z$, and by $\sigma_{i}(Z)$ its $i$th singular value. The operator ($l_2$) norm, Frobenius norm, and max ($l_\infty$) norm of $Z=(z_{ij})$ are $\parallel Z\parallel=\sqrt{\lambda_{\max}(Z'Z)}$, $\parallel Z\parallel_{\textrm{F}}=\sqrt{\textrm{Tr}(Z'Z)}$, and $\Vert Z\Vert_{\max}=\max_{i,j}\mid z_{ij}\mid$. Finally, $\mathbb{P}_{Z}=Z(Z'Z)^{-1}Z'$ and $\mathbb{M}_{Z}=\mathbb{I}_{d}-\mathbb{P}_{Z}$ for any $Z$ with $d$ rows, where $\mathbb{I}_{d}$ is the $d\times d$ identity.

\subsection{Model Setup}

We predict a vector of low-frequency targets $y_{t+h}$, $h$ steps ahead, from $N$ high-frequency predictors $x_{t}$ with sample size $T$, where each predictor is observed at most $m$ times between $t-1$ and $t$. A MIDAS regression aggregates the high-frequency variables through a lag polynomial. We consider the factor autoregressive distributed lag (ADL)-MIDAS model:
\begin{equation}
y_{t+h}=\alpha B(L^{1/m};\theta_{f})f_{t}+\alpha_{w}w_{t}+\varepsilon_{t+h}=\alpha\stackrel[j=0]{J}{\sum}b^{f}_{j}({\theta}_f)f_{t-j/m}+\alpha_{w}w_{t}+\varepsilon_{t+h}\qquad t=1,...,T-h,\label{equation 1}
\end{equation}

\noindent where $\alpha=(\alpha_{1},\alpha_{2},...,\alpha_{K})$,
$\theta_{f}=(\theta'_{f,1},\theta'_{f,2},...,\theta'_{f,K})'$ with $\theta_{f,k}=(\theta^f_{k,1},\theta^f_{k,2},...,\theta^f_{k,p_f})'$
a $p_f\times1$ vector of weighting parameters for the $k$-th factor,
for $k=1,...,K$. $B(L^{1/m};\theta_{f})=\sum_{j=0}^{J}b_{j}^{f}(\theta_{f})L^{j/m}$
and $L^{j/m}f_{t}=f_{t-j/m}$, where $f_{t-j/m}$ is $K\times1$ vector
of latent high-frequency factors. Here, $b^{f}_{j}({\theta}_f)$ is a $K\times K$
diagonal matrix weighting function that temporally aggregates the
regressors and their lags, and $b_{j}^{f}(\theta_{f})\equiv\textrm{diag}(b_{j,1}^{f}(\theta_{f,1}),b_{j,2}^{f}(\theta_{f,2}),...,b_{j,K}^{f}(\theta_{f,K}))$,
such that $b_{j,k}^{f}(\theta_{f,k})$ is the weight for the $j$-th lag
of the $k$-th factor. Here, $j$ represents the lag order of the
high-frequency variable and $j=0,...,J$. The $w_{t}$ is a $\mathcal{W}\times1$
vector of predetermined variables such as a constant and lags of $y_{t+h}$.
$\varepsilon_{t+h}$ is a vector of prediction errors. 

We assume that the high-frequency predictors $x_t$ follow a linear factor model:
\begin{equation}
x_{i,t-j/m}=\beta{}_{i}f_{t-j/m}+e_{i,t-j/m},\qquad i=1,...,N,t=1,...,T,j=m-1,...,0,\label{equation 2}
\end{equation}

\noindent where $\beta'_{i}$, $i=1,...,N,$ are $K\times1$ vectors
of factor loadings and $e_{i,t-j/m}$ is an idiosyncratic error term
satisfying $\mathbb{E}(e_{i,t-j/m})=0$, $\mathbb{E}(f_{t-j/m}e'_{i,t-j/m})=0,$
and $\mathbb{E}(w_{t}e'_{i,t-j/m})=0$.

We define
\[
F_{t}(\theta_{f})\equiv B(L^{1/m};\theta_{f})f_{t}=\sum_{j=0}^{J}b_{j}^{f}(\theta_{f})f_{t-j/m},
\]
as the aggregated factor, so that equation (\ref{equation 1}) can
be written as 
\begin{equation}
y_{t+h}=\alpha F_{t}(\theta_{f})+\alpha_{w}w_{t}+\varepsilon_{t+h},\qquad t=1,...,T-h.\label{equation 3}
\end{equation}

To identify $\alpha$, we assume that $b_{j,k}^{f}(\theta_{f,k})\in(0,1)$
and $\sum_{j=0}^{J}b_{j,k}^{f}(\theta_{f,k})=1$. A common weighting scheme
in the MIDAS regression model is the exponential Almon lag with two
parameters ($p=2$) such that 
\[
b_{j,k}^{f}(\theta_{f,k})=\frac{exp(\theta_{1,f,k}j+\theta_{2,f,k}j^{2})}{\sum_{j=0}^{J}exp(\theta_{1,f,k}j+\theta_{2,f,k}j^{2})}.
\]

Further weighting schemes are discussed in \cite{ghysels2007midas} and \cite{ghysels2009multi}. Parameters are estimated by nonlinear least squares (NLS), since the MIDAS regression is nonlinear in them, except for the unrestricted MIDAS (U-MIDAS) of \cite{foroni2015unrestricted}, estimable by OLS.

For convenience, we use the high-frequency time index $t_{h}=1,...,T_{H}$ with $T_{H}=mT$, where $t_{h}=m((t-1)+n/m)$ for $n=1,...,m$, so (\ref{equation 2}) reads $x_{i,t_{h}}=\beta_{i}f_{t_{h}}+e_{i,t_{h}}$. Equations (\ref{equation 2}) and (\ref{equation 3}) then take the matrix form
\[
\bold{x}_{H}=\beta\bold{f}_{H}+\bold{e}_{H},\qquad\overline{Y}=\alpha\underline{F}(\theta_{f})+\alpha_{w}\underline{W}+\overline{\epsilon}.
\]

\subsection{Assumptions}

We now state the assumptions on the data generating processes, comparable to those in \cite{huang2022scaled} and \cite{giglio2023prediction} to ease comparison with these same-frequency benchmarks. Hereafter $C$ is a generic constant, and in all asymptotics $N,T\rightarrow\infty$ while $h$, $K$, and $\mathcal{W}$ are fixed.
\begin{assumption}
\noindent\label{assu:1}The factors $\bold{f}_{H}$ satisfy $\sup_{t_{h}}\mathbb{E}\parallel f_{t_{h}}\parallel^{4}\leq C$
and $\sideset{\frac{1}{T_{H}}}{_{t_{h}=1}^{T_{H}}}\sum f_{t_{h}}f_{t_{h}}'\overset{\textrm{P}}{\rightarrow}\mathbf{\Sigma}_{f}$,
where $\mathbf{\Sigma}_{f}$ is a $K\times K$ positive definite matrix
with $\lambda_{K}(\mathbf{\Sigma}_{f})\apprge1$ and $\lambda_{1}(\mathbf{\Sigma}_{f})\apprle1$.
In addition, $\bold{f}_{H}$, the prediction error $\epsilon$ and the
predetermined regressor $W$ satisfy the following bounds:
\[
\begin{array}{c}
\parallel T_{H}^{-1}\underline{\bold{f}}_{H}\underline{\bold{f}}_{H}'-\mathbf{\Sigma}_{f}\parallel\apprle_{\textrm{P}}T_{H}^{-1/2},\parallel\bold{f}_{H}\parallel_{\max}\apprle_{\textrm{P}}\log(T_{H})^{1/2},\parallel T^{-1}\underline{WW}'-\mathbf{\Sigma}_{w}\parallel\apprle_{\textrm{P}}T^{-1/2},\\
\parallel W\underline{F}'(\theta_f)\parallel\apprle_{\textrm{P}}T{}^{1/2},\parallel\epsilon\parallel\apprle_{\textrm{P}}T{}^{1/2},\parallel\epsilon\parallel_{\max}\apprle_{\textrm{P}}\log(T)^{1/2},
\parallel\overline{\epsilon}\underline{F}'(\theta_f)\parallel\apprle_{\textrm{P}}T{}^{1/2},\parallel\overline{\epsilon}\underline{W}'\parallel\apprle_{\textrm{P}}T{}^{1/2},
\end{array}
\]
where $\mathbf{\Sigma}_{w}$ is a $\mathcal{W}\times\mathcal{W}$
positive definite matrix with $\lambda_{\mathcal{W}}(\mathbf{\Sigma}_{w})\apprge1$
and $\lambda_{1}(\mathbf{\Sigma}_{w})\apprle1$, $\underline{F}(\theta_f)$
is the aggregated factor.
\end{assumption}
Assumption \ref{assu:1} imposes mild conditions on the time-series properties of $f_{t_{h}}$, $\varepsilon_{t}$, and $w_{t}$, satisfied by stationary, strongly mixing processes with sufficient moments, and ensures that the $K$ left-singular values of $\bold{f}_{H}$ are bounded away from zero and infinity. In our context it is therefore the factor loadings that determine factor strength.
\begin{assumption}
\noindent\label{assu:2}The $N\times K$ factor loading matrix $\beta$
satisfies $\parallel\beta\parallel_{\max}\apprle1$ and
$\lambda_{K}(\beta'_{[I_{0}]}\beta{}_{[I_{0}]})\apprge N_{0}$
for some index set $I_{0}\subset[N]$, where $N_{0}=\mid I_{0}\mid\rightarrow\infty.$
\end{assumption}
Assumption \ref{assu:2} implies that one can construct a subset $I_{0}$
of predictors on which all latent factors are pervasive, while still
allowing a certain proportion of irrelevant factors. It is considerably
weaker than pervasiveness across all predictors, where
$\lambda_{1}(\beta'\beta)\asymp...\asymp\lambda_{K}(\beta'\beta)\asymp N$:
these eigenvalues may now grow at different rates slower than $N$, since
$N_{0}/N$ can diminish rapidly, and no restriction is imposed on
$\beta_{[I_{0}^{c}]}$ (see \citet{giglio2023prediction} for discussion). It excludes only
the extreme case where all entries of $\beta$ vanish uniformly, i.e.,
$\sup_{I,\mid I\mid\rightarrow\infty}\mid I\mid^{-1}\lambda_{K}(\beta'_{[I]}\beta{}_{[I]})=o_{\textrm{P}}(1)$,
in which case the desired $I_{0}$ would not exist.
\begin{assumption}
\noindent\label{assu:3}The idiosyncratic error term $e_{i,t_{h}}$
satisfies $e_{i,t_{h}}=\sigma_{i}e_{i,t_{h}}^{*}$, where $e_{i,t_{h}}^{*}$
is i.i.d. with $\sup_{i,t_{h}}\mathbb{E}\parallel e_{i,t_{h}}^{*}\parallel^{8}\leq C$
and $C^{-1}\leq\sigma_{i}\leq C$, where $\sigma_{i}$ is the standard
deviation. Additionally, the idiosyncratic error matrix $\bold{e}_{H}$ satisfies
$\parallel\bold{e}_{H}\parallel_{\max}\apprle_{\textrm{P}}(\log T_{H})^{1/2}+(\log N)^{1/2}$
and $\parallel\bold{e}_{H,[I]}\parallel\apprle_{\textrm{P}}\mid I\mid^{1/2}+T_{H}^{1/2}$
for any given non-random subset $I\subset[N]$.
\end{assumption}
Assumption \ref{assu:3} constrains the time-series dependence and heteroskedasticity of $e_{i,t_{h}}$. Under the i.i.d.\ structure on $e^{*}_{i,t_h}$, the first inequality follows from a standard large deviation theorem \citep[see, e.g.,][]{fan2011high} and the second from random matrix theory. This structure is imposed only for transparency, as the analysis uses only the two bounds, which should continue to hold under weak time-series and cross-sectional dependence. We do not assume a uniform bound over all index sets, since the rate $|I|^{1/2}+T_{H}^{1/2}$ may not hold universally.
\begin{assumption}
\noindent\label{assu:4}Suppose that $\{f_{t_{h}}\}$, $\{e_{i,t_{h}}\}$
and $\{\beta_{i}\}$ can be weakly dependent. For any non-random subset
$I\subset[N]$, the factor loading $\beta_{[I]}$ and the idiosyncratic
error $\bold{e}_{H,[I]}$ satisfy:

(i) $\parallel\underline{\bold{e}}_{H,[I]}Z'\parallel\apprle_{\textrm{P}}\lvert I\rvert^{1/2}T^{1/2},\,\parallel\underline{\bold{e}}_{H,[I]}Z'\parallel_{\max}\apprle_{\textrm{P}}(\log N)^{1/2}T^{1/2},$

(ii) \mbox{$\parallel\beta'_{[I]}\bold{e}_{H,[I]}\parallel\apprle_{\textrm{P}}\lvert I\rvert^{1/2}T^{1/2},\,\parallel\beta'_{[I]}\bold{e}_{H,[I]}\parallel_{\max}\apprle_{\textrm{P}}\lvert I\rvert^{1/2}(\log T)^{1/2},\,\parallel\beta'_{[I]}\underline{\bold{e}}_{H,[I]}Z'\parallel_{\max}\apprle_{\textrm{P}}\lvert I\rvert^{1/2}T^{1/2},$}

(iii) $\parallel(e_{T_{H}})'_{[I]}\underline{\bold{e}}_{H,[I]}Z'\parallel\apprle_{\textrm{P}}\lvert I\rvert+\lvert I\rvert^{1/2}T^{1/2},\,\parallel\beta'_{[I]}(e_{T_{H}})_{[I]}\parallel\apprle_{\textrm{P}}\lvert I\rvert^{1/2},$

\noindent\textsl{where $Z$ can be either $\underline{\bold{f}}_{H}$,
$\underline{W}$ or $\overline{\epsilon}$. If $Z=\underline{\bold{f}}_{H}$,
then in the bounds all occurrences of $T$ are replaced by $T_{H}$.
If the variables have different frequencies, high-frequency ones are
aggregated by an appropriate weighting scheme.}
\end{assumption}
In Assumption \ref{assu:4}, the $l_{2}$-norm bounds are implied by \cite{bai2003inferential} when $I=[N]$, and the max-norm results follow from a large deviation theorem as in \cite{fan2011high}.

\cite{bai2023approximate} and \cite{wang2017asymptotics} show that PCA remains consistent under non-pervasive loadings as long as $N/(T\lambda_{K}(\beta'\beta))\to0$, the latter identifying a bias term in the borderline case $N\asymp T\lambda_{K}(\beta'\beta)$. Assumptions \ref{assu:2} and \ref{assu:3} are the key identification conditions of our weak factor model; since we do not require $e_{t_{h}}$ to be stationary, $N/(T\lambda_{K}(\beta'\beta))$ may diverge, so traditional PCA typically fails unless the errors are homoscedastic. We therefore assume a subset $I_{0}\subset[N_{0}]$ with $\mid I_{0}\mid/(T\lambda_{K}(\beta_{[I_{0}]}'\beta_{[I_{0}]}))\rightarrow0$, identifying the factors within it. All these settings require $\lambda_{K}(\beta'\beta)\to\infty$, in contrast to the extremely weak factor model of \citet{onatski2012asymptotics}, which restricts $\lambda_{K}(\beta'\beta)\lesssim1$.
\begin{assumption}
\noindent\label{assu:5} The prediction error $\varepsilon_{t+h}$ and the associated score vector $g_{\kappa,t}$ satisfy the following moment and central limit theorem (CLT) conditions:

(i) $\mathbb{E}(\varepsilon_{t+h})=0,$ $\mathbb{E}||\varepsilon_{t+h}||^{2}<C,$

(ii) $\mathbb{E}\parallel g_{\kappa,t}\parallel^{4}\leq C,$ and $\frac{1}{T}\sum_{t=1}^{T}g_{\kappa,t}g'_{\kappa,t}\overset{\textrm{p}}{\longrightarrow}\mathbf{\Sigma}_{\kappa}>0,$

(iii) $\mathbb{E}\parallel\frac{1}{\sqrt{T}}\sum_{t=1}^{T}g_{\kappa,t}\varepsilon_{t+h}\parallel^{2}<C,$ and as $T\rightarrow\infty,$ $\frac{1}{\sqrt{T}}\sum_{t=1}^{T}g_{\kappa,t}\varepsilon_{t+h}\overset{\textrm{d}}{\longrightarrow}\mathcal{N}(0,\Omega_{\kappa}),$

\noindent\textsl{where $g_{\kappa,t}=\partial g(F_{t},\kappa)/\partial\kappa,$ $g(F_{t},\kappa)=\alpha F_{t}(\theta)+\alpha_{w}w_{t},$ $\kappa=(\alpha'_{l},\theta')'$ with $\alpha_{l}=(\alpha,\alpha_{w})', \theta=(\theta_{f},\theta_{x})'$. Here, $\theta_x$ denotes the weighting parameters for the high-frequency predictors, $F_{t}=(f'_{t},f'_{t-1/m},...,f'_{t-J/m})',$ and $\Omega_{\kappa}$ is the long-run variance of the empirical score process.}
\end{assumption}

Assumption \ref{assu:5} imposes the regularity conditions on the prediction error and score vector required for joint NLS estimation of $\alpha_{l}$ and $\theta$: the moment and covariance conditions ensure the sample score function behaves well uniformly in the parameter space, while the CLT makes the aggregated score process asymptotically Gaussian, as in \cite{koh2023inference} and \cite{gonccalves2014bootstrapping}. The parameter $\theta_x$ is defined in the next subsection.

\subsection{Algorithm and Implementation}

We now describe the estimation algorithm. The literature offers two representative PCA regression algorithms \citep{stock2002forecasting,giglio2023prediction}; since the former is more prone to overfitting and instability \citep{giglio2023prediction}, we adopt the latter, constructing $\hat{f}_{T_{H}}$ from the in-sample estimated weights for all subsequent PCA computations within the MIDAS model.\footnote{For PCA-based methods at the same frequency, see \cite{hounyo2023forecasting}.} For notational convenience, superscripts on predicted values denote the forecasting method (e.g., $\widehat{y}_{T+h}^{\textrm{SsPCA}}$) and subscripts on estimated factors the extraction method (e.g., $\underline{\bold{f}}_{H,\textrm{SsPCA}}$); we omit the `SsPCA' label unless needed, so an unlabeled estimator is understood to use SsPCA.

Existing supervised methods each have a limitation: SPCA \citep{giglio2023prediction,giglio2021test} selects an informative subset but cannot downweight irrelevant factors, while sPCA \citep{huang2022scaled} upweights informative predictors but cannot fully exclude irrelevant ones. \cite{hounyo2023forecasting} combine the two in SsPCA, whose weighting amplifies relevant signals and shrinks irrelevant ones and whose selection removes the remaining noise; Section \ref{Section 3.4} compares them in detail, and we extend SsPCA to the mixed-frequency setting.

The weighting step requires further notations. Let ${\theta}_{x}=(\theta'_{x,1},\theta'_{x,2},...,\theta'_{x,N})'$ with $\theta_{x,i}=(\theta^{x}_{i,1},\theta^{x}_{i,2},...,\theta^{x}_{i,p_x})'$
a $p_x\times1$ vector,
for $i=1,...,N$.  Define $b^{x}_{j}({\theta}_x)$ as an $N\times N$
diagonal matrix weighting function with $b^x_{j}(\theta_x)\equiv\textrm{diag}(b_{j,1}^x(\theta_{x,1}),b^x_{j,2}(\theta_{x,2}),...,b^x_{j,N}(\theta_{x,N}))$. So we have the aggregated predictors   $X_{t}({\theta}_x)=\stackrel[j=0]{J}{\sum}b^{x}_{j}({\theta}_{x})\bold{x}_{t-j/m}$,  and we assume that $b_{j,i}^x(\theta_{x,i})\in(0,1)$
and $\sum_{j=0}^{J}b_{j,i}^x(\theta_{x,i})=1$.

Formally, we define $\Upsilon_{i}$ as the regression slope of the
target variable $y_{t+h}$ on each aggregated (standardized) predictor 
$X_{i,t}(\theta_{x,i})$, and collect these slopes to construct the
diagonal matrix $\Upsilon=\textrm{diag}(\Upsilon_{1},\Upsilon_{2},...,\Upsilon_{N})$.
The scaling coefficient $\widehat{\Upsilon}_{i}$ is estimated from the regression:
\begin{equation}
y_{t+h}=c_{i}+\Upsilon_{i}\stackrel[j=0]{J}{\sum}b_{j,i}^x(\theta_{x,i})x_{i,t-j/m}+u_{i,t+h},\qquad i=1,...,N,t=1,...,T-h,\label{equation 4}
\end{equation}

\noindent where $c_{i}$ and $u_{i,t+h}$ are the constant and error term.

Our goal is to form a panel of scaled high-frequency predictors $(\Upsilon_{1}x_{1,t_{h}},\allowbreak...,\allowbreak\Upsilon_{N}x_{N,t_{h}})$, which shares the same latent factor structure:
\begin{equation}
\Upsilon_{i}x_{i,t_{h}}=\Upsilon_{i}\beta{}_{i}f_{t_{h}}+\Upsilon_{i}e_{i,t_{h}},\qquad i=1,...,N,t_{h}=1,...,T_{H},\label{equation 5}
\end{equation}
i.e. the scaled predictors share the same factors $f_{t_{h}}$ as the originals, in matrix form $\Upsilon\bold{x}_{H}=\Upsilon\beta\bold{f}_{H}+\Upsilon\bold{e}_{H}$.

For a given $\theta_x$, the aggregated predictor is $\underline{X}_{i}(\theta_{x,i})=\left(X_{i,1}(\theta_{x,i}),...,X_{i,T-h}(\theta_{x,i})\right)$ with $X_{i,t}(\theta_{x,i})=\sum_{j=0}^{J}b_{j,i}^{x}(\theta_{x,i})x_{i,t-j/m}$. Regressing $\overline{Y}$ on each standardized $\underline{X}_{i}({\theta}_{x,i})$ gives $\hat{\Upsilon}_{i}({\theta}_{x,i})=\overline{Y}\underline{X}_{i}({\theta}_{x,i})'(\underline{X}_{i}({\theta}_{x,i})\underline{X}_{i}({\theta}_{x,i})')^{-1}$, collected into the diagonal matrix $\hat{\Upsilon}({\theta}_{x})$.

We employ $\mathbb{M}_{\underline{W}'}$ to remove the effect of $\underline{W}$ from $\overline{Y}$, and describe the
SsPCA algorithm within the framework of MIDAS as follows:
\begin{lyxalgorithm}
\textbf{\textup{(Double Supervised Learning: SsPCA-MIDAS)\label{Algorithm 1}}}

\noindent Inputs: $\overline{Y},\bold{x}_{H}=(\underline{\bold{x}}_{H},\bold{x}_{T_{H}-mh+1}^{T_{H}}),\underline{W},\textrm{ and }w_{T}.$

S1. Supervised weighting. Obtain the scaling coefficients $\hat{\Upsilon}_{i}$ by regressing
$\overline{Y}$ on each aggregated predictor $\underline{X}_{i}(\theta_{x,i})$.
Stack them to form the diagonal scaling matrix $\hat{\Upsilon}(\theta_{x})$,
and construct the scaled high-frequency predictors $\bold{x}_{H,\textrm{scaled}}(\theta_{x})=\hat{\Upsilon}(\theta_{x})\bold{x}_{H}$.

S2. Supervised selection. Initialization: $Y_{(1)}:=\overline{Y}\mathbb{M}_{\underline{W}'}$,
$\bold{x}_{(1)}:=\underline{\bold{x}}_{\textrm{H,scaled}}(\theta_x)$, and, for given
values of the parameter ${\theta}_{x}$, $X_{(1)}({\theta}_{x}):=\underline{X}_{\textrm{scaled}}({\theta}_{x})$.
For $k=1,2,...,K$ iterate
the following steps using $X_{(k)}({\theta}_{x})$, $Y_{(k)}$:

$\qquad$a. Supervised selection: select a subset $\widehat{I}_{(k)}({\theta}_{x}):$
\begin{equation}
\widehat{I}_{(k)}(\theta_x):=\left\{ i\mid T^{-1}\Vert(X_{(k)}(\theta_x))_{[i]}Y'_{(k)}\Vert_{\max}\,\ge\,\widehat{c}_{qN}^{(k)}\right\} \subset[N],\label{equation 6}
\end{equation}
where $\widehat{c}_{qN}^{(k)}$ is the $(1-q)$th-quantile of $\left\{ T^{-1}\mid(X_{(k)}(\theta_x))_{[i]}Y'_{(k)}\mid\right\} _{i\in[N]}$, and $q$ is a scaling factor between 0 and 1. This hints selecting predictors with sufficiently
high absolute values of correlation or covariance.

$\qquad$b. Apply the standard PCA, singular value decomposition (SVD), to this subset $\bold{x}_{(k),[\widehat{I}_{(k)}(\theta_x)]}$
but only getting the first left singular vector  $\widehat{\varsigma}_{(k)}(\theta_{x,[\widehat{I}_{(k)}(\theta_x)]})$.
Hence, we can estimate the $k$-th latent factor\footnote{Note that $\widehat{\underline{\bold{f}}}_{H,\text{SsPCA}}$ can be
written as $\widehat{{\underline{\bold{f}}}}_{H,\text{SsPCA}}\equiv\left(\underline{\hat{f}}_{1,\text{SsPCA}},\underline{\hat{f}}_{2,\text{SsPCA}},...,\underline{\hat{f}}_{T_{H}-mh,\text{SsPCA}}\right)$,
where $\hat{f}_{t_{h},\text{SsPCA}}$ is a $K\times1$ vector of estimated
high-frequency factors. For a given value $\theta_f$, the corresponding matrix of aggregated estimated factors
${{\hat{\underline{F}}}}_{\text{SsPCA}}\left(\theta_f\right)=\left(\hat{F}_{1,\text{SsPCA}}\left(\theta_f\right),\hat{F}_{2,\text{SsPCA}}\left(\theta_f\right),...,\hat{F}_{T-h,\text{SsPCA}}\left(\theta_f\right)\right)$,
where $\hat{F}_{t,\text{SsPCA}}\left(\theta_f\right)$ is given by $\hat{F}_{t,\text{SsPCA}}\left(\theta_f\right)=\sum_{j=0}^{J}b^f_{j}(\theta_f)\hat{f}_{t-j/m,\text{SsPCA}}$.} as $\widehat{\underline{\bold{f}}}_{H,(k)}(\theta_x)=\widehat{\varsigma}{}_{(k)}(\theta_{x,[\widehat{I}_{(k)}(\theta_x)]})'\bold{x}_{(k),[\widehat{I}_{(k)}(\theta_x)]}$.
$\widehat{\underline{\bold{f}}}_{H,(k)}(\theta_x)$ can also be rewritten
as  $\widehat{\underline{\bold{f}}}_{H,(k)}(\theta_x) =\widehat{\zeta}{}_{(k)}'(\theta_x)\bold{x}_{(k)}$,
where  $\widehat{\zeta}_{(k)}(\theta_x)=(\mathbb{I}_{N}\sideset{-}{_{i=1}^{k-1}}\sum\widehat{\beta}_{(i)}(\theta_x)$
$\widehat{\zeta}{'}_{(i)}(\theta_x))_{[\widehat{I}_{(k)}(\theta_x)]}^{'}$
$\widehat{\varsigma}_{(k)}(\theta_{x,[\widehat{I}_{(k)}(\theta_x)]})$ is constructed recursively using $\widehat{\beta}_{(k-1)}(\theta_x)$
(defined in c).

$\qquad$c. Projection step: Project $\bold{x}_{(k)}$ onto this factor $\widehat{\underline{\bold{f}}}_{H,(k)}(\theta_x)$ and project
$Y_{(k)}$ onto the corresponding aggregated factor $\widehat{\underline{F}}_{(k)}(\theta)$, where $\theta=(\theta_f,\theta_x)'$,
separately. Then estimate coefficients $\widehat{\beta}_{(k)}(\theta_x)=\bold{x}_{(k)}\widehat{\underline{\bold{f}}}_{H,(k)}'(\theta_x)(\widehat{\underline{\bold{f}}}_{H,(k)}(\theta_x)$ $\widehat{\underline{\bold{f}}}_{H,(k)}'(\theta_x))^{-1}$
and $\widehat{\alpha}_{(k)}(\theta)=Y_{(k)}\widehat{\underline{F}}_{(k)}'(\theta)(\widehat{\underline{F}}_{(k)}(\theta)\widehat{\underline{F}}_{(k)}'(\theta))^{-1}$.
Then compute residuals of  $\bold{x}_{(k)}$  and $Y_{(k)}$,
denoted by $Y_{(k+1)}=Y_{(k)}-\widehat{\alpha}_{(k)}(\theta)\widehat{\underline{F}}_{(k)}(\theta)$
and $\bold{x}_{(k+1)}=\bold{x}_{(k)}-\widehat{\beta}_{(k)}(\theta_x)\widehat{\underline{\bold{f}}}_{H,(k)}(\theta_x)$.

$\qquad$d. Iterate K times in this stage: compute the (univariate)
correlation of the $Y_{(k+1)}$ and $\bold{x}_{(k+1)}$; select
predictors with sufficiently high correlation; extract another one
latent factor; projection, etc. Stop at $k=\widehat{K}$, where $\widehat{K}$
is chosen based on a proper stopping rule.

$\qquad$$\blacktriangleright$ the algorithm terminates as soon as
\begin{equation}
\widehat{c}_{qN}^{(k+1)}<c\:for\:some\:threshold\:c.\label{equation 7}
\end{equation}

S3. Once all factors are collected $\widehat{\underline{F}}_{\textrm{SsPCA}}(\theta)=(\widehat{\underline{F}}{}_{(1)}'(\theta),\widehat{\underline{F}}{}_{(2)}'(\theta),...,\widehat{\underline{F}}{}_{(\widehat{K})}'(\theta))'$,
estimate the coefficients $\alpha_{l}=(\alpha,\alpha_{w})'$
and weighting parameters $\theta$ jointly using NLS, which involve solving the optimization problem $\left(\hat{\alpha}_{l},\hat{\theta}\right)=\underset{\alpha_l,\theta}{\textrm{argmin}}\parallel\overline{Y}-\alpha\underline{\hat{F}}_{\textrm{SsPCA}}(\theta)-\alpha_{w}\underline{W}\parallel^{2}$.
The resulting prediction for $\ensuremath{y_{T+h}}$ is given by $\widehat{y}_{T+h}^{\textrm{SsPCA}}=\hat{\alpha}\hat{F}_{T,\textrm{SsPCA}}(\hat{\theta})+\widehat{\alpha}_{w}w_{T}\equiv\widehat{\alpha}\sum_{j=0}^{J}$
$b^f_{j}(\hat{\theta}_f)\widehat{\zeta}_{\textrm{SsPCA}}'(\hat\theta_x)\bold{x}_{T-j/m,\textrm{scaled}}+\widehat{\alpha}_{w}w_{T}$, where $\widehat{\zeta}_{\textrm{SsPCA}}(\hat\theta_x):=(\widehat{\zeta}_{(1)}(\hat\theta_x),\widehat{\zeta}_{(2)}(\hat\theta_x),...,\widehat{\zeta}_{(\widehat{K})}(\hat\theta_x))$.

outputs: $\widehat{y}_{T+h}^{\textrm{SsPCA}}$, $\widehat{\zeta}_{\text{SsPCA}}(\hat{\theta}_x)$,
$\widehat{\alpha}_{l}$, $\hat{\theta}$, factors $\widehat{\underline{F}}_{\textrm{SsPCA}}(\hat{\theta})$ and their loadings
$\widehat{\beta}(\hat{\theta}_x)$,. 
\end{lyxalgorithm}

Intuitively, at each iteration, the algorithm extracts one factor from the most informative subset and then projects it out. The procedure stops once the most correlated remaining predictors fall below the threshold $c$, signaling that no further predictive information can be recovered, so that the number of extracted factors $\hat{K}$ is determined endogenously by $c$.

We thus tune two parameters: $q$, controlling the subset size for PCA, and $c$, the stopping rule. Following \cite{giglio2023prediction}, we tune $\lfloor qN\rfloor$ rather than $q$ (many $q$ map to the same integer) and $K$ rather than $c$ ($K$ is integer-valued, interpretable, and identifiable from a cutoff among the largest eigenvalues). The proofs use $q$ and $c$; the simulations and applications use $\lfloor qN\rfloor$ and $K$, with the threshold set where the MSE is minimized.

The algorithms for the comparison methods PCA, sPCA, and SPCA appear in Appendix \ref{Appendix E} (Algorithms \ref{Algorithm 2}--\ref{Algorithm 4}). We also consider two supervised alternatives: a block-wise boosting algorithm adapted from \cite{bai2009boosting}, motivated by the evidence in \cite{lahiri2022boosting} that boosting with factors performs best in the MIDAS context, and Partial Least Squares (PLS), common in finance \citep{kelly2013market} but, to our knowledge, not previously used with MIDAS; their algorithms are in Appendix \ref{Appendix E} (Algorithms \ref{Algorithm 5}--\ref{Algorithm 6}). Tuning-parameter choices for all methods, which govern the bias-variance trade-off, are described in Appendix \ref{Appendix F}.

\section{Asymptotic Theory\label{Section 3}}

We now investigate the asymptotic properties of SsPCA. Our analysis builds on the framework of \cite{giglio2023prediction} and must additionally accommodate the supervised scaling coefficients \citep{huang2022scaled} and the NLS estimation of the MIDAS aggregation weights \citep{koh2023inference}, neither of which arises in their setting.

\subsection{Consistency in Prediction\label{Section 3.1}}

To establish prediction consistency, we first examine factor-estimation consistency. In our weak-factor framework, prediction consistency follows from consistent recovery of the factors relevant to the target, unlike the strong-factor case of \cite{stock2002forecasting}, where all factors are recovered via PCA.

Because SsPCA inherits SPCA's selection mechanism, Algorithm \ref{Algorithm 1} yields the selected subsets $\widehat{I}_{(k)}$, $k=1,2,\dots$. We iteratively construct their population counterparts for any chosen $c$ and $q$, which define the factor space recovered by SsPCA. Without loss of generality we set\footnote{Since $\mathbf{\Sigma}_{f}$ is typically taken as $\mathbb{I}_{K}$ and we assume $\sum_{j=0}^{J}b^f_{j}(\theta_f)=\mathbb{I}_{K}$, it follows that $\mathbf{\Sigma}_{F(\theta_f)}=\mathbf{\Sigma}_{f}$.} $\mathbf{\Sigma}_{F(\theta_f)}=\mathbf{\Sigma}_{f}=\mathbb{I}_{K}$; the general case replaces $\beta$ and $\alpha$ with $\beta^{*}=\beta\mathbf{\Sigma}_{f}^{1/2}$ and $\alpha^{*}=\alpha\mathbf{\Sigma}_{F(\theta_f)}^{1/2}$. We also impose $||\Upsilon||_{\max}\lesssim1$, so scaling rescales predictors without changing their directions.

Specifically, we begin by defining $\alpha_{i}^{(1)}:=\parallel(\hat{\Upsilon}\beta)_{[i]}\alpha'\parallel_{\max}$
and set $I_{1}:=\{i\mid\alpha_{i}^{(1)}\geq c_{qN}^{(1)}\}$, where
$c_{qN}^{(1)}$ is the $\lfloor qN\rfloor$th largest value among
$\{\alpha_{i}^{(1)}\}_{i=1,...,N}$. We then denote the largest singular
value of $(\hat{\Upsilon}\beta)_{(1)}:=(\hat{\Upsilon}\beta)_{[I_{1}]}$
by $\lambda_{(1)}^{1/2}$ and the corresponding left and right singular
vectors by $\varsigma_{(1)}$ and $\iota_{(1)}$. For $k>1$, we define
$\alpha_{i}^{(k)}:=\parallel(\hat{\Upsilon}\beta)_{[i]}\prod_{j<k}\mathbb{M}_{\iota_{(j)}}\alpha'\parallel_{\max}$,
$I_{k}:=\{i\mid\alpha_{i}^{(k)}\geq c_{qN}^{(k)}\}$, and $\lambda_{(k)}^{1/2}$,
$\varsigma_{(k)}$, $\iota_{(k)}$ represent the leading singular
value, left, and right singular vectors of $(\hat{\Upsilon}\beta)_{(k)}:=(\hat{\Upsilon}\beta)_{[I_{k}]}\prod_{j<k}\mathbb{M}_{\iota_{(j)}}$.
The process continues until step $\widetilde{K}$ (not necessarily equal to $K$ or $\widehat{K}$) when $c_{qN}^{(\widetilde{K}+1)}<c$. The $I_{k}$ are what SsPCA would select if applied directly to $\hat{\Upsilon}\beta\in\mathbb{R}^{N\times K}$ and $\alpha\in\mathbb{R}^{1\times K}$, being determined by $\alpha,\beta,\Upsilon,\mathbf{\Sigma}_{f},c,q$, and $N$, whereas the $\widehat{I}_{k}$ are random, obtained by applying SsPCA to $\underline{X}_{H,\textrm{scaled}}$ and $\overline{Y}$.

To ensure the singular vectors $\iota_{(j)}$ are well-defined and identifiable, the top two singular values of $(\hat{\Upsilon}\beta)_{(k)}$ must remain distinct at each stage $k$, and distinct $c_{qN}^{(k)}$ are needed so the $I_{k}$ are identifiable. We call two sequences $a_{N}$ and $z_{N}$ asymptotically distinct if $a_{N}\leq(1+\delta)^{-1}z_{N}$ for some $\delta>0$, and impose:
\begin{assumption}
\noindent\label{assu:6}For any given $k$, the following three pairs
of sequences of variables, $\sigma_{1}(\beta_{k})$ and $\sigma_{2}(\beta_{k})$,
$c_{qN}^{(k)}$ and $c_{qN+1}^{(k)}$, and $c_{qN}^{(\widetilde{K}+1)}$
and $c$ are asymptotically distinct, as $N\rightarrow\infty$.
\end{assumption}
This assumption is mild, excluding only corner cases commonly disregarded in the high-dimensional PCA literature (e.g. Assumption 2.1 of \citet{wang2017asymptotics} and Assumption 5 of \citet{giglio2023prediction}).

To identify the NLS estimator of the weighting parameter $\theta$ in the MIDAS framework, we impose $\sqrt{T}/(qN)\rightarrow0$, paralleling the $\sqrt{T}/N\rightarrow0$ of \cite{koh2023inference}, who establishes the asymptotics of the NLS estimator with latent strong factors, building on \cite{andreou2010regression} for observed factors. We discuss the role of this condition for inference, and the regime in which it fails, in Remark \ref{rem:1} below Theorem \ref{thm:4}.

We can show the following result:
\begin{thm}
\label{thm:1}Suppose that $y_{t}$ satisfies (\ref{equation 3})
and scaled predictors $\Upsilon_{i}x_{i,t_{h}}$ follows (\ref{equation 5}),
and that Assumptions \ref{assu:1}-\ref{assu:6} hold. If $\log(NT)(N_{0}^{-1}+T^{-1})\rightarrow0$,
then for any tuning parameters $c$ and $q$ that satisfy
\begin{equation}
c\rightarrow0,\;c^{-1}(\log NT)^{1/2}(q^{-1/2}N^{-1/2}+T^{-1/2})\rightarrow0,\;qN/N_{0}\rightarrow0,\;and\;\sqrt{T}/(qN)\rightarrow0,\label{equation 8}
\end{equation}
\noindent\begin{flushleft}
we have $\widetilde{K}\leq K,$ $P(\widehat{I}_{k}=I_{k})\rightarrow1,$
for any $1\leq k\leq\widetilde{K}$, and $P(\hat{K}=\widetilde{K})\rightarrow1.$
Moreover, the aggregated factors recovered by SsPCA-MIDAS are consistent.
That is, for any $1\leq k\leq\widetilde{K}$,
\begin{equation}
\parallel\underline{\hat{F}}_{(k)}(\hat{\theta})\parallel^{-1}\parallel\underline{\hat{F}}_{(k)}(\hat{\theta})-\underline{\hat{F}}_{(k)}(\hat{\theta})\mathbb{P}_{\underline{F}(\theta_f)'}\parallel\apprle_{\textrm{P}}q^{-1/2}N^{-1/2}+T^{-1}.\label{equation 9}
\end{equation}
\par\end{flushleft}

\end{thm}
Theorem \ref{thm:1} establishes consistency of the factors estimated by SsPCA-MIDAS, generalizing Theorem 1 of \citet{giglio2023prediction} to our setting, where $\hat{\theta}=(\hat{\theta}_x,\hat{\theta}_f)$ are the jointly estimated weighting parameters: $\hat{\theta}_x$ enters the estimated factors through predictor-side weighting, while $\hat{\theta}_f$ aggregates the latent factors.

Note that the assumptions of Theorem \ref{thm:1} do not ensure a consistent estimate of the number of factors $K$. Since SsPCA excludes factors uninformative for $y$, it recovers only the relevant factor space, so, like SPCA, it drops the remaining $K-\tilde{K}$ factors while $\hat{K}$ consistently estimates $\widetilde{K}$. Inequality (\ref{equation 9}) has a clear geometric reading: its left-hand side is $\sin\hat{\Theta}_{(k)}$, the sine of the angle between the stage-$k$ estimated factor and the true factor space $\mathbb{P}_{\underline{F}(\theta_f)'}$, which vanishes asymptotically.

Regarding the tuning parameters, condition (\ref{equation 8}) implies $c\rightarrow0$, $c\sqrt{N}\rightarrow\infty$, and $c\sqrt{qN}\rightarrow\infty$. The threshold $c$ must be small enough to allow iterations until the selected predictors lose predictive power, yet large enough to control screening errors: the sample-covariance error between $\bold{x}_{(1),\textrm{scaled}}(\theta_x)$ and $Y_{(1)}$ is of order $T^{-1/2}$, projection residuals add error when $\widetilde{K}>1$, and the factor-estimation error of order $(qN)^{-1/2}+T^{-1}$ contaminates screening, so $c$ should dominate $T^{-1/2}+(qN)^{-1/2}$. The subset size $\lfloor qN\rfloor$ should scale like $N_{0}$; since $N_{0}$ is not precisely defined, we require $qN/N_{0}\rightarrow0$ so the selected predictors lie within the $N_{0}$-subset supporting a strong factor structure.

The factor-estimation error is unaffected by the scaling $\Upsilon$, and $qN$ plays the role that $N$ plays in the strong-factor case of \cite{bai2003inferential}, where the factor space is recovered at rate $N^{-1/2}+T^{-1}$. Since Assumption \ref{assu:2} does not require equal factor strength, some factors could converge faster if the subset size were tuned to each factor's strength. As our goal is prediction rather than factor recovery, we adopt the $\hat{I}_{k}$ rule in (\ref{equation 6}), which reduces tuning sensitivity and yields more stable out-of-sample predictions than rules selecting a different number of predictors at each stage.

 We define $\hat{\gamma}=\hat{\alpha}\hat{\zeta}'$. With no relevant factors omitted,  the prediction of $\hat{y}_{T+h}$
is consistent, as we illustrate next.
\begin{thm}
\label{thm:2}Under the same assumptions as Theorem \ref{thm:1},
we have $\hat{\alpha}_{w}-\alpha_{w}\overset{\textrm{P}}{\longrightarrow}0$,
$\parallel\hat{\gamma}\hat{\Upsilon}\beta-\alpha\parallel\overset{\textrm{P}}{\longrightarrow}0,$
and consequently, $\hat{y}_{T+h}\overset{\textrm{P}}{\rightarrow}\mathbb{E}_{T}(y_{T+h})=\alpha F_{T}(\theta_f)+\alpha_{w}w_{T}$.
\end{thm}
Theorem \ref{thm:2} establishes prediction consistency by examining the estimation errors $\hat{\alpha}_{w}-\alpha_{w}$ and $\hat{\gamma}\hat{\Upsilon}\beta-\alpha$. For the latter, define $H=\hat{\zeta}'\beta\in \mathbb{R}^{\hat{K}\times K}$, with $\hat{\zeta}$ from Algorithm \ref{Algorithm 1}, so $\hat{\gamma}\hat{\Upsilon}\beta=\hat{\alpha}H$. The adjustment matrix $H$ arises from the indeterminacy of latent factor models, and the theorem implies $\parallel\hat{\alpha}H-\alpha\parallel=o_{\textrm{P}}(1)$, i.e., $\alpha$ is consistently estimated up to $H$. The proof also shows that for $k\leq\tilde{K}$, $\left\Vert \underline{\hat{F}}_{(k)}(\hat{\theta})\right\Vert ^{-1}\left\Vert \underline{\hat{F}}_{(k)}(\hat{\theta})-h_{k}\underline{F}(\theta_f)\right\Vert \lesssim_{\textrm{P}}\,q^{-1/2}N^{-1/2}+T^{-1}$, where $h_{k}$ is the $k$-th row of $H$; hence $\hat{\alpha}\underline{\hat{F}}(\hat{\theta})\approx\hat{\alpha}H\underline{F}(\theta_f)\approx\alpha\underline{F}(\theta_f)$, which with $\hat{\alpha}_{w}-\alpha_{w}=o_{\textrm{P}}(1)$ yields prediction consistency.

Since factors omitted by SsPCA are uncorrelated with $y_{t+h}$, consistency of $\hat{y}_{T+h}$ does not require $\tilde{K}=K$, and the result relies on neither normal errors nor equal factor strength. The extra condition $\sqrt{T}/(qN)\rightarrow0$ reflects the cost of estimating the MIDAS weights, which the richer high-frequency information helps offset in practice.

\subsection{Recovery of All Factors }

In this subsection we develop the asymptotic distribution of $\hat{y}_{T+h}$ from Algorithm \ref{Algorithm 1}. The conditions of Theorem \ref{thm:2} suffice for prediction consistency but not for convergence at the optimal rate $T^{-1/2}$, the main difficulty being the full recovery of all factors.

In the screening step, predictors are selected by their sample covariance with $y_{t+h}$; even if $y_{t+h}$ is independent of $X_{t,\textrm{scaled}}(\hat{\theta}_x)$, this covariance can be large, so a threshold $c>T^{-1/2}(\log N)^{1/2}$ is needed to control Type I error. But when the signal-to-noise ratio is low, e.g. $\alpha\asymp T^{-1/2}$ ($y_{t+h}$ close to noise), screening may stop prematurely because the relevant covariances are at best of order $T^{-1/2}(\log N)^{1/2}<c$, giving no factor discovery. Then $\hat{y}_{T+h}=0$, which is consistent since the bias $\left|\mathbb{E}_{T}(y_{T+h})-0\right|\asymp T^{-1/2}$, but the CLT fails.

This arises when the procedure fails to recover all factors driving the target. If all factors are identified, the bias is negligible and the CLT holds regardless of $\alpha$. To rule out asymptotically missing factors we require $\lambda_{\min}(\alpha^{\prime}\alpha)\gtrsim1$, so each factor contributes non-negligibly to the target\footnote{As in \cite{giglio2023prediction}, this condition is formulated for a general multi-target setting with $\alpha\in\mathbb{R}^{D\times K}$ and $D\geq K$, and the number-of-factors recovery argument remains valid at this generality.}, while no more factors than necessary are selected asymptotically, since the iteration stops once all covariances vanish. A consistent estimator of the number of factors then lets us recover the factor space and perform inference on the prediction target.

Inference in strong-factor models typically requires consistent estimation of the number of factors \citep{bai2002determining}, which is not guaranteed in finite samples. We adopt the PCA regression of \citet{giglio2023prediction}, robust to overestimating the number of factors and supported by finite-sample simulations. The next theorem gives the asymptotic results under conditions ensuring complete recovery of all factors:
\begin{thm}
\label{thm:3}Under the same assumptions as Theorem \ref{thm:2},
if we additionally require that $\lambda_{\min}(\alpha'\alpha)\gtrsim1$,
then for any tuning parameters $q$ and $c$ in equations (\ref{equation 6})
and (\ref{equation 7}) satisfying condition (\ref{equation 8}), we have
(i) $\hat{K}$ defined in Algorithm \ref{Algorithm 1} satisfies
$P(\hat{K}=K)\rightarrow1$; (ii) the factor space is consistently recovered in such a way that $\left\Vert \mathbb{P}_{\underline{\hat{F}}'(\hat{\theta})}-\mathbb{P}_{\underline{F}'(\theta_f)}\right\Vert =O_{\textrm{P}}(q^{-1/2}N^{-1/2}+T^{-1})$;
and (iii) the estimator $\hat{\gamma}$ constructed via Algorithm \ref{Algorithm 1}
satisfies $\left\Vert \hat{\gamma}\hat{\Upsilon}\beta-\alpha-T^{-1}\overline{\epsilon}\underline{F}'(\theta_f)\mathbf{\Sigma}_{F(\theta_f)}^{-1}\right\Vert =O_{\textrm{P}}(q^{-1}N^{-1}+T^{-1}).$
\end{thm}
Part $(i)$ shows that our procedure recovers the true number of factors asymptotically under weak factors; with Theorem \ref{thm:1} it implies $\tilde{K}=K$ under the strengthened assumptions, removing the need to distinguish the two. Unlike \cite{onatski2010determining}, our framework also recovers the space spanned by weak factors, as in $(ii)$, at the factor-estimation rate $(qN)^{1/2}\land T$ rather than $N^{1/2}\land T$ in the strong-factor case. Part $(iii)$ extends Theorem \ref{thm:2} by replacing $\alpha$ with $\alpha+T^{-1}\overline{\epsilon}\underline{F}'(\theta_{f})\mathbf{\Sigma}_{F(\theta_{f})}^{-1}$, the regression estimator of $\alpha$ under observable factors, and shows the latent-factor estimation error is at most $O_{\textrm{P}}(q^{-1}N^{-1}+T^{-1})$.

\subsection{Inference on the Prediction Target }

Without observable regressors $w$, we treat $\theta_f$ and $\theta_x$ as known (justified in Remark \ref{rem:1}) and decompose the prediction error as $\hat{y}_{T+h}-\mathbb{E}_{T}(y_{T+h})=(\hat{\gamma}\hat{\Upsilon}\beta-\alpha)F_{T}(\theta_f)+\hat{\gamma}\hat{\Upsilon}E_{T}(\theta_x)$, where the second term is of order $(qN)^{-1/2}$. By Theorem \ref{thm:3}$(iii)$, if $q^{-1}N^{-1}T\rightarrow0$ this term is $o_{\textrm{P}}(T^{-1/2})$ relative to the first, and $(\hat{\gamma}\hat{\Upsilon}\beta-\alpha)F_{T}(\theta_f)=T^{-1}\overline{\epsilon}\underline{F}(\theta_f)'\mathbf{\Sigma}_{F(\theta_f)}^{-1}F_{T}(\theta_f)+O_{\textrm{P}}(T^{-1})$, so root-$T$ inference on $\mathbb{E}_{T}(y_{T+h})$ is available. For a more accurate finite-sample approximation, we retain both terms without constraining the relative magnitudes of $qN$ and $T$.

For this purpose, we need the following assumption:
\begin{assumption}
\label{assu:7}As $N,T\rightarrow\infty$, $T^{-1/2}\overline{\epsilon}\underline{F}'(\theta_f),$
$T^{-1/2}\overline{\epsilon}\underline{W}'$, and $(qN)^{-1/2}\Psi\hat{\Upsilon}E_{T}(\theta_x)$
are jointly asymptotically normally distributed, satisfying:
\[
\left(\begin{array}{c}
vec(T^{-1/2}\overline{\epsilon}\underline{F}'(\theta_f))\\
vec(T^{-1/2}\overline{\epsilon}\underline{W}')\\
(qN)^{-1/2}\Psi\hat{\Upsilon}E_{T}(\theta_x)
\end{array}\right)\overset{\textrm{d}}{\longrightarrow}\mathcal{N}\left(\left(\begin{array}{c}
0\\
0\\
0
\end{array}\right),\Pi=\left(\begin{array}{ccc}
\Pi_{11} & \Pi_{12} & 0\\
\Pi'_{12} & \Pi_{22} & 0\\
0 & 0 & \Pi_{33}
\end{array}\right)\right),
\]

\noindent where $\Psi$ is a $K\times N$ matrix whose $k$th row
is equal to $\iota'_{(k)}(\hat{\Upsilon}\beta)'_{[I_{k}]}(\mathbb{I}_{N})_{[I_{k}]}$
and $\iota_{(k)}$ is the first right singular vector of $(\hat{\Upsilon}\beta)_{(k)}=(\hat{\Upsilon}\beta)_{[I_{k}]}\prod_{j<k}\mathbb{M}_{\iota_{(j)}}$
as defined in Section \ref{Section 3.1}.
\end{assumption}

Assumption \ref{assu:7} describes the joint asymptotic distribution of $\overline{\epsilon}\underline{F}'(\theta_f),$ $\overline{\epsilon}\underline{W}'$, and $\Psi\hat{\Upsilon}E_{T}(\theta_x)$, a high-level condition like those in \cite{bai2003inferential} and \cite{giglio2023prediction}. As the first two components are finite-dimensional, the CLT for mixing processes applies as $T\to\infty$. The $k$th row of $\Psi\hat{\Upsilon}E_T(\theta_x)$, $\iota_{(k)}'(\hat{\Upsilon}\beta)_{[I_k]}'(\hat{\Upsilon}E_T(\theta_x))_{[I_k]}$, is a weighted average of $\hat{\Upsilon}E_{iT}(\theta_{x,i})$, $i\in I_k$, converging at rate $(qN)^{1/2}$ since $|I_k|=qN$.

Since $\Pi_{33}$ is not restricted to be diagonal, Assumption \ref{assu:7} accommodates cross-sectional correlation and heteroskedasticity in the aggregated idiosyncratic components, with the i.i.d. case as a special instance. As shown in Appendix \ref{Appendix B.5}, $\Pi_{33}=q^{-1}N^{-1}\Psi\Sigma_{E(\theta)}\Psi'$ with $\Sigma_{E(\theta)}=Cov[E_t(\theta)]$ general, and the feasible estimator of $\Phi_2$ there does not presume cross-sectional independence.

To establish the CLT for our analysis, we proceed as follows:
\begin{thm}
\label{thm:4}Suppose the same assumption as in Theorem \ref{thm:3}
hold. If, in addition, Assumption \ref{assu:7} holds, we have
\[
\varPhi^{-1/2}(\hat{y}_{T+h}-\mathbb{E}_{T}(y_{T+h}))\overset{\textrm{d}}{\longrightarrow}\mathcal{N}(0,1),
\]

\noindent where $\varPhi=T^{-1}\varPhi_{1}+q^{-1}N^{-1}\varPhi_{2},$
and $\varPhi_{1}$ and $\varPhi_{2}$ are given by:
\[
\begin{array}{c}
\varPhi_{1}=(F_{T}'(\theta_f),w'_{T})\mathbf{\Sigma}_{F(\theta_f),w}^{-1}\left(\begin{array}{cc}
\Pi_{11} & \Pi_{12}\\
\Pi'_{12} & \Pi_{22}
\end{array}\right)\mathbf{\Sigma}_{F(\theta_f),w}^{-1}(F_{T}'(\theta_f),w'_{T})',\\
\varPhi_{2}=\alpha L(\Lambda/qN)^{-1}\Omega'\Pi_{33}\Omega(\Lambda/qN)^{-1}L'\alpha,
\end{array}
\]

\noindent where $\Pi_{ij}$ is specified by Assumption \ref{assu:7},
$\mathbf{\Sigma}_{F(\theta_f),w}=\textrm{diag}\left(\mathbf{\Sigma}_{F(\theta_f)},\mathbf{\Sigma}_{w}\right)$,
$\Lambda=\textrm{diag}(\lambda_{(1)},...,\lambda_{(K)})$, $L$
is a $K\times K$ matrix whose $k$th column is given by $\iota_{(k)}$, and
$\Omega=(\omega_{1},...,\omega_{K})$ is a $K\times K$ matrix with $\omega_{1}=\varrho_{1}$
and $\omega_{k}=\varrho_{k}\sideset{-}{_{i=1}^{k-1}}\sum\lambda_{(i)}^{-1}\iota'_{(k)}(\hat{\Upsilon}\beta)'_{[I_{k}]}(\hat{\Upsilon}\beta){}_{[I_{k}]}\iota{}_{(i)}\omega_{i}$,
where $\varrho_{k}$ is a $K$-dimensional unit vector with 1 at the
$k$-th entry and $0$ elsewhere, $\lambda_{(k)}^{1/2}$ is the largest singular value of $(\hat{\Upsilon}\beta)_{(k)}$
and $\iota_{(k)}$ is the corresponding right singular vector as defined
in Section \ref{Section 3.1}.
\end{thm}
The asymptotic rate of $\hat{y}_{T+h}$ is driven jointly by $T$ and $qN$, so feasible inference requires consistent estimators of both $\varPhi_{1}$ and $\varPhi_{2}$. Estimating $\varPhi_{1}$ is straightforward from the outputs of Algorithm \ref{Algorithm 1}, while $\varPhi_{2}$ is harder because it depends on the large covariance matrix of $e_{T}(\theta_x)$; details are in Appendix \ref{Prof B.5}.

\begin{remark}[Estimated MIDAS weights]\label{rem:1}
Theorem \ref{thm:4} treats the weighting parameters $\theta=(\theta_f,\theta_x)'$ as known, which is asymptotically justified under the condition $\sqrt{T}/(qN)\to 0$ maintained in Theorem \ref{thm:1}. By Theorem 2.1 of \citet{koh2023inference}, the NLS estimator has an asymptotic bias proportional to $c_0$, the limit of $\sqrt{T}/N$ in her setting and of $\sqrt{T}/(qN)$ in ours; the maintained condition sets $c_0=0$, so $\hat{\theta}-\theta=O_P(T^{-1/2})$ with no first-order bias (Lemma \ref{lem:16}(iii)), and its contribution is absorbed into the $O_P(T^{-1}+q^{-1}N^{-1})$ remainder in the proof of Theorem \ref{thm:4}. Since $\lfloor qN\rfloor\ll N$, this is more demanding than its strong-factor counterpart. When instead $\sqrt{T}/(qN)\to c_0>0$, the plug-in error induces an asymptotic bias, and the bootstrap correction of \citet{koh2023inference} could in principle be adapted, though a formal treatment requires a different asymptotic framework and is left for future research. The regime $\sqrt{T}\ll qN$ is precisely the one for which our procedure is designed: short samples of the low-frequency target combined with large high-frequency panels. In our simulations ($T\in\{30,60\}$, $|I_0|\approx 100$) and empirical application ($T=60$, $\lfloor qN\rfloor$ from a grid starting at $100$), $\sqrt{T}/(qN)$ is at most $0.08$; consistent with this, the standardized prediction errors in Figure \ref{Figure D.9} are well approximated by the standard normal.
\end{remark}

Finally, since the prediction in Algorithm \ref{Algorithm 1} (Step S3)  projects
$y_{T+h}$ directly onto $X_{T}(\theta_x)$ and $w_{T}$,
it is straightforward to compute and readily applicable out of sample,
avoiding the need for complex inference on latent factors.

\subsection{Comparative Advantages of SsPCA\label{Section 3.4}}

We now explain why combining supervised scaling with supervised selection is needed for reliable factor recovery and prediction under weak factors. Table \ref{Table A.1} in Appendix \ref{Appendix A} summarizes the four methods across key dimensions, and we discuss them below.


The four methods differ in how they use target information and which
factors they can recover. Standard PCA requires pervasiveness, that the
eigenvalues of $\beta'\beta$ grow at rate $N$; under weak factors they
grow more slowly, $N/(T\lambda_K(\beta'\beta))$ may diverge, and the
principal components become inconsistent, potentially nearly orthogonal
to the true factors \citep{giglio2023prediction}. Being unsupervised, PCA
also cannot prioritize factors relevant to $y_{t+h}$ over those that
merely explain predictor variance. sPCA scales each predictor by its
predictive slope $\hat{\Upsilon}_i$ before applying PCA, but assigns a
nonzero weight of order $T^{-1/2}$ to every predictor; when only an
$N_0\ll N$ subset loads on the relevant factors, the accumulated noise
from the remaining predictors renders both the factors and the forecast
inconsistent (Propositions \ref{prop:3} and \ref{prop:4}). SPCA instead
selects the $\lfloor qN\rfloor$ predictors most covariant with the target
residual and extracts a factor from that subset, achieving consistent
recovery under Assumption \ref{assu:2}, where pervasiveness is required
only within $I_0\subset[N]$; because screening uses raw covariances,
however, strong irrelevant factors can be selected spuriously and
contaminate the extracted directions. SsPCA combines both steps: scaling
by $\hat{\Upsilon}_i\approx\beta_i\alpha'$ inflates the weights on
relevant predictors and shrinks those on irrelevant ones, so the
subsequent SPCA screening more accurately targets the factors driving
$y_{t+h}$. Proposition \ref{prop:2} gives conditions under which SsPCA
attains strictly lower asymptotic MSFE than SPCA, and Propositions
\ref{prop:3} and \ref{prop:4} show it remains consistent where sPCA
fails. Since its advantage over PCA and PLS follows from established
results, the formal comparison in Appendix \ref{Appendix A} focuses on
SsPCA versus SPCA and sPCA.

\subsection{Asymptotic Properties of Boosting with SsPCA Factors\label{Section 3.5}}

\citet[Proposition 1]{bai2009boosting} establish the consistency of boosting applied to PCA-extracted factors, showing via an error-decomposition argument that the cumulative factor-estimation error vanishes so boosting stays consistent with estimated factors. Building on their same-frequency framework, we extend the analysis to aggregated factors estimated by SsPCA in a mixed-frequency setting. Because boosting repeatedly refits using the estimated factors, the estimation error accumulates over iterations; Proposition \ref{prop:5} (Appendix \ref{Appendix A.3}) controls this through the stopping rule $M(q^{-1/2}N^{-1/2}+T^{-1})\rightarrow0$, under which boosting consistently estimates the sparse model structure with observed or estimated aggregated factors, component-wise or block-wise.

\section{Monte Carlo Simulations\label{Section 4} }

We assess the finite-sample performance of SsPCA-MIDAS against other PCA-based methods and PLS in the factor-MIDAS framework, and evaluate boosting applied to the factors from each PCA-based method. Weak factors are generated as in \citet[Section 3]{hounyo2023forecasting}, i.e. with statistically significant loadings on about $5\%$ of the $N$ cross-sectional units, consistent with our definition.

Specifically, we consider a 3-factor ADL-MIDAS DGP:
\begin{equation}
y_{t+h}=\alpha_{0}+\alpha_{1}y_{t}+\alpha\stackrel[j=0]{J}{\sum}b^f_{j}(\theta_{f,g})g_{t-j/m}+\varepsilon_{t+h},
\end{equation}
\begin{equation}
x_{i,t-j/m}=\beta{}_{i}f_{t-j/m}+e_{i,t-j/m},j=m-1,...,0,
\end{equation}

\noindent where $f_{t-j/m}=(g_{1,t-j/m},g_{2,t-j/m},h_{t-j/m})'$: $g_{1,t-j/m}$ and $h_{t-j/m}$ are strong factors (the latter unrelated to the target) and $g_{2,t-j/m}$ is potentially weak. All three are drawn i.i.d. $\mathcal{N}(0,1)$. The strong loadings $\beta_{i,1}$ and $\beta_{i,h}$ are i.i.d. $\mathcal{N}(0,1)$; for $g_{2,t-j/m}$, the exposure $\beta_{i,2}$ is drawn from a mixture, with probability $\pi$ from $\mathcal{N}(0,1)$ and
$1-\pi$ from $\mathcal{N}(0,0.1^{2})$. The parameter $\pi$ governs the strength of $g_{2,t-j/m}$: we set $\pi=0.5$ when it is strong and $\pi=0.05$ when weak.

The idiosyncratic errors $e_{i,t-j/m}$ are heteroskedastic but independent over $i$ and $t-j/m$, with standard deviations $\sigma_{i}$ ($i=1,...,N$) drawn from $U[0,1]$; for time heterogeneity we multiply each $\sigma_{i}$ by $\sigma_{t_{H}}\sim U[0.5,1.5]$. The measurement errors $\varepsilon_{t+h}$ are i.i.d. $\mathcal{N}(0,1)$. The high-frequency variable is observed at most $m=3$ times per low-frequency period (quarterly target, monthly predictors), and we set $J=11$, so the target loads on $11$ lagged monthly factors. The weighting functions $b^f_{j}(\theta_f)$ and $b^x_{j}(\theta_x)$ use the two-parameter exponential Almon lag with $\theta_{1}=7\times10^{-4}$, $\theta_{2}=-5\times10^{-2}$ (fast-decaying weights). We set $\alpha_{0}=0$, $\alpha_{1}=0.2$ (i.e. $\alpha_{w}=(0,0.2)$), and let $\alpha$ vary across three scenarios:

\textbf{Scenario 1.} Predicting both $g_{1,t-j/m}$ and $g_{2,t-j/m}$, the latter strong, $\alpha=(1,1)$.

\textbf{Scenario 2.} Predicting both $g_{1,t-j/m}$ and $g_{2,t-j/m}$, but the latter weak, $\alpha=(1,1)$.

\textbf{Scenario 3.} Predicting only the weak factor $g_{2,t-j/m}$, $\alpha=(0,1)$, making $g_{1,t-j/m}$ irrelevant.

In all scenarios $T\in\{30,60\}$, so $T_{H}\in\{90,180\}$ since $m=3$. We set $N=200$, $\pi=0.5$ in Scenario 1 and $N=2000$, $\pi=0.05$ in Scenarios 2 and 3, so $\pi N=100$ throughout.

For tuning, PCA, sPCA, and PLS adjust only the number of factors $K$, set between $1$ and $5$ and chosen by 3-fold CV MSE. SPCA and SsPCA add the subset size $\lfloor qN\rfloor$, searched from $10$ to $200$ in steps of $10$ in Scenario 1 and from $100$ to $2000$ in steps of $100$ in Scenarios 2 and 3, with the optimal $(\lfloor qN\rfloor,K)$ again chosen by 3-fold CV MSE.

For the boosting extensions, applied after factor extraction, we fix the boosting-specific parameters to avoid extra tuning cost: the shrinkage $\nu=0.1$ (slow learning) and $M\in\{50,100\}$ iterations. These choices align with the forecasting literature, where useful stopping iterations are typically small at such learning rates \citep{wohlrabe2014assessing,lahiri2022boosting}.

We compare all methods and their boosting extensions in Tables \ref{Table 1}, \ref{Table D.1}, \ref{Table D.2}, and \ref{Table D.3}, based on OOS MSFE and bias at $h=1$ and $h=4$. We use the first $2/3$ of the sample for training and the remaining $1/3$ ($P=T/3$) for testing. For each origin $s=T-P-h+1,\ldots,T-h$, the model is re-estimated through period $s$ with $h$-step error $\tilde{u}_{s+1}=y_{s+1}-\hat{y}_{s+1|s}$, giving $\mathrm{OOS\ MSFE}_h=P^{-1}\sum_{s=T-P-h+1}^{T-h}\tilde{u}_{s+1}^{2}$ and $\mathrm{Bias}_h=P^{-1}\sum_{s=T-P-h+1}^{T-h}\tilde{u}_{s+1}$.

\begin{table}[!t]
\caption{Finite Sample Comparison of Predictions (MSFE)}
\small
\begin{centering}
\setlength{\tabcolsep}{2.5pt}
\begin{tabular}{c|c|cccccc|cccc}
\hline 
\multicolumn{12}{c}{{\small Scenario 1: $N=200,\pi=0.5$}}\tabularnewline
\hline 
\hline 
\multicolumn{1}{c}{} & {\small$T_{H}$} & Oracle & {\small PCA} & {\small SPCA} & {\small sPCA} & {\small SsPCA} & {\small PLS} & {\small Bo-PCA} & {\small Bo-SPCA} & {\small Bo-sPCA} & {\small Bo-SsPCA}\tabularnewline
\hline 
\multirow{2}{*}{{\small$h=1$}} & {\small$90$} & 0.681 & 0.972 & {\small 0.779} & {\small 0.920} & {\small\textbf{0.765}} & {\small 1.020} & {\small 1.002} & {\small 0.852} & {\small 0.951} & {\small\textbf{0.828}}\tabularnewline
\cline{2-12}
 & {\small$180$} & 0.703 & {\small 0.920} & {\small 0.825} & {\small 0.872} & {\small\textbf{0.814}} & {\small 0.948} & {\small 0.942} & {\small 0.859} & {\small 0.892} & {\small\textbf{0.843}}\tabularnewline
\hline 
\multirow{2}{*}{{\small$h=4$}} & {\small$90$} & 0.702 & {\small 1.046} & {\small 0.802} & {\small 0.972} & {\small\textbf{0.789}} & {\small 1.040} & {\small 1.069} & {\small 0.886} & {\small 1.003} & {\small\textbf{0.877}}\tabularnewline
\cline{2-12}
 & {\small$180$} & 0.718 & {\small 0.973} & {\small 0.853} & {\small 0.913} & {\small\textbf{0.852}} & {\small 0.990} & {\small 1.000} & {\small 0.891} & {\small 0.942} & {\small\textbf{0.888}}\tabularnewline
\hline 
\multicolumn{12}{c}{{\small Scenario 2: $N=2000,\pi=0.05$}}\tabularnewline
\hline 
\multicolumn{1}{c}{} & {\small$T_{H}$} & Oracle & {\small PCA} & {\small SPCA} & {\small sPCA} & {\small SsPCA} & \multicolumn{1}{c}{{\small PLS}} & {\small Bo-PCA} & {\small Bo-SPCA} & {\small Bo-sPCA} & {\small Bo-SsPCA}\tabularnewline
\hline 
\multirow{2}{*}{{\small$h=1$}} & {\small$90$} & 0.687 & {\small 0.955} & {\small 0.821} & {\small 0.872} & {\small\textbf{0.783}} & {\small 0.933} & {\small 0.988} & {\small 0.894} & {\small 0.919} & {\small\textbf{0.852}}\tabularnewline
\cline{2-12}
 & {\small$180$} & 0.698 & {\small 0.914} & {\small 0.856} & {\small 0.854} & {\small\textbf{0.839}} & {\small 0.905} & {\small 0.940} & {\small 0.885} & {\small 0.890} & {\small\textbf{0.872}}\tabularnewline
\hline 
\multirow{2}{*}{{\small$h=4$}} & {\small$90$} & 0.706 & {\small 1.018} & {\small 0.834} & {\small 0.912} & {\small\textbf{0.806}} & {\small 0.963} & {\small 1.054} & {\small 0.912} & {\small 0.966} & {\small\textbf{0.879}}\tabularnewline
\cline{2-12}
 & {\small$180$} & 0.722 & {\small 0.974} & {\small 0.888} & {\small 0.896} & {\small\textbf{0.864}} & {\small 0.941} & {\small 1.002} & {\small 0.927} & {\small 0.939} & {\small\textbf{0.908}}\tabularnewline
\hline 
\multicolumn{12}{c}{{\small Scenario 3: $N=2000,\pi=0.05$}}\tabularnewline
\hline 
\multicolumn{1}{c}{} & {\small$T_{H}$} & Oracle & {\small PCA} & {\small SPCA} & {\small sPCA} & {\small SsPCA} & {\small PLS} & {\small Bo-PCA} & {\small Bo-SPCA} & {\small Bo-sPCA} & {\small Bo-SsPCA}\tabularnewline
\hline 
\multirow{2}{*}{{\small$h=1$}} & {\small$90$} & 0.770 & {\small 1.005} & {\small 0.862} & {\small 0.905} & {\small\textbf{0.794}} & {\small 1.012} & {\small 1.035} & {\small 0.910} & {\small 0.942} & {\small\textbf{0.855}}\tabularnewline
\cline{2-12}
 & {\small$180$} & 0.793 & {\small 0.970} & {\small 0.909} & {\small 0.910} & {\small\textbf{0.870}} & {\small 0.989} & {\small 0.990} & {\small 0.938} & {\small 0.933} & {\small\textbf{0.901}}\tabularnewline
\hline 
\multirow{2}{*}{{\small$h=4$}} & {\small$90$} & 0.785 & {\small 1.064} & {\small 0.881} & {\small 0.949} & {\small\textbf{0.822}} & {\small 1.039} & {\small 1.079} & {\small 0.949} & {\small 0.991} & {\small\textbf{0.895}}\tabularnewline
\cline{2-12}
 & {\small$180$} & 0.814 & {\small 1.018} & {\small 0.951} & {\small 0.951} & {\small\textbf{0.888}} & {\small 1.028} & {\small 1.041} & {\small 0.987} & {\small 0.978} & {\small\textbf{0.932}}\tabularnewline
\hline 
\end{tabular}
\par\end{centering}
$\,$

\label{Table 1}{\footnotesize\textbf{Notes: }}{\footnotesize The table
reports the OOS MSFE of PCA, SPCA, sPCA, SsPCA and PLS, and of boosting
applied to the factors each of them extracts (Bo-PCA, Bo-SPCA, Bo-sPCA
and Bo-SsPCA), for $h=1,4$, with the number of boosting iterations fixed
at $M=50$. All values are averages over 1,000 Monte Carlo repetitions.
The Oracle uses the true factors in the MIDAS regression.}{\footnotesize\par}
\end{table}

To assess sensitivity to the subset-size tuning parameter, Table \ref{Table D.6} of Appendix \ref{Appendix D} perturbs $\lfloor qN \rfloor$ around its CV optimum in Scenario 2 with $h=4$ and $T_H=180$, which features coexisting strong and weak factors. The MSFE is lowest at the CV-selected value and changes only gradually nearby, so SsPCA's performance is not driven by a narrow or idiosyncratic choice of $\lfloor qN \rfloor$.

The results in Tables \ref{Table 1} and \ref{Table D.1}--\ref{Table D.3} align with our theory: SsPCA performs best in almost all scenarios, and boosting with SsPCA factors outperforms the other boosting variants. As $g_{2,t-j/m}$ weakens and $g_{1,t-j/m}$ becomes irrelevant from Scenario 1 to 3, SsPCA's advantage grows, especially at smaller $T$. The Oracle, which uses the true factors in the MIDAS regression, is infeasible and bounds the attainable MSFE from below; SsPCA tracks it most closely among feasible methods, so its factor-estimation loss is modest.

SPCA and sPCA both beat PCA but cannot be ranked directly, since SPCA's extra tuning parameter adds overfitting risk; SsPCA combines their strengths, and although it tunes two parameters the gains outweigh the cost, as for boosting with SsPCA. Boosting does not uniformly help---it accumulates error over iterations and can require up to three tuning parameters---yet boosting with PCA-based factors performs strongly in some settings \citep{lahiri2022boosting}, as our empirical results confirm. PLS tends to overfit here, since it uses $y$ directly, especially with many predictors \citep{giglio2023prediction,hounyo2023forecasting}; preprocessing such as thresholding before PLS \citep{hounyo2023forecasting} can mitigate this, so PCA-based approaches are more robust to noisy factors.

For prediction bias, SsPCA and boosting with SsPCA are less uniformly dominant but still attain the lowest bias in nearly half of the configurations, reflecting the bias-variance trade-off. On factor recovery, measured by $d(\hat{\mathbf{f}}_H,\mathbf{f}_H)=||\mathbb{P}_{\hat{\mathbf{f}}_H'}-\mathbb{P}_{\mathbf{f}_H'}||$ (Tables \ref{Table D.4} and \ref{Table D.5} of Appendix \ref{Appendix D}), SsPCA or boosting with SsPCA achieves the smallest distance in Scenario 3, where $g_{2,t-j/m}$ is hardest to detect; in Scenarios 1 and 2 it is marginally behind sPCA, whose all-predictor scaling aids recovery when strong factors are present, but the gap is negligible.

Finally, Figure \ref{Figure D.9} in Appendix \ref{Appendix D} presents histograms of the standardized prediction errors implied by the CLT of Theorem \ref{thm:4}, using the settings of Table \ref{Table 1} with $T_H=180$. For SsPCA the histograms match the standard normal density closely, confirming the CLT; these designs feature small $\sqrt{T}/(qN)$, consistent with Remark \ref{rem:1}. SsPCA and Bo-SsPCA also exhibit the smallest standard deviation across all scenarios.

\section{Empirical Application\label{Section 5}}

We evaluate SsPCA and boosting with SsPCA against the alternatives from our simulations, within the factor-MIDAS setting, targeting eight U.S. variables: four financial (S\&P 500 index, CBOE Volatility Index (VIX), crude oil price, and housing price) and four macroeconomic (real GDP growth, inflation, IP growth, and unemployment), across multiple horizons.

\subsection{Data}

Our analysis combines three datasets. We draw the quarterly target series from FRED-QD \citep{mccracken2016fred} and use the $126$ monthly series of FRED-MD \citep{mccracken2020fred} for part of the predictor set; both are maintained by the Federal Reserve Bank of St. Louis\footnote{See https://www.stlouisfed.org/research/economists/mccracken/fred-databases; we use the `current' version downloaded on June 27th, 2025.} and span the major categories of U.S. macro-financial data. We apply the databases' recommended transformations for stationarity. Third, we use the Global Factor Data of \cite{jensen2023there}, containing $153$ stock-level characteristics for $93$ countries or regions in $13$ thematic clusters, with monthly capped value-weighted returns per factor\footnote{See https://jkpfactors.com/; we use the `current' version downloaded on June 27th, 2025.}. We select the eleven countries or regions most closely linked to the United States, including the U.S. itself, listed with justifications in Table \ref{Table D.7} of Appendix \ref{Appendix D}, obtaining $153$ monthly factor series each; all series are transformed to stationarity or pass stationarity tests.

For a balanced panel, we use monthly data from October 1996 to December 2024, i.e. quarterly targets from 1996 Q3 to 2024 Q4; coverage and missing-value treatment are described following Table \ref{Table D.7}. The final predictor set contains $1{,}555$ series ($126$ from FRED-MD, the rest from the Global Factor Data), letting us examine global effects on the U.S. economy and construct scenarios with weak or complex factor structures.

\subsection{Out of Sample Forecast Evaluation}

We evaluate out-of-sample performance for the eight targets over 2011 Q4 to 2024 Q4. At each origin we re-estimate the model on a rolling window (using the most recently available revised data, once per month) for horizons $h=1,2,4,8$; $h=1$ is a nowcast, since the current quarter's target is not yet fully observed. Balancing the shorter windows suited to volatile financial variables against the longer ones suited to slower macroeconomic variables, and consistent with the simulation design, we use a rolling window of $180$ months ($T_{H}$) for predictors and $60$ quarters ($T$) for targets; Appendix \ref{Appendix F} gives a worked example.

The specification follows the simulation, except that aligning monthly predictors with the quarterly target occasionally involves leading months $l=1,2,3$, so we write
\begin{equation}
y_{t+h}=\alpha_{0}+\alpha_{1}y_{t}+\alpha\stackrel[j=0-l]{J-l}{\sum}b^f_{j}(\theta_{f,g})g_{t-j/m}+\varepsilon_{t+h},
\end{equation}
fixing one quarterly lag (consistent with BIC, which selects one lag on average), with $J=11$ and the two-parameter exponential Almon lag. The benchmark is an AR model with BIC-selected lags (maximum $12$); all methods extract up to the first $10$ latent factors, with $K$ and $\lfloor qN\rfloor$ chosen by rolling-window CV MSE and $\lfloor qN\rfloor$ searched from $100$ to $1500$ in steps of $100$, and for boosting we fix $\nu=0.1$ and $M=100$. Performance is the OOS relative MSFE, $\textrm{RMSFE(method)}=\textrm{MSFE(method)}/\textrm{MSFE(AR,BIC)}$, with values below one indicating improvement over the benchmark.

\subsection{Results}

\subsubsection{Mixed-frequency Forecasting Performance}

Table \ref{Table 2} reports the OOS RMSFE across all targets and horizons. Across most macroeconomic variables, SsPCA and boosting with SsPCA perform best: the scaling step emphasizes predictors more strongly associated with the target, while subset selection removes unrelated high-variance components that can distort medium- and long-horizon forecasts.

\begin{table}[!t]
\renewcommand{\arraystretch}{1.2}
\setlength{\tabcolsep}{2.4pt}

\caption{OOS Performance of Different Methods in Forecasting Macro and Financial
Targets}

\begin{centering}
{\footnotesize{}%
\begin{tabular}{c|c|c|ccccc|cccc}
\hline 
\multicolumn{1}{c}{} & \multicolumn{1}{c}{{\footnotesize Targets}} & \multicolumn{1}{c}{{\footnotesize$h$}} & {\footnotesize PCA} & {\footnotesize SPCA} & {\footnotesize sPCA} & {\footnotesize SsPCA} & {\footnotesize PLS} & {\footnotesize Bo-PCA} & {\footnotesize Bo-SPCA} & {\footnotesize Bo-sPCA} & {\footnotesize Bo-SsPCA}\tabularnewline
\hline 
\hline 
\multirow{16}{*}{{\footnotesize Macro}} & \multirow{4}{*}{{\footnotesize GDP Growth}} & {\footnotesize$1$} & {\footnotesize 0.987} & {\footnotesize 0.954} & {\footnotesize\textbf{\textit{0.792}}} & {\footnotesize 0.830} & {\footnotesize 1.062} & {\footnotesize 0.991} & {\footnotesize 0.948} & {\footnotesize 0.816} & {\footnotesize\textbf{0.801}}\tabularnewline
\cline{3-12}
 &  & {\footnotesize$2$} & {\footnotesize 0.993} & {\footnotesize 0.948} & {\footnotesize\textbf{0.808}} & {\footnotesize 0.824} & {\footnotesize 1.119} & {\footnotesize 0.993} & {\footnotesize 0.957} & {\footnotesize 0.862} & {\footnotesize\textbf{\textit{0.792}}}\tabularnewline
\cline{3-12}
 &  & {\footnotesize$4$} & {\footnotesize 0.965} & {\footnotesize 0.868} & {\footnotesize 0.763} & {\footnotesize\textbf{\textit{0.623}}} & {\footnotesize 1.062} & {\footnotesize 0.961} & {\footnotesize 0.867} & {\footnotesize 0.744} & {\footnotesize\textbf{0.651}}\tabularnewline
\cline{3-12}
 &  & {\footnotesize$8$} & {\footnotesize 1.016} & {\footnotesize 0.970} & {\footnotesize 0.815} & {\footnotesize\textbf{\textit{0.742}}} & {\footnotesize 1.093} & {\footnotesize 1.018} & {\footnotesize 0.989} & {\footnotesize 0.857} & {\footnotesize\textbf{0.756}}\tabularnewline
\cline{2-12}
 & \multirow{4}{*}{{\footnotesize Inflation}} & {\footnotesize$1$} & {\footnotesize 0.943} & {\footnotesize 0.780} & {\footnotesize 0.762} & {\footnotesize\textbf{0.603}} & {\footnotesize 1.086} & {\footnotesize 0.916} & {\footnotesize 0.782} & {\footnotesize 0.731} & {\footnotesize\textbf{\textit{0.582}}}\tabularnewline
\cline{3-12}
 &  & {\footnotesize$2$} & {\footnotesize 0.957} & {\footnotesize 0.822} & {\footnotesize 0.733} & {\footnotesize\textbf{\textit{0.618}}} & {\footnotesize 1.343} & {\footnotesize 0.962} & {\footnotesize 0.818} & {\footnotesize 0.721} & {\footnotesize\textbf{0.627}}\tabularnewline
\cline{3-12}
 &  & {\footnotesize$4$} & {\footnotesize 0.957} & {\footnotesize 0.855} & {\footnotesize 0.890} & {\footnotesize\textbf{\textit{0.812}}} & {\footnotesize 1.163} & {\footnotesize 0.959} & {\footnotesize 0.903} & {\footnotesize 0.899} & {\footnotesize\textbf{0.842}}\tabularnewline
\cline{3-12}
 &  & {\footnotesize$8$} & {\footnotesize 0.963} & {\footnotesize 0.798} & {\footnotesize 0.936} & {\footnotesize\textbf{0.708}} & {\footnotesize 1.347} & {\footnotesize 0.959} & {\footnotesize 0.856} & {\footnotesize 0.886} & {\footnotesize\textbf{\textit{0.699}}}\tabularnewline
\cline{2-12}
 & \multirow{4}{*}{{\footnotesize IP Growth}} & {\footnotesize$1$} & {\footnotesize 1.016} & {\footnotesize 0.994} & {\footnotesize\textbf{0.907}} & {\footnotesize 0.922} & {\footnotesize 1.057} & {\footnotesize 1.017} & {\footnotesize 0.997} & {\footnotesize 0.915} & {\footnotesize\textbf{\textit{0.887}}}\tabularnewline
\cline{3-12}
 &  & {\footnotesize$2$} & {\footnotesize 0.986} & {\footnotesize 0.956} & {\footnotesize 0.843} & {\footnotesize\textbf{0.823}} & {\footnotesize 1.099} & {\footnotesize 0.995} & {\footnotesize 0.936} & {\footnotesize 0.852} & {\footnotesize\textbf{\textit{0.781}}}\tabularnewline
\cline{3-12}
 &  & {\footnotesize$4$} & {\footnotesize 0.943} & {\footnotesize 0.907} & {\footnotesize 0.704} & {\footnotesize\textbf{0.680}} & {\footnotesize 1.037} & {\footnotesize 0.963} & {\footnotesize 0.874} & {\footnotesize 0.721} & {\footnotesize\textbf{\textit{0.647}}}\tabularnewline
\cline{3-12}
 &  & {\footnotesize$8$} & {\footnotesize 1.022} & {\footnotesize 0.962} & {\footnotesize 0.835} & {\footnotesize\textbf{\textit{0.790}}} & {\footnotesize 1.060} & {\footnotesize 1.024} & {\footnotesize 0.985} & {\footnotesize 0.876} & {\footnotesize\textbf{0.818}}\tabularnewline
\cline{2-12}
 & \multirow{4}{*}{{\footnotesize Unemployment}} & {\footnotesize$1$} & {\footnotesize 0.966} & {\footnotesize 0.882} & {\footnotesize\textbf{0.838}} & {\footnotesize 0.842} & {\footnotesize 0.987} & {\footnotesize 0.967} & {\footnotesize 0.865} & {\footnotesize 0.849} & {\footnotesize\textbf{\textit{0.804}}}\tabularnewline
\cline{3-12}
 &  & {\footnotesize$2$} & {\footnotesize 0.933} & {\footnotesize 0.835} & {\footnotesize 0.769} & {\footnotesize\textbf{\textit{0.700}}} & {\footnotesize 0.824} & {\footnotesize 0.929} & {\footnotesize 0.849} & {\footnotesize 0.797} & {\footnotesize\textbf{0.740}}\tabularnewline
\cline{3-12}
 &  & {\footnotesize$4$} & {\footnotesize 0.861} & {\footnotesize 0.771} & {\footnotesize 0.654} & {\footnotesize\textbf{0.588}} & {\footnotesize 0.761} & {\footnotesize 0.874} & {\footnotesize 0.772} & {\footnotesize 0.648} & {\footnotesize\textbf{\textit{0.572}}}\tabularnewline
\cline{3-12}
 &  & {\footnotesize$8$} & {\footnotesize 0.895} & {\footnotesize 0.752} & {\footnotesize 0.680} & {\footnotesize\textbf{\textit{0.591}}} & {\footnotesize 0.718} & {\footnotesize 0.894} & {\footnotesize 0.778} & {\footnotesize 0.716} & {\footnotesize\textbf{0.594}}\tabularnewline
\hline 
\multirow{16}{*}{{\footnotesize Finance}} & \multirow{4}{*}{{\footnotesize S\&P500}} & {\footnotesize$1$} & {\footnotesize 0.873} & {\footnotesize 0.760} & {\footnotesize 0.786} & {\footnotesize\textbf{0.756}} & {\footnotesize 1.005} & {\footnotesize 0.875} & {\footnotesize 0.787} & {\footnotesize 0.817} & {\footnotesize\textbf{\textit{0.661}}}\tabularnewline
\cline{3-12}
 &  & {\footnotesize$2$} & {\footnotesize 0.855} & {\footnotesize 0.707} & {\footnotesize 0.677} & {\footnotesize\textbf{0.612}} & {\footnotesize 0.891} & {\footnotesize 0.854} & {\footnotesize 0.679} & {\footnotesize 0.675} & {\footnotesize\textbf{\textit{0.522}}}\tabularnewline
\cline{3-12}
 &  & {\footnotesize$4$} & {\footnotesize 0.856} & {\footnotesize 0.654} & {\footnotesize 0.523} & {\footnotesize\textbf{0.516}} & {\footnotesize 0.823} & {\footnotesize 0.842} & {\footnotesize 0.608} & {\footnotesize 0.523} & {\footnotesize\textbf{\textit{0.481}}}\tabularnewline
\cline{3-12}
 &  & {\footnotesize$8$} & {\footnotesize 0.917} & {\footnotesize 0.720} & {\footnotesize 0.728} & {\footnotesize\textbf{0.583}} & {\footnotesize 0.877} & {\footnotesize 0.921} & {\footnotesize 0.666} & {\footnotesize 0.720} & {\footnotesize\textbf{\textit{0.559}}}\tabularnewline
\cline{2-12}
 & \multirow{4}{*}{{\footnotesize VIX}} & {\footnotesize$1$} & {\footnotesize 0.378} & {\footnotesize 0.334} & {\footnotesize 0.378} & {\footnotesize\textbf{0.322}} & {\footnotesize 0.435} & {\footnotesize 0.375} & {\footnotesize 0.347} & {\footnotesize 0.371} & {\footnotesize\textbf{\textit{0.314}}}\tabularnewline
\cline{3-12}
 &  & {\footnotesize$2$} & {\footnotesize 0.619} & {\footnotesize\textbf{0.450}} & {\footnotesize 0.505} & {\footnotesize 0.455} & {\footnotesize 0.575} & {\footnotesize 0.622} & {\footnotesize 0.483} & {\footnotesize 0.493} & {\footnotesize\textbf{\textit{0.427}}}\tabularnewline
\cline{3-12}
 &  & {\footnotesize$4$} & {\footnotesize 0.680} & {\footnotesize 0.507} & {\footnotesize 0.596} & {\footnotesize\textbf{0.485}} & {\footnotesize 0.655} & {\footnotesize 0.709} & {\footnotesize 0.555} & {\footnotesize 0.578} & {\footnotesize\textbf{\textit{0.448}}}\tabularnewline
\cline{3-12}
 &  & {\footnotesize$8$} & {\footnotesize 0.826} & {\footnotesize 0.534} & {\footnotesize 0.614} & {\footnotesize\textbf{0.465}} & {\footnotesize 0.528} & {\footnotesize 0.820} & {\footnotesize 0.587} & {\footnotesize 0.603} & {\footnotesize\textbf{\textit{0.458}}}\tabularnewline
\cline{2-12}
 & \multirow{4}{*}{{\footnotesize Oil Price}} & {\footnotesize$1$} & {\footnotesize 0.971} & {\footnotesize 0.826} & {\footnotesize 0.712} & {\footnotesize\textbf{\textit{0.640}}} & {\footnotesize 1.059} & {\footnotesize 0.947} & {\footnotesize 0.814} & {\footnotesize 0.748} & {\footnotesize\textbf{0.642}}\tabularnewline
\cline{3-12}
 &  & {\footnotesize$2$} & {\footnotesize 0.968} & {\footnotesize 0.877} & {\footnotesize 0.675} & {\footnotesize\textbf{0.632}} & {\footnotesize 1.069} & {\footnotesize 0.970} & {\footnotesize 0.869} & {\footnotesize 0.698} & {\footnotesize\textbf{\textit{0.593}}}\tabularnewline
\cline{3-12}
 &  & {\footnotesize$4$} & {\footnotesize 0.929} & {\footnotesize 0.808} & {\footnotesize 0.760} & {\footnotesize\textbf{\textit{0.627}}} & {\footnotesize 0.974} & {\footnotesize 0.932} & {\footnotesize 0.813} & {\footnotesize 0.789} & {\footnotesize\textbf{0.653}}\tabularnewline
\cline{3-12}
 &  & {\footnotesize$8$} & {\footnotesize 0.993} & {\footnotesize 0.808} & {\footnotesize 0.858} & {\footnotesize\textbf{0.707}} & {\footnotesize 1.131} & {\footnotesize 0.996} & {\footnotesize 0.825} & {\footnotesize 0.849} & {\footnotesize\textbf{\textit{0.682}}}\tabularnewline
\cline{2-12}
 & \multirow{4}{*}{{\footnotesize Housing Price}} & {\footnotesize$1$} & {\footnotesize 0.390} & {\footnotesize 0.368} & {\footnotesize 0.361} & {\footnotesize\textbf{0.347}} & {\footnotesize 0.454} & {\footnotesize 0.404} & {\footnotesize 0.356} & {\footnotesize 0.372} & {\footnotesize\textbf{\textit{0.339}}}\tabularnewline
\cline{3-12}
 &  & {\footnotesize$2$} & {\footnotesize 0.705} & {\footnotesize 0.608} & {\footnotesize 0.607} & {\footnotesize\textbf{0.531}} & {\footnotesize 0.744} & {\footnotesize 0.719} & {\footnotesize 0.628} & {\footnotesize 0.592} & {\footnotesize\textbf{\textit{0.515}}}\tabularnewline
\cline{3-12}
 &  & {\footnotesize$4$} & {\footnotesize 0.373} & {\footnotesize 0.284} & {\footnotesize 0.336} & {\footnotesize\textbf{\textit{0.283}}} & {\footnotesize 0.357} & {\footnotesize 0.395} & {\footnotesize\textbf{0.299}} & {\footnotesize 0.333} & {\footnotesize 0.304}\tabularnewline
\cline{3-12}
 &  & {\footnotesize$8$} & {\footnotesize 0.578} & {\footnotesize 0.315} & {\footnotesize 0.422} & {\footnotesize\textbf{0.282}} & {\footnotesize 0.352} & {\footnotesize 0.558} & {\footnotesize 0.334} & {\footnotesize 0.391} & {\footnotesize\textbf{\textit{0.281}}}\tabularnewline
\hline 
\end{tabular}}{\footnotesize\par}
\par\end{centering}
$\,$

{\footnotesize\textbf{\label{Table 2}Notes: }}{\footnotesize The table
reports the out-of-sample RMSFE (relative to the AR,BIC model) of
PCA, SPCA, sPCA, SsPCA and PLS, and of boosting applied to the factors
each of them extracts (Bo-PCA, Bo-SPCA, Bo-sPCA and Bo-SsPCA), in
forecasting four macroeconomic and four financial target variables at
horizons $h=1,2,4,8$ over the sample 2011 Q4 to 2024 Q4. For each target
variable and horizon, italics mark the best value across all methods and
bold marks the best value within the non-boosted and the boosting methods
separately, so that the overall best value is both bold and
italic.}{\footnotesize\par}
\end{table}

In a few cases, mostly at $h=1$, conventional sPCA yields lower RMSFE than SsPCA, likely reflecting a trade-off between shrinkage aggressiveness and robustness under the coarse $\lfloor qN\rfloor$ grid: although SsPCA can in principle select all predictors and match sPCA, short-horizon forecasts of variables such as GDP growth offer scarce high-frequency signal, making the model more tuning-sensitive. Even then, boosting with SsPCA often reverses the result, refining the clean factors and recovering weak short-term signals; consistent with the simulations, it attains the best accuracy at a given horizon among all nine methods in many instances, by refitting residuals on factor-lag blocks to correct factor--target misalignment while early stopping controls variance.

In financial markets, SsPCA again performs best across most variables, with boosting on SsPCA frequently achieving the lowest RMSFE; its advantage is even more pronounced than for macroeconomic variables, with only two instances where SPCA slightly outperforms, matching \cite{hounyo2023forecasting} for same-frequency forecasting. A plausible explanation is that supervised scaling and selection are especially effective where predictive signal arises from a mixture of persistent risk-premium and rapidly shifting sentiment components, dispersed across many predictors and easily masked by volatility or cross-asset spillovers---precisely where the two supervision steps help isolate the structural drivers of asset prices, aided by the many asset-pricing and global financial factors in the predictor set.

\subsubsection{Predictors Selected by SsPCA}

Finally, we examine the predictors SsPCA selects. Figures \ref{Figure D.1}--\ref{Figure D.8} in Appendix \ref{Appendix D} display, as heatmaps, the top $50$ predictors extracted through the first factor for all eight targets at $h=1$ and $h=4$, and Tables \ref{Table D.8} and \ref{Table D.9} list the top $10$ per target, split at the March 2020 COVID break; details are in the figure and table notes.

Several patterns emerge. Global factors are selected persistently, consistent with international spillovers and the global financial cycle \citep{rey2015dilemma,miranda2020us}; selection shifts markedly around the onset of COVID-19 in March 2020 and the recovery in January 2023, indicating a break in the information set; the selected predictors are horizon-dependent; and Rates-group variables are consistently important. Interest-rate and credit-spread measures dominate the top ranks, alongside balance-sheet and real-activity indicators, and the pre- to post-COVID comparison shows a clear rotation from the rates block toward real-activity and supply-side measures and toward trade- and technology-linked Asian economies. SsPCA thus reallocates attention as the dominant source of fluctuations changes rather than relying on a fixed predictor set. A detailed, economically motivated discussion is given in Appendix \ref{Appendix D.1}.

\section{Conclusion\label{Section 6}}

We develop SsPCA-MIDAS to address weak factors in mixed-frequency forecasting and establish its asymptotic theory. Unlike conventional PCA-based approaches, SsPCA combines supervised weighting to emphasize relevant predictors with supervised selection to filter out irrelevant ones, improving factor estimation and preserving predictive relevance under weak factors.

Our analysis shows that SsPCA yields consistent factor estimates, improves factor recovery, ensures prediction consistency, and delivers asymptotic normality, permitting inference on the prediction target; it also explains why SsPCA often outperforms existing supervised PCA methods. The simulations confirm consistent gains in mixed-frequency forecasting, especially when weak factors dominate, together with better factor recovery and more stable, lower-variance standardized prediction errors, and boosting with the cleaner SsPCA factors yields further, asymptotically consistent, improvements. The empirical application to U.S. macro-financial forecasting confirms the method's practical relevance: SsPCA delivers systematic accuracy gains over PCA-based and other benchmarks, with further improvements from boosting, while selecting economically meaningful predictors and adapting to structural shifts such as the COVID-19 crisis.

Looking ahead, mixed-frequency machine-learning methods relying on latent factors could benefit from SsPCA, and future work could develop theoretical explanations for this effect, extend beyond factor-MIDAS to broader macroeconomic forecasting, financial modeling, and asset pricing, and explore computational efficiency, factor-selection strategies, and more complex data environments.


\appendix
\setcounter{footnote}{1}
\setcounter{table}{0} \setcounter{figure}{0}
\renewcommand{\thetable}{D.\arabic{table}}
\renewcommand{\thefigure}{D.\arabic{figure}}
\renewcommand{\thefootnote}{\fnsymbol{footnote}}
\clearpage
\begin{center}{\Large\bf Appendices}\end{center}
\medskip
\noindent The following appendices collect the technical material that
supports the paper. Appendix~\ref{Appendix A} establishes the comparative
advantages of SsPCA over SPCA and sPCA, together with the asymptotic
properties of boosting with SsPCA factors. Appendix~\ref{Appendix B}
provides the mathematical proofs of the theorems and propositions.
Appendix~\ref{Appendix C} presents the technical lemmas and their proofs.
Appendix~\ref{Appendix D} reports supplementary figures and tables and
additional discussion. Appendix~\ref{Appendix E} collects the algorithms
of the competing procedures, and Appendix~\ref{Appendix F} details the
tuning-parameter selection.
\medskip

\spacingset{1.8} 

\section{Comparative Advantages of SsPCA and Asymptotic Properties of Boosting\label{Appendix A}}

\renewcommand{\theequation}{A\arabic{equation}}
\setcounter{equation}{0}

This appendix collects the formal results underlying Sections \ref{Section 3.4}
and \ref{Section 3.5} of the main text. Appendices \ref{Appendix A.1} and
\ref{Appendix A.2} establish the comparisons summarized in Table \ref{Table A.1},
and Appendix \ref{Appendix A.3} states the asymptotic result for boosting with
SsPCA factors.

\begin{table}[!ht]
\centering
{\renewcommand{\thetable}{A.1}%
\caption{Comparison of PCA-Based Methods}
\label{Table A.1}}
\renewcommand{\arraystretch}{1.25}%
\setlength{\tabcolsep}{5pt}%
\small
\linespread{0.8}\selectfont
\begin{tabular}{
>{\raggedright\arraybackslash}p{3.2cm}
>{\raggedright\arraybackslash}p{2.75cm}
>{\raggedright\arraybackslash}p{2.75cm}
>{\raggedright\arraybackslash}p{2.75cm}
>{\raggedright\arraybackslash}p{2.75cm}
}
\toprule
 & \textbf{PCA}
 & \textbf{sPCA}
 & \textbf{SPCA}
 & \textbf{SsPCA} \\
\midrule
Supervision
  & None
  & Scaling only
  & Selection only
  & Scaling+Selection \\

Uses target $y$?
  & No
  & Yes(weights)
  & Yes(screening)
  & Yes(both) \\

Factor strength required
  & Pervasive ($\lambda_K(\beta'\beta) \asymp N$)
  & Pervasive in $N_0$, applied to all $N$
  & Pervasive in subset $I_0$
  & Pervasive in subset $I_0$ \\

Handles weak factors?
  & No
  & Partially
  & Yes
  & Yes \\

Handles irrelevant factors?
  & No
  & Partially(shrinks but retains)
  & Partially(selects but no scaling)
  & Yes(shrinks then removes) \\

Factor recovery
  & Inconsistent under weak factors\citep{giglio2023prediction}
  & Inconsistent when $N/(N_0 T_H^2)\!\to\!\delta\!>\!0$ (Prop. \ref{prop:3})
  & Consistent in $I_0$\citep{giglio2023prediction}
  & Consistent in $I_0$ (Theorem \ref{thm:1}) \\

Prediction consistency
  & No\citep{giglio2023prediction}
  & No(Prop.\ref{prop:4})
  & Yes\citep{giglio2023prediction}
  & Yes(Theorem \ref{thm:2}) \\

Asymptotic MSFE
  & Highest
  & Lower than PCA \citep{huang2022scaled}
  & Higher than SsPCA under conditions (Prop. \ref{prop:2})
  & Lowest under conditions (Props. \ref{prop:1}-\ref{prop:4}) \\

Mixed-frequency extension
  & E.g., \cite{koh2023inference}
  & This paper
  & This paper
  & This paper \\

Tuning parameters
  & $K$
  & $K$
  & $K$, $\lfloor qN \rfloor$
  & $K$, $\lfloor qN \rfloor$ \\
\bottomrule
\end{tabular}
\end{table}
\addtocounter{table}{-1}

\subsection{Comparison with SPCA\label{Appendix A.1}}

Our analysis so far has assumed that both SPCA and SsPCA can identify a subset of predictors in which all factors are pervasive. We now consider cases where the predictor pool still includes a certain degree of irrelevant factors after selection, so that the contrast between raw covariance screening ($\overline{Y}$ and $\underline{X}(\theta_x)$) in SPCA and screening after scaling in SsPCA becomes consequential. To formalize this argument, we extend the partially relevant latent factor structure of \cite{huang2022scaled} to the mixed-frequency setting:
\begin{equation}
x_{i,t_{h}}=\beta{}_{i}f_{t_{h}}+e_{i,t_{h}}\equiv\phi{}_{i}g_{t_{h}}+\psi_{i}h_{t_{h}}+e_{i,t_{h}},\label{eq:A1}
\end{equation}
\begin{equation}
y_{t+h}=\alpha\stackrel[j=0]{J}{\sum}b^f_{j}(\theta_{f,g})g_{t-j/m}+\alpha_{w}w_{t}+\varepsilon_{t+h},
\end{equation}
where $f_{t_{h}}=(g'_{t_{h}},h'_{t_{h}})'$ are $r$-dimensional latent
factors, of which $g_{t_{h}}$ are $r_{1}$- dimensional relevant
factors associated with the target $y_{t+h}$ and $h_{t_{h}}$
are $(r-r_{1})$- dimensional irrelevant factors. $\beta_{i}=(\phi{}_{i},\psi_{i})'$
denote factor loadings for each predictor $i=1,...,N$.

We compare the two methods through their asymptotic mean squared forecast error (MSFE) as the evaluation criterion: $\textrm{MSFE}=\underset{N_{0},T\rightarrow\infty}{\textrm{p}\lim}\frac{1}{T}\stackrel[t=1]{T}{\sum}\left(y_{t+h}-\hat{y}_{t+h}\right)^{2}$.
The following proposition provides the first-order MSFE expressions for the SPCA and SsPCA forecasts\footnote{The analysis here concerns predictor structure and applies equally in same-frequency and mixed-frequency settings.}. 
\begin{prop}
\label{prop:1}Given Assumptions \ref{assu:1}\textendash \ref{assu:7},
as $N_{0}\rightarrow\infty$, $T\rightarrow\infty$, and $\sqrt{N_{0}}/T\rightarrow0$,
the first-order MSFEs of the SPCA and SsPCA forecasts have the following
expressions:
\[
\begin{array}{c}
\textrm{MSF}\textrm{E}_{\textrm{SPCA}}=\frac{1}{T}3\sigma_{\varepsilon}^{2}+\frac{1}{T}\stackrel[t=1]{T}{\sum}\alpha^{*}{}'(\mathfrak{B}'\mathfrak{B})^{-1}\mathfrak{T}_{t}^{\textrm{SPCA}}(\mathfrak{B}'\mathfrak{B})^{-1}\alpha^{*},\\
\textrm{MSF}\textrm{E}_{\textrm{SsPCA}}=\frac{1}{T}3\sigma_{\varepsilon}^{2}+\frac{1}{T}\stackrel[t=1]{T}{\sum}\alpha^{*}{}'(\mathfrak{B}'\mathbb{W}\mathfrak{B})^{-1}\mathfrak{B}'\mathfrak{T}_{t}^{\textrm{SsPCA}}(\mathfrak{B}'\mathbb{W}\mathfrak{B})^{-1}\alpha^{*},
\end{array}
\]

\noindent where $\alpha^{*}=(\alpha,0)',$ $\mathfrak{B}=(\beta_{1},\beta_{2},...,\beta_{N_{0}})'$,
$\mathfrak{T}_{t}^{\textrm{SPCA}}=\frac{1}{N_{0}}\sum_{i=1}^{N_{0}}E(\beta_{i}\beta'_{i}e_{i,t_{h}}^{2})$,
$\mathfrak{T}_{t}^{\textrm{SsPCA}}=\frac{1}{N_{0}}\sum_{i=1}^{N_{0}}E(\phi_{i}^{4}\beta_{i}\beta'_{i}e_{i,t_{h}}^{2})$,
and $\mathbb{W}=\textrm{diag}(\phi_{1}^{2},\phi_{2}^{2},...,\phi_{N_{0}}^{2}).$
\end{prop}
Proposition \ref{prop:1} suggests that, given the current assumptions, it is difficult to determine which method is better. If $e_{i,t_{h}}$
is homoskedastic across $i\in N_{0}$, SPCA may achieve lower MSFE
than SsPCA. In contrast, if $\phi_{i}^{2}$ is proportional to the inverse of $var(e_{i,t_{h}})$,
the MSFE of SsPCA can outperform that of SPCA. Supporting evidence comes from the simulation findings in
\cite{huang2022scaled}, which show that sPCA outperforms PCA in about
70\% of cases under a strong factor design, with our subset $N_{0}$
playing a role similar to their full sample $N$.

This outcome is intuitive under the normalization condition $\mathbf{\Sigma}_{f}=\mathbb{I}_{K}$,
which implies $\phi_{i}^{2}+\psi_{i}^{2}+\sigma_{i}^{2}\approx1,$
where $\sigma_{i}^{2}=var(e_{i,t_{h}}).$ A larger $|\phi_{i}|$ generally
corresponds to a smaller $\sigma_{i}^{2}$. In SsPCA, the weight assigned
to the $i$th predictor is $\hat{\Upsilon}_{i}=\beta_{i}\alpha'+O_\textrm{P}(T^{-1/2})$.
Consequently, the idiosyncratic noise component in the scaled predictor for a generic $i$ has variance approximately $\phi_{i}^{2}\sigma_{i}^{2}$. Intuitively, SPCA
can be related to OLS regressions \citep[analogous to PCA, see, e.g.,][]{bai2003inferential}
and SsPCA can be related to the generalized least square (GLS) regressions
with a specific weighting matrix $\mathbb{W}=\textrm{diag}(\phi_{1}^{2},\phi_{2}^{2},...,\phi_{N_{0}}^{2}).$
Theoretically, the most efficient estimator would use the inverse
residual covariance matrix as a weight $\mathbb{W}$. Thus, the closer
$\mathbb{W}$ aligns with the inverse residual covariance matrix,
the higher the efficiency \citep[see][for further details]{pelger2022interpretable}. Given that $\phi_{i}^{2}$ is generally
negatively related to $\sigma_{i}^{2}$, the weighting matrix $\mathbb{W}$
closely approximates the inverse residual covariance matrix with high
probability. This suggests that, in such cases, the SsPCA factors
are estimated more accurately because the scaled errors are more likely to
be homoskedastic.

Next, we briefly examine the one-factor case, where the first SsPCA factor is used for forecasting as an aggregate index of all predictors. The identification in a partially relevant factor model differs from a fully relevant one, an issue that has received little attention in the literature apart from \cite{huang2022scaled,huang2023bond}.
We also establish sufficient conditions under which the SsPCA forecast
outperforms the SPCA forecast in this setting.
\begin{prop}
\label{prop:2}Suppose $E(e_{i,t_{h}}^{2})=\sigma_{e}^{2}$ for all
$i\in N_{0}=\mid I_{0}\mid$ and $t_{h}$. We have

Case 1: $0<\mathsection^{o}\leq\mathsection$ if $|\phi_{i}|>|\psi_{i}|$
and $\phi_{i}\psi_{i}\geq0$ for all $i$.

Case 2: $0>\mathsection^{o}\geq\mathsection$ if $|\phi_{i}|>|\psi_{i}|$
and $\phi_{i}\psi_{i}<0$ for all $i$.

Case 3: $0>\mathsection\geq\mathsection^{o}$ if $|\phi_{i}|<|\psi_{i}|$
and $\phi_{i}\psi_{i}\geq0$ for all $i$.

Case 4: $0<\mathsection\leq\mathsection^{o}$ if $|\phi_{i}|<|\psi_{i}|$
and $\phi_{i}\psi_{i}<0$ for all $i$.

\noindent Here,
\[
\mathsection^{o}=\frac{\sum_{i=1}^{N_{0}}\phi_{i}^{o}\psi_{i}^{o}}{\sum_{i=1}^{N_{0}}\left(\phi_{i}^{o2}-\psi_{i}^{o2}\right)}=\frac{\sum_{i=1}^{N_{0}}\phi_{i}^{3}\psi_{i}}{\sum_{i=1}^{N_{0}}\phi_{i}^{2}\left(\phi_{i}^{2}-\psi_{i}^{2}\right)},and\;\mathsection=\frac{\sum_{i=1}^{N_{0}}\phi_{i}\psi_{i}}{\sum_{i=1}^{N_{0}}\left(\phi_{i}^{2}-\psi_{i}^{2}\right)}.
\]

The SsPCA forecast outperforms the SPCA forecast in the case of using
the first factor to conduct forecasting, that is, $\textrm{MSF}\textrm{E}_{\textrm{SsPCA}}\leq\textrm{MSF}\textrm{E}_{\textrm{SPCA}}$
if the factor loadings $\phi_{i}$ and $\psi_{i}$ in Equation (\ref{eq:A1})
belong to one of these four cases.
\end{prop}
Proposition \ref{prop:2} provides the counterpart of Proposition 4 in
\cite{huang2022scaled}, where \textquoteleft qualitatively\textquoteright{}
SsPCA plays the role of sPCA, while SPCA plays the role of PCA. It
examines four different cases based on the signs of $\phi_{i}\psi_{i}$
and $|\phi_{i}|-|\psi_{i}|$ to facilitate the discussion on the signs
of the numerator and denominator in the expression $\mathsection=\sum_{i=1}^{N_{0}}\phi_{i}\psi_{i}/\sum_{i=1}^{N_{0}}\left(\phi_{i}^{2}-\psi_{i}^{2}\right)$.
To illustrate, consider the Case 1, where both the numerator and denominator
of $\mathsection$ are positive. Define $z_{i}=\psi_{i}/\phi_{i}$
as the noise-to-signal ratio for predictor $i$, and let $w_{i}=\phi_{i}\psi_{i}/\left(\phi_{i}^{2}-\psi_{i}^{2}\right)=z_{i}/(1-z_{i}^{2})$.
Since $w_{i}$ is an increasing function of $z_{i}$ when $z_{i}<1$,
a predictor with strong predictive power, meaning a lower noise-to-signal
ratio, is associated with a smaller $w_{i}$ value. Let $w_{i}^{u}$
and $w_{i}^{d}$ denote the numerator and denominator of $w_{i}$,
respectively, so that $\mathsection$ can be rewritten as $\mathsection=\left(w_{1}^{u}+w_{2}^{u}+...+w_{N_{0}}^{u}\right)/\left(w_{1}^{d}+w_{2}^{d}+...+w_{N_{0}}^{d}\right)$.
The weighting mechanism of SsPCA assigns greater weight to predictor
$i$ when its ratio $w_{i}=w_{i}^{u}/w_{i}^{d}$ is smaller, which
in turn leads to a reduction in $\mathsection$, as a direct consequence
of the general inequality $a/b\leq\sum_{i}a_{i}/\sum_{i}b_{i}$ when
$a/b\leq a_{i}/b_{i}$ for all $i$. The same argument applies to the other three cases.

Unlike \cite{huang2022scaled}, who assume that all predictors carry strong factors, we restrict the analysis to the informative subset $N_0$.  Since our conclusion is independent of the size of $N_0$, we do not conduct additional simulations.

We need to acknowledge that these analyses rely on stationarity. When predictors exhibit strong
persistence, such as local-to-unit or fractionally integrated
processes, SsPCA may fail, as do PCA 
\citep[see discussions in][]{bai2004estimating,bai2004panic} and its extensions, and we leave refinements
addressing predictor persistence for future research.

\subsection{Comparison with sPCA\label{Appendix A.2}}

In this subsection, we discuss the conditions under which sPCA can
still fail even with scaling coefficients in the presence of weak
factors, whereas SsPCA does not. To illustrate this point, it suffices
to consider a single-factor model and suppose that $x_{t_h}$ follows a sparse loading structure.
Therefore, after scaling, we have:
\begin{equation}
\begin{array}{cc}
\hat{\Upsilon}x_{t_h}=\hat{\Upsilon}\left[\begin{array}{c}
\underline{\,\beta_{1}\,}\\
\,\\
0\\
\,
\end{array}\right]f_{t_h}+\hat{\Upsilon}e_{t_h},\quad & y_{t+h}=\alpha F_{t}(\theta_f),\end{array}\label{eq:A3}
\end{equation}
where $\beta_{1}$ is the first $N_{0}$ entries of $\beta$ with
$||\beta_{1}||\asymp N_{0}^{1/2}$. Additionally, $f_{t_h}\overset{i.i.d.}{\sim}N(0,1)$
and $\mathbf{e}_H=\rho \mathbf{z}_H,$ where $\rho$ is an $N\times T_H$ matrix with i.i.d.
$N(0,1)$ entries and $\mathbf{z}_H$ is a $T_H\times T_H$ matrix satisfying $||\mathbf{z}_H||\lesssim1.$
\begin{prop}
\label{prop:3}In the setup of (\ref{eq:A3}), suppose
that $N/(N_{0}T_H^{2})\rightarrow\delta\geq0$ and $||\beta||\rightarrow\infty$
and define $\mathfrak{M}:=T_H^{-1}(\Upsilon\underline{\mathbf{f}}_H)'\Upsilon\underline{\mathbf{f}}_H+\delta(\Upsilon \mathbf{z}_{H,1})'\Upsilon \mathbf{z}_{H,1},$
where $\mathbf{z}_{H,1}$ is the first $T_H-mh$ columns of $\mathbf{z}_H$. Then, if the two
leading eigenvalues of $\mathfrak{M}$ are distinct in the
sense that $(\lambda_{1}(\mathfrak{M})-\lambda_{2}(\mathfrak{M}))/\lambda_{1}(\mathfrak{M})\gtrsim_{\textrm{P}}1$,
the estimated factor $\underline{\hat{\mathbf{f}}}_{H,\textrm{sPCA}}$ satisfies
\[
\parallel\mathbb{P}_{\underline{\hat{\mathbf{f}}}'_{H,\textrm{sPCA}}}-\mathbb{P}_{\eta_{_{\textrm{sPCA}}}}\parallel\overset{\textrm{P}}{\longrightarrow}0,
\]
 where $\eta_{_{\textrm{sPCA}}}$ is the first eigenvector of $\mathfrak{M}$.
In the special case where $\mathbf{z}'_{H,1}\mathbf{z}_{H,1}=\mathbb{I}_{T_H-mh},$ it satisfies
that
\[
\parallel\mathbb{P}_{\underline{\hat{\mathbf{f}}}'_{H,\textrm{sPCA}}}-\mathbb{P}_{\underline{\mathbf{f}}'_H}\parallel\overset{\textrm{P}}{\longrightarrow}0.
\]
\end{prop}
Proposition \ref{prop:3} first presents that even when the number
of factors is correctly specified as $1$, the high-frequency factor estimated using
sPCA is generally inconsistent. This occurs because the eigenvector
$\eta_{\textrm{sPCA}}$ deviates from that of $T_H^{-1}(\Upsilon\underline{\mathbf{f}}_H)'\Upsilon\underline{\mathbf{f}}_H$,
as it is polluted by noise $\Upsilon \mathbf{z}_H$. It is worth noting that
\citet[cf. Proposition 2]{huang2022scaled} showed in their paper
that when $N/(N_{0}T_H^{2})\rightarrow0$ and $N/(\sqrt{N_{0}}T_H)\rightarrow\delta$,
the sPCA forecast outperforms the PCA forecast. For SsPCA, while it
may fail under extremely weak conditions, as described in \cite{onatski2012asymptotics},
it remains consistent in factor estimation as long as the informative
subset $N_{0}$ is correctly identified, which holds even when $N/(N_{0}T_H^{2})\rightarrow\delta$,
since $N_{0}/(N_{0}T_H^{2})=1/T_H^{2}\rightarrow0$. Although sPCA benefits
from scaling coefficients, it still applies them to all predictors,
making it less effective in certain weak factor settings. In contrast,
SsPCA retains only informative predictors, allowing it to remain consistent
even in weaker cases where sPCA may fail. 

In the specific scenario where the error term is homoskedastic and
lacks serial correlation, i.e., $\mathbf{z}'_{H,1}\mathbf{z}_{H,1}=\mathbb{I}_{T_H-mh}$, the
estimated factor becomes consistent, as the term $\delta(\Upsilon \mathbf{z}_{H,1})'\Upsilon \mathbf{z}_{H,1}$
in $\mathfrak{M}$ does not affect the eigenvectors of $T_H^{-1}(\Upsilon\underline{\mathbf{f}}_H)'\Upsilon\underline{\mathbf{f}}_H.$
This result is similar to Section 4 of \cite{bai2003inferential} which
establishes that factors can be estimated consistently by PCA under
homoskedasticity and serially independent errors, even when $T_H$ remains
fixed. That being said, although factors can be estimated consistently
in this specific case, the forecast of $y_{T+h}$ using sPCA remains
inconsistent.
\begin{prop}
\label{prop:4}
Under the same assumptions as in Proposition \ref{prop:3}, if we
further assume $\mathbf{z}'_{H,1}\mathbf{z}_{H,1}=\mathbb{I}_{T_H-mh},$ then we have $\hat{y}_{T+h}^{\textrm{sPCA}}\overset{\textrm{P}}{\longrightarrow}(1+\delta)^{-1}\mathbb{E}_{T}(y_{T+h}).$
\end{prop}
The inconsistency arises because, while $\hat{\mathbf{f}}_{H,\textrm{sPCA}}$,
which corresponds to the right singular vector of $\hat{\Upsilon}\underline{\mathbf{x}}_H$,
remains consistent in this special case, the left singular vector
$\hat{\varsigma}$ and the singular values are not. After MIDAS aggregation with consistently estimated weights, this singular-value bias propagates into the low-frequency forecast. This leads to biased predictions, highlighting the limitations of PC regressions without supervised selection in the presence of weak factor structure.

\subsection{Asymptotic Properties of Boosting with SsPCA Factors\label{Appendix A.3}}

\begin{prop}
\label{prop:5}Let $\tilde{\Phi}_{(M)}$ be the boosting estimate
for the conditional mean when all aggregated factors are observable,
and let $\hat{\Phi}_{(M)}$ be the boosting estimate when some or
all aggregated factors are estimated using the method of SsPCA-MIDAS.
Under the same assumptions as Theorem \ref{thm:1}, if boosting terminates
at step $M$ with $M(q^{-1/2}N^{-1/2}+T^{-1})\rightarrow0$, as $M,qN,T\rightarrow\infty,$
then $\frac{1}{T}\sum_{t=1}^{T}|\hat{\Phi}_{(M)}-\tilde{\Phi}_{(M)}|^{2}=o_{\textrm{P}}(1)$,
and for each given $\underline{\hat{F}}(\hat{\theta})$, $|\hat{\Phi}_{(M),\underline{\hat{F}}(\hat{\theta})}-\tilde{\Phi}_{(M),\underline{F}(\theta_f)}|=o_{\textrm{P}}(1)$. 
\end{prop}

\section{Mathematical Proofs\label{Appendix B}}

\renewcommand{\theequation}{B\arabic{equation}}
\setcounter{equation}{0} 

Our mathematical proofs build upon the asymptotic framework in \cite{giglio2023prediction},
incorporating insights on the scaling coefficient from \cite{huang2022scaled}
and drawing upon the theoretical foundations of boosting with PCA
factors from \cite{bai2009boosting}. Given the context of mixed-data
sampling, we further refine the analysis by integrating recent developments
in inference for factor-MIDAS from \cite{koh2023inference}.

We first establish the proof for the same-frequency case, which was
not provided in \cite{hounyo2023forecasting}, while the model setup
and methodological details can be found therein. In matrix form, the
model can be represented as
\[
\overline{\mathbf{y}}=\alpha\underline{\mathbf{f}}+\alpha_{w}\underline{\mathbf{w}}+\overline{\epsilon},
\]
\[
\Upsilon\mathbf{x}=\Upsilon\beta\mathbf{f}+\Upsilon\mathbf{e}.
\]

We then proceed to the proof for the mixed-frequency case, with the model setup detailed in the main text.

For notational simplicity, we drop overlines and underlines whenever no confusion arises. For example, we write $\mathbf{x}_{H},\mathbf{f}_{H},\mathbf{e}_{H},Y,\epsilon$  in place of $\underline{\mathbf{x}}_{H},\underline{\mathbf{f}}_{H},\underline{\mathbf{e}}_{H},\overline{Y}$, and $\overline{\epsilon}$, respectively. We similarly omit the subscript $H$ whenever the frequency is clear from context. We use $T_{h}$ as a shorthand for the forecast-truncated sample size: $T-h$ at low frequency and $T_H-mh$ at high frequency;
moreover, a subscript $T$ attached to a variable denotes the final
observation period. We use an overdot to distinguish some variables used in the mixed-frequency framework. When no confusion arises, we use $\theta$ to denote either $\theta_f$ or $\theta_x$, depending on the context. In addition, without loss of generality, we
assume that $\mathbf{\Sigma}_{f}=\mathbb{I}_{K}$ in the proof, in
that we can always normalize the factors by $\mathbf{\Sigma}_{f}^{-1/2}$
and redefine $\alpha$ in (\ref{equation 3}) and $\beta$ in (\ref{equation 5})
accordingly. Furthermore, since we assume that $\sum_{j=0}^{J}b_{j}(\theta_f)=\mathbb{I}_{K}$,
it follows that $\sum_{j=0}^{J}b_{j}(\theta)\mathbf{\Sigma}_{f}=\mathbf{\Sigma}_{F(\theta)}=\mathbb{I}_{K}$.
Therefore, the covariance structure of the aggregated factor $\mathbf{\Sigma}_{F(\theta)}$
remains equivalent to that of the original factor $\mathbf{\Sigma}_{f}$. Although the proofs are presented for a scalar target,
the argument underlying the recovery of the number of factors remains
valid for a $D$-dimensional target with $\alpha\in\mathbb{R}^{D\times K}$.

\subsection{Proof of Theorem 1}
\begin{proof}
We begin with the DGP, excluding $w_{t}$, in the same-frequency setting
first. Throughout the proof, we use $\tilde{\mathbf{x}}_{(k)}:=(\hat{\Upsilon}\mathbf{x}_{(k)})_{[\hat{I}_{k}]}$
to denote the matrix on which we perform SVD in each step of Algorithm
\ref{Algorithm 1}. The first left and right singular vectors of $\tilde{\mathbf{x}}_{(k)}$
are denoted by $\widehat{\varsigma}_{(k)}$ and $\hat{\xi}_{(k)}$,
while the largest singular value of $\tilde{\mathbf{x}}_{(k)}$ is denoted
by $\sqrt{T_{h}\hat{\lambda}_{(k)}}$. Consequently, $\hat{\lambda}_{(k)}=T_{h}^{-1}\parallel\tilde{\mathbf{x}}_{(k)}\parallel^{2}$
. Moreover, by definition
\begin{equation}
\widehat{\varsigma}_{(k)}=T_{h}^{-1/2}\hat{\lambda}_{(k)}^{-1/2}\tilde{\mathbf{x}}_{(k)}\hat{\xi}_{(k)},\quad\hat{\xi}_{(k)}=T_{h}^{-1/2}\hat{\lambda}_{(k)}^{-1/2}\tilde{\mathbf{x}}_{(k)}'\widehat{\varsigma}_{(k)}.\label{eq:B1}
\end{equation}
Therefore, our estimated factor at $k$-th step is $\hat{\mathbf{f}}_{(k)}=\widehat{\varsigma}'_{(k)}\tilde{\mathbf{x}}_{(k)}=T_{h}^{1/2}\hat{\lambda}_{(k)}^{1/2}\hat{\xi}'_{(k)}$. As a result, the coefficients obtained by regressing $\mathbf{x}$ and $Y$
onto this factor are, respectively:
\begin{equation}
(\hat{\Upsilon}\hat{\beta})_{(k)}=T_{h}^{-1/2}\hat{\lambda}_{(k)}^{-1/2}\hat{\Upsilon}\mathbf{x}_{(k)}\hat{\xi}_{(k)}\quad and\quad\hat{\alpha}_{(k)}=T_{h}^{-1/2}\hat{\lambda}_{(k)}^{-1/2}\mathbf{y}_{(k)}\hat{\xi}_{(k)}.\label{eq:B2}
\end{equation}
 Then we define $\tilde{D}_{(k)}\in\mathbb{R}^{qN\times N}$ iteratively
by
\[
\tilde{D}_{(k)}=(\mathbb{I}_{N})_{[\hat{I}_{k}]}-\stackrel[i=1]{k-1}{\sum}T_{h}^{-1/2}\hat{\lambda}_{(i)}^{-1/2}(\hat{\Upsilon}\mathbf{x})_{[\hat{I}_{k}]}\hat{\xi}_{(i)}\widehat{\varsigma}'_{(i)}\tilde{D}_{(i)},
\]
with $\tilde{D}_{(1)}=(\mathbb{I}_{N})_{[\hat{I}_{k}]}$. By induction,
we can show that $\tilde{\mathbf{x}}_{(k)}=\tilde{D}_{(k)}\hat{\Upsilon}\mathbf{x}$.
In fact, according to Lemma \ref{lem:1}, we have $\hat{\xi}'_{(i)}\hat{\xi}_{(j)}=0$
for $i\neq j\leq\hat{K}$, which indicates that $\hat{\mathbf{f}}_{(k)}$ are
pairwise orthogonal for all $k$. Using this property and the definition
of $\tilde{\mathbf{x}}_{(k)}$, we have
\begin{equation}
\tilde{\mathbf{x}}_{(k)}=(\hat{\Upsilon}\mathbf{x}_{(k)})_{[\hat{I}_{k}]}=(\hat{\Upsilon}\mathbf{x})_{[\hat{I}_{k}]}\stackrel[i=1]{k-1}{\prod}\mathbb{\mathbb{M}}_{\hat{\mathbf{f}}'_{(i)}}=(\hat{\Upsilon}\mathbf{x})_{[\hat{I}_{k}]}\left(\mathbb{I}_{T_{h}}-\stackrel[i=1]{k-1}{\sum}\hat{\xi}{}_{(i)}\hat{\xi}'_{(i)}\right),\label{eq:B3}
\end{equation}
for $k>1$ and when $k=1$,
\[
\tilde{\mathbf{x}}_{(1)}=(\hat{\Upsilon}\mathbf{x})_{[\hat{I}_{1}]}=(\hat{\Upsilon}\beta)_{[\hat{I}_{1}]}\mathbf{f}+(\hat{\Upsilon}\mathbf{e})_{[\hat{I}_{1}]}.
\]
Using (\ref{eq:B1}), if $\tilde{\mathbf{x}}_{(i)}=\tilde{D}_{(i)}\hat{\Upsilon}\mathbf{x}$
for any $i<k$, we can write (\ref{eq:B3}) as
\[
\tilde{\mathbf{x}}_{(k)}=(\hat{\Upsilon}\mathbf{x})_{[\hat{I}_{k}]}\left(\mathbb{I}_{T_{h}}-\stackrel[i=1]{k-1}{\sum}\hat{\xi}{}_{(i)}\hat{\xi}'_{(i)}\right)=(\hat{\Upsilon}\mathbf{x})_{[\hat{I}_{k}]}-\stackrel[i=1]{k-1}{\sum}T_{h}^{-1/2}\hat{\lambda}_{(i)}^{-1/2}(\hat{\Upsilon}\mathbf{x})_{[\hat{I}_{k}]}\hat{\xi}{}_{(i)}\widehat{\varsigma}'_{(i)}\tilde{\mathbf{x}}_{(i)}=\tilde{D}_{(k)}\hat{\Upsilon}\mathbf{x}.
\]
By definition, $\tilde{\mathbf{x}}_{(1)}=(\hat{\Upsilon}\mathbf{x})_{[\hat{I}_{1}]}=\tilde{D}_{(1)}\hat{\Upsilon}\mathbf{x}$
is evident. Using induction, this leads to $\tilde{\mathbf{x}}_{(k)}=\tilde{D}_{(k)}\hat{\Upsilon}\mathbf{x}$.
Consequently, the estimated factors satisfy
\begin{equation}
\hat{\mathbf{f}}_{(k)}=\widehat{\varsigma}'_{(k)}\tilde{\mathbf{x}}_{(k)}=\widehat{\varsigma}'_{(k)}\tilde{D}_{(k)}\hat{\Upsilon}\mathbf{x},
\end{equation}
for all $k$, and by definition, we have $\hat{\zeta}_{(k)}=(\widehat{\varsigma}'_{(k)}\tilde{D}_{(k)})'.$
Additionally, using (\ref{eq:B2}) the estimated coefficient $\hat{\gamma}$
can be written as
\begin{equation}
\hat{\gamma}=\stackrel[k=1]{\hat{K}}{\sum}\hat{\alpha}_{(k)}\hat{\zeta}'_{(k)}=\stackrel[k=1]{\hat{K}}{\sum}\hat{\alpha}_{(k)}\widehat{\varsigma}'_{(k)}\tilde{D}_{(k)}=\stackrel[k=1]{\hat{K}}{\sum}T_{h}^{-1/2}\hat{\lambda}_{(k)}^{-1/2}\mathbf{y}\hat{\xi}_{(k)}\widehat{\varsigma}'_{(k)}\tilde{D}_{(k)}.\label{eq:B5}
\end{equation}
We proceed to define $\tilde{\beta}_{(k)}=\tilde{D}_{(k)}\hat{\Upsilon}\beta$
and $\tilde{\mathbf{e}}_{(k)}=\tilde{D}_{(k)}\hat{\Upsilon}\mathbf{e}$, then $\tilde{\mathbf{x}}_{(k)}$
can be written in the form of
\begin{equation}
\tilde{\mathbf{x}}_{(k)}=\tilde{\beta}_{(k)}\mathbf{f}+\tilde{\mathbf{e}}_{(k)}.\label{eq:B6}
\end{equation}
We also define the population analog of $\tilde{D}_{(k)}$ for each
$k$ by
\[
D_{(k)}=(\mathbb{I}_{N})_{[I_{k}]}-\stackrel[i=1]{k-1}{\sum}\lambda_{(i)}^{-1/2}(\hat{\Upsilon}\beta)_{[I_{k}]}\iota{}_{(i)}\varsigma'_{(i)}D_{(i)},\quad D_{(1)}=(\mathbb{I}_{N})_{[I_{1}]},
\]
where the leading singular value of $(\hat{\Upsilon}\beta)_{(k)}$
is $\lambda_{(k)}^{1/2}$, and the corresponding left and right singular
vectors are $\varsigma{}_{(k)}$ and $\iota{}_{(k)}$, respectively.
Using a similar induction reasoning, it can be shown that
\[
(\hat{\Upsilon}\beta)_{(k)}=(\hat{\Upsilon}\beta)_{[I_{k}]}\stackrel[i=1]{k-1}{\prod}\mathbb{\mathbb{M}}_{\iota_{(i)}}=D_{(k)}\hat{\Upsilon}\beta.
\]
Intuitively, $\tilde{\beta}_{(k)}$ and $\tilde{D}_{(k)}$ are sample
analogs of $(\hat{\Upsilon}\beta)_{(k)}$ and $D_{(k)}$.

Similar representations to (\ref{eq:B6}) can be constructed for $\mathbf{y}_{(k)}:=\mathbf{y}\prod_{i=1}^{k-1}\mathbb{\mathbb{M}}_{\hat{\mathbf{f}}'_{(i)}}$
for each $k$. Specifically, we have
\begin{equation}
\mathbf{y}_{(k)}:=\mathbf{y}\left(\mathbb{I}_{T_{h}}-\stackrel[i=1]{k-1}{\sum}\hat{\xi}{}_{(i)}\hat{\xi}'_{(i)}\right)=\tilde{\alpha}_{(k)}\mathbf{f}+\tilde{\epsilon}_{(k)},\label{eq:B7}
\end{equation}
where $\tilde{\alpha}_{(k)}\in\mathbb{R}^{1\times K}$ and $\tilde{\beta}_{(k)}\in\mathbb{R}^{|\hat{I}_{k}|\times K}$
are defined as 
\[
\tilde{\alpha}_{(k)}:=\alpha-\stackrel[i=1]{k-1}{\sum}T_{h}^{-1/2}\hat{\lambda}_{(i)}^{-1/2}\mathbf{y}\hat{\xi}{}_{(i)}\widehat{\varsigma}'_{(i)}\tilde{\beta}{}_{(i)}\quad and\quad\tilde{\epsilon}_{(k)}:=\epsilon-\stackrel[i=1]{k-1}{\sum}T_{h}^{-1/2}\hat{\lambda}_{(i)}^{-1/2}\mathbf{y}\hat{\xi}{}_{(i)}\widehat{\varsigma}'_{(i)}\tilde{\mathbf{e}}{}_{(i)}.
\]
By Lemma \ref{lem:2}, we have $P(\widehat{I}_{k}=I_{k})\rightarrow1$
for $k\leq\widetilde{K}$ and $P(\hat{K}=\widetilde{K})\rightarrow1$.
Therefore, with a probability approaching one, we can assume $\widehat{I}_{k}=I_{k}$
for any $k$, and $\hat{K}=\widetilde{K}$ in the subsequent analysis.

To prove Theorem \ref{thm:1} (if it is same-frequency), using (\ref{eq:B6}), the estimated factors can be written as
\[
\hat{\mathbf{f}}_{(k)}=\hat{\varsigma}'\tilde{\mathbf{x}}_{(k)}=\hat{\varsigma}_{(k)}'\tilde{\beta}_{(k)}\mathbf{f}+\hat{\varsigma}_{(k)}'\tilde{\mathbf{e}}_{(k)}.
\]
Using Lemma \ref{lem:5}(i), $||\hat{\mathbf{f}}_{(k)}||=\sqrt{T_{h}\hat{\lambda}_{(k)}}$,
and $||\mathbb{M}_{\mathbf{f}'}||\leq1$, we can get
\[
||\hat{\mathbf{f}}_{(k)}||^{-1}||\hat{\mathbf{f}}_{(k)}\mathbb{M}_{\mathbf{f}'}||\leq||\hat{\mathbf{f}}_{(k)}||^{-1}||\hat{\varsigma}_{(k)}'\tilde{\mathbf{e}}_{(k)}\mathbb{M}_{\mathbf{f}'}||\leq||\hat{\mathbf{f}}_{(k)}||^{-1}||\hat{\varsigma}_{(k)}'\tilde{\mathbf{e}}_{(k)}||\lesssim_{\textrm{P}}q^{-1/2}N^{-1/2}+T^{-1}.
\]
Next, we turn to the mixed-frequency setting, where it is necessary
to establish the consistency of the aggregated factor estimator $\hat{F}_{(k)}(\hat{\theta})$.
The proof is similar to that of the same-frequency case since the MIDAS aggregation step only applies a finite weighted average (fixed $J$) to the estimated high-frequency factors.

Let $\tilde{\mathbf{x}}_{H,(k)}$ represents the high-frequency scaled predictors,
the estimated aggregated factor at the $k$-th step is $\hat{F}_{(k)}(\hat{\theta})=\sum_{j=0}^{J}b_{j,k}(\hat{\theta})\hat{\mathbf{f}}_{t-j/m,(k)}$, where $\hat{\mathbf{f}}_{t-j/m,(k)}=\widehat{\varsigma}'_{(k)}\tilde{\mathbf{x}}_{t-j/m,(k)}=T_{h}^{1/2}\hat{\lambda}_{(k)}^{1/2}\hat{\xi}'_{t-j/m,(k)}.$
We previously distinguished between $\theta_f$ and $\theta_x$  due to differences in cross-sectional
dimension. Given this, the derivation related to the high-frequency sector remains
unchanged, and the coefficient from regressing $Y$ onto the aggregated
factor is:
\[
\dot{\hat{\alpha}}_{(k)}=T_{h}^{-1/2}\hat{\lambda}_{(k)}^{-1/2}Y_{(k)}\hat{\xi}_{(k)}(\hat{\theta}).
\]
Then, the estimated coefficient $\dot{\hat{\gamma}}$ can be written
as
\[
\dot{\hat{\gamma}}=\stackrel[k=1]{\hat{K}}{\sum}\dot{\hat{\alpha}}_{(k)}\widehat{\varsigma}'_{(k)}\tilde{D}_{(k)}=\stackrel[k=1]{\hat{K}}{\sum}T_{h}^{-1/2}\hat{\lambda}_{(k)}^{-1/2}Y\hat{\xi}_{(k)}(\hat{\theta})\widehat{\varsigma}'_{(k)}\tilde{D}_{(k)}.
\]
We continue to construct the low-frequency target $Y_{(k)}$. Based
on (\ref{eq:B7}), we have
\[
Y_{(k)}:=Y\left(\mathbb{I}_{T_{h}}-\stackrel[i=1]{k-1}{\sum}\hat{\xi}{}_{(i)}(\hat{\theta})\hat{\xi}'_{(i)}(\hat{\theta})\right)=\dot{\tilde{\alpha}}_{(k)}F(\theta)+\dot{\tilde{\epsilon}}_{(k)},
\]
where we retain $\hat{\xi}{}_{(i)}(\hat{\theta})\hat{\xi}'_{(i)}(\hat{\theta})$ as the projection
basis. Therefore, we can define
\[
\dot{\tilde{\alpha}}_{(k)}:=\alpha-\stackrel[i=1]{k-1}{\sum}T_{h}^{-1/2}\hat{\lambda}_{(i)}^{-1/2}Y\hat{\xi}_{(i)}(\hat{\theta})\widehat{\varsigma}'_{(i)}\tilde{\beta}{}_{(i)}\quad and\quad\dot{\tilde{\epsilon}}_{(k)}:=\epsilon-\stackrel[i=1]{k-1}{\sum}T_{h}^{-1/2}\hat{\lambda}_{(i)}^{-1/2}Y\hat{\xi}_{(i)}(\hat{\theta})\widehat{\varsigma}'_{(i)}\tilde{E}{}_{(i)}.
\]
In the mixed-frequency case, the estimated aggregated factor is obtained by applying the MIDAS weighting scheme to finitely many lagged high-frequency estimated factors. Since $J$ is fixed and the MIDAS weights satisfy $\sum_{j=0}^{J}b_{j}(\theta_f)=\mathbb{I}_{K}$, the aggregation step does not alter the leading term. Moreover, by Lemma \ref{lem:16}(iii), $\qquad\hat{\theta}\overset{\textrm{p}}{\longrightarrow}\theta$, and by continuity of the weighting function, $\max_{0\leq j\leq J}|b_{j,k}(\hat{\theta})-b_{j,k}(\theta)|=o_{\textrm{P}}(1)$. Therefore, the selection consistency result in Lemma \ref{lem:2} continues to hold after MIDAS aggregation, that is, $P(\widehat{I}_{k}=I_{k})\rightarrow1$
for $k\leq\widetilde{K}$ and $P(\hat{K}=\widetilde{K})\rightarrow1$.

We are now ready to prove the main result in Theorem \ref{thm:1},
using the estimated aggregated factor, which can be written as:
\[
\begin{split}\hat{F}_{(k)}(\hat{\theta})&=\sum_{j=0}^{J}b_{j,k}(\hat{\theta})\hat{\mathbf{f}}_{t-j/m,(k)}=\sum_{j=0}^{J}b_{j,k}(\hat{\theta})\widehat{\varsigma}'_{(k)}\tilde{\mathbf{x}}_{t-j/m(k)}\\&=\sum_{j=0}^{J}b_{j,k}(\hat{\theta})\hat{\varsigma}_{(k)}'\tilde{\beta}_{(k)}\mathbf{f}_{t-j/m}+\sum_{j=0}^{J}b_{j,k}(\hat{\theta})\hat{\varsigma}_{(k)}'\tilde{\mathbf{e}}_{t-j/m,(k)}.\end{split}
\]
Using Lemma \ref{lem:5}(i), $||\hat{\mathbf{f}}_{H,(k)}||=\sqrt{T_{h}\hat{\lambda}_{(k)}}$,
$||\mathbb{M}_{\mathbf{f_H}'}||\leq1$, Assumption \ref{assu:2}, $\sum_{j=0}^{J}b_{j,k}(\theta)=1$, Lemma \ref{lem:16}, we have $||\mathbb{M}_{F(\theta)'}||\leq1$. Then we can deduce
\[
||\hat{F}_{(k)}(\hat{\theta})||^{-1}||\hat{F}_{(k)}(\hat{\theta})\mathbb{M}_{F(\theta)'}||\leq||\hat{F}_{(k)}(\hat{\theta})||^{-1}\,||\sum_{j=0}^{J}b_{j,k}(\hat{\theta})\hat{\varsigma}_{(k)}'\tilde{\mathbf{e}}_{t-j/m,(k)}\mathbb{M}_{F(\theta)'}||
\]
\[
\leq||\hat{F}_{(k)}(\hat{\theta})||^{-1}\,\sum_{j=0}^{J}|b_{j,k}(\hat{\theta})|\,||\hat{\varsigma}_{(k)}'\tilde{\mathbf{e}}_{t-j/m,(k)}||\lesssim_{\textrm{P}}q^{-1/2}N^{-1/2}+T^{-1}.
\]
\end{proof}

\subsection{Proof of Theorem 2}
\begin{proof}
As before, we begin by establishing the consistency of prediction
under the same-frequency setting. By definition of $\mathbf{x}_{(k),\textrm{scaled}}$
in Algorithm \ref{Algorithm 1}, we have
\[
\mathbf{x}_{(k),\textrm{scaled}}=\hat{\Upsilon}\mathbf{x}_{(k)}=\hat{\Upsilon}\mathbf{x}_{(k-1)}\mathbb{M}_{\hat{\mathbf{f}}'_{(k-1)}}=\hat{\Upsilon}\mathbf{x}\stackrel[i=1]{k-1}{\prod}\mathbb{M}_{\hat{\mathbf{f}}'_{(i)}}=\hat{\Upsilon}\mathbf{x}\left(\mathbb{I}_{T_{h}}-\stackrel[i=1]{k-1}{\sum}\hat{\xi}_{(i)}\hat{\xi}'_{(i)}\right).
\]
Thus, using (\ref{eq:B7}), we have
\[
\hat{\Upsilon}\mathbf{x}_{(k)}\mathbf{y}'_{(k)}=\hat{\Upsilon}\mathbf{x}\left(\mathbb{I}_{T_{h}}-\stackrel[i=1]{k-1}{\sum}\hat{\xi}_{(i)}\hat{\xi}'_{(i)}\right)\mathbf{y}'_{(k)}=\hat{\Upsilon}\mathbf{xy'}_{(k)}.
\]
Here, we know $\mathbf{y}_{(k)}\hat{\xi}_{(i)}=0$ for $i<k$ by Lemma \ref{lem:1}.
Therefore, the covariance $(\hat{\Upsilon}\mathbf{x}_{(k)})_{[i]}\mathbf{y}'_{(k)}$
for each predictor equals to $(\hat{\Upsilon}\mathbf{x})_{[i]}\mathbf{y}'_{(k)}$. According
to the stopping rule, if our algorithm stops at $\tilde{K}$, there
are at most $qN-1$ predictors among all satisfying $T_{h}^{-1}||(\hat{\Upsilon}\mathbf{x})_{[i]}\mathbf{y}'_{(\tilde{K}+1)}||_{\max}\geq c.$
Let $S$ denote the set of these predictors. For $i\in S^{c},$ we
have 
\begin{equation}
||T_{h}^{-1}(\hat{\Upsilon}\mathbf{x})_{[i]}\mathbf{y}'_{(\tilde{K}+1)}||_{F}^{2}\lesssim||T_{h}^{-1}\hat{\Upsilon}\mathbf{x}\mathbf{y}'_{(\tilde{K}+1)}||_{\max}^{2}\lesssim_{\textrm{P}}1,\label{eq:B8}
\end{equation}
where we use $||\beta||_{\max}\lesssim1$ from Assumption \ref{assu:2}
and Lemma \ref{lem:2}(vi) in the last step. On the other hand, based
on the set $I_{0}$ in Assumption \ref{assu:2}, we have
\[
\underset{i\in I_{0}}{\sum}||T_{h}^{-1}(\hat{\Upsilon}\mathbf{x})_{[i]}\mathbf{y}'_{(\tilde{K}+1)}||_\textrm{F}^{2}=\underset{i\in I_{0}\cap S}{\sum}||T_{h}^{-1}(\hat{\Upsilon}\mathbf{x})_{[i]}\mathbf{y}'_{(\tilde{K}+1)}||_\textrm{F}^{2}+\underset{i\in I_{0}\cap S^{c}}{\sum}||T_{h}^{-1}(\hat{\Upsilon}\mathbf{x})_{[i]}\mathbf{y}'_{(\tilde{K}+1)}||_\textrm{F}^{2}
\]
\begin{equation}
\lesssim_{\textrm{P}}|I_{0}\cap S|+|I_{0}\cap S^{c}|c^{2}=o(N_{0}),\label{eq:B9}
\end{equation}
where we use (\ref{eq:B8}), $|S|\leq qN-1,c\rightarrow0,$ and $qN/N_{0}\rightarrow0.$
Therefore, (\ref{eq:B9}) leads to $||\mathbf{y}_{(\tilde{K}+1)}(\hat{\Upsilon}\mathbf{x})'{}_{[I_{0}]}||=o_{\textrm{P}}(N_{0}).$
Additionally, using (\ref{eq:B7}) and that $\hat{\Upsilon}\mathbf{x}=\hat{\Upsilon}\beta \mathbf{f}+\hat{\Upsilon}\mathbf{e}$,
we can decompose
\begin{equation}
\mathbf{y}_{(\tilde{K}+1)}(\hat{\Upsilon}\mathbf{x})'{}_{[I_{0}]}=\tilde{\alpha}_{(\tilde{K}+1)}\mathbf{ff}'(\hat{\Upsilon}\beta)'_{[I_{o}]}+\tilde{\alpha}_{(\tilde{K}+1)}\mathbf{f}(\hat{\Upsilon}\mathbf{e})'_{[I_{o}]}+\tilde{\epsilon}_{(\tilde{K}+1)}\mathbf{f}'(\hat{\Upsilon}\beta)'_{[I_{o}]}+\tilde{\epsilon}_{(\tilde{K}+1)}(\hat{\Upsilon}\mathbf{e})'_{[I_{o}]}.\label{eq:B10}
\end{equation}
Using (\ref{eq:B9}), (\ref{eq:B10}), Lemma \ref{lem:7}, Lemma \ref{lem:16}(i),
and the fact that $||\beta_{[I_{0}]}||\lesssim N_{0}^{1/2}$, we have
\begin{equation}
||\tilde{\alpha}_{(\tilde{K}+1)}\left(\mathbf{ff}'(\hat{\Upsilon}\beta)'_{[I_{o}]}+\mathbf{f}(\hat{\Upsilon}\mathbf{e})'_{[I_{o}]}\right)||=o_{\textrm{P}}(N_{0}^{1/2}T)+O_{\textrm{P}}(T^{-1/2})o_{\textrm{P}}(N_{0}^{1/2}T)=o_{P}(N_{0}^{1/2}T).\label{eq:B11}
\end{equation}
Moreover, using Assumption \ref{assu:4}(i), Assumption \ref{assu:1}
and Weyl's theorem, we derive
\[
|\sigma_{K}\left(\mathbf{ff}'(\hat{\Upsilon}\beta)'_{[I_{o}]}+\mathbf{f}(\hat{\Upsilon}\mathbf{e})'_{[I_{o}]}\right)-\sigma_{K}\left(T_{h}(\hat{\Upsilon}\beta){}_{[I_{o}]}\right)|\leq||\mathbf{f}(\hat{\Upsilon}\mathbf{e})'_{[I_{o}]}||+||T_{h}^{-1}\mathbf{ff}'-\mathbb{I}_{K}||\,||T_{h}(\hat{\Upsilon}\beta){}_{[I_{o}]}||
\]
\begin{equation}
\leq||\hat{\Upsilon}_{[I_{0}]}||\,||\mathbf{f}\mathbf{e}'_{[I_{o}]}||+||T_{h}^{-1}\mathbf{ff}'-\mathbb{I}_{K}||\,||T_{h}\beta{}_{[I_{o}]}||\,||\hat{\Upsilon}_{[I_{0}]}||\lesssim_{\textrm{P}}N_{0}^{1/2}T^{1/2}.\label{eq:B12}
\end{equation}
Because Assumption \ref{assu:2} implies that $\sigma_{K}(\beta_{[I_{0}]})\asymp N_{0}^{1/2}$
and $\sigma_{K}((\hat{\Upsilon}\beta)_{[I_{0}]})\asymp N_{0}^{1/2}$,
we have $\sigma_{K}\left(\mathbf{ff}'(\hat{\Upsilon}\beta)'_{[I_{o}]}+\mathbf{f}(\hat{\Upsilon}\mathbf{e})'_{[I_{o}]}\right)\asymp N_{0}^{1/2}T^{1/2}.$
Using this result, (\ref{eq:B11}) and the inequality $||\tilde{\alpha}_{(\tilde{K}+1)}\left(\mathbf{ff}'(\hat{\Upsilon}\beta)'_{[I_{o}]}+\mathbf{f}(\hat{\Upsilon}\mathbf{e})'_{[I_{o}]}\right)||\geq\sigma_{K}(\mathbf{ff}'(\hat{\Upsilon}\beta)'_{[I_{o}]}$
$+\mathbf{f}(\hat{\Upsilon}\mathbf{e})'_{[I_{o}]})||\tilde{\alpha}_{(\tilde{K}+1)}||$, we can get $||\tilde{\alpha}_{(\tilde{K}+1)}||\overset{P}{\rightarrow}0.$
That is, by definition of $\tilde{\alpha}_{(\tilde{K}+1)}$ in (\ref{eq:B7}),
\begin{equation}
\parallel\alpha-\stackrel[i=1]{\tilde{K}}{\sum}\mathbf{y}\tilde{\xi}_{(i)}\frac{\hat{\varsigma}'_{(i)}\tilde{\beta}_{(i)}}{\sqrt{T_{h}\hat{\lambda}_{(i)}}}\parallel=o_{\textrm{P}}(1).\label{eq:B13}
\end{equation}
Then, (\ref{eq:B5}) and $\tilde{\beta}_{(k)}=\tilde{D}_{(k)}\hat{\Upsilon}\beta$
imply that 
\[
\hat{\gamma}\hat{\Upsilon}\beta=\stackrel[i=1]{\tilde{K}}{\sum}T_{h}^{-1/2}\hat{\lambda}_{(i)}^{-1/2}\mathbf{y}\hat{\xi}_{(i)}\hat{\varsigma}'_{(i)}\tilde{\beta}_{(i)}.
\]
Thus, (\ref{eq:B13}) is equivalent to $||\hat{\gamma}\hat{\Upsilon}\beta-\alpha||=o_{\textrm{P}}(1)$.

As illustrated in Lemma \ref{lem:8}, Assumptions \ref{assu:1}, \ref{assu:3},
and \ref{assu:4} hold when $\mathbf{f}$, $\epsilon$, and $\mathbf{e}$ are replaced
by $\mathbf{f}\mathbb{M}_{\mathbf{w}'}$, $\epsilon\mathbb{M}_{\mathbf{w}'}$, $\mathbf{e}\mathbb{M}_{\mathbf{w}'}$,
respectively. Consequently, all the lemmas and the result $||\hat{\gamma}\hat{\Upsilon}\beta-\alpha||=o_{\textrm{P}}(1)$,
even under the specification of \cite{giglio2023prediction} where
$w_{t}$ is incorporated into $x_{t}$. The prediction error for $y_{T+h}$
is then written as
\[
\hat{y}_{T+h}-\mathbb{E}_{T}(y_{T+h})=\hat{\gamma}\hat{\Upsilon}x_{T}+\hat{\alpha}_{w}w_{T}-\alpha f_{T}-\alpha_{w}w_{T}
\]
\begin{equation}
=(\hat{\gamma}\hat{\Upsilon}\beta-\alpha)\left(f_{T}-\mathbf{fw}'(\mathbf{ww}')^{-1}w_{T}\right)+\hat{\gamma}\hat{\Upsilon}(e_{T}-\mathbf{ew}'(\mathbf{ww}')^{-1}w_{T})+\epsilon \mathbf{w}'(\mathbf{ww}')^{-1}w_{T}.\label{eq:B14}
\end{equation}

Using (\ref{eq:B5}) and $||\mathbf{y}||\leq||\alpha \mathbf{f}||+||\epsilon||\lesssim_{\textrm{P}}T^{1/2}$
by Assumption \ref{assu:1}, we have
\begin{equation}
||\hat{\gamma}\hat{\Upsilon}e_{T}||\leq\underset{k\leq\tilde{K}}{\sum}T_{h}^{-1/2}\hat{\lambda}_{(k)}^{-1/2}||\mathbf{y}||\,||\hat{\xi}_{(k)}||\,||\hat{\varsigma}'{}_{(k)}\tilde{D}_{(k)}\hat{\Upsilon}e_{T}||\lesssim_{\textrm{P}}\underset{k\leq\tilde{K}}{\sum}\hat{\lambda}_{(k)}^{-1/2}||\hat{\varsigma}'{}_{(k)}\tilde{D}_{(k)}\hat{\Upsilon}e_{T}||,
\end{equation}
and
\begin{equation}
T_{h}^{-1}||\hat{\gamma}\hat{\Upsilon}\mathbf{ew}'||\leq\underset{k\leq\tilde{K}}{\sum}T_{h}^{-3/2}\hat{\lambda}_{(k)}^{-1/2}||\mathbf{y}||\,||\hat{\xi}_{(k)}||\,||\hat{\varsigma}'{}_{(k)}\tilde{D}_{(k)}\hat{\Upsilon}\mathbf{ew}'||\lesssim_{\textrm{P}}\underset{k\leq\tilde{K}}{\sum}T_{h}^{-1}\hat{\lambda}_{(k)}^{-1/2}||\hat{\varsigma}'{}_{(k)}\tilde{\mathbf{e}}_{(k)}\mathbf{w}'||.\label{eq:B16}
\end{equation}
Using $\hat{\lambda}_{(k)}\asymp_{P}qN$ from Lemma \ref{lem:2},
and Lemma \ref{lem:5}(ii) (iv), we have
\begin{equation}
T_{h}^{-1}\hat{\lambda}_{(k)}^{-1/2}||\hat{\varsigma}{}_{(k)}\tilde{\mathbf{e}}_{(k)}\mathbf{w}'||\lesssim_{\textrm{P}}q^{-1}N^{-1}+T^{-1},\;\hat{\lambda}_{(k)}^{-1/2}||\hat{\varsigma}'{}_{(k)}\tilde{D}_{(k)}\hat{\Upsilon}e_{T}||\lesssim_{\textrm{P}}q^{-1/2}N^{-1/2}+T^{-1/2}.\label{eq:B17}
\end{equation}
Thus, $||\hat{\gamma}\hat{\Upsilon}e_{T}||=o_{\textrm{P}}(1)$. Moreover,
with $||(\mathbf{ww}')^{-1}||\lesssim_{\textrm{P}}T^{-1}$ from Assumption
\ref{assu:1}, we have $||\hat{\gamma}\hat{\Upsilon}\mathbf{ew}'(\mathbf{ww})'^{-1}||=o_{\textrm{P}}(1).$
Combine with $||\mathbf{fw}'||\lesssim_{\textrm{P}}T^{1/2}$, $||\epsilon \mathbf{w}'||\lesssim_{\textrm{P}}T^{1/2}$
from Assumption \ref{assu:1} and $||\hat{\gamma}\hat{\Upsilon}\beta-\alpha||=o_{\textrm{P}}(1),$
we show that each term of (\ref{eq:B14}) vanishes, and hence $\hat{y}_{T+h}-\mathbb{E}_{T}(y_{T+h})\overset{\textrm{P}}{\longrightarrow}0.$

We now outline the prediction consistency in the mixed-frequency case,
emphasizing the differences. By using the definition of $\mathbf{x}_{(k),\textrm{scaled}}$,
weighting function $b_{j}(\hat{\theta})$, the low-frequency
target $Y_{(k)}$, we have
\[
(\hat{\Upsilon}X_{(k)}(\hat{\theta}))Y'_{(k)}=\hat{\Upsilon}X(\hat{\theta})\left(\mathbb{I}_{T_{h}}-\stackrel[i=1]{k-1}{\sum}\hat{\xi}_{(i)}(\hat{\theta})\hat{\xi}'_{(i)}(\hat{\theta})\right)Y'_{(k)}=\hat{\Upsilon}X(\hat{\theta})Y'_{(k)}.
\]
Therefore, for each predictor, we need to consider the covariance
structure formed by the aggregated scaled predictors and the target,
given by $(\hat{\Upsilon}X_{(k)}(\hat{\theta}))_{[i]}Y'_{(k)}$.
The resulting covariance for each predictor is therefore $(\hat{\Upsilon}X(\hat{\theta}))_{[i]}Y'_{(k)}$.
Accordingly, the stopping rule now states that there can be at most
$qN-1$ predictors for which $T_{h}^{-1}||(\hat{\Upsilon}X(\hat{\theta}))_{[i]}Y'_{(\tilde{K}+1)}||_{\max}\geq c$. 

Since $\sum_{j=0}^{J}b_{j}(\theta_{x})=\mathbb{I}_{N}$ and $||X(\hat{\theta})||=||\sum_{j=0}^{J}b_{j}(\hat{\theta})\mathbf{x}_{t-j/m}||\leq\sum_{j=0}^{J}||b_{j}(\hat{\theta})||\,||\mathbf{x}_{t-j/m}||$,
equations (\ref{eq:B8}), (\ref{eq:B9}) still hold under the replacement
of the scaled predictors with their aggregated counterparts. Consequently,
we have $||Y_{(\tilde{K}+1)}(\hat{\Upsilon}X(\hat{\theta}))'{}_{[I_{0}]}||=o_{\textrm{P}}(N_{0})$
and 
\[
Y_{(\tilde{K}+1)}(\hat{\Upsilon}X(\hat{\theta}))'{}_{[I_{0}]}=\dot{\tilde{\alpha}}_{(\tilde{K}+1)}F(\theta)F(\hat{\theta})'(\hat{\Upsilon}\beta)'_{[I_{o}]}
\]
\[
+\dot{\tilde{\alpha}}_{(\tilde{K}+1)}F(\theta)(\hat{\Upsilon}E(\hat{\theta}))'_{[I_{o}]}+\dot{\tilde{\epsilon}}_{(\tilde{K}+1)}F(\hat{\theta})'(\hat{\Upsilon}\beta)'_{[I_{o}]}
\]
\[
+\dot{\tilde{\epsilon}}_{(\tilde{K}+1)}(\hat{\Upsilon}E(\hat{\theta}))'_{[I_{o}]}.
\]
By analogy with (\ref{eq:B11}), we can clearly show that
\[
||\dot{\tilde{\alpha}}_{(\tilde{K}+1)}\left(F(\theta)F(\hat{\theta})'(\hat{\Upsilon}\beta)'_{[I_{o}]}+F(\theta)(\hat{\Upsilon}E(\hat{\theta}))'_{[I_{o}]}\right)||=o_{\textrm{P}}(N_{0}^{1/2}T),
\]
so this bound remains unchanged. Additionally, using Assumption \ref{assu:1}
and the property of $\sum_{j=0}^{J}b_{j}(\theta_f)$, we have $||T_{h}^{-1}F(\theta)F(\theta)'-\mathbb{I}_{K}||\lesssim_{\textrm{P}}T^{-1/2}$.
Building on the same way as in (\ref{eq:B12}), we can then derive
\[
|\sigma_{K}\left(F(\theta)F(\hat{\theta})'(\hat{\Upsilon}\beta)'_{[I_{o}]}+F(\theta)(\hat{\Upsilon}E(\hat{\theta}))'_{[I_{o}]}\right)-\sigma_{K}\left(T_{h}(\hat{\Upsilon}\beta){}_{[I_{o}]}\right)|
\]
\[
\leq||F(\theta)(\hat{\Upsilon}E(\hat{\theta}))'_{[I_{o}]}||+||T_{h}^{-1}F(\theta)F(\hat{\theta})'-\mathbb{I}_{K}||\,||T_{h}(\hat{\Upsilon}\beta){}_{[I_{o}]}||
\]
\[
\leq||\hat{\Upsilon}_{[I_{0}]}||\,||F(\theta)E_{[I_{o}]}(\hat{\theta})'||+||T_{h}^{-1}F(\theta)F(\hat{\theta})'-\mathbb{I}_{K}||\,||T_{h}\beta{}_{[I_{o}]}||\,||\hat{\Upsilon}_{[I_{0}]}||\lesssim_{\textrm{P}}N_{0}^{1/2}T^{1/2},
\]
where $||T_{h}^{-1}F(\theta)F(\hat{\theta})'-\mathbb{I}_{K}||\lesssim_{\textrm{P}}T^{-1/2}$
since we can decompose the term as $T_{h}^{-1}F(\theta)F(\hat{\theta})'-\mathbb{I}_{K}=(T_{h}^{-1}F(\theta)F(\theta)'-\mathbb{I}_{K})+T_{h}^{-1}F(\theta)(F(\hat{\theta})-F(\theta))'$,
where the second term is $o_{\textrm{P}}(1)$ since $\hat{\theta}-\theta=o_{\textrm{P}}(1)$ by Lemma \ref{lem:16}(iii), and hence $F(\hat{\theta})-F(\theta)=o_{\textrm{P}}(1)$.

Given $\sigma_{K}(\beta_{[I_{0}]})\asymp N_{0}^{1/2}$,
\[
\sigma_{K}\left(F(\theta)F(\hat{\theta})'(\hat{\Upsilon}\beta)'_{[I_{o}]}+F(\theta)(\hat{\Upsilon}E(\hat{\theta}))'_{[I_{o}]}\right)\asymp N_{0}^{1/2}T^{1/2}
\]
still holds. Then using this and previous results, and the inequality
\[
\begin{split}&||\dot{\tilde{\alpha}}_{(\tilde{K}+1)}\left(F(\theta)F(\hat{\theta})'(\hat{\Upsilon}\beta)'_{[I_{o}]}+F(\theta)(\hat{\Upsilon}E(\hat{\theta}))'_{[I_{o}]}\right)||\\&\geq\sigma_{K}\left(F(\theta)F(\hat{\theta})'(\hat{\Upsilon}\beta)'_{[I_{o}]}+F(\theta)(\hat{\Upsilon}E(\hat{\theta}))'_{[I_{o}]}\right),\end{split}
\]
we can also have $||\dot{\tilde{\alpha}}_{(\tilde{K}+1)}||\overset{\textrm{P}}{\longrightarrow}0$,
which means
\[
\parallel\alpha-\stackrel[i=1]{\tilde{K}}{\sum}Y\tilde{\xi}_{(i)}(\hat{\theta})\frac{\hat{\varsigma}'_{(i)}\tilde{\beta}_{(i)}}{\sqrt{T_{h}\hat{\lambda}_{(i)}}}\parallel=o_{\textrm{P}}(1).
\]
Using the coefficient $\dot{\hat{\gamma}}$ and $\tilde{\beta}_{(k)}=\tilde{D}_{(k)}\hat{\Upsilon}\beta$,
we have 
\[
\dot{\hat{\gamma}}\hat{\Upsilon}\beta=\stackrel[i=1]{\tilde{K}}{\sum}T_{h}^{-1/2}\hat{\lambda}_{(i)}^{-1/2}Y\hat{\xi}_{(i)}(\hat{\theta})\hat{\varsigma}'_{(i)}\tilde{\beta}_{(i)}.
\]
Therefore, $\parallel\alpha-\sum_{i=1}^{\tilde{K}}T_{h}^{-1/2}\hat{\lambda}_{(i)}^{-1/2}Y\hat{\xi}_{(i)}(\hat{\theta})\hat{\varsigma}'_{(i)}\tilde{\beta}_{(i)}||=o_{\textrm{P}}(1)$
is equivalent to $||\dot{\hat{\gamma}}\hat{\Upsilon}\beta-\alpha||=o_{\textrm{P}}(1)$.

By analogy with (\ref{eq:B14}), we can write the mixed-frequency
prediction error for $\dot{y}_{T+h}$ as 
\[
\dot{\hat{y}}_{T+h}-\mathbb{E}_{T}(\dot{y}_{T+h})=\dot{\hat{\gamma}}\hat{\Upsilon}X_{T}(\hat{\theta})+\hat{\alpha}_{w}\dot{w}_{T}-\alpha F_{T}(\theta)-\alpha_{w}\dot{w}_{T}
\]
\[
=\dot{\hat{\gamma}}\hat{\Upsilon}\beta\left(F_{T}(\hat{\theta})-F(\hat{\theta})W'(WW')^{-1}\dot{w}_{T}\right)-\alpha\left(F_{T}(\theta)-F(\theta)W'(WW')^{-1}\dot{w}_{T}\right)
\]
\[
+\dot{\hat{\gamma}}\hat{\Upsilon}\left(E_{T}(\hat{\theta})-E(\hat{\theta})W'(WW')^{-1}\dot{w}_{T}\right)+\epsilon W'(WW')^{-1}\dot{w}_{T}.
\]
Since $\sum_{j=0}^{J}b_{j}(\theta_f)=\mathbb{I}_{K}$, $||\alpha F(\theta)||+||\epsilon||\lesssim_{\textrm{P}}T^{1/2}$. Therefore, $||Y||\leq||\alpha F(\theta)||+||\epsilon||\lesssim_{\textrm{P}}T^{1/2}$.
Using this result, and definition of $\dot{\hat{\gamma}}$, we have
\[
||\dot{\hat{\gamma}}\hat{\Upsilon}E_{T}(\hat{\theta})||\leq\underset{k\leq\tilde{K}}{\sum}T_{h}^{-1/2}\hat{\lambda}_{(k)}^{-1/2}||Y||\,||\hat{\xi}_{(k)}(\hat{\theta})||\,||\hat{\varsigma}'{}_{(k)}\tilde{D}_{(k)}\hat{\Upsilon}E_{T}(\hat{\theta})||
\]
\[
\lesssim_{\textrm{P}}\underset{k\leq\tilde{K}}{\sum}\hat{\lambda}_{(k)}^{-1/2}||\hat{\varsigma}'{}_{(k)}\tilde{D}_{(k)}\hat{\Upsilon}E_{T}(\hat{\theta})||,
\]
and likewise,
\[
T_{h}^{-1}||\dot{\hat{\gamma}}\hat{\Upsilon}E(\hat{\theta})W'||\lesssim_{\textrm{P}}\underset{k\leq\tilde{K}}{\sum}T_{h}^{-1}\hat{\lambda}_{(k)}^{-1/2}||\hat{\varsigma}'{}_{(k)}\tilde{E}_{(k)}(\hat{\theta})W'||.
\]
Since the MIDAS aggregation step is a finite weighted average with normalized weights,
we replace $\tilde{E}_{(k)}$ in Lemma \ref{lem:5}(ii) with $\tilde{E}_{(k)}(\hat{\theta})$,
and $e_{T}$ in Lemma \ref{lem:5}(iv) with $E_{T}(\hat{\theta})$,
and their original bounds still hold. Hence, using $\hat{\lambda}_{(k)}\asymp_{\textrm{P}}qN$,
we have
\[
T_{h}^{-1}\hat{\lambda}_{(k)}^{-1/2}||\hat{\varsigma}'{}_{(k)}\tilde{E}_{(k)}(\hat{\theta})W'||\lesssim_{\textrm{P}}q^{-1}N^{-1}+T^{-1},
\]
and 
\[
\hat{\lambda}_{(k)}^{-1/2}||\hat{\varsigma}'{}_{(k)}\tilde{D}_{(k)}\hat{\Upsilon}E_{T}(\hat{\theta})||\lesssim_{\textrm{P}}q^{-1/2}N^{-1/2}+T^{-1/2}.
\]
Therefore, $||\dot{\hat{\gamma}}\hat{\Upsilon}E_{T}(\hat{\theta})||=o_{\textrm{P}}(1)$
and we also show that $||\dot{\hat{\gamma}}\hat{\Upsilon}E(\hat{\theta})W'(WW')^{-1}||=o_{\textrm{P}}(1)$.
Then we deduce the term $\dot{\hat{\gamma}}\hat{\Upsilon}\beta F_{T}(\hat{\theta})-\alpha F_{T}(\theta)=(\dot{\hat{\gamma}}\hat{\Upsilon}\beta-\alpha)F_{T}(\hat{\theta})+\alpha\left(F_{T}(\hat{\theta})-F_{T}(\theta)\right)$.
Because $||\dot{\hat{\gamma}}\hat{\Upsilon}\beta-\alpha||=o_{\textrm{P}}(1)$,
this term vanishes. Together with $||F(\theta)W'||\lesssim_{\textrm{P}}T^{1/2},||\epsilon W'||\lesssim_{\textrm{P}}T^{1/2}$
from Assumption \ref{assu:1}, $\hat{\theta}-\theta=o_{\textrm{P}}(1)$ from Lemma \ref{lem:16}(iii), we can conclude that each term in the
expression for the mixed-frequency prediction error vanishes asymptotically.
Therefore, $\dot{\hat{y}}_{T+h}-\mathbb{E}_{T}(\dot{y}_{T+h})\overset{\textrm{P}}{\longrightarrow}0$.
\end{proof}

\subsection{Proof of Theorem 3}
\begin{proof}
As shown in the proof of Theorem \ref{thm:1}, we impose that $\hat{K}=\tilde{K}$
and $\hat{I}_{k}=I_{k}$, because Lemma \ref{lem:2} shows that both
events occur with probability converging to 1 and our analysis indicates
that this holds in both the same- and mixed-frequency cases. Lemma
\ref{lem:4}(iv) further indicates that under the condition $\lambda_{K}(\alpha'\alpha)\gtrsim1$,
we have $\tilde{K}=K$ in the same-frequency case. This result also
extends to the mixed-frequency case under our assumptions on the weighting
function. Combine with $P(\hat{K}=\tilde{K})\rightarrow1$, we can
obtain (i) of Theorem \ref{thm:3}. In what follows, we directly impose
$\hat{K}=K$.

Similarly, we first focus on the same-frequency case. Following the
same reasoning as above (\ref{eq:B14}), we only need to analyze the
case excluding $w_{t}$. As $\hat{\mathbf{f}}_{(k)}=T_{h}^{1/2}\hat{\lambda}_{(k)}^{1/2}\hat{\xi}'_{(k)}$,
Theorem \ref{thm:1} implies $||\hat{\xi}'_{(k)}\mathbb{M}_{\mathbf{f}'}||\lesssim_{\textrm{P}}q^{-1/2}N^{-1/2}+T^{-1}$
for $k\leq K$. Let $v$ denote $\mathbf{f}'(\mathbf{ff}')^{-1/2},$ we have
\begin{equation}
||\hat{\xi}-\mathbb{P}_{\mathbf{f}'}\hat{\xi}||=||\hat{\xi}-vv'\hat{\xi}||\lesssim_{\textrm{P}}q^{-1/2}N^{-1/2}+T^{-1},\label{eq:B18}
\end{equation}
where $\hat{\xi}$ is a $T\times K$ matrix with each column equal
to $\hat{\xi}_{(k)}$. As shown in (\ref{eq:B1}), $\hat{\xi}_{(k)}$
is derived from scaling predictors, and although this introduces magnitude
of $T^{-1/2}$ compared to deriving $\hat{\xi}_{(k)}$ from unscaled
ones, it does not affect the leading terms retained in the proof.
(\ref{eq:B18}) implies that $||\hat{\xi}'vv'\hat{\xi}-\mathbb{I}_{K}||\lesssim_{\textrm{P}}q^{-1/2}N^{-1/2}+T^{-1}.$
By Weyl's inequality, $|\sigma_{i}(\hat{\xi}'v)-1|\lesssim_{\textrm{P}}q^{-1/2}N^{-1/2}+T^{-1}$,
for $1\leq i\leq K$, as so
\[
||v-\hat{\xi}\hat{\xi}'v||\leq\sigma_{K}^{-1}(v'\hat{\xi})||vv'\hat{\xi}-\hat{\xi}\hat{\xi}'vv'\hat{\xi}||\lesssim_{\textrm{P}}||vv'\hat{\xi}-\hat{\xi}||+||\hat{\xi}(\hat{\xi}'vv'\hat{\xi}-\mathbb{I}_{K})||\lesssim_{\textrm{P}}q^{-1/2}N^{-1/2}+T^{-1}.
\]

\noindent Then, using above equation, (\ref{eq:B18}). and the fact
that $||v||=1$ and $||\hat{\xi}||=1,$ we have
\[
||\mathbb{P}_{\hat{\mathbf{f}}'}-\mathbb{P}_{\mathbf{f}'}||=||\hat{\xi}\hat{\xi}'-vv'||\leq||\hat{\xi}(\hat{\xi}-vv'\hat{\xi})'||+||(\hat{\xi}\hat{\xi}'v-v)v'||\lesssim_{\textrm{P}}q^{-1/2}N^{-1/2}+T^{-1}.
\]
Next, a more detailed examination of $\hat{\gamma}$ is required.
As noted in the proof of Lemma \ref{lem:2},
\begin{equation}
\hat{\gamma}\hat{\Upsilon}\beta=\stackrel[k=1]{\tilde{K}}{\sum}T_{h}^{-1/2}\hat{\lambda}_{(k)}^{-1/2}\mathbf{y}\hat{\xi}_{(k)}\hat{\varsigma}'_{(k)}\tilde{\beta}_{(k)}.\label{eq:B19}
\end{equation}
Denote $L_{1}=(\iota_{11},...,\iota_{\hat{K}1})\in R^{K\times\hat{K}}$,
$L_{2}=(\iota_{12},...,\iota_{\hat{K}2})\in R^{K\times\hat{K}}$,
where
\begin{equation}
\iota_{k1}=T^{-1/2}\mathbf{f}\hat{\xi}_{(k)},\quad\iota_{k2}=\hat{\lambda}_{(k)}^{-1/2}\tilde{\beta}'_{(k)}\hat{\varsigma}{}_{(k)}.\label{eq:B20}
\end{equation}
By lemma \ref{lem:9},
\begin{equation}
||T_{h}^{-1/2}\epsilon\hat{\xi}_{(k)}-T_{h}^{-1}\epsilon \mathbf{f}'\iota_{k2}||\lesssim_{\textrm{P}}T^{-1}+q^{-1}N^{-1}.\label{eq:B21}
\end{equation}
As we impose that $\hat{K}=\tilde{K}=K$, together with (\ref{eq:B19}),
(\ref{eq:B20}), and (\ref{eq:B21}), with $||L_{1}||\lesssim_{\textrm{P}}1$,
$||L_{2}||\lesssim_{\textrm{P}}1$ from Lemma \ref{lem:10}, we have
\begin{equation}
||\hat{\gamma}\hat{\Upsilon}\beta-\alpha L_{1}L'_{2}-T_{h}^{-1}\epsilon \mathbf{f}'L_{2}L'_{2}||\lesssim_{\textrm{P}}T^{-1}+q^{-1}N^{-1}.
\end{equation}
Using Lemma \ref{lem:10}(iv)(v), we obtain $||\hat{\gamma}\hat{\Upsilon}\beta-\alpha-T_{h}^{-1}\epsilon \mathbf{f}'||\lesssim_{\textrm{P}}T^{-1}+q^{-1}N^{-1}.$ 

As for the mix-frequency case, since the latent factors extracted
from the high-frequency scaled predictors satisfy $||\mathbb{P}_{\hat{\mathbf{f}}_{H}'}-\mathbb{P}_{\mathbf{f}'_{H}}||\lesssim_{\textrm{P}}q^{-1/2}N^{-1/2}+T^{-1},$
and the MIDAS aggregation uses fixed, bounded, and normalized weights with $\hat{\theta}$ consistent by Lemma \ref{lem:16}(iii),
the aggregated factor space is also consistently recovered. That is
\[
\left\Vert \mathbb{P}_{\underline{\hat{F}}(\hat{\theta})'}-\mathbb{P}_{\underline{F}(\theta)'}\right\Vert \lesssim_{\textrm{P}}q^{-1/2}N^{-1/2}+T^{-1}.
\]
We next consider $\dot{\hat{\gamma}}\hat{\Upsilon}\beta=\stackrel[k=1]{\tilde{K}}{\sum}T_{h}^{-1/2}\hat{\lambda}_{(k)}^{-1/2}Y\hat{\xi}_{(k)}(\hat{\theta})\hat{\varsigma}'_{(k)}\tilde{\beta}_{(k)}$,
and define $\dot{L}_{1}=(\dot{\iota}_{11},...,\dot{\iota}_{\hat{K}1})\in R^{K\times\hat{K}}$,
$\dot{\iota}_{k1}=T^{-1/2}F(\hat{\theta})\hat{\xi}_{(k)}$, while
keep $L_{2}$ unchanged. Since $\mathbf{\Sigma}_{f}\equiv\mathbf{\Sigma}_{f(\theta)}=\mathbb{I}_{K}$
and $||T_{h}^{-1}\mathbf{f}_H\mathbf{f}_H'-\mathbb{I}_{K}||\lesssim_{\textrm{P}}T^{-1/2},$
we have $||\mathbf{f}_H-F(\theta)||=O_{\textrm{P}}(1)$. Then,
\[
\begin{split}||T_{h}^{-1/2}\epsilon\hat{\xi}_{(k)}-T_{h}^{-1}\epsilon F(\theta)'\iota_{k2}||&\leq||T_{h}^{-1/2}\epsilon\hat{\xi}_{(k)}-T_{h}^{-1}\epsilon \mathbf{f}_H'\iota_{k2}||+||T_{h}^{-1/2}\epsilon(\mathbf{f}_H-F(\theta)')\iota_{k2}||\\&\lesssim_{\textrm{P}}T^{-1}+q^{-1}N^{-1}.\end{split}
\]
By our definition of MIDAS aggregation and the consistency of $\hat{\theta}$,
$||\dot{L}_{1}||\lesssim_{\textrm{P}}1$. Likewise, replacing all
instances of $L_{1}$ in Lemma \ref{lem:10} with $\dot{L}_{1}$,
does not alter the leading term. Using this and all above results,
we impose $\hat{K}=\tilde{K}=K$, we have
\[
||\dot{\hat{\gamma}}\hat{\Upsilon}\beta-\alpha \dot{L}_{1}L'_{2}-T_{h}^{-1}\epsilon F(\theta)'L_{2}L'_{2}||\lesssim_{\textrm{P}}T^{-1}+q^{-1}N^{-1}.
\]
Using Lemma \ref{lem:10} (iv) and $\dot{L}_{1}$ case of Lemma
\ref{lem:10}(v), that is $||\dot{L}_{1}L_{2}-\mathbb{I}_{K}||\lesssim_{\textrm{P}}T^{-1}+q^{-1}N^{-1}.$
We can get $||\dot{\hat{\gamma}}\hat{\Upsilon}\beta-\alpha-T_{h}^{-1}\epsilon F(\theta)'||\lesssim_{\textrm{P}}T^{-1}+q^{-1}N^{-1}.$
\end{proof}

\subsection{Proof of Theorem 4}
\begin{proof}
We start with the same-frequency case. Following the reasoning in
the proof of Theorem \ref{thm:2}, we have $||\mathbf{f}\mathbf{w}'(\mathbf{ww}')^{-1}||\lesssim_{\textrm{P}}T^{-1/2}$
from Assumption \ref{assu:1} and $||\hat{\gamma}\hat{\Upsilon}\mathbf{e}\mathbf{w}'(\mathbf{ww'})^{-1}||\lesssim_{\textrm{P}}T^{-1}+q^{-1}N^{-1}$
as shown in (\ref{eq:B16}) and (\ref{eq:B17}). Combine with $||\hat{\gamma}\hat{\Upsilon}\beta-\alpha-T_{h}^{-1}\epsilon \mathbf{f}'||\lesssim_{\textrm{P}}T^{-1}+q^{-1}N^{-1}$,
we can obtain from (\ref{eq:B14}) that:
\[
\hat{y}_{T+h}-\mathbb{E}_{T}(y_{T+h})=T_{h}^{-1}\epsilon \mathbf{f}'f_{T}+\epsilon \mathbf{w}'(\mathbf{ww}')^{-1}w_{T}+\hat{\gamma}\hat{\Upsilon}e_{T}+O_{\textrm{P}}(T^{-1}+q^{-1}N^{-1}).
\]
By Assumption \ref{assu:1}, we have $|\lambda_{i}\left(T_{h}^{-1}\mathbf{\Sigma}_{w}^{-1/2}\mathbf{ww}'\mathbf{\Sigma}_{w}\right)-1|\lesssim_{\textrm{P}}T^{-1/2}$
and thus
\[
||\epsilon \mathbf{w}'(\mathbf{ww}')^{-1}w_{T}-T_{h}^{-1}\epsilon \mathbf{w}'\mathbf{\Sigma}_{w}^{-1}w_{T}||\leq T_{h}^{-1}||\epsilon \mathbf{w}'||\,||(T_{h}^{-1}\mathbf{ww}')^{-1}-\mathbf{\Sigma}_{w}^{-1}||\,||w_{T}||
\]
\[
\lesssim_{P}T_{h}^{-1/2}||T_{h}^{-1}\mathbf{\Sigma}_{w}^{-1/2}\mathbf{ww}'\mathbf{\Sigma}_{w}^{-1/2}-1||=T_{h}^{-1/2}\underset{i\leq D}{\max}|\lambda_{i}\left(T_{h}^{-1}\mathbf{\Sigma}_{w}^{-1/2}\mathbf{ww}'\mathbf{\Sigma}_{w}^{-1/2}\right)^{-1}-1|
\]
\begin{equation}
\lesssim_{\textrm{P}}T^{-1}.
\end{equation}
For $\hat{\gamma}\hat{\Upsilon}e_{T}$, by (\ref{eq:B5}), we have
$\hat{\gamma}\hat{\Upsilon}e_{T}=\sum_{k=1}^{K}\hat{\alpha}_{(k)}\hat{\varsigma}'_{(k)}\tilde{D}_{(k)}\hat{\Upsilon}e_{T}$
and thus
\begin{equation}
||\hat{\gamma}\hat{\Upsilon}e_{T}-\stackrel[k=1]{K}{\sum}\lambda_{(k)}^{-1/2}\alpha\iota_{(k)}\varsigma'_{(k)}D_{(k)}\hat{\Upsilon}e_{T}||\leq\stackrel[k=1]{K}{\sum}||\hat{\alpha}_{(k)}\hat{\varsigma}'_{(k)}\tilde{D}_{(k)}\hat{\Upsilon}e_{T}-\lambda_{(k)}^{-1/2}\alpha\iota_{(k)}\varsigma'_{(k)}D_{(k)}\hat{\Upsilon}e_{T}||.\label{eq:B24}
\end{equation}
Lemma \ref{lem:11}(vi) gives
\begin{equation}
q^{-1/2}N^{-1/2}|\hat{\varsigma}'_{(k)}\tilde{D}_{(k)}\hat{\Upsilon}e_{T}-\varsigma'_{(k)}D_{(k)}\hat{\Upsilon}e_{T}|\lesssim_{\textrm{P}}T^{-1}+q^{-1}N^{-1}.\label{eq:B25}
\end{equation}
Additionally, (\ref{eq:B2}) and Lemma \ref{lem:1} give $\hat{\lambda}_{(k)}^{1/2}\hat{\alpha}=T_{h}^{-1/2}\mathbf{y}\hat{\xi}_{(k)}=\alpha\iota_{k1}+T_{h}^{-1/2}\epsilon\hat{\xi}_{(k)}.$
With (\ref{eq:B21}), $||\epsilon \mathbf{f}'||\lesssim_{\textrm{P}}T^{1/2}$,
and $||\iota_{k2}||\lesssim_{\textrm{P}}1$ from Lemma \ref{lem:10}(i),
this equation leads to
\[
||\hat{\lambda}_{(k)}^{-1/2}\hat{\alpha}_{(k)}-\alpha\iota_{k1}||\lesssim||T_{h}^{-1/2}\epsilon\hat{\xi}_{(k)}-T_{h}^{-1}\epsilon \mathbf{f}'\iota_{k2}||+||T_{h}^{-1}\epsilon \mathbf{f}'\iota_{k2}||\lesssim_{\textrm{P}}T^{-1/2}+q^{-1}N^{-1}.
\]
Using $||\iota_{k2}-\iota_{(k)}||\lesssim_{\textrm{P}}T^{-1/2}+q^{-1/2}N^{-1/2}$
implied by Lemma \ref{lem:10}(iii) and $\hat{\lambda}_{(k)}^{1/2}\asymp_{\textrm{P}}q^{1/2}N^{1/2}$
implied by Lemma \ref{lem:2}(iii), we have 
\begin{equation}
||\hat{\alpha}_{(k)}-\hat{\lambda}_{(k)}^{-1/2}\alpha\iota_{(k)}||\leq||\hat{\alpha}_{(k)}-\hat{\lambda}_{(k)}^{-1/2}\alpha\iota_{k2}||+||\hat{\lambda}_{(k)}^{-1/2}\alpha(\iota_{(k)}-\iota_{k2})||\lesssim_{\textrm{P}}T^{-1/2}q^{-1/2}N^{-1/2}+q^{-1}N^{-1}.\label{eq:B26}
\end{equation}
Furthermore, applying Lemma \ref{lem:2}(iii), we obtain
\[
||\hat{\lambda}_{(k)}^{-1/2}-\lambda_{(k)}^{-1/2}||\leq\hat{\lambda}_{(k)}^{-1/2}|\hat{\lambda}_{(k)}^{1/2}/\lambda_{(k)}^{1/2}-1|\lesssim_{\textrm{P}}T^{-1/2}q^{-1/2}N^{-1/2}+q^{-1}N^{-1}.
\]
Since $||\iota_{(k)}||=1$ and Lemma \ref{lem:2}(iii), the above
two inequalities lead to
\begin{equation}
||\hat{\alpha}_{(k)}-\lambda_{(k)}^{-1/2}\alpha\iota_{(k)}||\lesssim_{\textrm{P}}T^{-1/2}q^{-1/2}N^{-1/2}+q^{-1}N^{-1}.\label{eq:B27}
\end{equation}

For each term in the summation of (\ref{eq:B24}), we have 
\[
||\hat{\alpha}_{(k)}\hat{\varsigma}'_{(k)}\tilde{D}_{(k)}\hat{\Upsilon}e_{T}-\lambda_{(k)}^{-1/2}\alpha\iota_{(k)}\varsigma'_{(k)}D_{(k)}\hat{\Upsilon}e_{T}||
\]
\begin{equation}
\leq||\hat{\alpha}_{(k)}(\hat{\varsigma}'_{(k)}\tilde{D}_{(k)}\hat{\Upsilon}e_{T}-\varsigma'_{(k)}D_{(k)}\hat{\Upsilon}e_{T})||+||(\hat{\alpha}_{(k)}-\lambda_{(k)}^{-1/2}\alpha\iota_{(k)})\varsigma'_{(k)}D_{(k)}\hat{\Upsilon}e_{T}||.\label{eq:B28}
\end{equation}
Note that (\ref{eq:B26}) also implies $||\hat{\alpha}_{(k)}||\lesssim_{\textrm{P}}q^{-1/2}N^{-1/2}$
as $\lambda_{(k)}\asymp_{\textrm{P}}qN,$ and that (\ref{eq:B25})
implies the first term in (\ref{eq:B28}) is $O_{\textrm{P}}(T^{-1}+q^{-1}N^{-1}).$
Moreover, $|\varsigma'_{(k)}D_{(k)}\hat{\Upsilon}e_{T}|\lesssim_{\textrm{P}}1$
from Lemma \ref{lem:5}(iv) and (\ref{eq:B27}) show that the second
term in (\ref{eq:B28}) is also $O_{\textrm{P}}(T^{-1}+q^{-1}N^{-1})$.
Therefore, (\ref{eq:B24}) becomes
\begin{equation}
||\hat{\gamma}\hat{\Upsilon}e_{T}-\stackrel[k=1]{K}{\sum}\lambda_{(k)}^{-1/2}\alpha\iota_{(k)}\varsigma'_{(k)}D_{(k)}\hat{\Upsilon}e_{T}||\lesssim_{\textrm{P}}T^{-1}+q^{-1}N^{-1}.
\end{equation}

In summary, we have established that
\[
\hat{y}_{T+h}-\mathbb{E}_{T}(y_{T+h})=\frac{\epsilon \mathbf{f}'}{T_{h}}f_{T}+\frac{\epsilon \mathbf{w}'}{T_{h}}\mathbf{\Sigma}_{w}^{-1}w_{T}+\stackrel[k=1]{K}{\sum}\lambda_{(k)}^{-1/2}\alpha\iota_{(k)}\varsigma'_{(k)}D_{(k)}\hat{\Upsilon}e_{T}+O_{\textrm{P}}(T^{-1}+q^{-1}N^{-1}).
\]
In the usual case that $\mathbf{\Sigma}_{f}$ may not be $\mathbb{I}_{K}$,
the first term becomes $T_{h}^{-1}\epsilon \mathbf{f}'\mathbf{\Sigma}_{f}^{-1}f_{T}.$
Using the fact $\varsigma_{(k)}=\lambda_{(k)}^{-1/2}(\hat{\Upsilon}\beta)_{(k)}\iota_{(k)}=\lambda_{(k)}^{-1/2}(\hat{\Upsilon}\beta)_{[I_{k}]}\iota_{(k)}$
and the iterative definition of $D_{(k)}$, we can see that $\lambda_{(k)}^{-1/2}\varsigma'_{(k)}D_{(k)}\hat{\Upsilon}e_{T}$
is exactly the $k$th row of $\Lambda^{-1}\Omega'\Psi\hat{\Upsilon}e_{T}$
with $\Lambda,\Omega,$ and $\Psi$ defined in Theorem \ref{thm:4}.
Using Delta method and Assumption \ref{assu:7}, it is straightforward
to obtain the desired CLT.

Next, we consider the mixed-frequency case. Using analogous arguments, together with the definition of MIDAS aggregation and the consistency of $\hat{\theta}$ in Lemma \ref{lem:16}(iii),
we have $||F(\theta)W'(WW')^{-1}||\lesssim_{\textrm{P}}T^{-1/2}$, $||\dot{\hat{\gamma}}\hat{\Upsilon}E(\theta)W'(WW)^{-1}||\lesssim_{\textrm{P}}T^{-1}+q^{-1}N^{-1},$
$|\lambda_{i}\left(T_{h}^{-1}\mathbf{\Sigma}_{w}^{-1/2}WW'\mathbf{\Sigma}_{w}\right)-1|\lesssim_{\textrm{P}}T^{-1/2}$,
and $||\dot{\hat{\gamma}}\hat{\Upsilon}\beta-\alpha-T_{h}^{-1}\epsilon F(\theta)'||\lesssim_{\textrm{P}}T^{-1}+q^{-1}N^{-1}.$
Under these conditions, we have the error as
\[
\dot{\hat{y}}_{T+h}-\mathbb{E}_{T}(\dot{y}_{T+h})=T_{h}^{-1}\epsilon F(\theta)'F_{T}(\hat{\theta})+\dot{\hat{\gamma}}\hat{\Upsilon}E_{T}(\hat{\theta})+O_{\textrm{P}}(T^{-1}+q^{-1}N^{-1}).
\]
In analyzing $\dot{\hat{\gamma}}\hat{\Upsilon}E_{T}(\hat{\theta})$,
we should focus on the $\dot{\hat{\gamma}}\hat{\Upsilon}E_{T}(\hat{\theta})=\sum_{k=1}^{K}\dot{\hat{\alpha}}_{(k)}\hat{\varsigma}'_{(k)}\tilde{D}_{(k)}\hat{\Upsilon}E_{T}(\hat{\theta})$,
and thus
\[
\begin{split}&||\dot{\hat{\gamma}}E_{T}(\hat{\theta})-\stackrel[k=1]{K}{\sum}\lambda_{(k)}^{-1/2}\alpha\iota_{(k)}\varsigma'_{(k)}D_{(k)}\hat{\Upsilon}E_{T}(\hat{\theta})||\\&\leq\stackrel[k=1]{K}{\sum}||\dot{\hat{\alpha}}_{(k)}\hat{\varsigma}'_{(k)}\tilde{D}_{(k)}\hat{\Upsilon}E_{T}(\hat{\theta})-\lambda_{(k)}^{-1/2}\alpha\iota_{(k)}\varsigma'_{(k)}D_{(k)}\hat{\Upsilon}E_{T}(\hat{\theta})||.\end{split}
\]
Lemma \ref{lem:11} continues to hold under the inclusion of the aggregation
function so that (\ref{eq:B25}) becomes $q^{-1/2}N^{-1/2}|\hat{\varsigma}'_{(k)}\tilde{D}_{(k)}$
$\hat{\Upsilon}E_{T}(\hat{\theta})-\varsigma'_{(k)}D_{(k)}\hat{\Upsilon}E_{T}(\hat{\theta})|\lesssim_{\textrm{P}}T^{-1}+q^{-1}N^{-1}.$
Moreover, under the MIDAS normalization $\sum_{j=0}^{J}b_{j}(\theta_f)=\mathbb{I}_{K},$ $||T_{h}^{-1/2}\epsilon\hat{\xi}_{(k)}-T_{h}^{-1}\epsilon F(\theta)'\iota_{k2}||\lesssim_{\textrm{P}}T^{-1}+q^{-1}N^{-1}$,
$||\epsilon F(\theta)'||\lesssim_{\textrm{P}}T^{1/2}$, and $||\iota_{k2}||\lesssim_{\textrm{P}}1$,
equation $\hat{\lambda}_{(k)}^{1/2}\dot{\hat{\alpha}}=\alpha\dot{\iota}_{k1}+T_{h}^{-1/2}\epsilon\hat{\xi}_{(k)}$
leads to
\[
\begin{split}||\hat{\lambda}_{(k)}^{-1/2}\dot{\hat{\alpha}}_{(k)}-\alpha\dot{\iota}_{k1}||&\lesssim||T_{h}^{-1/2}\epsilon\hat{\xi}_{(k)}-T_{h}^{-1}\epsilon F(\theta)'\iota_{k2}||+||T_{h}^{-1}\epsilon F(\theta)'\iota_{k2}||\\&\lesssim_{\textrm{P}}T^{-1/2}+q^{-1}N^{-1}.\end{split}
\]
Combine above equation with $||\hat{\lambda}_{(k)}^{-1/2}-\lambda_{(k)}^{-1/2}||\lesssim_{\textrm{P}}T^{-1/2}q^{-1/2}N^{-1/2}+q^{-1}N^{-1}$
and $||\iota_{(k)}||=1,$ we have
\[
||\dot{\hat{\alpha}}_{(k)}-\lambda_{(k)}^{-1/2}\alpha\iota_{(k)}||\lesssim_{\textrm{P}}T^{-1/2}q^{-1/2}N^{-1/2}+q^{-1}N^{-1}.
\]
Revisit the long equation about $\dot{\hat{\gamma}}\hat{\Upsilon}E_{T}(\hat{\theta})$,
we have
\[
\begin{split}&||\dot{\hat{\alpha}}_{(k)}\hat{\varsigma}'_{(k)}\tilde{D}_{(k)}\hat{\Upsilon}E_{T}(\hat{\theta})-\lambda_{(k)}^{-1/2}\alpha\iota_{(k)}\varsigma'_{(k)}D_{(k)}\hat{\Upsilon}E_{T}(\hat{\theta})||\\&\leq||\dot{\hat{\alpha}}_{(k)}(\hat{\varsigma}'_{(k)}\tilde{D}_{(k)}\hat{\Upsilon}E_{T}(\hat{\theta})-\varsigma'_{(k)}D_{(k)}\hat{\Upsilon}E_{T}(\hat{\theta}))||+||(\dot{\hat{\alpha}}_{(k)}-\lambda_{(k)}^{-1/2}\alpha\iota_{(k)})\varsigma'_{(k)}D_{(k)}\hat{\Upsilon}E_{T}(\hat{\theta})||.\end{split}
\]
Given the same conditions following (\ref{eq:B28}), and taking into
account $\sum_{j=0}^{J}b_{j}(\hat{\theta})$, the long equation becomes
$||\dot{\hat{\gamma}}\hat{\Upsilon}E_{T}(\hat{\theta})-\sum_{k=1}^{K}\lambda_{(k)}^{-1/2}\alpha\iota_{(k)}\varsigma'_{(k)}D_{(k)}\hat{\Upsilon}E_{T}(\hat{\theta})||\lesssim_{\textrm{P}}T^{-1}+q^{-1}N^{-1}.$
Finally, we have
\[
||\dot{\hat{\gamma}}\hat{\Upsilon}E_{T}(\hat{\theta})-\stackrel[k=1]{K}{\sum}\lambda_{(k)}^{-1/2}\alpha\iota_{(k)}\varsigma'_{(k)}D_{(k)}\hat{\Upsilon}E_{T}(\theta)||\lesssim_{\textrm{P}}T^{-1}+q^{-1}N^{-1}.
\]
To sum up, we have established that
\[
\begin{split}\dot{\hat{y}}_{T+h}-\mathbb{E}_{T}(\dot{y}_{T+h})&=\frac{\epsilon F(\theta)'}{T_{h}}F_{T}(\theta)+\frac{\epsilon W'}{T_{h}}\mathbf{\Sigma}_{w}^{-1}w_{T}\\&\quad+\stackrel[k=1]{K}{\sum}\lambda_{(k)}^{-1/2}\alpha\iota_{(k)}\varsigma'_{(k)}D_{(k)}\hat{\Upsilon}E_{T}(\theta)+O_{\textrm{P}}(T^{-1}+q^{-1}N^{-1}).\end{split}
\]
If $\mathbf{\Sigma}_{f(\theta)}\neq\mathbb{I}_{K},$ the first term
becomes $T_{h}^{-1}\epsilon F(\theta)'\mathbf{\Sigma}_{f(\theta)}^{-1}F_{T}(\theta)$.
Similar to the same-frequency, we can see that $\lambda_{(k)}^{-1/2}\varsigma'_{(k)}D_{(k)}\hat{\Upsilon}E_{T}(\theta)$
is exactly the $k$th row of $\Lambda^{-1}\Omega'\Psi\hat{\Upsilon}E_{T}(\theta)$.
By applying the Delta method and Assumption \ref{assu:7}, the desired
CLT follows straightforwardly.
\end{proof}

\subsection{Estimation of $\varPhi_{1}$ and $\varPhi_{2}$\label{Prof B.5}}
\label{Appendix B.5}
This subsection extends the construction in Section A of \cite{giglio2023prediction} to the mixed-frequency setting; we present the mixed-frequency case directly, since the alternative case can be readily established. From the outputs of Algorithm \ref{Algorithm 1},
we have defined $\hat{\underline{F}}(\hat{\theta})$, $\hat{\beta}$,
and $\hat{\alpha}.$ Consequently, we can also estimate $\hat{\epsilon}=Y-\hat{\alpha}\hat{\underline{F}}(\hat{\theta})-\hat{\alpha}_{w}\underline{W}$
and $\hat{\underline{\mathbf{e}}}_H=\underline{\mathbf{x}}_H-\hat{\beta}\hat{\underline{\mathbf{f}}}_H$
(i.e., $\hat{\Upsilon}\hat{\underline{\mathbf{e}}}_H=\hat{\Upsilon}\underline{\mathbf{x}}_H-\hat{\Upsilon}\hat{\beta}\hat{\underline{\mathbf{f}}}_H$).
This allows us to construct Newey-West-type estimators for $\Pi_{11}$,
$\Pi_{12}$ and $\Pi_{22}$, as each component can be estimated using
their corresponding sample analogs derived above. Furthermore, estimators
for $\mathbf{\Sigma}_{F(\theta)}$ and $\mathbf{\Sigma}_{w}$ can
be obtained as $\mathbf{\hat{\Sigma}}_{F(\theta)}=T_{h}^{-1}\hat{\underline{F}}(\hat{\theta})\hat{\underline{F}}(\hat{\theta})'$
and $\hat{\mathbf{\Sigma}}_{w}=T_{h}^{-1}\underline{W}\underline{W}'$,
respectively. With $\hat{F}_{T}(\hat{\theta})=\widehat{\zeta}'_{(k)}X_{T}(\hat{\theta})$,
$\hat{\varPhi}_{1}$ can be constructed as follows:
\[
\hat{\varPhi}_{1}=(\hat{F}{}_{T}(\hat{\theta})',w_{T}')\mathbf{\hat{\Sigma}^{-1}}_{F(\theta),w}\left(\begin{array}{cc}
\hat{\Pi}_{11} & \hat{\Pi}_{12}\\
\hat{\Pi}'_{12} & \hat{\Pi}_{22}
\end{array}\right)\mathbf{\hat{\Sigma}^{-1}}_{F(\theta),w}(\hat{F}{}_{T}(\hat{\theta})',w_{T}')'.
\]
Although $\hat{\underline{F}}(\hat{\theta})$ recovers the latent factors only up to a rotation, the rotation matrix appears on both sides of the quadratic form defining $\hat{\varPhi}_1$ and therefore cancels, so $\hat{\varPhi}_1$ shares the limiting behavior of the infeasible estimator built from the true $F(\theta)$. A practical consequence is that the construction does not require estimating high-dimensional objects such as $\boldsymbol{\Sigma}_{E(\theta)}$. Consistency of $\hat{\varPhi}_1$ then follows by adapting the argument in \citet{giglio2021asset} to the mixed-frequency setting, and we omit the details.

For $\varPhi_{2}$, we may also adopt a thresholding estimator for
$\mathbf{\Sigma}_{E(\theta)}=Cov[E_{t}(\theta)]$ following the method
proposed by \citet{fan2013large}, as implemented by \cite{giglio2023prediction}.
Specifically, $\mathbf{\hat{\Sigma}}_{E(\theta)}$ can be constructed
as
\[
\begin{array}{cc}
(\mathbf{\hat{\Sigma}}_{E(\theta)})_{ij}=\begin{cases}
(\mathbf{\tilde{\Sigma}}_{E(\theta)})_{ij}, & i=j\\
s_{ij}\left((\mathbf{\tilde{\Sigma}}_{E(\theta)})_{ij}\right), & i\neq j
\end{cases}, & \mathbf{\tilde{\Sigma}}_{E(\theta)}=T_{h}^{-1}\end{array}\hat{\underline{E}}(\hat{\theta})\hat{\underline{E}}(\hat{\theta})',
\]
where $s_{ij}(\cdot)$ is a general thresholding function with an
entry-dependent threshold $\tau_{ij}$ satisfying (i) $s_{ij}(z)=0$
when $|z|\leq\tau_{ij}$ (ii) $|s_{ij}(z)-z|\leq\tau_{ij}$. The adaptive
threshold $\tau_{ij}$ can be defined as $\tau_{ij}=C\left(\frac{1}{\sqrt{qN}}+\sqrt{\frac{\log N}{T}}\right)\sqrt{\hat{\emptyset}_{ij}},$
where $C>0$ is a sufficiently large constant and 
\[
\hat{\emptyset}_{ij}=\frac{1}{T_{h}}\sum_{t\leq T_{h}}(\underline{\hat{E}}_{i}(\hat{\theta})\underline{\hat{E}}_{j}(\hat{\theta})-(\mathbf{\tilde{\Sigma}}_{E(\theta)})_{ij})^{2},
\]
where $\underline{\hat{E}}_{i}(\hat{\theta})$ denotes the entries of $\hat{\underline{E}}(\hat{\theta})$. Using $\hat{\boldsymbol{\Sigma}}_{E(\theta)}$, $\varPhi_2$ is estimated by $\hat{\varPhi}_2 = qN\,\hat{\gamma}\,\hat{\boldsymbol{\Sigma}}_{E(\theta)}\,\hat{\gamma}'$.

The following theorem provides the mixed-frequency counterpart of Theorem A.1 in \cite{giglio2023prediction}. The regularity conditions parallel theirs, and the conclusion holds under the same rate restrictions.
\setcounter{thm}{4}
\begin{thm}
\label{thm:5}
Under the assumptions of Theorem \ref{thm:4}, if we further assume that

(i) $E_{t}(\theta)$ is stationary with $\mathbb{E}[E_{t}(\theta)]=0$ and
$\mathbf{\Sigma}_{E(\theta)}=Cov[E_{t}(\theta)]$ satisfying $C_{1}>\lambda_{1}(\mathbf{\Sigma}_{E(\theta)})\geq\lambda_{N}(\mathbf{\Sigma}_{E(\theta)})>C_{2}$
and $\min_{i,j}Var(E_{it}(\theta)E_{jt}(\theta))>C_{2}$ for some
constant $C_{1},C_{2}>0$, 

(ii) $E_{t}(\theta)$ has exponential tail, i.e., there exist $r_{1}>0$
and $C>0$, such that for any $s>0$ and $i\leq N,$ $P(|E_{it}(\theta)|>s)\leq\exp(-(s/C)^{r_{1}}).$

(iii) $E_{t}(\theta)$ is strong mixing, i.e., there exist positive
constants $r_{2}$ and $C$ such that for all $t\in\mathbb{Z}^{+},$
$\alpha(t)\leq\exp(-Ct^{r_{2}}),$ where $\alpha(T)=\sup_{A\in\mathcal{F}_{-\infty}^{0},Z\in\mathcal{F}_{T}^{\infty}}|P(A)P(Z)-P(AZ)|$
and $\mathcal{F}_{-\infty}^{0},\mathcal{F}_{T}^{\infty}$ are $\sigma-$algebras
generated by $\{E_{t}(\theta)\}_{-\infty\leq t_{h}\leq0}$, $\{E_{t}(\theta)\}_{T\leq t_{h}\leq\infty}$.

(iv) $(\log N)^{6(3r_{1}^{-1}+r_{2}^{-1}+1)}=o(T)$, $T=o(q^{2}N^{2}).$

Then $\mathbf{\hat{\Sigma}}_{E(\theta)}$ satisfies $||\hat{\Upsilon}||\,||\mathbf{\hat{\Sigma}}_{E(\theta)}-\mathbf{\Sigma}_{E(\theta)}||\lesssim_{\textrm{P}}m_{\mathfrak{q},N}\left(\frac{1}{\sqrt{qN}}+\sqrt{\frac{\log N}{T}}\right)^{1-\mathfrak{q}},$
where $m_{\mathfrak{q},N}=\max_{i\leq N}\sum_{j\leq N}|(\mathbf{\Sigma}_{E(\theta)})_{ij}|^{\mathfrak{q}}$ for some sparsity parameter $\mathfrak{q}\in[0,1)$ (as in \citealp{fan2013large}), and $q$ retains its meaning as the SsPCA screening fraction.
In addition, if $m_{\mathfrak{q},N}=\left(\frac{1}{\sqrt{qN}}+\sqrt{\frac{\log N}{T}}\right)^{1-\mathfrak{q}}=o(1),$
then $\hat{\varPhi}_{2}\overset{\textrm{P}}{\longrightarrow}\varPhi_{2}.$
\end{thm}

\subsubsection*{Proof of Theorem 5}
\begin{proof}
We apply Theorem 5 in \citet{fan2013large}, the error bound $||\mathbf{\hat{\Sigma}}_{E(\theta)}-\mathbf{\Sigma}_{E(\theta)}||$
can be established by showing that $||\hat{E}(\hat{\theta})-E(\theta)||_{\max}=o_{\textrm{p}}(1)$
and 
\[
\underset{i\leq N}{\max}T_{h}^{-1}\sum_{t}|E_{it}(\theta)-\widehat{E}_{it}(\hat{\theta})|^{2}=O_{\textrm{P}}\left(q^{-1}N^{-1}+\frac{\log N}{T}\right).
\]
These two estimates can be shown by combining Lemma \ref{lem:13}
(iii)(iv) with the aggregation function. If $m_{\mathfrak{q},N}=\left(\frac{1}{\sqrt{qN}}+\sqrt{\frac{\log N}{T}}\right)^{1-\mathfrak{q}}=o(1)$,
then $||\mathbf{\hat{\Sigma}}_{E(\theta)}-\mathbf{\Sigma}_{E(\theta)}||=o_{\textrm{P}}(1).$
With $||\hat{\varsigma}'_{(k)}\tilde{D}_{(k)}-\varsigma'_{(k)}D_{(k)}||\lesssim_{\textrm{P}}T^{-1/2}+q^{-1/2}N^{-1/2}$
from Lemma \ref{lem:11}(iv) and $\dot{\hat{\gamma}}=\sum_{k\leq K}\dot{\hat{\alpha}}_{(k)}\hat{\varsigma}'_{(k)}\tilde{D}_{(k)},$
rewrite the proof of equation $||\dot{\hat{\gamma}}E_{T}(\hat{\theta})-\sum_{k=1}^{K}\lambda_{(k)}^{-1/2}\alpha\iota_{(k)}\varsigma'_{(k)}D_{(k)}E_{T}(\hat{\theta})||\lesssim_{\textrm{P}}T^{-1}+q^{-1}N^{-1}$,
we have 
\begin{equation}
||\dot{\hat{\gamma}}-\underset{k\leq K}{\sum}\lambda_{(k)}^{-1/2}\alpha\iota_{(k)}\varsigma'_{(k)}D_{(k)}||\lesssim_{\textrm{P}}T^{-1/2}q^{-1/2}N^{-1/2}+q^{-1}N^{-1}.\label{eq:B30}
\end{equation}
The discrepancy in the rate between the preceding equation and (\ref{eq:B30})
stems from the distinction between Lemma \ref{lem:11}(iv) and (vi).
Recall that $\lambda_{(k)}^{-1/2}\varsigma'_{(k)}D_{(k)}$ corresponds
precisely to the $k$th row of $\Lambda^{-1}\Omega'\Psi$, the left
hand side of (\ref{eq:B30}) is equivalent to $||\dot{\hat{\gamma}}-\alpha L\Lambda^{-1}\Omega'\Psi||.$
Moreover, under the assumption $\mathbf{\Sigma}_{E(\theta)}=Cov[E_{t}(\theta)]$,
$\Pi_{33}$ equals $q^{-1}N^{-1}\Psi\mathbf{\Sigma}_{E(\theta)}\Psi'$.
Let $\dot{\tilde{\gamma}}$ represent $\alpha L\Lambda^{-1}\Omega'\Psi$,
then we have
\[
\hat{\varPhi}_{2}-\varPhi_{2}=qN\left(\dot{\hat{\gamma}}\mathbf{\hat{\Sigma}}_{E(\theta)}\dot{\hat{\gamma}}'-\dot{\tilde{\gamma}}\mathbf{\Sigma}_{E(\theta)}\dot{\tilde{\gamma}}'\right).
\]
Therefore, we have
\begin{equation}
||\hat{\varPhi}_{2}-\varPhi_{2}||\leq qN||\dot{\hat{\gamma}}(\mathbf{\hat{\Sigma}}_{E(\theta)}-\mathbf{\Sigma}_{E(\theta)})\dot{\hat{\gamma}}'||+||(\dot{\hat{\gamma}}-\dot{\tilde{\gamma}})\mathbf{\Sigma}_{E(\theta)}\dot{\hat{\gamma}}'||+||\dot{\tilde{\gamma}}\mathbf{\Sigma}_{E(\theta)}(\dot{\hat{\gamma}}-\dot{\tilde{\gamma}})'||.\label{eq:B31}
\end{equation}
Using the definition of $D_{(k)}$, $||\beta_{[I_{k}]}||\lesssim(qN)^{1/2}$,
and $\lambda_{(k)}\asymp qN$, we have $||D_{(k)}||\lesssim1$ and
thus $||\dot{\hat{\gamma}}||\lesssim q^{-1/2}N^{-1/2}.$ The MIDAS aggregation uses fixed, bounded, and normalized weights, with $\hat\theta$ consistent by Lemma \ref{lem:16}(iii). Combining this with $||\mathbf{\hat{\Sigma}}_{E(\theta)}-\mathbf{\Sigma}_{E(\theta)}||=o_{\textrm{P}}(1)$,
(\ref{eq:B30}), $||\mathbf{\Sigma}_{e}||\lesssim1$ from the assumption and $||\dot{\hat{\gamma}}||\lesssim q^{-1/2}N^{-1/2}$,
all three terms in (\ref{eq:B31}) are $o_{\textrm{P}}(1)$.
\end{proof}

\subsection{Proof of Proposition 1}
\begin{proof}
The proof of Proposition \ref{prop:1} builds on Theorem 3 of
\citet{bai2006confidence}, with two modifications: we work on the informative subset of size $N_0$ rather than the full sample, and we account for the rotational structure explicitly. We carry out the argument in the same-frequency setting; the extension to the mixed-frequency case follows because the leading asymptotic term depends only on the cross-sectional structure of the predictors $(\beta_i, e_{i,t_h})$ and is unaffected by the MIDAS aggregation $\sum_{j=0}^{J} b_j^f(\theta_f)$, which acts linearly on the factors with bounded weights summing to $\mathbb{I}_{K}$.

Focusing on SPCA, after subset selection of size $N_{0}$, standard calculations following \cite{bai2003inferential} yield, under the condition $\sqrt{N_{0}}/T\rightarrow0$,
\[
\sqrt{N_{0}}(\hat{\mathbf{f}}_{H}-H_{\textrm{SPCA}}\mathbf{f}_{H})=\hat{V}^{-1}\frac{1}{T}\hat{\mathbf{f}}_{H}'\mathbf{f}_{H}\frac{1}{\sqrt{N_{0}}}\stackrel[i=1]{N_{0}}{\sum}\beta_{i}e_{i,t_h}+o_{\textrm{P}}(1),
\]
where $H_{\textrm{SPCA}}=\frac{1}{N_{0}T}\hat{V}^{-1}\hat{\mathbf{f}}_{H}'\mathbf{f}_{H}\mathfrak{B'}\mathfrak{B}$
and $\hat{V}$ is the covariance matrix of estimated factor loadings. The probability limits of $\hat{V}$
and $\frac{1}{T}\hat{\mathbf{f}}_{H}'\mathbf{f}_{H}$ are established in \cite{bai2003inferential}.

We reformulate the expansion using the rotational matrix $H_{\textrm{SPCA}}$:
\[
\sqrt{N_{0}}(\hat{\mathbf{f}}_{H}-H_{\textrm{SPCA}}\mathbf{f}_{H})=H_{\textrm{SPCA}}(\frac{1}{N_{0}}\mathfrak{B'}\mathfrak{B})^{-1}\frac{1}{\sqrt{N_{0}}}\stackrel[i=1]{N_{0}}{\sum}\beta_{i}e_{i,t_h}+o_{\textrm{P}}(1).
\]
This formulation, also used by \cite{huang2022scaled} in the same-frequency setting, makes the rotational indeterminacy explicit and facilitates the MSFE calculation. The MSFE expression for SPCA in Proposition \ref{prop:1} then follows by substituting this expansion into the forecast error and taking the asymptotic variance.

The SsPCA case follows analogously, with $\mathfrak{B'}\mathfrak{B}$ replaced by $\mathfrak{B}'\mathbb{W}\mathfrak{B}$ in the rotational matrix, reflecting the scaling weights $\mathbb{W} = \operatorname{diag}(\phi_1^2, \ldots, \phi_{N_0}^2)$. The resulting MSFE expression has the form stated in Proposition \ref{prop:1}.
\end{proof}

\subsection{Proof of Proposition 2}
\begin{proof}
Given Lemma \ref{lem:14} and Lemma \ref{lem:15}, we first consider
the case of $|\phi_{i}|>|\psi_{i}|$ and $\phi_{i}\psi_{i}\geq0$
for all $i$. Without loss of generality, we can assume $\phi_{i}>\psi_{i}\geq0$.
Otherwise, we can transform the variables by replacing $\phi_{i}$
with $-\phi_{i}$ and $\psi_{i}$ with $-\psi_{i}$. Consider the
case that $N_{0}=2$, our objective is to show 
\[
\frac{\phi_{1}^{2}\phi_{1}\psi_{1}+\phi_{2}^{2}\phi_{2}\psi_{2}}{\phi_{1}^{2}(\phi_{1}^{2}-\psi_{1}^{2})+\phi_{2}^{2}(\phi_{2}^{2}-\psi_{2}^{2})}\leq\frac{\phi_{1}\psi_{1}+\phi_{2}\psi_{2}}{(\phi_{1}^{2}-\psi_{1}^{2})+(\phi_{2}^{2}-\psi_{2}^{2})}.
\]
Simple calculations show that the given inequality can be rewritten
as
\[
(\phi_{1}^{2}-\phi_{2}^{2})\left(\frac{\phi_{1}\psi_{1}}{\phi_{1}^{2}-\psi_{1}^{2}}-\frac{\phi_{2}\psi_{2}}{\phi_{2}^{2}-\psi_{2}^{2}}\right)\leq0.
\]
It suffices to confirm that $\frac{\phi\psi}{\phi^{2}-\psi^{2}}$
decreases as $\phi^{2}$ grows. Given that we normalize $x_{i,t_{h}}$,
$i\in N_{0}=\mid I_{0}\mid$, it follows that $\phi_{i}^{2}+\psi_{i}^{2}+\sigma_{i}^{2}=1$
for all $i$. Consequently, a larger $\phi^{2}$ results in a smaller
$\psi^{2}$, which in turn leads to a smaller $\frac{\phi\psi}{\phi^{2}-\psi^{2}}=\frac{z}{1-z^{2}}$
with $z=\psi/\phi$. Based on this fact, we have established the result
for $N_{0}=2$. Next, we proceed to prove the general case using induction.
Assume that the result holds for $N_{0}=N_{0}^{*}-1,$ that is,
\begin{equation}
\frac{\stackrel[i=1]{N_{0}^{*}-1}{\sum}\phi_{i}^{2}\phi_{i}\psi_{i}}{\stackrel[i=1]{N_{0}^{*}-1}{\sum}\phi_{i}^{2}(\phi_{i}^{2}-\psi_{i}^{2})}\leq\frac{\stackrel[i=1]{N_{0}^{*}-1}{\sum}\phi_{i}\psi_{i}}{\stackrel[i=1]{N_{0}^{*}-1}{\sum}(\phi_{i}^{2}-\psi_{i}^{2})}.\label{eq:B32}
\end{equation}
Consider the case where $N_{0}=N_{0}^{*}.$ Among $\phi_{i}$ for
$i=1,2,...,N_{0}^{*}$, there must be a largest $\phi$ (which may
not be unique). Without loss of generality, we can assume that $\phi_{N_{0}^{*}}$
is the largest. Otherwise, we can swap the positions of the largest
$\phi_{i}$ and $\phi_{N_{0}^{*}}$. From (\ref{eq:B32}), we obtain
\[
\frac{\stackrel[i=1]{N_{0}^{*}-1}{\sum}\phi_{i}^{2}\phi_{i}\psi_{i}}{\stackrel[i=1]{N_{0}^{*}-1}{\sum}\phi_{i}\psi_{i}}\leq\frac{\stackrel[i=1]{N_{0}^{*}-1}{\sum}\phi_{i}^{2}(\phi_{i}^{2}-\psi_{i}^{2})}{\stackrel[i=1]{N_{0}^{*}-1}{\sum}(\phi_{i}^{2}-\psi_{i}^{2})}\leq\phi_{N_{0}^{*}}^{2},
\]
where we temporarily define $a\equiv\frac{\sum_{i=1}^{N_{0}^{*}}\phi_{i}^{2}\phi_{i}\psi_{i}}{\sum_{i=1}^{N_{0}^{*}}\phi_{i}\psi_{i}}$,
$b\equiv\frac{\sum_{i=1}^{N_{0}^{*}}\phi_{i}^{2}(\phi_{i}^{2}-\psi_{i}^{2})}{\sum_{i=1}^{N_{0}^{*}}(\phi_{i}^{2}-\psi_{i}^{2})}$,
and the second inequality follows from the fact that $\phi_{N_{0}^{*}}^{2}$
is the largest. Note that we need to prove
\[
\frac{\stackrel[i=1]{N_{0}^{*}}{\sum}\phi_{i}^{2}\phi_{i}\psi_{i}}{\stackrel[i=1]{N_{0}^{*}}{\sum}\phi_{i}^{2}(\phi_{i}^{2}-\psi_{i}^{2})}\leq\frac{\stackrel[i=1]{N_{0}^{*}}{\sum}\phi_{i}\psi_{i}}{\stackrel[i=1]{N_{0}^{*}}{\sum}(\phi_{i}^{2}-\psi_{i}^{2})},
\]
which is equivalent to
\[
\frac{\stackrel[i=1]{N_{0}^{*}-1}{\sum}\phi_{i}^{2}\phi_{i}\psi_{i}+\phi_{N_{0}^{*}}^{2}\phi_{N_{0}^{*}}\psi_{N_{0}^{*}}}{\stackrel[i=1]{N_{0}^{*}-1}{\sum}\phi_{i}^{2}(\phi_{i}^{2}-\psi_{i}^{2})+\phi_{N_{0}^{*}}^{2}(\phi_{N_{0}^{*}}^{2}-\psi_{N_{0}^{*}}^{2})}\leq\frac{\stackrel[i=1]{N_{0}^{*}-1}{\sum}\phi_{i}\psi_{i}+\phi_{N_{0}^{*}}\psi_{N_{0}^{*}}}{\stackrel[i=1]{N_{0}^{*}-1}{\sum}(\phi_{i}^{2}-\psi_{i}^{2})+(\phi_{N_{0}^{*}}^{2}-\psi_{N_{0}^{*}}^{2})}.
\]
Based on the definition of $a$ and $b$, the above inequality is
identical to
\[
\frac{a\stackrel[i=1]{N_{0}^{*}-1}{\sum}\phi_{i}\psi_{i}+\phi_{N_{0}^{*}}^{2}\phi_{N_{0}^{*}}\psi_{N_{0}^{*}}}{b\stackrel[i=1]{N_{0}^{*}-1}{\sum}(\phi_{i}^{2}-\psi_{i}^{2})+\phi_{N_{0}^{*}}^{2}(\phi_{N_{0}^{*}}^{2}-\psi_{N_{0}^{*}}^{2})}\leq\frac{\stackrel[i=1]{N_{0}^{*}-1}{\sum}\phi_{i}\psi_{i}+\phi_{N_{0}^{*}}\psi_{N_{0}^{*}}}{\stackrel[i=1]{N_{0}^{*}-1}{\sum}(\phi_{i}^{2}-\psi_{i}^{2})+(\phi_{N_{0}^{*}}^{2}-\psi_{N_{0}^{*}}^{2})}.
\]
Since $a\leq b,$ it follows that
\[
\frac{a\stackrel[i=1]{N_{0}^{*}-1}{\sum}\phi_{i}\psi_{i}+\phi_{N_{0}^{*}}^{2}\phi_{N_{0}^{*}}\psi_{N_{0}^{*}}}{b\stackrel[i=1]{N_{0}^{*}-1}{\sum}(\phi_{i}^{2}-\psi_{i}^{2})+\phi_{N_{0}^{*}}^{2}(\phi_{N_{0}^{*}}^{2}-\psi_{N_{0}^{*}}^{2})}\leq\frac{b\stackrel[i=1]{N_{0}^{*}-1}{\sum}\phi_{i}\psi_{i}+\phi_{N_{0}^{*}}^{2}\phi_{N_{0}^{*}}\psi_{N_{0}^{*}}}{b\stackrel[i=1]{N_{0}^{*}-1}{\sum}(\phi_{i}^{2}-\psi_{i}^{2})+\phi_{N_{0}^{*}}^{2}(\phi_{N_{0}^{*}}^{2}-\psi_{N_{0}^{*}}^{2})}.
\]
Then we still need to prove
\[
\frac{b\stackrel[i=1]{N_{0}^{*}-1}{\sum}\phi_{i}\psi_{i}+\phi_{N_{0}^{*}}^{2}\phi_{N_{0}^{*}}\psi_{N_{0}^{*}}}{b\stackrel[i=1]{N_{0}^{*}-1}{\sum}(\phi_{i}^{2}-\psi_{i}^{2})+\phi_{N_{0}^{*}}^{2}(\phi_{N_{0}^{*}}^{2}-\psi_{N_{0}^{*}}^{2})}\leq\frac{\stackrel[i=1]{N_{0}^{*}-1}{\sum}\phi_{i}\psi_{i}+\phi_{N_{0}^{*}}\psi_{N_{0}^{*}}}{\stackrel[i=1]{N_{0}^{*}-1}{\sum}(\phi_{i}^{2}-\psi_{i}^{2})+(\phi_{N_{0}^{*}}^{2}-\psi_{N_{0}^{*}}^{2})}.
\]
Some basic calculations show that the above inequality is equivalent
to
\begin{equation}
\left(\phi_{N_{0}^{*}}^{2}-b\right)\left(\frac{\phi_{N_{0}^{*}}\psi_{N_{0}^{*}}}{\phi_{N_{0}^{*}}^{2}-\psi_{N_{0}^{*}}^{2}}-\frac{\stackrel[i=1]{N_{0}^{*}-1}{\sum}\phi_{i}\psi_{i}}{\stackrel[i=1]{N_{0}^{*}-1}{\sum}(\phi_{i}^{2}-\psi_{i}^{2})}\right)\leq0.\label{eq:B33}
\end{equation}
Let $z_{i}=\psi_{i}/\phi_{i}$. Since $\phi_{N_{0}^{*}}$
is the largest, which means for each $i$, 
\[
\frac{\phi_{N_{0}^{*}}\psi_{N_{0}^{*}}}{\phi_{N_{0}^{*}}^{2}-\psi_{N_{0}^{*}}^{2}}=\frac{z_{N_{0}^{*}}}{1-z_{N_{0}^{*}}^{2}}\leq\frac{z_{i}}{1-z_{i}^{2}}=\frac{\phi_{i}\psi_{i}}{\phi_{i}^{2}-\psi_{i}^{2}}.
\]
By the mediant inequality, we have
\[
\frac{\phi_{N_{0}^{*}}\psi_{N_{0}^{*}}}{\phi_{N_{0}^{*}}^{2}-\psi_{N_{0}^{*}}^{2}}-\frac{\stackrel[i=1]{N_{0}^{*}-1}{\sum}\phi_{i}\psi_{i}}{\stackrel[i=1]{N_{0}^{*}-1}{\sum}(\phi_{i}^{2}-\psi_{i}^{2})}\leq0.
\]
Since $\phi_{N_{0}^{*}}^{2}\geq b$, (\ref{eq:B33}) holds. This
completes the proof for the first case. The remaining three cases
can be proved by the same argument, so their details are omitted.

Now, consider the first two cases, which indicate that $\sum_{i=1}^{N_{0}^{*}}\phi_{i}^{2}>\sum_{i=1}^{N_{0}^{*}}\psi_{i}^{2}$
and $|\mathsection^{o}|\leq|\mathsection|$. Based on the formulas
for MSFE and $\mathcal{O}$, we observe that MSFE is a decreasing
function of $\mathcal{O}$ and that $\mathcal{O}^{2}$ is a decreasing
function of $|\mathsection|$. To minimize MSFE, $|\mathsection|$
should be as small as possible. Consequently, SsPCA outperforms SPCA
in these two cases. In the remaining two cases, where $\sum_{i=1}^{N_{0}^{*}}\phi_{i}^{2}<\sum_{i=1}^{N_{0}^{*}}\psi_{i}^{2}$
and $|\mathsection^{o}|\geq|\mathsection|$, an examination of the
formulas shows that MSFE remains a decreasing function of $\mathcal{O}^{2}$,
but $\mathcal{O}^{2}$ is now an increasing function of $|\mathsection|$.
Given this, it follows that SsPCA also outperforms SPCA in the latter
two cases, thereby completing the proof of Proposition \ref{prop:2}.
\end{proof}

\subsection{Proof of Propositions 3 and 4}
\begin{proof}
For any orthogonal matrix $\Gamma\in\mathbb{R}^{N\times N}$, the
estimator based on sPCA for $\Gamma \mathbf{x}_H$ are identical to those based
on $\mathbf{x}_H$. Therefore, without loss of generality (WLOG), we can assume $\beta=(\lambda^{1/2},0,...,0)',$
where $\lambda=||\beta||^{2},$ and this assumption does not affect
$\mathbf{z}_H$.

Then, we can write $\hat{\Upsilon}\mathbf{x}_H$ in the following form:
\begin{equation}
\hat{\Upsilon}\mathbf{x}_H=\hat{\Upsilon}\beta \mathbf{f}_H+\hat{\Upsilon}\mathbf{e}_H=\hat{\Upsilon}\beta \mathbf{f}_H+\hat{\Upsilon}\rho \mathbf{z}_{H,1}=\left(\begin{array}{c}
\hat{\Upsilon}_{1}\sqrt{\lambda}\mathbf{f}_H+\hat{\Upsilon}_{1}\rho_{1}\mathbf{z}_{H,1}\\
\hat{\Upsilon}_{2}\rho_{2}\mathbf{z}_{H,1}
\end{array}\right),\label{eq:B34}
\end{equation}
where $\rho_{1},\hat{\Upsilon}_{1}$ are the first rows of $\rho$
and $\hat{\Upsilon}$, respectively, while $\rho_{2}$ and $\hat{\Upsilon}_{2}$
consist of their remaining rows. Accordingly, we write the first left
singular vector of $\hat{\Upsilon}\mathbf{x}_H$ as $\hat{\varsigma}=(\hat{\varsigma}_{1},\hat{\varsigma}'_{2})',$
where $\hat{\varsigma}_{1}$ is the first element of $\hat{\varsigma}$
and $\hat{\varsigma}_{2}$ is a vector of the remaining $N-1$ entries
of $\hat{\varsigma}$, write $\hat{\xi}$ as the first right singular
vector of $\hat{\Upsilon}\mathbf{x}_H$, and denote the first singular value
as $\sqrt{T\hat{\lambda}}$. By straightforward algebra we have
\begin{equation}
\begin{array}{cc}
\hat{\varsigma}_{1}=\frac{\hat{\Upsilon}_{1}(\sqrt{\lambda}\mathbf{f}_H+\rho_{1}\mathbf{z}_{H,1})\hat{\xi}}{\sqrt{T\hat{\lambda}}}, & \hat{\varsigma}_{2}=\frac{\hat{\Upsilon}_{2}\rho_{2}\mathbf{z}_{H,1}\hat{\xi}}{\sqrt{T\hat{\lambda}}}\end{array}.\label{eq:B35}
\end{equation}
Note that the scaling coefficients $\hat{\Upsilon}_{1}$ and $\hat{\Upsilon}_{2}$
only serve to rescale and do not affect the direction of the left
and right singular vectors. Because the entries of $\mathbf{f}_H$ are i.i.d.
$N(0,1)$, we have large deviation inequality $|T_{h}^{-1}\mathbf{ff}'-1|\lesssim_{\textrm{P}}T^{-1/2}.$
Furthermore, by Weyl\textquoteright s inequality, this result also
implies that $||\mathbf{f}||-T_{h}^{1/2}\lesssim_{\textrm{P}}1$.

Similarly, we can obtain $|T_{h}^{-1}\rho_{1}\rho'_{1}-1|\lesssim_{\textrm{P}}T^{-1/2}$,
$||\rho_{1}||-T_{h}^{1/2}\lesssim_{\textrm{P}}1$, $||\hat{\Upsilon}_{1}||\asymp_{\textrm{P}}1$
and $||\hat{\Upsilon}_{2}||\lesssim_{\textrm{P}}T^{-1/2}$. The last
conditions follow from the original stated assumption. Moreover, applying
Lemma A.1 from \cite{wang2017asymptotics} and the assumption about
$||\hat{\Upsilon}||$, we have $||N^{-1}(\hat{\Upsilon}\mathbf{e}_H)'\hat{\Upsilon}\mathbf{e}_H-(\hat{\Upsilon}\mathbf{z}_{H,1})'\hat{\Upsilon}\mathbf{z}_{H,1}||\leq||\hat{\Upsilon}||^{2}\,||\mathbf{z}_{H,1}||^{2}||N^{-1}\rho'\rho-\mathbb{I}_{T_{h}}||\lesssim_{\textrm{P}}\sqrt{T/N}.$
Next, by performing straightforward calculations using the previously
established inequalities, we obtain
\[
\begin{split}&\parallel\frac{(\hat{\Upsilon}_{1}\mathbf{f}_H)'\rho_{1}\hat{\Upsilon}_{1}\mathbf{z}_{H,1}+(\hat{\Upsilon}_{1}\mathbf{z}_{H,1})'\rho'_{1}\hat{\Upsilon}_{1}\mathbf{f}_H}{T_{h}\lambda}+\frac{(\hat{\Upsilon}\mathbf{e}_H)'\hat{\Upsilon}\mathbf{e}_H-N(\hat{\Upsilon}\mathbf{z}_{H,1})'\hat{\Upsilon}\mathbf{z}_{H,1}}{T_{h}\lambda}\parallel\\&\lesssim_{\textrm{P}}\frac{1}{\sqrt{\lambda}}+\frac{\sqrt{NT}}{T\lambda}\lesssim_{\textrm{P}}\frac{1}{\sqrt{\lambda}}.\end{split}
\]
Combine with (\ref{eq:B34}) and the assumption about $||\hat{\Upsilon}||$,
we have
\begin{equation}
||\hat{\Upsilon}||^{2}\parallel\frac{\mathbf{x}_H'\mathbf{x}_H}{T_{h}\lambda}-\frac{\mathbf{f}_H'\mathbf{f}_H}{T_{h}}-\frac{N\mathbf{z}'_{H,1}\mathbf{z}_{H,1}}{T_{h}\lambda}\parallel\lesssim_{\textrm{P}}\frac{1}{\sqrt{\lambda}}+\frac{1}{2\sqrt{\lambda T}}+\frac{1}{T\sqrt{\lambda}}\lesssim_{\textrm{P}}\frac{1}{\sqrt{\lambda}}.\label{eq:B36}
\end{equation}
Let $\eta$ represent the leading eigenvector of the matrix $\mathfrak{M}:=T_{h}^{-1}(\Upsilon \mathbf{f}_H)'\Upsilon \mathbf{f}_H+\delta(\Upsilon \mathbf{z}_{H,1})'\Upsilon \mathbf{z}_{H,1}.$
Given the assumption that $N/(T^{2}\lambda)\rightarrow\delta,$ $(\lambda_1(\mathfrak{M}) - \lambda_2(\mathfrak{M}))/\lambda_1(\mathfrak{M}) \gtrsim_P 1$
and (\ref{eq:B36}), by the sin-theta theorem in \citet{davis1970rotation},
we have $||\mathbb{P}_{\eta}-\mathbb{P}_{\hat{\xi}}||=||\mathbb{P}_{\eta}-\mathbb{P}_{\hat{\mathbf{f}}_H'}||=o_{\textrm{P}}(1).$

Under the condition that $\mathbf{z}'_{H,1}\mathbf{z}_{H,1}=\mathbb{I}_{T_{h}}$, the eigenvalues
of $\mathfrak{M}$ are expressed as
\begin{equation}
\lambda_{i}=\begin{cases}
T_{h}^{-1}\Upsilon \mathbf{f}_H(\Upsilon \mathbf{f}_H)'+\delta & i=1;\\
\delta & i\geq2.
\end{cases}
\end{equation}
and the leading eigenvector is $(\Upsilon \mathbf{f}_H)'/||\Upsilon \mathbf{f}_H||.$ Since
the largest eigenvalue of $(\hat{\Upsilon}\mathbf{x}_H)'(\hat{\Upsilon}\mathbf{x}_H)/(T_{h}\lambda)$
is $\hat{\lambda}/\lambda$, with its corresponding eigenvector $\hat{\xi}$,
equation (\ref{eq:B36}) and Weyl's theorem yield that
\begin{equation}
\frac{\hat{\lambda}}{\lambda}=\frac{\Upsilon \mathbf{f}_H(\Upsilon \mathbf{f}_H)'}{T_{h}}+\frac{N}{T_{h}\lambda}+O_{\textrm{P}}(\frac{1}{\sqrt{\lambda}})=1+\delta+o_{\textrm{P}}(1),\label{eq:B38}
\end{equation}
and the sin-theta theorem hints that
\begin{equation}
|\mathbb{P}_{\mathbf{f}_H'}-\mathbb{P}_{\hat{\xi}}||=||(\Upsilon \mathbf{f}_H)'[\Upsilon \mathbf{f}_H(\Upsilon \mathbf{f}_H)']^{-1}\Upsilon \mathbf{f}_H-\hat{\xi}\hat{\xi}'||=o_{\textrm{P}}(1).\label{eq:B39}
\end{equation}
Additionally, (\ref{eq:B39}) indicates that $[\Upsilon \mathbf{f}_H(\Upsilon \mathbf{f}_H)']^{-1}(\Upsilon \mathbf{f}_H\hat{\xi})^{2}=\hat{\xi}'(\Upsilon \mathbf{f}_H)'[\Upsilon \mathbf{f}_H(\Upsilon \mathbf{f}_H)']^{-1}\Upsilon \mathbf{f}_H\hat{\xi}=1+o_{\textrm{P}}(1).$
Together with $|T_{h}^{-1}(\Upsilon \mathbf{f}_H)'\Upsilon \mathbf{f}_H-1|\lesssim T^{-1/2}$,
and given that the sign of $\hat{\xi}$ does not impact the estimator
$\hat{y}_{T+h}$, we can appropriately choose $\hat{\xi}$ and its
relative scaling coefficients $\hat{\Upsilon}$ such that
\begin{equation}
\frac{\hat{\Upsilon}\mathbf{f}_H\hat{\xi}}{\sqrt{T_{h}}}-1=o_{\textrm{P}}(1).\label{eq:B40}
\end{equation}
Since we can express
\[
\hat{y}_{T+h}^{\textrm{sPCA}} = \hat{\alpha}\,\hat{F}_{T,\textrm{sPCA}}(\hat{\theta}) = \hat{\alpha}\sum_{j=0}^{J} b_{j}(\hat{\theta})\,\hat{\varsigma}'\mathbf{x}_{T-j/m}=\hat{\alpha}\sum_{j=0}^{J} b_{j}(\hat{\theta})\,\hat{\varsigma}'(\beta\mathbf{f}_{T-j/m}+\mathbf{e}_{T-j/m})
\]
Consequently,
\[
\hat{y}_{T+h}^{\text{sPCA}} = \hat{\alpha}\sum_{j=0}^{J} b_{j}(\hat{\theta})\,\hat{\varsigma}'\mathbf{x}_{T-j/m}=\frac{Y\hat{\xi}\sum_{j=0}^{J} b_{j}(\hat{\theta})\,\hat{\varsigma}'\mathbf{x}_{T-j/m}}{\sqrt{T_h\hat{\lambda}}}=\alpha \frac{F(\theta)\hat{\xi}\sum_{j=0}^{J} b_{j}(\hat{\theta})\,\hat{\varsigma}'\mathbf{x}_{T-j/m}}{\sqrt{T_h\hat{\lambda}}}
\]
\begin{equation}
 = \alpha\,\frac{\hat{\varsigma}'\beta\, F_T(\hat{\theta}) + \hat{\varsigma}'E_T(\hat{\theta})}{\sqrt{\hat{\lambda}}}\bigl(1 + o_\textrm{P}(1)\bigr).\label{eq:B41}
\end{equation}
Using (\ref{eq:B35}), we have
\[
\frac{\hat{\varsigma}'\beta}{\sqrt{\hat{\lambda}}}=\frac{\sqrt{\lambda}\hat{\varsigma}_{1}}{\sqrt{\hat{\lambda}}}=\frac{\lambda}{\hat{\lambda}}\frac{\hat{\Upsilon}_{1}(\mathbf{f}_H+\lambda^{-1/2}\rho_{1}\mathbf{z}_{H,1})\hat{\xi}}{\sqrt{T_{h}}}=\frac{\lambda}{\hat{\lambda}}\left(\frac{\hat{\Upsilon}_{1}\mathbf{f}_H\hat{\xi}}{\sqrt{T_{h}}}+\frac{\rho_{1}\hat{\Upsilon}_{1}\mathbf{z}_{H,1}\hat{\xi}}{\sqrt{T_{h}\lambda}}\right).
\]
Using (\ref{eq:B38}), (\ref{eq:B40}), $||\mathbf{z}_{H,1}||\leq1$, $||\hat{\Upsilon}_{1}||\lesssim_{\textrm{P}}1$
and $||\rho_{1}||\lesssim_{\textrm{P}}\sqrt{T}$, it follows that
\begin{equation}
\frac{\hat{\varsigma}'\beta}{\sqrt{\hat{\lambda}}}\overset{\textrm{P}}{\longrightarrow}\frac{1}{1+\delta}.\label{eq:B42}
\end{equation}
Furthermore, since $\textrm{Cov}(e_{s},e_{t})=0$ for $s\neq t$,
it follows that $E_{T}(\hat{\theta})$ is independent of $\hat{\varsigma}$, leading
to $\hat{\varsigma}'E_{T}(\hat{\theta})=O_{\textrm{P}}(1)$. Combining this with
(\ref{eq:B41}), (\ref{eq:B42}), and $\hat{\theta} \xrightarrow{\textrm{P}} \theta$ (Lemma \ref{lem:16}(iii)), we obtain $\hat{y}^{\textrm{sPCA}}_{T+h}\overset{\textrm{P}}{\longrightarrow}\frac{\alpha F_{T}(\theta)}{1+\delta}=(1+\delta)^{-1}\mathbb{E}_{T}(y_{T+h}).$
\end{proof}

\subsection{Proof of Propositions 5}

Before proving this proposition, we need to define the degrees of
freedom associated with boosting when it stops at iteration $m$,
given by
\[
d.f._{(m)}=\textrm{trace}(\mathcal{A}_{(m)}),
\]
 where 
\[
\mathcal{A}_{(m)}=\mathcal{A}_{(m-1)}+\nu\mathbb{P}_{(m)}\left(\mathbb{I}_{T}-\mathcal{A}_{(m-1)}\right)=\mathbb{I}_{T}-\stackrel[j=0]{m}{\prod}\left(\mathbb{I}_{T}-\nu\mathbb{P}_{(j)}\right).
\]
In component-wise boosting, $\mathbb{P}_{(m)} = \underline F(\theta)_{\hat i_m}\bigl(\underline F(\theta)'_{\hat i_m}\underline F(\theta)_{\hat i_m}\bigr)^{-1}\underline F(\theta)'_{\hat i_m}$ is the projection matrix formed by the $\hat i_m$th aggregated factor $\underline F(\theta)_{\hat i_m}$; in block-wise boosting, $\mathbb{P}_{(m)}$ is defined analogously as the projection onto the block of aggregated factors selected at step $m$, and the subsequent argument carries over with appropriate dimensional adjustments ($m=1,2,...,M)$. The staring value $\mathcal{A}_{(0)}=\frac{1}{\nu}\mathbb{P}_{(0)}=\imath_{T}\imath'_{T}/T,$
where $\imath_{T}$ is a $T\times1$ vector of 1's. Also note that
the vector of fitted values at stage $m$ is $\tilde{\Phi}_{(m)}=\mathcal{A}_{(m)}Y$.
\begin{proof}
The fitted value $\hat{\Phi}_{(M)}$ equals $\mathcal{\hat{A}}_{(M)}\overline{Y}$
and $\tilde{\Phi}_{(M)}=\mathcal{A}_{(M)}\overline{Y}:$
\[
\frac{1}{T}\stackrel[t=1]{T}{\sum}|\hat{\Phi}_{(M),\hat{\underline{F}}(\hat{\theta})}-\tilde{\Phi}_{(M),\underline{F}(\theta)}|^{2}=\frac{1}{T}||\hat{\Phi}_{(M)}-\tilde{\Phi}_{(M)}||^{2}\leq||\mathcal{\hat{A}}_{(M)}-\mathcal{A}_{(M)}||^{2}\left(||\overline{Y}||^{2}/T\right).
\]
Note $||\overline{Y}||^{2}/T=\frac{1}{T}\sum_{t=1}^{T}||y_{t}||^{2}=O_{\textrm{P}}(1)$,
and
\[
\begin{split}\mathcal{\hat{A}}_{(M)}-\mathcal{A}_{(M)}&=\stackrel[m=1]{M}{\prod}\hat{z}_{(m)}-\stackrel[m=1]{M}{\prod}z_{(m)}\\&=\left(\hat{z}_{(1)}-z_{(1)}\right)Z_{(1)}+\mathcal{A}_{(1)}\left(\hat{z}_{(2)}-z_{(2)}\right)Z_{(2)}+...+\mathcal{A}_{(M-1)}\left(\hat{z}_{(M)}-z_{(M)}\right)Z_{(M)},\end{split}
\]
where $\hat{z}_{(m)}=\mathbb{I}_{T}-\hat{\mathbb{P}}_{(m)}$, $z_{(m)}=\mathbb{I}_{T}-\mathbb{P}_{(m)},$
and $Z_{(m)}=\prod_{j=m+1}^{M}\hat{z}_{(j)}$. But $z_{(j)}$ and
$\hat{z}_{(j)}$ are projection matrices whose largest eigenvalue
is one, and hence $||z_{(j)}||\leq1$ and $||\hat{z}_{(j)}||\leq1$.
It follows that $||Z_{(m)}||\leq1$ and $||\mathcal{A}_{(m)}||\leq1$
for all $m$. Moreover, $\left(\hat{z}_{(j)}-z_{(m)}\right)=\mathbb{P}_{(m)}-\hat{\mathbb{P}}_{(m)}$
and $||\mathbb{P}_{(m)}-\hat{\mathbb{P}}_{(m)}||\lesssim_{\textrm{P}}q^{-1/2}N^{-1/2}+T^{-1}$
by Theorem \ref{thm:1}. We have that
\[
\begin{split}||\mathcal{A}_{(j-1)}\left(\hat{z}_{(j)}-z_{(j)}\right)Z_{(j)}||&\leq||\mathcal{A}_{(j-1)}||\,||\hat{z}_{(j)}-z_{(j)}||\,||Z_{(j)}||\\&\leq||\hat{z}_{(j)}-z_{(j)}||=O_{\textrm{P}}(q^{-1/2}N^{-1/2}+T^{-1}),\end{split}
\]
and $||\mathcal{\hat{A}}_{(M)}-\mathcal{A}_{(M)}||\lesssim_{\textrm{P}}M(q^{-1/2}N^{-1/2}+T^{-1}).$
Therefore, if $M(q^{-1/2}N^{-1/2}+T^{-1})\rightarrow0$, then $\frac{1}{T}\sum_{t=1}^{T}|\hat{\Phi}_{(M),\underline{\hat{F}}(\hat{\theta})}-\tilde{\Phi}_{(M),\underline{F}(\theta)}|^{2}\overset{\textrm{P}}{\longrightarrow}0.$

Next, $\tilde{\Phi}_{(M)}=\tilde{\Phi}_{(0)}+\underline{F}(\theta)'\mathcal{\tilde{B}}_{M}$
and $\hat{\Phi}_{(M)}=\hat{\Phi}_{(0)}+\hat{\underline{F}}(\hat{\theta})'\mathcal{\hat{B}}_{M}$,
where $\tilde{\Phi}_{(0)}=\hat{\Phi}_{(0)}=mean(\bar{Y})$. Note that
(for $M\geq1$) 
\[
\mathcal{\tilde{B}}_{M}=\nu\left(\tilde{b}_{\hat{i}_{1}}^{(1)}+\tilde{b}_{\hat{i}_{2}}^{(2)}+...+\tilde{b}_{\hat{i}_{M}}^{(M)}\right),\quad\mathcal{\hat{B}}_{M}=\nu\left(\hat{b}_{\hat{i}_{1}}^{(1)}+\hat{b}_{\hat{i}_{2}}^{(2)}+...+\hat{b}_{\hat{i}_{M}}^{(M)}\right),
\]
where $\tilde{b}_{\hat{i}_{m}}^{(m)}=\left(\underline{F}(\theta)'_{\hat{i}_{m}}\underline{F}(\theta)_{\hat{i}_{m}}\right)^{-1}\underline{F}(\theta)'_{\hat{i}_{m}}\overline{Y}$
and $\hat{b}_{\hat{i}_{m}}^{(m)}=\left(\underline{\hat{F}}(\hat{\theta})'_{\hat{i}_{m}}\underline{\hat{F}}(\hat{\theta})_{\hat{i}_{m}}\right)^{-1}\underline{\hat{F}}(\hat{\theta})'_{\hat{i}_{m}}\overline{Y}$,
$m=1,2,...,M$.

Consequently,
\[
|\hat{\Phi}_{(M),\underline{\hat{F}}(\hat{\theta})}-\tilde{\Phi}_{(M),\underline{F}(\theta)}|\leq|\nu|\,||\underline{F}(\theta)||\,\left[||\hat{b}_{\hat{i}_{1}}^{(1)}-\tilde{b}_{\hat{i}_{1}}^{(1)}||+||\hat{b}_{\hat{i}_{2}}^{(2)}-\tilde{b}_{\hat{i}_{2}}^{(2)}||+...+||\hat{b}_{\hat{i}_{M}}^{(M)}-\tilde{b}_{\hat{i}_{M}}^{(M)}||\right]
\]
By using Lemma A1(ii) of \cite{bai2009boosting}, we have 
\[
||\hat{b}_{\hat{i}_{m}}^{(m)}-\tilde{b}_{\hat{i}_{m}}^{(m)}||\leq\parallel\left(\underline{\hat{F}}(\hat{\theta})'_{\hat{i}_{m}}\underline{\hat{F}}(\hat{\theta})_{\hat{i}_{m}}\right)^{-1}\underline{\hat{F}}(\hat{\theta})'_{\hat{i}_{m}}-\left(\underline{F}(\theta)'_{\hat{i}_{m}}\underline{F}(\theta)_{\hat{i}_{m}}\right)^{-1}\underline{F}(\theta)'_{\hat{i}_{m}}||\,||\overline{Y}||
\]
\[
\leq O_{\textrm{P}}(q^{-1/2}N^{-1/2}+T^{-1})T^{-1/2}||\overline{Y}||=O_{\textrm{P}}(q^{-1/2}N^{-1/2}+T^{-1}).
\]
It follows that
\[
|\hat{\Phi}_{(M),\underline{\hat{F}}(\hat{\theta})}-\tilde{\Phi}_{(M),\underline{F}(\theta)}|\leq|\nu|\,||\underline{F}(\theta)||MO_{\textrm{P}}(q^{-1/2}N^{-1/2}+T^{-1}),
\]
and the above converges to zero if $M(q^{-1/2}N^{-1/2}+T^{-1})\rightarrow0$.
\end{proof}

\section{Technical Lemmas and Their Proofs\label{Appendix C}}

\renewcommand{\theequation}{C\arabic{equation}}
\setcounter{equation}{0} 

For simplicity, we also assume in the following lemmas that $\mathbf{\Sigma}_{f}=\mathbb{I}_{K}$.
Furthermore, with the exception of Lemma \ref{thm:2}, we assume that $\hat{K}=\tilde{K}$
and $\hat{I}_{k}=I_{k}$ for $k\leq\hat{K}$, which hold with the
probability converges to one as we will show in Lemma \ref{thm:2}.
\begin{lem}
\label{lem:1}The singular vectors $\hat{\xi}_{(k)}$s in Algorithm
\ref{Algorithm 1} satisfy $\hat{\xi}'_{(j)}\hat{\xi}_{(k)}=\delta_{jk}$
for $j,k\leq\hat{K}$.
\end{lem}
\begin{proof}
If $j=k$, this result holds from the definition of $\hat{\xi}_{(k)}$.
If $j<k$, recall that $\tilde{\mathbf{x}}_{(k)}$ is defined in (\ref{eq:B3})
and $\hat{\xi}_{(k)}$ is the first right singular vector of $\tilde{\mathbf{x}}_{(k)}$,
we have 
\[
\tilde{\mathbf{x}}_{(k)}=(\hat{\Upsilon}\mathbf{x})_{[\hat{I}_{k}]}\underset{i<k}{\prod}\left(\mathbb{I}_{T}-\hat{\xi}{}_{(i)}\hat{\xi}'_{(i)}\right)\quad and\quad\hat{\xi}_{(k)}=\underset{v\in\mathbb{R}^{T}}{\arg\max}\frac{\parallel\tilde{\mathbf{x}}_{(k)}v\parallel}{\parallel v\parallel}.
\]

\noindent if $\hat{\xi}'_{(k)}\hat{\xi}_{(j)}=c_{0}\neq0$ for some
$j<k$, then
\begin{equation}
\parallel\tilde{\mathbf{x}}_{(k)}(\hat{\xi}_{(k)}-c_{0}\hat{\xi}_{(j)})\parallel=\parallel\tilde{\mathbf{x}}_{(k)}\hat{\xi}_{(k)}-c_{0}\hat{\xi}_{(k)}\hat{\xi}_{(j)})\parallel=\parallel\tilde{\mathbf{x}}_{(k)}\hat{\xi}_{(k)}\parallel,\label{eq:C1}
\end{equation}

\noindent because the definition of $\tilde{\mathbf{x}}_{(k)}$ hints that
$\tilde{\mathbf{x}}_{(k)}\hat{\xi}_{(k)}=0$ for $j<k$. Alternatively, since
$\hat{\xi}'_{(j)}\hat{\xi}_{(k)}=c_{0}\neq0$, we have $(\hat{\xi}_{(k)}-c_{0}\hat{\xi}_{(j)})'\hat{\xi}_{(j)}=0$,
and consequently,
\begin{equation}
\parallel\hat{\xi}_{(k)}\parallel^{2}=\parallel\hat{\xi}_{(k)}-c_{0}\hat{\xi}_{(j)}\parallel^{2}+\parallel c_{0}\hat{\xi}_{(j)}\parallel^{2}>\parallel\hat{\xi}_{(k)}-c_{0}\hat{\xi}_{(j)}\parallel^{2}.\label{eq:C2}
\end{equation}
Obviously, if $\parallel\tilde{\mathbf{x}}_{(k)}\parallel=0,$ the SsPCA procedure
will terminate. Thus we have $\parallel\tilde{\mathbf{x}}_{(k)}\parallel>0$
for $k\leq\hat{K}$.

\noindent Combined with (\ref{eq:C1}) and (\ref{eq:C2}), we have
\[
\parallel\tilde{\mathbf{x}}_{(k)}\parallel=\frac{\parallel\tilde{\mathbf{x}}_{(k)}\hat{\xi}_{(k)}\parallel}{\parallel\hat{\xi}_{(k)}\parallel}\leq\frac{\parallel\tilde{\mathbf{x}}_{(k)}(\hat{\xi}_{(k)}-c_{0}\hat{\xi}_{(j)})\parallel}{\parallel\hat{\xi}_{(k)}-c_{0}\hat{\xi}_{(j)}\parallel},
\]
which contradicts with the definition of $\hat{\xi}_{(k)}$. Therefore,
$\hat{\xi}'_{(j)}\hat{\xi}_{(k)}=0$ for $j<k$.
\end{proof}
\begin{lem}
\label{lem:2}Under assumptions of Theorem \ref{thm:1}, for $k\leq\tilde{K}$,
$I_{k}$, $\tilde{K}$ and $(\hat{\Upsilon}\beta)_{(k)}$ satisfy

(i) $P(\widehat{I}_{k}=I_{k})\rightarrow1.$

(ii) $||\tilde{\mathbf{x}}_{(k)}-(\hat{\Upsilon}\beta)_{(k)}\mathbf{f}||\lesssim_{\textrm{P}}q^{1/2}N^{1/2}+T^{1/2}+T^{-1/2}q^{1/2}N^{1/2}.$

(iii) $|\hat{\lambda}_{(k)}^{1/2}/\lambda_{(k)}^{1/2}-1|\lesssim_{\textrm{P}}q^{-1/2}N^{-1/2}+T^{-1/2}$,
and $\hat{\lambda}_{(k)}\asymp_{\textrm{P}}\lambda_{(k)}\asymp qN.$

(iv) $||\mathbb{P}_{\hat{\mathbf{f}}'_{(k)}}-T_{h}^{-1}\mathbf{f}'\mathbb{P}_{\iota_{(k)}}\mathbf{f}||\asymp||T_{h}^{-1/2}\mathbf{f}\hat{\xi}_{(k)}-\iota_{(k)}||\lesssim_{\textrm{P}}q^{-1/2}N^{-1/2}+T^{-1/2}.$

(v) $P(\hat{K}=\widetilde{K})\rightarrow1.$

\noindent For $k\leq\tilde{K}+1$, we have

(vi) \textup{$||T_{h}^{-1}\hat{\Upsilon}\mathbf{x}\prod_{j=1}^{k-1}\mathbf{y}'_{(k)}-\hat{\Upsilon}\beta\prod_{j=1}^{k-1}\mathbb{\mathbb{M}}_{\iota_{(j)}}\alpha'||_{\max}\lesssim_{\textrm{P}}(\log NT)^{1/2}(q^{-1/2}N^{-1/2}+T^{-1/2}).$}
\end{lem}
\begin{proof}
We prove (i)\textendash (iv) by induction. To begin, we show that
these hold when $k=1$:

(i) Note that $\widehat{I}_{1}$ is selected based on $T_{h}^{-1}\hat{\Upsilon}\mathbf{xy}'$
and $I_{1}$ is chosen based on $\hat{\Upsilon}\beta\alpha'$. With
simple algebra, we derive 
\[
T_{h}^{-1}\hat{\Upsilon}\mathbf{xy}'-\hat{\Upsilon}\beta\alpha'=\hat{\Upsilon}\beta(T_{h}^{-1}\mathbf{ff}'-\mathbb{I}_{K})\alpha'+T_{h}^{-1}\hat{\Upsilon}\beta\mathbf{f}\epsilon'+T_{h}^{-1}\hat{\Upsilon}\mathbf{ef}'\alpha'+T_{h}^{-1}\hat{\Upsilon}\mathbf{e}\epsilon'.
\]
With Assumptions \ref{assu:1}, \ref{assu:2}, \ref{assu:4}, and
Lemma \ref{lem:16} (i), we obtain the following bound:
\[
\begin{array}{c}
||T_{h}^{-1}\hat{\Upsilon}\mathbf{xy}'-\hat{\Upsilon}\beta\alpha'||_{\max}\\
\lesssim||\hat{\Upsilon}||\,||\beta||_{\max}||T_{h}^{-1}\mathbf{ff}'-\mathbb{I}_{K}||\,||\alpha||+T_{h}^{-1}||\hat{\Upsilon}||\,||\beta||_{\max}||F\epsilon'||\\
+T_{h}^{-1}||\hat{\Upsilon}||\,||\mathbf{ef}'||_{\max}||\alpha||+T_{h}^{-1}||\hat{\Upsilon}||\,||\mathbf{e}\epsilon'||_{\max}\\
\lesssim||\beta||_{\max}^{2}||T_{h}^{-1}\mathbf{ff}'-\mathbb{I}_{K}||\,||\alpha||^{2}+||\varkappa||\,||\beta||_{\max}||T_{h}^{-1}\mathbf{ff}'-\mathbb{I}_{K}||\,||\alpha||\\
+T_{h}^{-1}||\beta||_{\max}^{2}||\mathbf{f}\epsilon'||\,||\alpha||+T_{h}^{-1}||\varkappa||\,||\beta||_{\max}||\mathbf{f}\epsilon'||\\
+T_{h}^{-1}||\beta||_{\max}||\mathbf{ef}'||_{\max}||\alpha||^{2}+T_{h}^{-1}||\varkappa||\,||\mathbf{ef}'||_{\max}||\alpha||\\
+T_{h}^{-1}||\beta||_{\max}||\mathbf{e}\epsilon'||_{\max}||\alpha||+T_{h}^{-1}||\varkappa||\,||\mathbf{e}\epsilon'||_{\max}\lesssim_{\textrm{P}}(\log N)^{1/2}T^{-1/2}.
\end{array}
\]
From Assumption \ref{assu:6}, we have $c_{qN}^{(1)}-c_{qN+1}^{(1)}\gtrsim c_{qN}^{(1)}$
and the definition of $\tilde{K}$ implies that $c_{qN}^{(k)}\geq c$
for $k\leq\tilde{K}$. Therefore, we have $c_{qN}^{(1)}-c_{qN+1}^{(1)}\gtrsim c$.
Define the events
\begin{equation}
\begin{array}{c}
A_{1}:=\left\{ ||T_{h}^{-1}(\hat{\Upsilon}\mathbf{x})_{[i]}\mathbf{y}'||_{\max}>(c_{qN}^{(1)}+c_{qN+1}^{(1)})/2\;for\:all\:i\in I_{1}\right\} \\
A_{2}:=\left\{ ||T_{h}^{-1}(\hat{\Upsilon}\mathbf{x})_{[i]}\mathbf{y}'||_{\max}<(c_{qN}^{(1)}+c_{qN+1}^{(1)})/2\;for\:all\:i\in I_{1}^{c}\right\} \\
A_{3}:=\left\{ ||T_{h}^{-1}(\hat{\Upsilon}\mathbf{x})_{[i]}\mathbf{y}'-(\hat{\Upsilon}\beta)_{[i]}\alpha'||_{\max}\geq(c_{qN}^{(1)}-c_{qN+1}^{(1)})/2\;for\:some\:i\in[N]\right\} .
\end{array}\label{eq:C3}
\end{equation}
It is straightforward to observe that $\{\hat{I}_{1}=I_{1}\}\supset A_{1}\cap A_{2}.$
Furthermore, based on the definition of $I_{1}$, we have $||(\hat{\Upsilon}\beta)_{[i]}\alpha'||_{\max}\geq c_{qN}^{(1)}$
for all $i\in I_{1}$ and $||(\hat{\Upsilon}\beta)_{[i]}\alpha'||_{\max}\leq c_{qN+1}^{(1)}$
for all $i\in I_{1}^{c}$. Thus, if $A_{1}^{c}$ occurs, we have $||T_{h}^{-1}(\hat{\Upsilon}\mathbf{x})_{[i]}\mathbf{y}'-(\hat{\Upsilon}\beta)_{[i]}\alpha'||_{\max}\geq(c_{qN}^{(1)}-c_{qN+1}^{(1)})/2,$ for some $i\in I_{1},$ which implies $A_{1}^{c}\subset A_{3}.$ Similarly,
we have $A_{2}^{c}\subset A_{3}$. By $\{\hat{I}_{1}=I_{1}\}\supset A_{1}\cap A_{2}$
and $A_{1}^{c}\cup A_{2}^{c}\subset A_{3}$, we have 
\begin{equation}
P(\widehat{I}_{k}=I_{k})\geq P(A_{1}\cap A_{2})=1-P(A_{1}^{c}\cup A_{2}^{c})\geq1-P(A_{3}).\label{eq:C4}
\end{equation}
Using $c^{-1}(\log N)^{1/2}T^{-1/2}\rightarrow0$ and $c_{qN}^{(1)}-c_{qN+1}^{(1)}\gtrsim c$,
we can get $P(A_{3})\rightarrow0$. Then, $P(\widehat{I}_{1}=I_{1})\rightarrow1.$

(ii) Given that $P(\widehat{I}_{1}=I_{1})\rightarrow1$, we can impose
$\widehat{I}_{1}=I_{1}$. Then, we have $\tilde{\mathbf{x}}_{(1)}=(\hat{\Upsilon}\mathbf{x})_{[I_{1}]}$
by (\ref{eq:B3}), and Assumption \ref{assu:3} and Lemma \ref{lem:16}
(i) give $||\tilde{\mathbf{x}}_{(1)}-(\hat{\Upsilon}\beta)_{(1)}\mathbf{f}||=||(\hat{\Upsilon}\mathbf{e})_{[I_{1}]}||\leq||\beta||_{\max}||\mathbf{e}_{[I_{1}]}||\,||\alpha||+||\varkappa||\,||\mathbf{e}_{[I_{1}]}||\lesssim_{\textrm{P}}q^{1/2}N^{1/2}+T^{1/2}+T^{-1/2}q^{1/2}N^{1/2}$.

(iii) From Lemma \ref{lem:3} and Lemma \ref{lem:4}, we have $\sigma_{j}(\widehat{\Upsilon}\beta\mathbf{f})/\sigma_{j}(\widehat{\Upsilon}\beta)=T_{h}^{1/2}+O_{\textrm{P}}(1)$
and $\lambda_{(1)}^{1/2}=||(\widehat{\Upsilon}\beta)_{(1)}||\asymp q^{1/2}N^{1/2}$,
which leads to 
\begin{equation}
\mid||(\widehat{\Upsilon}\beta)_{(1)}\mathbf{f}||-T_{h}^{1/2}\lambda_{(1)}^{1/2}\mid=|\sigma_{1}((\widehat{\Upsilon}\beta)_{(1)}\mathbf{f})-T_{h}^{1/2}\sigma_{1}((\widehat{\Upsilon}\beta)_{(1)})|\lesssim_{\textrm{P}}q^{1/2}N^{1/2}.\label{eq:C5}
\end{equation}

\noindent Moreover, the result in (ii) implies that
\begin{equation}
\mid||\tilde{\mathbf{x}}_{(1)}||-||(\widehat{\Upsilon}\beta)_{(1)}\mathbf{f}||\mid\leq||\tilde{\mathbf{x}}_{(1)}-(\widehat{\Upsilon}\beta)_{(1)}\mathbf{f}||\lesssim_{\textrm{P}}q^{1/2}N^{1/2}+T^{1/2}+T^{-1/2}q^{1/2}N^{1/2}.\label{eq:C6}
\end{equation}
Using (\ref{eq:C5}), (\ref{eq:C6}) and $\lambda_{(1)}^{1/2}\asymp q^{1/2}N^{1/2},$
we have 
\[
\begin{split}|\frac{\hat{\lambda}_{(1)}^{1/2}}{\lambda_{(1)}^{1/2}}-1|=|\frac{||\tilde{\mathbf{x}}_{(1)}||}{T_{h}^{1/2}\lambda_{(1)}^{1/2}}-1|&\leq\frac{\mid||\tilde{\mathbf{x}}_{(1)}||-||(\widehat{\Upsilon}\beta)_{(1)}\mathbf{f}||\mid}{T_{h}^{1/2}\lambda_{(1)}^{1/2}}+\frac{\mid||(\widehat{\Upsilon}\beta)_{(1)}\mathbf{f}||-T_{h}^{1/2}\lambda_{(1)}^{1/2}\mid}{T_{h}^{1/2}\lambda_{(1)}^{1/2}}\\&\lesssim_{\textrm{P}}q^{-1/2}N^{-1/2}+T^{-1/2}.\end{split}
\]
Therefore, $\hat{\lambda}_{(1)}\asymp_{\textrm{P}}\lambda_{(1)}$
and $\hat{\lambda}_{(1)}\asymp_{\textrm{P}}qN$.

(iv) Let $\tilde{\xi}_{(1)}$ denote the first right singular vector
of $(\widehat{\Upsilon}\beta)_{(1)}\mathbf{f}$. Lemma \ref{lem:3} establishes
\begin{equation}
||\mathbb{P}_{\tilde{\xi}_{(1)}}-T_{h}^{-1}\mathbf{f}'\mathbb{P}_{\iota_{(1)}}\mathbf{f}||\lesssim_{\textrm{P}}T^{-1/2},\label{eq:C7}
\end{equation}
and $\sigma_{j}((\widehat{\Upsilon}\beta)_{(1)}\mathbf{f})/\sigma_{j}((\widehat{\Upsilon}\beta)_{(1)})=T_{h}^{1/2}+O_{\textrm{P}}(1)$
which further leads to
\begin{equation}
\begin{split}\sigma_{1}((\widehat{\Upsilon}\beta)_{(1)}\mathbf{f})-\sigma_{2}((\widehat{\Upsilon}\beta)_{(1)}\mathbf{f})&=T_{h}^{1/2}(\sigma_{1}((\widehat{\Upsilon}\beta)_{(1)})-\sigma_{2}((\widehat{\Upsilon}\beta)_{(1)}))+O_{\textrm{P}}(\sigma_{1}((\widehat{\Upsilon}\beta)_{(1)}))\\&\asymp_{\textrm{P}}T_{h}^{1/2}\sigma_{1}((\widehat{\Upsilon}\beta)_{(1)}),\end{split}\label{eq:C8}
\end{equation}
where we use the assumption that $\sigma_{2}((\widehat{\Upsilon}\beta)_{(1)})\leq(1+\delta)^{-1}\sigma_{1}((\widehat{\Upsilon}\beta)_{(1)})$
in the last equation.

Using $||\tilde{X}_{(1)}-(\widehat{\Upsilon}\beta)_{(1)}\mathbf{f}||\lesssim_{\textrm{P}}q^{1/2}N^{1/2}+T^{1/2}+T^{-1/2}q^{1/2}N^{1/2}$,
(\ref{eq:C8}), Lemma \ref{lem:4}, and \citeauthor{wedin1972perturbation}'s
\citeyearpar{wedin1972perturbation} sin-theta theorem
for singular vectors, we have
\begin{equation}
||\mathbb{P}_{\hat{\mathbf{f}}_{(1)}'}-\mathbb{P}_{\tilde{\xi}_{(1)}}||\lesssim_{\textrm{P}}\frac{q^{1/2}N^{1/2}+T^{1/2}+T^{-1/2}q^{1/2}N^{1/2}}{\sigma_{1}((\widehat{\Upsilon}\beta)_{(1)}\mathbf{f})-\sigma_{2}((\widehat{\Upsilon}\beta)_{(1)}\mathbf{f})}\lesssim_{\textrm{P}}q^{-1/2}N^{-1/2}+T^{-1/2}.\label{eq:C9}
\end{equation}
According to (\ref{eq:C7}) and (\ref{eq:C9}), we have the first
formula in (iv) holds for $k=1$. As $\mathbb{P}_{\hat{\mathbf{f}}'_{(k)}}=\hat{\xi}_{(k)}\hat{\xi}'_{(k)}$,
left and right multiplying this equation by $\hat{\xi}'_{(1)}$ and
$\hat{\xi}_{(1)}$, we have 
\[
|1-T_{h}^{-1}(\iota'_{(1)}\mathbf{f}\hat{\xi}_{(1)})^{2}|\lesssim_{\textrm{P}}q^{-1/2}N^{-1/2}+T^{-1/2},
\]
which leads to $|1-T_{h}^{-1/2}\iota'_{(1)}\mathbf{f}\hat{\xi}_{(1)}|\lesssim_{\textrm{P}}q^{-1/2}N^{-1/2}+T^{-1/2}$.
Left-multiplying it by $\iota_{(1)}$, we can get the second equation
in (iv). 

So far, we have demonstrated that (i)-(iv) hold for $k=1$. Now, assuming
that (i)-(iv) hold for $j\leq k-1$, we aim to prove that they also
remain hold for $j=k$.

(i) Again, we show that the gap between the sample covariances and
their population counterparts introduced in the SsPCA process is minimal.
At the $k$th step, this discrepancy can be expressed as
\[
||\hat{\Upsilon}\beta\prod_{j=1}^{k-1}\mathbb{\mathbb{M}}_{\iota_{(j)}}\alpha'-T_{h}^{-1}\hat{\Upsilon}(\beta\mathbf{f}+\mathbf{e})\prod_{j=1}^{k-1}\mathbb{\mathbb{M}}_{\hat{\mathbf{f}}'_{(j)}}(\alpha\mathbf{f}+\epsilon)'||_{\max}
\]
\[
\leq||\hat{\Upsilon}||\,||\beta\prod_{j=1}^{k-1}\mathbb{\mathbb{M}}_{\iota_{(j)}}\alpha'-T_{h}^{-1}\beta\mathbf{f}\prod_{j=1}^{k-1}\mathbb{\mathbb{M}}_{\hat{\mathbf{f}}'_{(j)}}\mathbf{f}'\alpha'||_{\max}
\]
\begin{equation}
+T_{h}^{-1}||\hat{\Upsilon}||\,||\beta\mathbf{f}\prod_{j=1}^{k-1}\mathbb{\mathbb{M}}_{\hat{\mathbf{f}}'_{(j)}}\epsilon'||_{\max}+T_{h}^{-1}||\hat{\Upsilon}||\,||\mathbf{e}\prod_{j=1}^{k-1}\mathbb{\mathbb{M}}_{\hat{\mathbf{f}}'_{(j)}}\mathbf{f}'\alpha'||_{\max}+T_{h}^{-1}||\hat{\Upsilon}||\,||\mathbf{e}\prod_{j=1}^{k-1}\mathbb{\mathbb{M}}_{\hat{\mathbf{f}}'_{(j)}}\epsilon'||_{\max}.\label{eq:C10}
\end{equation}
Since (iv) holds for $j\leq k-1$, we have
\begin{equation}
||\stackrel[j=1]{k-1}{\sum}\mathbb{P}_{\hat{\mathbf{f}}'_{(j)}}-T_{h}^{-1}\mathbf{f}'\stackrel[j=1]{k-1}{\sum}\mathbb{P}_{\iota_{(j)}}\mathbf{f}||=||\stackrel[j=1]{k-1}{\sum}\left(\mathbb{P}_{\hat{\mathbf{f}}'_{(j)}}-T_{h}^{-1}\mathbf{f}'\mathbb{P}_{\iota_{(j)}}\mathbf{f}\right)||\lesssim_{\textrm{P}}q^{-1/2}N^{-1/2}+T^{-1/2}.\label{eq:C11}
\end{equation}
Using Lemma \ref{lem:1} and Lemma \ref{lem:4} (i), we have 
\[
\prod_{j=1}^{k-1}\mathbb{\mathbb{M}}_{\iota_{(j)}}=\mathbb{I_{K}}-\stackrel[j=1]{k-1}{\sum}\mathbb{P}_{\iota_{(j)}},\quad and\quad\prod_{j=1}^{k-1}\mathbb{\mathbb{M}}_{\hat{\mathbf{f}}'_{(j)}}=\mathbb{I}_{T_{h}}-\stackrel[j=1]{k-1}{\sum}\mathbb{P}_{\hat{\mathbf{f}}'_{(j)}}.
\]
Using the above equations, (\ref{eq:C11}), and $||T_{h}^{-1}\mathbf{ff}'-\mathbb{I}_{K}||\lesssim_{\textrm{P}}T^{-1/2}$,
we have 
\begin{equation}
T_{h}^{-1/2}||\mathbf{f}\prod_{j=1}^{k-1}\mathbb{\mathbb{M}}_{\hat{\mathbf{f}}'_{(j)}}-\prod_{j=1}^{k-1}\mathbb{\mathbb{M}}_{\iota_{(j)}}\mathbf{f}||=T_{h}^{-1/2}||\mathbf{f}\prod_{j=1}^{k-1}\mathbb{\mathbb{P}}_{\hat{\mathbf{f}}'_{(j)}}-\prod_{j=1}^{k-1}\mathbb{\mathbb{P}}_{\iota_{(j)}}\mathbf{f}||\lesssim_{\textrm{P}}q^{-1/2}N^{-1/2}+T^{-1/2}.\label{eq:C12}
\end{equation}
Analogously, right multiplying $\mathbf{f}'$ to the term inside the $||\cdot||$
of (\ref{eq:C12}), we can get
\begin{equation}
T_{h}^{-1}||\mathbf{f}\prod_{j=1}^{k-1}\mathbb{\mathbb{M}}_{\hat{\mathbf{f}}'_{(j)}}\mathbf{f}'-\prod_{j=1}^{k-1}\mathbb{\mathbb{M}}_{\iota_{(j)}}\mathbf{f}||\lesssim_{\textrm{P}}q^{-1/2}N^{-1/2}+T^{-1/2}.\label{eq:C13}
\end{equation}
Next, we proceed to analyze the four terms in (\ref{eq:C10}) separately.
By applying (\ref{eq:C13}) along with Assumption \ref{assu:2}, we
obtain the first term involving $\hat{\Upsilon}$ based on Lemma \ref{lem:16}
(i), as below
\[
||\hat{\Upsilon}||\,||\beta\prod_{j=1}^{k-1}\mathbb{\mathbb{M}}_{\iota_{(j)}}\alpha'-T_{h}^{-1}\beta\mathbf{f}\prod_{j=1}^{k-1}\mathbb{\mathbb{M}}_{\hat{\mathbf{f}}'_{(j)}}\mathbf{f}'\alpha'||_{\max}\lesssim||\beta||_{\max}^{2}||\prod_{j=1}^{k-1}\mathbb{\mathbb{M}}_{\iota_{(j)}}-T_{h}^{-1}\mathbf{f}\prod_{j=1}^{k-1}\mathbb{\mathbb{M}}_{\hat{\mathbf{f}}'_{(j)}}\mathbf{f}'||\,||\alpha||^{2}
\]
\[
+||\varkappa||\,||\beta||_{\max}||\prod_{j=1}^{k-1}\mathbb{\mathbb{M}}_{\iota_{(j)}}-T_{h}^{-1}\mathbf{f}\prod_{j=1}^{k-1}\mathbb{\mathbb{M}}_{\hat{\mathbf{f}}'_{(j)}}\mathbf{f}'||\,||\alpha||\lesssim_{\textrm{P}}q^{-1/2}N^{-1/2}+T^{-1/2}.
\]
For the 2nd term, applying (\ref{eq:C12}), Assumption \ref{assu:2}, and considering the influence of scaling coefficient, we have
\[
T_{h}^{-1}||\hat{\Upsilon}||\,||\beta\mathbf{f}\prod_{j=1}^{k-1}\mathbb{\mathbb{M}}_{\hat{\mathbf{f}}'_{(j)}}\epsilon'||_{\max}\lesssim T_{h}^{-1}||\beta||_{\max}^{2}||\prod_{j=1}^{k-1}\mathbb{\mathbb{M}}_{\iota_{(j)}}||\,||\mathbf{f}\epsilon'||\,||\alpha||
\]
\[
+T_{h}^{-1}||\beta||_{\max}^{2}||\mathbf{f}\prod_{j=1}^{k-1}\mathbb{\mathbb{M}}_{\hat{\mathbf{f}}'_{(j)}}-\prod_{j=1}^{k-1}\mathbb{\mathbb{M}}_{\iota_{(j)}}\mathbf{f}||\,||\epsilon||\,||\alpha||+T_{h}^{-1}||\varkappa||\,||\beta||_{\max}||\prod_{j=1}^{k-1}\mathbb{\mathbb{M}}_{\iota_{(j)}}||\,||\mathbf{f}\epsilon'||
\]
\[
+T_{h}^{-1}||\varkappa||\,||\beta||_{\max}||\mathbf{f}\prod_{j=1}^{k-1}\mathbb{\mathbb{M}}_{\hat{\mathbf{f}}'_{(j)}}-\prod_{j=1}^{k-1}\mathbb{\mathbb{M}}_{\iota_{(j)}}\mathbf{f}||\,||\epsilon||\lesssim_{\textrm{P}}q^{-1/2}N^{-1/2}+T^{-1/2}.
\]
For the 3rd term, applying (\ref{eq:C12}), and considering the influence
of $\hat{\Upsilon}$, we have
\[
T_{h}^{-1}||\hat{\Upsilon}||\,||\mathbf{e}\prod_{j=1}^{k-1}\mathbb{\mathbb{M}}_{\hat{\mathbf{f}}'_{(j)}}\mathbf{f}'\alpha'||_{\max}\lesssim T_{h}^{-1}||\beta||_{\max}||\mathbf{e}||_{\max}T_{h}^{1/2}||\mathbf{f}\prod_{j=1}^{k-1}\mathbb{\mathbb{M}}_{\hat{\mathbf{f}}'_{(j)}}-\prod_{j=1}^{k-1}\mathbb{\mathbb{M}}_{\iota_{(j)}}\mathbf{f}||\,||\alpha||^{2}
\]
\[
+T_{h}^{-1}||\beta||_{\max}||\mathbf{ef}'||_{\max}||\prod_{j=1}^{k-1}\mathbb{\mathbb{M}}_{\iota_{(j)}}||\,||\alpha||^{2}+T_{h}^{-1}||\varkappa||\,||\mathbf{e}||_{\max}T_{h}^{1/2}||\mathbf{f}\prod_{j=1}^{k-1}\mathbb{\mathbb{M}}_{\hat{\mathbf{f}}'_{(j)}}-\prod_{j=1}^{k-1}\mathbb{\mathbb{M}}_{\iota_{(j)}}\mathbf{f}||\,||\alpha||
\]
\[
+T_{h}^{-1}||\varkappa||\,||\mathbf{ef}'||_{\max}||\prod_{j=1}^{k-1}\mathbb{\mathbb{M}}_{\iota_{(j)}}||\,||\alpha||\lesssim_{\textrm{P}}(\log NT)^{1/2}(q^{-1/2}N^{-1/2}+T^{-1/2}).
\]
For the 4th term, applying (\ref{eq:C11}), and considering the influence
of $\hat{\Upsilon}$, we have
\[
T_{h}^{-1}||\hat{\Upsilon}||\,||\mathbf{e}\prod_{j=1}^{k-1}\mathbb{\mathbb{M}}_{\hat{\mathbf{f}}'_{(j)}}\epsilon'||_{\max}\lesssim T_{h}^{-1}||\beta||_{\max}||\mathbf{e}\epsilon'||_{\max}\,||\alpha||
\]
\[
+T_{h}^{-2}||\beta||_{\max}||\mathbf{ef}'||_{\max}||\stackrel[j=1]{k-1}{\sum}\mathbb{P}_{\iota_{(j)}}||\,||\mathbf{f}\epsilon'||\,||\alpha||
\]
\[
+T_{h}^{-1}||\beta||_{\max}||\mathbf{e}||_{\max}||T_{h}^{-1}\mathbf{f}'\stackrel[j=1]{k-1}{\sum}\mathbb{P}_{\iota_{(j)}}\mathbf{f}-\stackrel[j=1]{k-1}{\sum}\mathbb{P}_{\hat{\mathbf{f}}'_{(j)}}||\,||\epsilon||\,||\alpha||
\]
\[
+T_{h}^{-1}||\varkappa||\,||\mathbf{e}\epsilon'||_{\max}+T_{h}^{-2}||\varkappa||\,||\mathbf{e}\mathbf{f}'||_{\max}||\stackrel[j=1]{k-1}{\sum}\mathbb{P}_{\iota_{(j)}}||\,||\mathbf{f}\epsilon'||
\]
\[
+T_{h}^{-1}||\varkappa||\,||\mathbf{e}||_{\max}||T_{h}^{-1}\mathbf{f}'\stackrel[j=1]{k-1}{\sum}\mathbb{P}_{\iota_{(j)}}\mathbf{f}-\stackrel[j=1]{k-1}{\sum}\mathbb{P}_{\hat{\mathbf{f}}'_{(j)}}||\,||\epsilon||\lesssim_{\textrm{P}}(\log NT)^{1/2}(q^{-1/2}N^{-1/2}+T^{-1/2}).
\]
Therefore, we can obtain
\begin{equation}
||T_{h}^{-1}\hat{\Upsilon}\mathbf{x}\prod_{j=1}^{k-1}\mathbb{\mathbb{M}}_{\hat{\mathbf{f}}'_{(j)}}\mathbf{y}'-\hat{\Upsilon}\beta\prod_{j=1}^{k-1}\mathbb{\mathbb{M}}_{\iota_{(j)}}\alpha'||_{\max}\lesssim_{\textrm{P}}(\log NT)^{1/2}(q^{-1/2}N^{-1/2}+T^{-1/2}).\label{eq:C14}
\end{equation}
As in the case where $k=1$, assuming that $c^{-1}(\log NT)^{1/2}(q^{-1/2}N^{-1/2}+T^{-1/2})\rightarrow0$
and under Assumption \ref{assu:6}, we can leverage the arguments
presented in (\ref{eq:C3}) and (\ref{eq:C4}) for the case of $k=1$
leading to $P(\widehat{I}_{k}=I_{k})\rightarrow1$.

(ii) We assume $\widehat{I}_{k}=I_{k}$ below. Then we have $\tilde{\mathbf{x}}_{(k)}=(\hat{\Upsilon}\mathbf{x})_{[I_{k}]}\prod_{j=1}^{k-1}\mathbb{\mathbb{M}}_{\hat{\mathbf{f}}'_{(j)}}, (\hat{\Upsilon}\beta)_{(k)}=(\hat{\Upsilon}\beta)_{[I_{k}]}\prod_{j=1}^{k-1}\mathbb{\mathbb{M}}_{\iota_{(j)}}$,
and these lead to
\[
\begin{split}\tilde{\mathbf{x}}_{(k)}-(\hat{\Upsilon}\beta)_{(k)}\mathbf{f}&=(\hat{\Upsilon}\mathbf{x})_{[I_{k}]}\prod_{j=1}^{k-1}\mathbb{\mathbb{M}}_{\hat{\mathbf{f}}'_{(j)}}-(\hat{\Upsilon}\beta)_{[I_{k}]}\prod_{j=1}^{k-1}\mathbb{\mathbb{M}}_{\iota_{(j)}}\mathbf{f}\\&=(\hat{\Upsilon}\beta)_{[I_{k}]}\left(\mathbf{f}\prod_{j=1}^{k-1}\mathbb{\mathbb{M}}_{\hat{\mathbf{f}}'_{(j)}}-\prod_{j=1}^{k-1}\mathbb{\mathbb{M}}_{\iota_{(j)}}\mathbf{f}\right)+(\hat{\Upsilon}\mathbf{e})_{[I_{k}]}\prod_{j=1}^{k-1}\mathbb{\mathbb{M}}_{\hat{\mathbf{f}}'_{(j)}}.\end{split}
\]
Using Assumption \ref{assu:2}, (\ref{eq:C12}), and considering the
influence of $\hat{\Upsilon}$, we obtain
\[
||\tilde{\mathbf{x}}_{(k)}-(\hat{\Upsilon}\beta)_{(k)}\mathbf{f}||\leq||\hat{\Upsilon}_{[I_{k}]}||\,||\beta_{[I_{k}]}||\,||\mathbf{f}\prod_{j=1}^{k-1}\mathbb{\mathbb{M}}_{\hat{\mathbf{f}}'_{(j)}}-\prod_{j=1}^{k-1}\mathbb{\mathbb{M}}_{\iota_{(j)}}\mathbf{f}||+||\hat{\Upsilon}_{[I_{k}]}||\,||\mathbf{e}_{[I_{k}]}||\,||\prod_{j=1}^{k-1}\mathbb{\mathbb{M}}_{\hat{\mathbf{f}}'_{(j)}}||
\]
\[
\lesssim_{\textrm{P}}q^{1/2}N^{1/2}+T^{1/2}+T^{-1/2}q^{1/2}N^{1/2}.
\]

(iii)(iv) The proofs of (iii) and (iv) are analogous to the case where
$k=1$.

To summarize, by induction, we have demonstrated that (i)-(iv) hold
for $k\leq\tilde{K}.$

(v) Recall that $\tilde{K}$ is determined by $(\hat{\Upsilon}\beta)_{[i]}\prod_{j<k}\mathbb{\mathbb{M}}_{\iota_{(j)}}\alpha'$,
whereas $\hat{K}$ is determined by $T_{h}^{-1}(\hat{\Upsilon}\mathbf{x})_{[i]}\prod_{j<k}\mathbb{\mathbb{M}}_{\hat{\mathbf{f}}'_{(j)}}\mathbf{y}'$.
Because (iv) holds for $j\leq\tilde{K}$ as shown above, using the
same proof for (\ref{eq:C14}), we have
\begin{equation}
||T_{h}^{-1}\hat{\Upsilon}\mathbf{x}\prod_{j=1}^{\tilde{K}}\mathbb{\mathbb{M}}_{\hat{\mathbf{f}}'_{(j)}}\mathbf{y}'-\hat{\Upsilon}\beta\prod_{j=1}^{\tilde{K}}\mathbb{\mathbb{M}}_{\iota_{(j)}}\alpha'||_{\max}\lesssim_{\textrm{P}}(\log NT)^{1/2}(q^{-1/2}N^{-1/2}+T^{-1/2}).\label{eq:C15}
\end{equation}
The assumption $c_{qN}^{(\tilde{K}+1)}\leq(1+\delta)^{-1}c$ in Assumption
\ref{assu:6} implies that $c-c_{qN}^{(\tilde{K}+1)}\asymp c.$ Combine
with $c^{-1}(\log NT)^{1/2}(q^{-1/2}N^{-1/2}+T^{-1/2})\rightarrow0$,
we can reuse the arguments for (\ref{eq:C3}) and (\ref{eq:C4}) with
the events
\[
B_{1}=\left\{ ||T_{h}^{-1}(\hat{\Upsilon}\mathbf{x})_{[i]}\prod_{j=1}^{\tilde{K}}\mathbb{\mathbb{M}}_{\hat{\mathbf{f}}'_{(j)}}\mathbf{y}'||_{\max}>(c+c_{qN}^{(\tilde{K}+1)})/2\;\textrm{for}\;\textrm{at}\:\textrm{most}\;qN-1\;\textrm{different}\;\textrm{is}\;\textrm{in}\;[N]\right\} ,
\]
\[
B_{2}=\left(||T_{h}^{-1}(\hat{\Upsilon}\mathbf{x})_{[i]}\prod_{j=1}^{\tilde{K}}\mathbb{\mathbb{M}}_{\hat{\mathbf{f}}'_{(j)}}\mathbf{y}'-(\hat{\Upsilon}\beta)_{[i]}\prod_{j=1}^{\tilde{K}}\mathbb{\mathbb{M}}_{\iota_{(j)}}\alpha'||_{\max}>(c-c_{qN}^{(\tilde{K}+1)})/2\;\textrm{for}\;\textrm{some}\;i\in[N]\right),
\]
to derive $P(\hat{K}=\widetilde{K})\geq P(B_{1})=1-P(B_{1}^{c})\geq1-P(B_{2})\rightarrow1.$

(vi) This result follows directly from (\ref{eq:C14}) and (\ref{eq:C15}).
\end{proof}
\begin{lem}
\label{lem:3}For any $N\times K$ scaling matrix $\hat{\Upsilon}\beta$,
if $||T_{h}^{-1}\mathbf{ff}'-\mathbb{I}_{K}||\lesssim_{\textrm{P}}T^{-1/2}$,
we have

(i) $\sigma_{j}(\hat{\Upsilon}\beta\mathbf{f})/\sigma_{j}(\hat{\Upsilon}\beta)=T_{h}^{1/2}+O_{\textrm{P}}(1)$
for $j\leq K$.

(ii) If $\sigma_{1}(\hat{\Upsilon}\beta)-\sigma_{2}(\hat{\Upsilon}\beta)\asymp\sigma_{1}(\hat{\Upsilon}\beta)$,
then $||\mathbb{P}_{\tilde{\xi}}-T_{h}^{-1}\mathbf{f}'\mathbb{P}_{\iota}\mathbf{f}||\lesssim_{\textrm{P}}T^{-1/2}$,
where $\tilde{\xi}$ and $\iota$ are the first right singular vectors
of $\hat{\Upsilon}\beta\mathbf{f}$ and $\hat{\Upsilon}\beta$, respectively.
\end{lem}
\begin{proof}
(i) For $j\leq K$, $\sigma_{j}(\hat{\Upsilon}\beta\mathbf{f})^{2}=\lambda_{j}[(\hat{\Upsilon}\beta)\mathbf{ff}'(\hat{\Upsilon}\beta)']=\lambda_{j}[(\hat{\Upsilon}\beta)'(\hat{\Upsilon}\beta)\mathbf{ff}'],$
which implies $\lambda_{j}[(\hat{\Upsilon}\beta)'(\hat{\Upsilon}\beta)]\lambda_{j}(\mathbf{f}\mathbf{f}')\leq\sigma_{j}(\hat{\Upsilon}\beta\mathbf{f})^{2}\leq\lambda_{j}[(\hat{\Upsilon}\beta)'(\hat{\Upsilon}\beta)]\lambda_{1}(\mathbf{f}\mathbf{f}')$.
By the assumption $||T_{h}^{-1}\mathbf{ff}'-\mathbb{I}_{K}||\lesssim_{\textrm{P}}T^{-1/2}$,
Lemma \ref{lem:16} (i), and Wely's inequality, we can have $T_{h}^{-1/2}\sigma_{j}(\hat{\Upsilon}\beta\mathbf{f})/\sigma_{j}(\hat{\Upsilon}\beta)=1+O_{\textrm{P}}(T^{-1/2})$.
Therefore, we derive $\sigma_{j}(\hat{\Upsilon}\beta\mathbf{f})/\sigma_{j}(\hat{\Upsilon}\beta)=T_{h}^{1/2}+O_{\textrm{P}}(1)$.

(ii) Let $\tilde{\varsigma}$ and $\varsigma$ be the first left singular
vectors of $\hat{\Upsilon}\beta\mathbf{f}$ and $\hat{\Upsilon}\beta$, respectively.
Equivalently, $\tilde{\varsigma}$ and $\varsigma$ are the eigenvectors
of $T_{h}^{-1}(\hat{\Upsilon}\beta)\mathbf{ff}'(\hat{\Upsilon}\beta)'$ and
$(\hat{\Upsilon}\beta)(\hat{\Upsilon}\beta)'$. Since $||(\hat{\Upsilon}\beta)(\hat{\Upsilon}\beta)'-T_{h}^{-1}(\hat{\Upsilon}\beta)\mathbf{ff}'(\hat{\Upsilon}\beta)'||\leq||\hat{\Upsilon}\beta||^{2}||T_{h}^{-1}\mathbf{ff}'-\mathbb{I}_{K}||\lesssim_{\textrm{P}}\sigma_{1}(\hat{\Upsilon}\beta)^{2}T^{-1/2}$
and $\sigma_{1}(\hat{\Upsilon}\beta)-\sigma_{2}(\hat{\Upsilon}\beta)\asymp\sigma_{1}(\hat{\Upsilon}\beta)$,
by sin-theta theorem
\[
||\varsigma\varsigma'-\tilde{\varsigma}\tilde{\varsigma}'||\lesssim\frac{||(\hat{\Upsilon}\beta)(\hat{\Upsilon}\beta)'-T_{h}^{-1}(\hat{\Upsilon}\beta)\mathbf{ff}'(\hat{\Upsilon}\beta)'||}{\sigma_{1}(\hat{\Upsilon}\beta)^{2}-\sigma_{2}(\hat{\Upsilon}\beta)^{2}-O[(\hat{\Upsilon}\beta)(\hat{\Upsilon}\beta)'-T_{h}^{-1}(\hat{\Upsilon}\beta)\mathbf{ff}'(\hat{\Upsilon}\beta)']}\lesssim T^{-1/2}.
\]
Using the connection between left and right singular vectors, we have
$\iota'=\varsigma'\hat{\Upsilon}\beta/\sigma_{1}(\hat{\Upsilon}\beta)$
and $\tilde{\xi}'=\tilde{\varsigma}'\hat{\Upsilon}\beta\mathbf{f}/||\hat{\Upsilon}\beta\mathbf{f}||.$
Therefore,
\begin{equation}
\begin{split}||\mathbb{P}_{\tilde{\xi}}-T_{h}^{-1}\frac{\sigma_{1}(\hat{\Upsilon}\beta)^{2}}{||\hat{\Upsilon}\beta\mathbf{f}||^{2}}\mathbf{f}'\mathbb{P}_{\iota}\mathbf{f}||&=||\widetilde{\xi}\tilde{\xi}'-\frac{\mathbf{f}'(\hat{\Upsilon}\beta)'\varsigma\varsigma'(\hat{\Upsilon}\beta)\mathbf{f}}{||\hat{\Upsilon}\beta\mathbf{f}||^{2}}||\\&=||\frac{\mathbf{f}'(\hat{\Upsilon}\beta)'\tilde{\varsigma}\tilde{\varsigma}'(\hat{\Upsilon}\beta)\mathbf{f}}{||\hat{\Upsilon}\beta\mathbf{f}||^{2}}-\frac{\mathbf{f}'(\hat{\Upsilon}\beta)'\varsigma\varsigma'(\hat{\Upsilon}\beta)\mathbf{f}}{||\hat{\Upsilon}\beta\mathbf{f}||^{2}}||\lesssim_{\textrm{P}}T^{-1/2}.\end{split}\label{eq:C16}
\end{equation}
By Weyl's inequality, $T_{h}^{-1}||\hat{\Upsilon}\beta\mathbf{f}||^{2}=\lambda_{1}(T_{h}^{-1}(\hat{\Upsilon}\beta)\mathbf{ff}'(\hat{\Upsilon}\beta)')=\sigma_{1}(\hat{\Upsilon}\beta)^{2}+O_{\textrm{P}}(\sigma_{1}(\hat{\Upsilon}\beta)^{2}T^{-1/2})$.
Based on (\ref{eq:C16}), we can get $||\mathbb{P}_{\tilde{\xi}}-T_{h}^{-1}\mathbf{f}'\mathbb{P}_{\iota}\mathbf{f}||\lesssim_{\textrm{P}}T^{-1/2}$. 
\end{proof}
\begin{lem}
\label{lem:4}Under assumptions of Theorem \ref{thm:1}, $\iota_{(k)}$, $(\hat{\Upsilon}\beta)_{(k)}$
and $\tilde{K}$ in Section \ref{Section 3.1} satisfy

\noindent (i) $\iota'_{(j)}\iota_{(k)}=\delta_{jk}$ for $j\leq k\leq\tilde{K}$. 

\noindent (ii) $\lambda_{(k)}^{1/2}=||(\hat{\Upsilon}\beta)_{(k)}||\asymp q^{1/2}N^{1/2}.$

\noindent (iii) $\tilde{K}\leq K.$

\noindent (iv) $\tilde{K}=K$, if we further have $\lambda_{K}(\alpha'\alpha)\gtrsim1.$
\end{lem}
\begin{proof}
\noindent (i) Recall that $\iota_{(k)}$ represents the first right
singular vector of $(\hat{\Upsilon}\beta)_{(k)}$, where $(\hat{\Upsilon}\beta)_{(k)}=(\hat{\Upsilon}\beta)_{[I_{k}]}\prod_{j<k}\mathbb{\mathbb{M}}_{\iota_{(j)}}$.
By applying the same argument as in the proof of Lemma \ref{lem:1},
it follows that $\iota'_{(j)}\iota_{(k)}=\delta_{jk}$ for $j,k\leq\tilde{K}$.

(ii) Recall that $\lambda_{(k)}^{1/2}$ represents the leading singular
value of $(\hat{\Upsilon}\beta)_{(k)}$. The selection rule at $k$th
step implies that 
\begin{equation}
|I_{k}|^{-1}\underset{i\in I_{k}}{\sum}||(\hat{\Upsilon}\beta)_{[i]}\underset{j<k}{\prod}\mathbb{\mathbb{M}}_{\iota_{(j)}}\alpha'||_{\max}^{2}\geq N_{0}^{-1}\underset{i\in I_{0}}{\sum}||(\hat{\Upsilon}\beta)_{[i]}\underset{j<k}{\prod}\mathbb{\mathbb{M}}_{\iota_{(j)}}\alpha'||_{\max}^{2}.\label{eq:C17}
\end{equation}
For any matrix $Z\in\mathbb{R}^{N\times D}$ and set $I\subset[N]$,
we have
\[
\underset{i\in I}{\sum}||Z_{[i]}||_{\max}^{2}\leq||Z||_{F}^{2}\leq D\underset{i\in I}{\sum}||Z_{[i]}||_{\max}^{2},
\]
and $||Z||^{2}\leq||Z||_{F}^{2}\leq D||Z||^{2}$. This gives us
\begin{equation}
||Z||^{2}\asymp\underset{i\in I}{\sum}||Z_{[i]}||_{\max}^{2}.\label{eq:C18}
\end{equation}
Using this result, (\ref{eq:C17}) becomes
\[
|I_{k}|^{-1}||(\hat{\Upsilon}\beta)_{[I_{k}]}\underset{j<k}{\prod}\mathbb{\mathbb{M}}_{\iota_{(j)}}\alpha'||^{2}\apprge N_{0}^{-1}||(\hat{\Upsilon}\beta)_{[I_{0}]}\underset{j<k}{\prod}\mathbb{\mathbb{M}}_{\iota_{(j)}}\alpha'||^{2}.
\]
Then, we have
\begin{equation}
\begin{split}\frac{||(\hat{\Upsilon}\beta)_{(k)}||}{\sqrt{|I_{k}|}}||\underset{j<k}{\prod}\mathbb{\mathbb{M}}_{\iota_{(j)}}\alpha'||&\geq\frac{1}{\sqrt{|I_{k}|}}||(\hat{\Upsilon}\beta)_{[I_{k}]}\underset{j<k}{\prod}\mathbb{\mathbb{M}}_{\iota_{(j)}}\alpha'||\\&\apprge\frac{1}{\sqrt{N_{0}}}||(\hat{\Upsilon}\beta)_{[I_{0}]}\underset{j<k}{\prod}\mathbb{\mathbb{M}}_{\iota_{(j)}}\alpha'||\geq\frac{\sigma_{K}((\hat{\Upsilon}\beta)_{[I_{0}]})}{\sqrt{N_{0}}}||\underset{j<k}{\prod}\mathbb{\mathbb{M}}_{\iota_{(j)}}\alpha'||,\end{split}\label{eq:C19}
\end{equation}
where we use $(\hat{\Upsilon}\beta)_{[I_{k}]}\underset{j<k}{\prod}\mathbb{\mathbb{M}}_{\iota_{(j)}}\alpha'=(\hat{\Upsilon}\beta)_{[I_{k}]}(\underset{j<k}{\prod}\mathbb{\mathbb{M}}_{\iota_{(j)}})^{2}\alpha'=(\hat{\Upsilon}\beta)_{(k)}\underset{j<k}{\prod}\mathbb{\mathbb{M}}_{\iota_{(j)}}\alpha'$
in the first inequality. From Assumption \ref{assu:2} and Lemma \ref{lem:16}
(i), we have $\sigma_{K}[(\hat{\Upsilon}\beta)_{[I_{0}]}]\gtrsim\sigma_{K}(\beta_{[I_{0}]})\gtrsim\sqrt{N_{0}}$,
so (\ref{eq:C19}) leads to $||(\hat{\Upsilon}\beta)_{(k)}||\gtrsim\sqrt{|I_{k}|}$.
Additionally, $||\beta||_{\max}\lesssim1$ leads to $||(\hat{\Upsilon}\beta)_{(k)}||\lesssim\sqrt{|I_{k}|}$.
Thus, we have $\lambda_{(k)}^{1/2}=||(\hat{\Upsilon}\beta)_{(k)}||\asymp|I_{k}|^{1/2}\asymp q^{1/2}N^{1/2}.$

(iii) As established in (i), the right singular vectors $\iota_{(k)}$
are pairwise orthogonal for $k\leq\tilde{K}$. It is not feasible
to have more than $K$ mutually orthogonal vectors in a $K$-dimensional
space. Therefore, $\tilde{K}\leq K$ for SsPCA.

(iv) Recall that $\tilde{K}$ is defined in Section \ref{Section 3.1}. Since the
SsPCA procedure stops at $\tilde{K}+1$, we have at most $qN-1$ rows
of $\hat{\Upsilon}\beta$ satisfying $||(\hat{\Upsilon}\beta)_{[i]}\underset{j<\tilde{K}}{\prod}\mathbb{\mathbb{M}}_{\iota_{(j)}}\alpha'||_{\max}\geq c$.
We denote the corresponding set as $S_{\beta}$, with $|S_{\beta}|\leq qN-1$.
This implies
\[
||(\hat{\Upsilon}\beta)_{[I_{0}]}\underset{j<\tilde{K}}{\prod}\mathbb{\mathbb{M}}_{\iota_{(j)}}\alpha'||^{2}\lesssim|I_{0}\cap S_{\beta}|+|I_{0}\cap S_{\beta}^{c}|c^{2}\leq qN+(N_{0}-qN)c^{2}=o(N_{0}),
\]
where we use (\ref{eq:C18}) and the assumptions $c\rightarrow0$,
$qN/N_{0}\rightarrow0$. With $\sigma_{K}(\beta_{[I_{0}]})\gtrsim\sqrt{N_{0}}$
from Assumption \ref{assu:2} and $\sigma_{K}[(\hat{\Upsilon}\beta)_{[I_{0}]}]\gtrsim\sqrt{N_{0}}$
implied by Lemma \ref{lem:16} (i), we have
\begin{equation}
||\underset{j<\tilde{K}}{\alpha\prod}\mathbb{\mathbb{M}}_{\iota_{(j)}}||\leq\sigma_{K}[(\hat{\Upsilon}\beta)_{[I_{0}]}]^{-1}||(\hat{\Upsilon}\beta)_{[I_{0}]}\underset{j<k}{\prod}\mathbb{\mathbb{M}}_{\iota_{(j)}}\alpha'||=o(1).\label{eq:C20}
\end{equation}
If $\tilde{K}\leq K-1$, using (i), we have $\alpha\prod_{j\leq\tilde{K}}\mathbb{\mathbb{M}}_{\iota_{(j)}}=\alpha-\alpha\sum_{j\leq\tilde{K}}\iota_{(j)}\iota'_{(j)},$
so that
\begin{equation}
\sigma_{K}(\alpha)\leq\sigma_{1}\left(\underset{j<\tilde{K}}{\alpha\prod}\mathbb{\mathbb{M}}_{\iota_{(j)}}\right)+\sigma_{K}\left(\alpha\underset{j\leq\tilde{K}}{\sum}\iota_{(j)}\iota'_{(j)}\right).\label{eq:C21}
\end{equation}
Since
\begin{equation}
Rank\left(\alpha\underset{j\leq\tilde{K}}{\sum}\iota_{(j)}\iota'_{(j)}\right)\leq\tilde{K}\leq K-1,
\end{equation}
we have $\sigma_{K}(\alpha\sum_{j\leq\tilde{K}}\iota_{(j)}\iota'_{(j)})=0$.
Thus, by (\ref{eq:C20}) and (\ref{eq:C21}), we further have $\sigma_{K}(\alpha)\leq\sigma_{1}(\alpha\prod_{j\leq\tilde{K}}\mathbb{\mathbb{M}}_{\iota_{(j)}})\rightarrow0.$
This conflicts with the assumption that $\lambda_{K}(\alpha'\alpha)\gtrsim1$.
So, we have established that $\tilde{K}\geq K$. Combine with (iii),
we have $\tilde{K}=K$.
\end{proof}
\begin{lem}
\label{lem:5}Under assumptions of Theorem \ref{thm:1}, for $k,l\leq\tilde{K}$,
we have

(i) $||\frac{\tilde{\mathbf{e}}'_{(k)}\hat{\varsigma}{}_{(k)}}{\sqrt{T_{h}\hat{\lambda}_{(k)}}}||\lesssim_{\textrm{P}}T^{-1}+q^{-1/2}N^{-1/2},\,||\frac{\tilde{\mathbf{e}}{}_{(k)}}{\sqrt{T_{h}\hat{\lambda}_{(k)}}}||\lesssim_{\textrm{P}}T^{-1/2}+q^{-1/2}N^{-1/2}.$

(ii) $||\frac{Z\tilde{\mathbf{e}}'_{(k)}\hat{\varsigma}{}_{(k)}}{T_{h}\sqrt{\hat{\lambda}_{(k)}}}||\lesssim_{\textrm{P}}T^{-1}+q^{-1}N^{-1},\,for\,Z=\mathbf{f},\epsilon,and\,W.$

(iii) $|\frac{\hat{\xi}'_{(l)}(\hat{\Upsilon}\mathbf{e})'_{[I_{k}]}\hat{\varsigma}{}_{(k)}}{\sqrt{T_{h}\hat{\lambda}_{(k)}}}|\lesssim_{\textrm{P}}T^{-1}+q^{-1}N^{-1},\,|\frac{\hat{\xi}'_{(l)}\tilde{\mathbf{e}}'_{(k)}\hat{\varsigma}{}_{(k)}}{\sqrt{T_{h}\hat{\lambda}_{(k)}}}|\lesssim_{\textrm{P}}T^{-1}+q^{-1}N^{-1}.$

(iv) $|\hat{\varsigma}'{}_{(k)}\tilde{D}_{(k)}\hat{\Upsilon}e_{T}|\lesssim_{\textrm{P}}1+T^{-1/2}q^{1/2}N^{1/2},\,|\varsigma'_{(k)}D_{(k)}\hat{\Upsilon}e_{T}|\lesssim_{\textrm{P}}1.$
\end{lem}
\begin{proof}
(i) Recall that from the definition of $\mathbf{e}_{(k)}$ (below (\ref{eq:B5})),
we have
\begin{equation}
\tilde{\mathbf{e}}_{(k)}=(\hat{\Upsilon}\mathbf{e})_{[I_{k}]}-\stackrel[i=1]{k-1}{\sum}\frac{(\hat{\Upsilon}\mathbf{x})_{[I_{k}]}\hat{\xi}{}_{(i)}}{\sqrt{T_{h}}}\frac{\hat{\varsigma}'{}_{(i)}\tilde{\mathbf{e}}_{(i)}}{\sqrt{\hat{\lambda}_{(i)}}}.\label{eq:C23}
\end{equation}
Then, a direct multiplication of $\hat{\varsigma}'{}_{(k)}/\sqrt{T_{h}\hat{\lambda}_{(k)}}$
from the left side of (\ref{eq:C23}) leads to
\[
\frac{\hat{\varsigma}'{}_{(k)}\tilde{\mathbf{e}}{}_{(k)}}{\sqrt{T_{h}\hat{\lambda}_{(k)}}}=\frac{\hat{\varsigma}'{}_{(k)}(\hat{\Upsilon}\mathbf{e})_{[I_{k}]}}{\sqrt{T_{h}\hat{\lambda}_{(k)}}}-\stackrel[i=1]{k-1}{\sum}\frac{\hat{\varsigma}'{}_{(k)}(\hat{\Upsilon}\mathbf{x})_{[I_{k}]}\hat{\xi}{}_{(i)}}{\sqrt{T_{h}\hat{\lambda}_{(k)}}}\frac{\hat{\varsigma}'{}_{(i)}\tilde{\mathbf{e}}_{(i)}}{\sqrt{T_{h}\hat{\lambda}_{(i)}}}.
\]
Therefore, using $\mathbf{x}_{[I_{k}]}\leq||\beta_{[I_{k}]}||\,||\mathbf{f}||+||\mathbf{e}_{[I_{k}]}||\lesssim_{\textrm{P}}q^{1/2}N^{1/2}T^{1/2},$
$\hat{\lambda}_{(k)}\asymp_{\textrm{P}}qN$, Lemma \ref{lem:6}(i),
and Lemma \ref{lem:16}(i), we have
\[
||\frac{\hat{\varsigma}'{}_{(k)}\tilde{\mathbf{e}}{}_{(k)}}{\sqrt{T_{h}\hat{\lambda}_{(k)}}}||\leq||\frac{\hat{\varsigma}'{}_{(k)}(\hat{\Upsilon}\mathbf{e})_{[I_{k}]}}{\sqrt{T_{h}\hat{\lambda}_{(k)}}}||+\stackrel[i=1]{k-1}{\sum}||\frac{(\hat{\Upsilon}\mathbf{x})_{[I_{k}]}}{\sqrt{T_{h}\hat{\lambda}_{(k)}}}||\,||\frac{\hat{\varsigma}'{}_{(i)}\tilde{\mathbf{e}}_{(i)}}{\sqrt{T_{h}\hat{\lambda}_{(i)}}}||
\]
\begin{equation}
\lesssim_{\textrm{P}}T^{-1}+q^{-1/2}N^{-1/2}+\stackrel[i=1]{k-1}{\sum}||\frac{\hat{\varsigma}'{}_{(i)}\tilde{\mathbf{e}}_{(i)}}{\sqrt{T_{h}\hat{\lambda}_{(i)}}}||.\label{eq:C24}
\end{equation}
If $||T_{h}^{-1/2}\hat{\lambda}_{(i)}^{-1/2}\hat{\varsigma}'{}_{(i)}\tilde{\mathbf{e}}_{(i)}||\lesssim_{\textrm{P}}T^{-1}+q^{-1/2}N^{-1/2}$
holds for $i\leq k-1$, then (\ref{eq:C24}) implies that this inequality
also holds for $k$. Additionally, when $k=1$, $\tilde{\mathbf{e}}_{(1)}=(\hat{\Upsilon}\mathbf{e})_{[I_{1}]}$
and this equation is implied from Lemma \ref{lem:6}(i). Therefore,
we have (i) hold for $k\leq\tilde{K}$ by induction.

Utilizing (\ref{eq:C23}) again, with Assumption \ref{assu:3}, we
have
\[
||\frac{\tilde{\mathbf{e}}{}_{(k)}}{\sqrt{T_{h}\hat{\lambda}_{(k)}}}||\leq||\frac{(\hat{\Upsilon}\mathbf{e})_{[I_{k}]}}{\sqrt{T_{h}\hat{\lambda}_{(k)}}}||+\stackrel[i=1]{k-1}{\sum}||\frac{(\hat{\Upsilon}\mathbf{x})_{[I_{k}]}}{\sqrt{T_{h}\hat{\lambda}_{(k)}}}||\,||\frac{\tilde{\mathbf{e}}_{(i)}}{\sqrt{T_{h}\hat{\lambda}_{(i)}}}||
\]
\begin{equation}
\lesssim_{\textrm{P}}T^{-1}+q^{-1/2}N^{-1/2}+\stackrel[i=1]{k-1}{\sum}||\frac{\tilde{\mathbf{e}}_{(i)}}{\sqrt{T_{h}\hat{\lambda}_{(i)}}}||.\label{eq:C25}
\end{equation}
When $k=1,$ Assumption \ref{assu:3} implies $||\frac{\tilde{\mathbf{e}}_{(k)}}{\sqrt{T_{h}\hat{\lambda}_{(k)}}}||\lesssim_{\textrm{P}}T^{-1/2}+q^{-1/2}N^{-1/2}.$
Then, using the same induction argument with (\ref{eq:C25}), we have
this inequality holds for $k\leq\tilde{K}$.

(ii) Analogously, by multiplying $\mathbf{f}'$ on the right of (\ref{eq:C23}),
we have
\[
\frac{\hat{\varsigma}'{}_{(k)}\tilde{\mathbf{e}}{}_{(k)}\mathbf{f}'}{T_{h}\sqrt{\hat{\lambda}_{(k)}}}=\frac{\hat{\varsigma}'{}_{(k)}(\hat{\Upsilon}\mathbf{e})_{[I_{k}]}\mathbf{f}'}{T_{h}\sqrt{\hat{\lambda}_{(k)}}}-\stackrel[i=1]{k-1}{\sum}\frac{\hat{\varsigma}'{}_{(k)}(\hat{\Upsilon}\mathbf{x})_{[I_{k}]}\hat{\xi}{}_{(i)}}{\sqrt{T_{h}\hat{\lambda}_{(k)}}}\frac{\hat{\varsigma}'{}_{(i)}\tilde{\mathbf{e}}_{(i)}\mathbf{f}'}{T_{h}\sqrt{\hat{\lambda}_{(i)}}}.
\]
Therefore, we have
\[
||\frac{\hat{\varsigma}'{}_{(k)}\tilde{\mathbf{e}}{}_{(k)}\mathbf{f}'}{T_{h}\sqrt{\hat{\lambda}_{(k)}}}||\leq||\frac{\hat{\varsigma}'{}_{(k)}(\hat{\Upsilon}\mathbf{e})_{[I_{k}]}\mathbf{f}'}{T_{h}\sqrt{\hat{\lambda}_{(k)}}}||-\stackrel[i=1]{k-1}{\sum}||\frac{(\hat{\Upsilon}\mathbf{x})_{[I_{k}]}}{\sqrt{T_{h}\hat{\lambda}_{(k)}}}||\,||\frac{\hat{\varsigma}'{}_{(i)}\tilde{\mathbf{e}}_{(i)}\mathbf{f}'}{T_{h}\sqrt{\hat{\lambda}_{(i)}}}||
\]
\begin{equation}
\lesssim_{\textrm{P}}T^{-1}+q^{-1}N^{-1}+\stackrel[i=1]{k-1}{\sum}||\frac{\hat{\varsigma}'{}_{(i)}\tilde{\mathbf{e}}_{(i)}\mathbf{f}'}{\sqrt{T_{h}\hat{\lambda}_{(i)}}}||.\label{eq:C26}
\end{equation}
When $k=1$, $||\frac{\hat{\varsigma}'{}_{(k)}\tilde{\mathbf{e}}_{(k)}\mathbf{f}'}{T_{h}\sqrt{\hat{\lambda}_{(k)}}}||\lesssim_{\textrm{P}}T^{-1}+q^{-1}N^{-1}$
is a result of Lemma \ref{lem:6}(ii). Then, a straightforward induction
argument using (\ref{eq:C26}) leads to this inequality for $k\leq\tilde{K}$.

Replacing $\mathbf{f}$ by $\epsilon$ or $W$ in the above proof, and using
Lemma \ref{lem:6}(ii), we have
\[
||\frac{\epsilon\tilde{\mathbf{e}}'_{(k)}\hat{\varsigma}{}_{(k)}}{T_{h}\sqrt{\hat{\lambda}_{(k)}}}||\lesssim_{\textrm{P}}T^{-1}+q^{-1}N^{-1}\quad or\quad||\frac{W\tilde{\mathbf{e}}'_{(k)}\hat{\varsigma}{}_{(k)}}{T_{h}\sqrt{\hat{\lambda}_{(k)}}}||\lesssim_{\textrm{P}}T^{-1}+q^{-1}N^{-1}.
\]

(iii) Recall that $\tilde{\mathbf{x}}_{(k)}=\tilde{\beta}_{(k)}\mathbf{f}+\tilde{\mathbf{e}}{}_{(k)}$
as defined in (\ref{eq:B3}), we have
\[
|\hat{\varsigma}'{}_{(l)}\tilde{\mathbf{x}}_{(l)}(\hat{\Upsilon}\mathbf{e})'_{[I_{k}]}\hat{\varsigma}{}_{(k)}|\leq||\hat{\varsigma}'{}_{(l)}\tilde{\beta}_{(l)}||\,||\mathbf{f}(\hat{\Upsilon}\mathbf{e})'_{[I_{k}]}\hat{\varsigma}{}_{(k)}||+||\hat{\varsigma}'{}_{(l)}\tilde{\mathbf{e}}_{(l)}||\,||(\hat{\Upsilon}\mathbf{e})'_{[I_{k}]}\hat{\varsigma}{}_{(k)}||.
\]
Along with (\ref{eq:B1}), and $||\hat{\lambda}_{(l)}^{-1/2}\hat{\varsigma}'{}_{(l)}\tilde{\beta}_{(l)}||\lesssim_{\textrm{P}}1$
which can be deduced from (\ref{eq:C32}), so we can have
\[
|\frac{\hat{\xi}'_{(l)}(\hat{\Upsilon}\mathbf{e})'_{[I_{k}]}\hat{\varsigma}{}_{(k)}}{\sqrt{T_{h}\hat{\lambda}_{(k)}}}|\leq||\frac{\hat{\varsigma}'{}_{(l)}\tilde{\beta}_{(l)}}{\sqrt{\hat{\lambda}_{(l)}}}||\,||\frac{\mathbf{f}(\hat{\Upsilon}\mathbf{e})'_{[I_{k}]}\hat{\varsigma}{}_{(k)}}{T_{h}\sqrt{\hat{\lambda}_{(k)}}}||+||\frac{\hat{\varsigma}'{}_{(l)}\tilde{\mathbf{e}}_{(l)}}{\sqrt{T_{h}\hat{\lambda}_{(l)}}}||\,||\frac{(\hat{\Upsilon}\mathbf{e})'_{[I_{k}]}\hat{\varsigma}{}_{(k)}}{\sqrt{T_{h}\hat{\lambda}_{(k)}}}||
\]
\begin{equation}
\lesssim_{\textrm{P}}T^{-1}+q^{-1}N^{-1}.
\end{equation}
Similarly, using $||\hat{\lambda}_{(k)}^{-1/2}\hat{\varsigma}'{}_{(k)}\tilde{\beta}_{(k)}||\lesssim_{\textrm{P}}1$,
results of (i) (ii) and Lemma \ref{lem:6}(i) completes the proof. Replacing
$(\hat{\Upsilon}\mathbf{e})_{[I_{k}]}$ by $\tilde{\mathbf{e}}_{(k)}$ above and using
the inequality that
\[
|\hat{\varsigma}'{}_{(l)}\tilde{\mathbf{x}}_{(l)}\tilde{\mathbf{e}}'_{(k)}\hat{\varsigma}{}_{(k)}|\leq||\hat{\varsigma}'{}_{(l)}\tilde{\beta}_{(l)}||\,||\mathbf{f}\tilde{\mathbf{e}}'_{(k)}\hat{\varsigma}{}_{(k)}||+||\hat{\varsigma}'{}_{(l)}\tilde{\mathbf{e}}'_{(l)}||\,||\tilde{\mathbf{e}}'_{(k)}\hat{\varsigma}{}_{(k)}||
\]
and (\ref{eq:B1}), we derive the second equation in (iii).

(iv) Similar to (ii), by induction, we have
\[
|\hat{\varsigma}'{}_{(k)}\tilde{D}_{(k)}\hat{\Upsilon}e_{T_{.}}|\leq|\hat{\varsigma}'{}_{(k)}(\hat{\Upsilon}e_{T_{.}})_{[I_{k}]}|+\stackrel[i=1]{k-1}{\sum}||\frac{\hat{(\Upsilon}\mathbf{x})_{[I_{k}]}}{\sqrt{T_{h}\hat{\lambda}_{(i)}}}||\,|\hat{\varsigma}'{}_{(i)}\tilde{D}_{(i)}\hat{\Upsilon}e_{T}|\lesssim_{\textrm{P}}1+T^{-1/2}q^{1/2}N^{1/2}.
\]
For the second inequality, from Lemma \ref{lem:4}, we have $\iota'_{(i)}\iota_{(j)}=0$
when $i\neq j$. Therefore, by definition, we have $\varsigma_{(k)}=\lambda_{(k)}^{-1/2}(\hat{\Upsilon}\beta)_{(k)}\iota_{(k)}=\lambda_{(k)}^{-1/2}\hat{(\Upsilon}\beta)_{[I_{k}]}\iota_{(k)}.$
Using $||\beta'_{[I_{k}]}(e_{T})_{[I_{k}]}||\lesssim_{\textrm{P}}q^{1/2}N^{1/2}$
from Assumption \ref{assu:4}(iii), $\lambda_{(k)}\asymp_{\textrm{P}}qN$
from Lemma \ref{lem:4}, $||\hat{\Upsilon}||\lesssim_{\textrm{P}}1$
from Lemma \ref{lem:16}, and the definition of $D_{(k)}$, we have
\[
|\varsigma'_{(k)}D_{(k)}\hat{\Upsilon}e_{T}|\leq|\varsigma'_{(k)}(\hat{\Upsilon}e_{T})_{[I_{k}]}|+\underset{i<k}{\sum}\lambda_{(i)}^{-1/2}||\hat{(\Upsilon}\beta){}_{[I_{k}]}||\,||\iota_{(i)}||\,|\varsigma'_{(i)}D_{(i)}\hat{\Upsilon}e_{T}|
\]
\[
\leq\lambda_{(k)}^{-1/2}||\iota_{(k)}||\,||\hat{(\Upsilon}\beta)_{[I_{k}]}'(\hat{\Upsilon}e_{T})_{[I_{k}]}||+\underset{i<k}{\sum}\lambda_{(i)}^{-1/2}||\hat{(\Upsilon}\beta){}_{[I_{k}]}||\,|\varsigma'_{(i)}D_{(i)}\hat{\Upsilon}e_{T}|
\]
\[
\leq\lambda_{(k)}^{-1/2}||\iota_{(k)}||\,||\hat{\Upsilon}||^{2}\,||\beta'_{[I_{k}]}(e_{T})_{[I_{k}]}||+||\hat{\Upsilon}||\underset{i<k}{\sum}\lambda_{(i)}^{-1/2}||\beta{}_{[I_{k}]}||\,|\varsigma'_{(i)}D_{(i)}\hat{\Upsilon}e_{T}|\lesssim_{\textrm{P}}1+\underset{i<k}{\sum}|\varsigma'_{(i)}D_{(i)}\hat{\Upsilon}e_{T}|
\]
As $|\varsigma'_{(1)}D_{(1)}\hat{\Upsilon}e_{T}|\leq\lambda_{(1)}^{-1/2}||\iota_{(1)}||\,||\hat{\Upsilon}||^{2}\,|\beta{}_{[I_{1}]}(e_{T})_{[I_{1}]}||\lesssim_{\textrm{P}}1$,
$|\varsigma'_{(k)}D_{(k)}\hat{\Upsilon}e_{T}|\lesssim_{\textrm{P}}1$
holds by induction.
\end{proof}
\begin{lem}
\label{lem:6}Under assumptions of Theorem \ref{thm:1}, for $k\leq\tilde{K},$
we have

(i) $\Vert(\hat{\Upsilon}\mathbf{e})'_{[I_{k}]}\hat{\varsigma}_{(k)}\Vert\lesssim_{\textrm{P}}T^{1/2}+T^{-1/2}q^{1/2}N^{1/2}.$

(ii) $||\mathbf{f}(\hat{\Upsilon}\mathbf{e})'_{[I_{k}]}\hat{\varsigma}_{(k)}\Vert\lesssim_{\textrm{P}}q^{1/2}N^{1/2}+T^{1/2}, \Vert\epsilon(\hat{\Upsilon}\mathbf{e})'_{[I_{k}]}\hat{\varsigma}_{(k)}\Vert\lesssim_{\textrm{P}}q^{1/2}N^{1/2}+T^{1/2}.$

(iii) $|\hat{\varsigma}'_{(k)}(\hat{\Upsilon}e_{T})_{[_{I_{k}}]}|\lesssim_{\textrm{P}}1+T^{-1/2}q^{1/2}N^{1/2}.$
\end{lem}
\begin{proof}
(i) Using Lemma \ref{lem:1}, we have
\begin{equation}
T_{h}^{1/2}\hat{\lambda}_{(k)}^{1/2}\hat{\varsigma}_{(k)}=(\hat{\Upsilon}\mathbf{x})_{[I_{k}]}\stackrel[j=1]{k-1}{\prod}\mathbb{\mathbb{M}}_{\hat{\xi}_{(j)}}\hat{\xi}_{(k)}=(\hat{\Upsilon}\mathbf{x})_{[I_{k}]}\hat{\xi}_{(k)}=(\hat{\Upsilon}\beta)_{[I_{k}]}\mathbf{f}\hat{\xi}_{(k)}+(\hat{\Upsilon}\mathbf{e})_{[I_{k}]}\hat{\xi}_{(k)}.
\end{equation}
Thus, applying Assumption \ref{assu:1}, Assumption \ref{assu:3},
Assumption \ref{assu:4}(ii), and Lemma \ref{lem:16}(i), we derive
\[
T_{h}^{1/2}\hat{\lambda}_{(k)}^{1/2}||\hat{\varsigma}'_{(k)}(\hat{\Upsilon}\mathbf{e})_{[I_{k}]}||\leq||\hat{\xi}'_{(k)}\mathbf{f}'(\hat{\Upsilon}\beta)'_{[I_{k}]}(\hat{\Upsilon}\mathbf{e})_{[I_{k}]}||+||\hat{\xi}'_{(k)}(\hat{\Upsilon}\mathbf{e})'_{[I_{k}]}(\hat{\Upsilon}\mathbf{e})_{[I_{k}]}||
\]
\begin{equation}
\leq||\hat{\Upsilon}||^{2}\,||\mathbf{f}||\Vert\beta'_{[I_{k}]}\mathbf{e}_{[I_{k}]}\Vert+||\hat{\Upsilon}||^{2}\,||\mathbf{e}'_{[I_{k}]}\mathbf{e}_{[I_{k}]}||\lesssim_{\textrm{P}}qN+q^{1/2}N^{1/2}T.
\end{equation}
Combined with $\hat{\lambda}_{(k)}\asymp_{\textrm{P}}qN$, we can
get 
\[
||\hat{\varsigma}'_{(k)}(\hat{\Upsilon}\mathbf{e})_{[I_{k}]}||\leq||\hat{\Upsilon}||\,||\hat{\varsigma}'_{(k)}\mathbf{e}_{[I_{k}]}||\lesssim_{\textrm{P}}T^{1/2}+T^{-1/2}q^{1/2}N^{1/2}.
\]

(ii) Similarly, by Assumption \ref{assu:1}, Assumption \ref{assu:3},
Assumption \ref{assu:4}(i)(ii), and Lemma \ref{lem:16}(i), we can
derive
\[
T_{h}^{1/2}\hat{\lambda}_{(k)}^{1/2}||\hat{\varsigma}'_{(k)}(\hat{\Upsilon}\mathbf{e})_{[I_{k}]}\mathbf{f}'||\leq||\hat{\xi}'_{(k)}\mathbf{f}'(\hat{\Upsilon}\beta)'_{[I_{k}]}(\hat{\Upsilon}\mathbf{e})_{[I_{k}]}\mathbf{f}'||+||\hat{\xi}_{(k)}(\hat{\Upsilon}\mathbf{e})'_{[I_{k}]}(\hat{\Upsilon}\mathbf{e})_{[I_{k}]}\mathbf{f}'||
\]
\begin{equation}
\leq||\hat{\Upsilon}||^{2}\,||\mathbf{f}||\Vert\beta'_{[I_{k}]}\mathbf{e}_{[I_{k}]}\mathbf{f}'\Vert+||\hat{\Upsilon}||^{2}\,||\mathbf{e}'_{[I_{k}]}||\,||\mathbf{e}_{[I_{k}]}\mathbf{f}'||\lesssim_{\textrm{P}}q^{1/2}N^{1/2}T+qNT^{1/2}.
\end{equation}
Together with $\hat{\lambda}_{(k)}\asymp_{P}qN$, we have the desired
result. Additionally, replacing $\mathbf{f}$ by $\epsilon$ and $W$, we have
the second and third equations in (ii).

(iii) By Assumption \ref{assu:1}, Assumption \ref{assu:3}, Assumption
\ref{assu:4}(iii), and Lemma \ref{lem:16}(i), we have 
\[
T_{h}^{1/2}\hat{\lambda}_{(k)}^{1/2}|\hat{\varsigma}'_{(k)}(\hat{\Upsilon}e_{T})_{[I_{k}]}|\leq|\hat{\xi}'_{(k)}\mathbf{f}'(\hat{\Upsilon}\beta)'_{[I_{k}]}(\hat{\Upsilon}e_{T})_{[I_{k}]}|+|\hat{\xi}'_{(k)}(\hat{\Upsilon}\mathbf{e})'_{[I_{k}]}(\hat{\Upsilon}e_{T})_{[I_{k}]}|
\]
\[
\leq||\hat{\Upsilon}||^{2}\,||\mathbf{f}||\,||\beta'_{[I_{k}]}(e_{T})_{[I_{k}]}||+||\hat{\Upsilon}||^{2}\,||\mathbf{e}_{[I_{k}]}||\,||(e_{T})_{[I_{k}]}||\lesssim_{\textrm{P}}q^{1/2}N^{1/2}T^{1/2}+qN.
\]
Combine with $\hat{\lambda}_{(k)}\asymp_{\textrm{P}}qN,$ we have
the desired result.
\end{proof}
\begin{lem}
\label{lem:7}Under assumptions of Theorem \ref{thm:1}, for $k\leq\tilde{K}+1$,
we have

(i) $||\tilde{\epsilon}_{(k)}\mathbf{f}'||\lesssim_{\textrm{P}}T^{1/2}+Tq^{-1}N^{-1}.$

(ii) $||\tilde{\epsilon}_{(k)}(\hat{\Upsilon}\mathbf{e})'_{[I_{0}]}||\lesssim_{\textrm{P}}N_{0}^{1/2}T^{1/2}+Tq^{-1/2}N^{-1/2}.$
\end{lem}
\begin{proof}
(i) From the definition (\ref{eq:B7}) of $\tilde{\epsilon}_{(k)}$,
we have
\[
\tilde{\epsilon}_{(k)}\mathbf{f}'=\epsilon \mathbf{f}'-\stackrel[i=1]{k-1}{\sum}\mathbf{y}\hat{\xi}_{(i)}\frac{\hat{\varsigma}'_{(i)}\tilde{\textbf{e}}_{(i)}\mathbf{f}'}{\sqrt{T_{h}\hat{\lambda}_{(i)}}}.
\]
Then along with Lemma \ref{lem:5}(ii), we have
\[
||\tilde{\epsilon}_{(k)}\mathbf{f}'||\leq||\epsilon \mathbf{f}'||+\stackrel[i=1]{k-1}{\sum}||\mathbf{y}\hat{\xi}_{(i)}||\,||\frac{\hat{\varsigma}'_{(i)}\tilde{\mathbf{e}}_{(i)}\mathbf{f}'}{\sqrt{T_{h}\hat{\lambda}_{(i)}}}||\lesssim_{\textrm{P}}T^{1/2}+Tq^{-1}N^{-1}.
\]

(ii) With (\ref{eq:B7}) once more, we have 
\[
\tilde{\epsilon}_{(k)}(\hat{\Upsilon}\mathbf{e})'_{[I_{0}]}=\epsilon(\hat{\Upsilon}\mathbf{e})'_{[I_{0}]}-\stackrel[i=1]{k-1}{\sum}\mathbf{y}\hat{\xi}_{(i)}\frac{\hat{\varsigma}'_{(i)}\tilde{\mathbf{e}}_{(i)}(\hat{\Upsilon}\mathbf{e})'_{[I_{0}]}}{\sqrt{T_{h}\hat{\lambda}_{(i)}}},
\]
which, combine with Lemma \ref{lem:5}(i), Lemma \ref{lem:16}(i),
and the assumptions on $q$, lead to
\[
||\tilde{\epsilon}_{(k)}(\hat{\Upsilon}\mathbf{e})'_{[I_{0}]}||\leq||\hat{\Upsilon}||\,||\epsilon \mathbf{e}'_{[I_{0}]}||+\stackrel[i=1]{k-1}{\sum}||\mathbf{y}\hat{\xi}_{(i)}||\,||\frac{\hat{\varsigma}'_{(i)}\tilde{\mathbf{e}}_{(i)}}{\sqrt{T_{h}\hat{\lambda}_{(i)}}}||\,||\mathbf{e}{}_{[I_{0}]}||\,||\hat{\Upsilon}||
\]
\[
\lesssim_{\textrm{P}}N_{0}^{1/2}T^{1/2}+(q^{-1/2}N^{-1/2}+T^{-1})(N_{0}^{1/2}T^{1/2}+T)
\]
\[
\lesssim_{\textrm{P}}N_{0}^{1/2}T^{1/2}+Tq^{-1/2}N^{-1/2}.
\]
\end{proof}
\begin{lem}
\label{lem:8}Under Assumptions \ref{assu:1}-\ref{assu:4}, for any
$I\subset[N],$ we have the following results of the boundary:

(i) $||T_{h}^{-1}\mathbf{f}\mathbb{M}_{\mathbf{\mathbf{w}}'}\mathbf{f}'-\mathbb{I}_{K}||\lesssim_{\textrm{P}}T^{-1/2},$
$||\epsilon\mathbb{M}_{\mathbf{w}'}||\lesssim_{\textrm{P}}T^{1/2}.$

(ii) $||\mathbf{f}\mathbb{M}_{\mathbf{w}'}||_{\max}\lesssim_{\textrm{P}}(\log T)^{1/2},$
$||\epsilon\mathbb{M}_{\mathbf{w}'}||_{\max}\lesssim_{\textrm{P}}(\log T)^{1/2}.$

(iii) $||\beta'_{[I]}\mathbf{e}_{[I]}\mathbb{M}_{\mathbf{w}'}||\lesssim_{\textrm{P}}|I|^{1/2}T^{1/2},$
$||\beta'_{[I]}\mathbf{e}_{[I]}\mathbb{M}_{\mathbf{w}'}||_{\max}\lesssim_{\textrm{P}}|I|^{1/2}(\log T)^{1/2}.$

(iv) $||\beta'_{[I]}\mathbf{e}_{[I]}\mathbb{M}_{\mathbf{w}'}\mathbf{f}'||\lesssim_{\textrm{P}}|I|^{1/2}T^{1/2},$
$||\beta'_{[I]}\mathbf{e}_{[I]}\mathbb{M}_{\mathbf{w}'}\epsilon'||\lesssim_{\textrm{P}}|I|^{1/2}T^{1/2}.$

(v) $||\mathbf{e}\mathbb{M}_{\mathbf{w}'}||_{\max}\lesssim_{\textrm{P}}(\log NT)^{1/2}.$

(vi) $||\mathbf{e}\mathbb{M}_{\mathbf{w}'}\mathbf{f}'||_{\max}\lesssim_{\textrm{P}}(\log N)^{1/2}T^{1/2},$
$||\mathbf{e}\mathbb{M}_{\mathbf{w}'}\epsilon'||_{\max}\lesssim_{\textrm{P}}(\log N)^{1/2}T^{1/2}.$

(vii) $||\mathbf{e}_{[I]}\mathbb{M}_{\mathbf{w}'}||\lesssim_{\textrm{P}}|I|^{1/2}+T^{1/2},$
$||\mathbf{e}_{[I]}\mathbb{M}_{\mathbf{w}'}Z'||\lesssim_{\textrm{P}}|I|^{1/2}T^{1/2},$
for $Z=\mathbf{f},\epsilon$.

(viii) $||\mathbf{f}\mathbb{M}_{\mathbf{w}'}\epsilon'||\lesssim_{\textrm{P}}T^{1/2},$
$||\mathbf{f}\mathbb{M}_{\mathbf{w}'}\epsilon'-\mathbf{f}\epsilon'||\lesssim_{\textrm{P}}1.$

(ix) $||(e_{T})'_{[I]}\underline{\mathbf{e}}_{[I]}\mathbb{M}_{\mathbf{w}'}\mathbf{f}'||\lesssim_{\textrm{P}}|I|^{1/2}T^{1/2}+|I|,$
$||(e_{T})'_{[I]}\underline{\mathbf{e}}_{[I]}\mathbb{M}_{\mathbf{w}'}\epsilon'||\lesssim_{\textrm{P}}|I|^{1/2}T^{1/2}+|I|.$
\end{lem}
\begin{proof}
(i) Having $||(\mathbf{ww}')^{-1}||\lesssim_{\textrm{P}}T^{-1}$ from Assumption
\ref{assu:1}, $||\mathbf{wf}'||\lesssim_{\textrm{P}}T^{1/2}$, we can get
\[
||T_{h}^{-1}\mathbf{f}\mathbb{M}_{\mathbf{w}'}\mathbf{f}'-\mathbb{I}_{K}||\leq||T_{h}^{-1}\mathbf{ff}'-\mathbb{I}_{K}||+T_{h}^{-1}||\mathbf{fw}'||^{2}||(\mathbf{ww}')^{-1}||\lesssim_{\textrm{P}}T^{-1/2},
\]
and
\[
||\epsilon\mathbb{M}_{\mathbf{w}'}||\leq||\epsilon||\lesssim_{\textrm{P}}T^{1/2}.
\]

(ii) Using bound on $||\mathbf{f}||_{\max}$ and that $||\mathbf{w}||\lesssim T^{1/2}$
by Assumption \ref{assu:1}, we have
\[
||\mathbf{f}\mathbb{M}_{\mathbf{w}'}||_{\max}\leq||\mathbf{f}||_{\max}+||\mathbf{fw}'||\,||(\mathbf{ww}')^{-1}||\,||\mathbf{w}||\lesssim_{\textrm{P}}(\log T)^{1/2}.
\]
By substituting $\mathbf{f}$ with $\epsilon$ in the preceding proof, the
second inequality is derived.

(iii) Using Assumption \ref{assu:4}(ii) and $||\mathbb{M}_{\mathbf{w}'}||\leq1$,
the first inequality holds directly. Also,
\[
||\beta'_{[I]}\mathbf{e}_{[I]}\mathbb{M}_{\mathbf{w}'}||_{\max}\lesssim||\beta'_{[I]}\mathbf{e}_{[I]}||_{\max}+||\beta'_{[I]}\mathbf{e}_{[I]}\mathbf{w}||\,||(\mathbf{ww}')^{-1}||\,||\mathbf{w}||\lesssim_{\textrm{P}}|I|^{1/2}(\log T)^{1/2},
\]
where we use Assumption \ref{assu:1} and Assumption \ref{assu:4}(ii)
in the last inequality.

(iv) With $||(\mathbf{ww}')^{-1}||\lesssim_{\textrm{P}}T^{-1}$, $||\mathbf{wf}'||\lesssim_{\textrm{P}}T^{1/2}$,
and by Assumption \ref{assu:4}, we can get
\[
||\beta'_{[I]}\mathbf{e}_{[I]}\mathbb{M}_{\mathbf{w}'}\mathbf{f}'||\leq||\beta'_{[I]}\mathbf{e}_{[I]}\mathbf{f}'||+||\beta'_{[I]}\mathbf{e}_{[I]}\mathbf{w}'||\,||(\mathbf{ww}')^{-1}||\,||\mathbf{wf}'||\lesssim_{\textrm{P}}|I|^{1/2}T^{1/2}.
\]
By substituting $\mathbf{f}$ with $\epsilon$ in the preceding proof, the
second inequality in (iv) is derived.

(v) Similar to (ii), using Assumption \ref{assu:3} and \ref{assu:4},
we have
\[
||\mathbf{e}\mathbb{M}_{\mathbf{w}'}||_{\max}\lesssim||\mathbf{e}||_{\max}+||\mathbf{ew}'||_{\max}||(\mathbf{ww}')^{-1}\mathbf{w}||\lesssim_{\textrm{P}}(\log N)^{1/2}+(\log T)^{1/2}.
\]

(vi) Similar to (iv), by Assumption \ref{assu:1} and \ref{assu:3},
we have
\[
||\mathbf{e}\mathbb{M}_{\mathbf{w}'}\mathbf{f}'||_{\max}\lesssim||\mathbf{ef}'||_{\max}+||\mathbf{ew}'||_{\max}||(\mathbf{ww}')^{-1}\mathbf{wf}'||\lesssim_{\textrm{P}}(\log N)^{1/2}T^{1/2}.
\]
By substituting $\mathbf{f}$ with $\epsilon$ in the above inequality, we
can also get $||\mathbf{e}\mathbb{M}_{\mathbf{w}'}\epsilon'||_{\max}\lesssim_{\textrm{P}}(\log N)^{1/2}T^{1/2}.$

(vii) With Assumption \ref{assu:3} and $||\mathbb{M}_{\mathbf{w}'}||\leq1,$
the first inequality holds directly. By Assumptions \ref{assu:1}
and \ref{assu:4}, we have
\[
||\mathbf{e}_{[I]}\mathbb{M}_{\mathbf{w}'}\mathbf{f}'||\leq||\mathbf{e}_{[I]}\mathbf{f}'||+||\mathbf{e}_{[I]}\mathbf{w}'||\,||(\mathbf{ww}')^{-1}||\,||\mathbf{wf}'||\lesssim_{\textrm{P}}|I|^{1/2}T^{1/2}.
\]
By substituting $\mathbf{f}$ with $\epsilon$ in the above inequality, we
can also get the third inequality.

(viii) Using Assumption \ref{assu:1} and $||\mathbb{M}_{\mathbf{w}'}||\leq1,$
we have $||\epsilon\mathbb{M}_{\mathbf{w}'}||\lesssim_{\textrm{P}}T^{1/2}$.
And,
\[
||\mathbf{f}\mathbb{M}_{\mathbf{w}'}\epsilon'-\mathbf{f}\epsilon'||=||\mathbf{f}\mathbb{P}_{\mathbf{w}'}\epsilon'||\leq||\mathbf{fw}'||\,||(\mathbf{ww}')^{-1}||\,||\mathbf{w}\epsilon'||\lesssim_{\textrm{P}}1.
\]
Therefore, combing $||\mathbf{f}\epsilon'||\lesssim_{\textrm{P}}T^{1/2}$ from
Assumption \ref{assu:1}, we have
\[
||\mathbf{f}\mathbb{M}_{\mathbf{w}'}\epsilon'||\leq||\mathbf{f}\mathbb{M}_{\mathbf{w}'}\epsilon'-\mathbf{f}\epsilon'||+||\mathbf{f}\epsilon'||\lesssim_{\textrm{P}}T^{1/2}.
\]

(ix) By Assumptions \ref{assu:1} and \ref{assu:4}, we have
\[
||(e_{T})'_{[I]}\underline{\mathbf{e}}_{[I]}\mathbb{M}_{\mathbf{w}'}\mathbf{f}'||\leq||(e_{T})'_{[I]}\underline{\mathbf{e}}_{[I]}\mathbf{f}'||+||(e_{T})'_{[I]}\underline{\mathbf{e}}_{[I]}\mathbf{w}'||\,||(\mathbf{ww}')^{-1}||\,||\mathbf{wf}'||\lesssim_{\textrm{P}}|I|+|I|^{1/2}T^{1/2}.
\]
Replacing $\mathbf{f}$ by $\epsilon$, we have the second inequality.

Furthermore, even after incorporating the effect of $\hat{\Upsilon}$,
all the boundary results in Lemma \ref{lem:8} remain unchanged.
\end{proof}
\begin{lem}
\label{lem:9}Under assumptions of Theorem \ref{thm:1}, for $k\leq\tilde{K},$
we have

(i) $||\hat{\xi}_{(k)}-T_{h}^{-1/2}\mathbf{f}'\iota_{k2}||\lesssim_{\textrm{P}}T^{-1}+q^{-1/2}N^{-1/2}.$

(ii) $||T_{h}^{-1/2}\epsilon\hat{\xi}_{(k)}-T_{h}^{-1}\epsilon \mathbf{f}'\iota_{k2}||\lesssim_{\textrm{P}}T^{-1}+q^{-1}N^{-1}.$
\end{lem}
\begin{proof}
(i) By the definitions of $\iota_{k2}$ and $\hat{\xi}_{(k)}$, $\tilde{\mathbf{x}}_{(k)}=\tilde{\beta}_{(k)}\mathbf{f}+\tilde{\mathbf{e}}_{(k)},$
we have
\begin{equation}
\hat{\xi}_{(k)}-T_{h}^{-1/2}\mathbf{f}'\iota_{k2}=\frac{\tilde{\mathbf{e}}'_{(k)}\hat{\varsigma}{}_{(k)}}{\sqrt{T_{h}\hat{\lambda}_{(k)}}}.\label{eq:C31}
\end{equation}
By Lemma \ref{lem:5}(i), we can derive $\hat{\xi}_{(k)}-T_{h}^{-1/2}\mathbf{f}'\iota_{k2}\lesssim_{\textrm{P}}T^{-1}+q^{-1/2}N^{-1/2}$.

(ii) Similarly, by Lemma \ref{lem:5}(ii), we have
\[
T_{h}^{-1/2}\epsilon\hat{\xi}_{(k)}-T_{h}^{-1}\epsilon \mathbf{f}'\iota_{k2}=\frac{\epsilon\tilde{\mathbf{e}}'_{(k)}\hat{\varsigma}{}_{(k)}}{T_{h}\sqrt{\hat{\lambda}_{(k)}}}\lesssim_{\textrm{P}}T^{-1}+q^{-1}N^{-1}.
\]
\end{proof}
\begin{lem}
\label{lem:10}Under assumptions of Theorem \ref{thm:1}, $L_{1},$
$L_{2}$ defined by (\ref{eq:B20}) satisfy

(i) $||L_{1}||\lesssim_{P}1$, $||L_{2}||\lesssim_{\textrm{P}}1.$

(ii) $||L'_{1}L_{2}-\mathbb{I}_{\tilde{K}}||\lesssim_{\textrm{P}}T^{-1}+q^{-1}N^{-1}.$

(iii) $||L_{1}-L_{2}||\lesssim_{P}T^{-1/2}+q^{-1}N^{-1},$ $||L_{1}-L||\lesssim_{\textrm{P}}T^{-1/2}+q^{-1/2}N^{-1/2}.$

(iv) $||L_{2}L'_{2}-\mathbb{I}_{K}||\lesssim_{\textrm{P}}T^{-1/2}+q^{-1}N^{-1},$
when $\tilde{K}=K.$

(v) $||L_{1}L'_{2}-\mathbb{I}_{K}||\lesssim_{\textrm{P}}T^{-1}+q^{-1}N^{-1},$
when $\tilde{K}=K.$
\end{lem}
\begin{proof}
(i) Using the definition of (\ref{eq:B20}) of $L_{1}$ and Assumption
\ref{assu:1}, we have
\[
||\iota_{k1}||=||\frac{\mathbf{f}\hat{\xi}_{(k)}}{\sqrt{T_{h}}}||\lesssim_{\textrm{P}}1,
\]
which leads to $||L_{1}||\lesssim_{P}1.$ Using the definition (\ref{eq:B20})
of $L_{2},$ we have
\begin{equation}
||\iota_{k2}||=||\frac{\tilde{\beta}'_{(k)}\hat{\varsigma}_{(k)}}{\sqrt{\hat{\lambda}_{(k)}}}||\leq q^{-1/2}N^{-1/2}||\tilde{\beta}{}_{(k)}||.\label{eq:C32}
\end{equation}
We know that $(\hat{\Upsilon}\mathbf{x})_{[I_{k}]}=||(\hat{\Upsilon}\beta)_{[I_{k}]}||\,||\mathbf{f}||+||(\hat{\Upsilon}\mathbf{e})_{[I_{k}]}||\lesssim_{\textrm{P}}q^{1/2}N^{1/2}T^{1/2}$
and $\hat{\lambda}_{(k)}^{1/2}\asymp_{\textrm{P}}q^{1/2}N^{1/2}$.
Therefore, 
\begin{equation}
||\tilde{\beta}{}_{(k)}||\leq||(\hat{\Upsilon}\beta)_{[I_{k}]}||+\stackrel[i=1]{k-1}{\sum}||\frac{(\hat{\Upsilon}\mathbf{x})_{[I_{k}]}\hat{\xi}_{(i)}}{\sqrt{T_{h}\hat{\lambda}_{(i)}}}||\,||\hat{\varsigma}'_{(i)}\tilde{\beta}{}_{(i)}||\lesssim_{\textrm{P}}q^{1/2}N^{1/2}+\stackrel[i=1]{k-1}{\sum}||\tilde{\beta}{}_{(i)}||.
\end{equation}
Since $||\tilde{\beta}_{(1)}||=||\tilde{\beta}_{[I_{1}]}||\lesssim q^{1/2}N^{1/2},$
we have $||\tilde{\beta}{}_{(k)}||\lesssim q^{1/2}N^{1/2}$ by induction.
Combine with (\ref{eq:C32}), we have $||\iota_{k2}||\lesssim_{\textrm{P}}1$
and thus $||L_{2}||\lesssim_{\textrm{P}}1.$

(ii) By (\ref{eq:B1}) and Lemma \ref{lem:1}, we have
\[
\delta_{lk}=\hat{\xi}'_{(l)}\hat{\xi}_{(k)}=\frac{\hat{\xi}'_{(l)}\mathbf{f}'\tilde{\beta}'{}_{(k)}\hat{\varsigma}_{(k)}}{\sqrt{T_{h}\hat{\lambda}_{(k)}}}+\frac{\hat{\xi}'_{(l)}\tilde{\mathbf{e}}'{}_{(k)}\hat{\varsigma}_{(k)}}{\sqrt{T_{h}\hat{\lambda}_{(k)}}}=\iota'_{l1}\iota_{k2}+\frac{\hat{\xi}'_{(l)}\tilde{\mathbf{e}}'{}_{(k)}\hat{\varsigma}_{(k)}}{\sqrt{T_{h}\hat{\lambda}_{(k)}}}.
\]
By Lemma \ref{lem:5}(iii), we have $|\iota'_{l1}\iota_{k2}-\delta_{lk}|\lesssim_{\textrm{P}}T^{-1}+q^{-1}N^{-1}$,
and thus $||L'_{1}L_{2}-\mathbb{I}_{\tilde{K}}||\lesssim_{\textrm{P}}T^{-1}+q^{-1}N^{-1}.$

(iii) Using (\ref{eq:B1}) and $\tilde{\mathbf{x}}_{(k)}=\tilde{\beta}_{(k)}\mathbf{f}+\tilde{\mathbf{e}}_{(k)},$
we have
\[
\mathbf{f}\hat{\xi}{}_{(k)}=\frac{\mathbf{ff}'\tilde{\beta}'{}_{(k)}}{\sqrt{T_{h}\hat{\lambda}_{(k)}}}\hat{\varsigma}_{(k)}+\frac{\mathbf{f}\tilde{\mathbf{e}}'{}_{(k)}\hat{\varsigma}_{(k)}}{\sqrt{T_{h}\hat{\lambda}_{(k)}}}.
\]
By the definitions of $\iota{}_{k1}$ and $\iota_{k2}$, it becomes
\begin{equation}
\iota_{k1}=\frac{\mathbf{ff}'}{T_{h}}\iota_{k2}+\frac{\mathbf{f}\tilde{\mathbf{e}}'{}_{(k)}\hat{\varsigma}_{(k)}}{T_{h}\sqrt{\hat{\lambda}_{(k)}}}.\label{eq:C34}
\end{equation}
With $||L_{2}||\lesssim_{\textrm{P}}1,$ Assumption \ref{assu:1}
and Lemma \ref{lem:5}(ii), (\ref{eq:C34}) leads to
\[
\iota{}_{k1}-\iota{}_{k2}\lesssim_{\textrm{P}}T^{-1/2}+q^{-1}N^{-1}.
\]

\noindent This finishes the proof. The second inequality of (iii)
comes from Lemma \ref{lem:2}(iv) directly.

(iv) When $\tilde{K}=K$, $L'_{1}L_{2}$ is a $K\times K$ matrix.
By (i), (ii), (iii), we have
\[
||L'_{2}L_{2}-\mathbb{I}_{K}||\leq||L'_{1}L_{2}-\mathbb{I}_{K}||+||L_{1}-L_{2}||\,||L_{2}||\lesssim_{\textrm{P}}T^{-1/2}+q^{-1}N^{-1}.
\]
Since $L_{2}$ is a $K\times K$ matrix, we have
\[
||L'_{2}L_{2}-\mathbb{I}_{K}||=\underset{1\leq i\leq K}{\max}|\lambda_{i}(L'_{2}L_{2})-1|=||L'_{2}L_{2}-\mathbb{I}_{K}||\lesssim_{\textrm{P}}T^{-1/2}+q^{-1}N^{-1}.
\]

(v) With respect to $L_{1}L'_{2}$, we have
\begin{equation}
\sigma_{K}(L_{2})||L'_{2}L_{1}-\mathbb{I}_{K}||\leq||(L'_{2}L_{1}-\mathbb{I}_{K})L_{2}||=||L_{2}(L'_{1}L_{2}-\mathbb{I}_{K})||\leq\sigma_{1}(L_{2})||L'_{1}L_{2}-\mathbb{I}_{K}||.\label{eq:C35}
\end{equation}
Since (iv) implies that $\sigma_{1}(L_{2})/\sigma_{K}(L_{2})\lesssim_{P}1$
when $\tilde{K}=K,$ (ii) and (\ref{eq:C35}) yield
\begin{equation}
||L_{1}L'_{2}-\mathbb{I}_{K}||=||L'_{2}L_{1}-\mathbb{I}_{K}||\leq\frac{\sigma_{1}(L_{2})}{\sigma_{K}(L_{2})}||L'_{1}L_{2}-\mathbb{I}_{K}||\lesssim_{\textrm{P}}T^{-1}+q^{-1}N^{-1}.
\end{equation}
\end{proof}
\begin{lem}
\label{lem:11}Under assumptions of Theorem \ref{thm:1}, for $k\leq\tilde{K},$
we have

(i) $||\hat{\varsigma}_{(k)}-\varsigma_{(k)}||\lesssim_{\textrm{P}}T^{-1/2}+q^{-1/2}N^{-1/2}.$

(ii) $q^{-1/2}N^{-1/2}||\hat{\varsigma}'_{(k)}(\hat{\Upsilon}\beta)_{[I_{k}]}-\varsigma'_{(k)}(\hat{\Upsilon}\beta)_{[I_{k}]}||\lesssim_{\textrm{P}}T^{-1/2}+q^{-1/2}N^{-1/2}.$

(iii) $|\hat{\varsigma}'_{(k)}(\hat{\Upsilon}e_{T})_{[I_{k}]}-\varsigma'_{(k)}(\hat{\Upsilon}e_{T})_{[I_{k}]}|\lesssim_{\textrm{P}}T^{-1}q^{1/2}N^{1/2}+q^{-1/2}N^{-1/2}.$

(iv) $||\hat{\varsigma}'_{(k)}\tilde{D}_{(k)}-\varsigma'_{(k)}D_{(k)}||\lesssim_{\textrm{P}}T^{-1/2}+q^{-1/2}N^{-1/2}.$

(v) $q^{-1/2}N^{-1/2}||\hat{\varsigma}'_{(k)}\tilde{D}_{(k)}\hat{\Upsilon}\beta-\varsigma'_{(k)}D_{(k)}\hat{\Upsilon}\beta||\lesssim_{\textrm{P}}T^{-1/2}+q^{-1/2}N^{-1/2}.$

(vi) $|\hat{\varsigma}'_{(k)}\tilde{D}_{(k)}\hat{\Upsilon}e_{T}-\varsigma'_{(k)}D_{(k)}\hat{\Upsilon}e_{T}|\lesssim_{\textrm{P}}T^{-1}q^{1/2}N^{1/2}+q^{-1/2}N^{-1/2}.$
\end{lem}
\begin{proof}
We prove (i) - (vi) by induction. Take the case where $k=1$ into
consideration. The definitions of $\hat{\varsigma}_{(k)}$ in (\ref{eq:B1})
and $\varsigma_{(k)}$ in Section \ref{Section 3.1} lead to
\begin{equation}
\hat{\varsigma}_{(k)}-\varsigma_{(k)}=T_{h}^{-1/2}\hat{\lambda}_{(k)}^{-1/2}(\tilde{D}_{(k)}\hat{\Upsilon}\beta \mathbf{f}\hat{\xi}_{(k)}+\tilde{D}_{(k)}\hat{\Upsilon}\mathbf{e}\hat{\xi}_{(k)})-\lambda_{(k)}^{-1/2}D_{(k)}\hat{\Upsilon}\beta\iota_{(k)},\label{eq:C37}
\end{equation}
when $k=1$, as $\tilde{D}_{(1)}=D_{(1)}=(\mathbb{I}_{N})_{[I_{1}]}$,
(\ref{eq:C37}) becomes
\begin{equation}
\hat{\varsigma}_{(1)}-\varsigma_{(1)}=(T_{h}^{-1/2}\hat{\lambda}_{(1)}^{-1/2}(\hat{\Upsilon}\beta)_{(1)}\mathbf{f}\hat{\xi}_{(1)}-\lambda_{(k)}^{-1/2}(\hat{\Upsilon}\beta)_{(1)}\iota_{(1)})+T_{h}^{-1/2}\hat{\lambda}_{(1)}^{-1/2}(\hat{\Upsilon}\mathbf{e})_{[I_{1}]}\hat{\xi}_{(1)}.\label{eq:C38}
\end{equation}
As Lemma \ref{lem:2} (iii) and (iv), and Lemma \ref{lem:16} imply
that $||T_{h}^{-1/2}\hat{\lambda}_{(1)}^{-1/2}\mathbf{f}\hat{\xi}_{(1)}-\lambda_{(k)}^{-1/2}\iota_{(1)}||\lesssim_{\textrm{P}}q^{-1}N^{-1}+T^{-1/2}q^{-1/2}N^{-1/2}$
and $||(\hat{\Upsilon}\beta)_{(1)}||\lesssim q^{1/2}N^{1/2}$. To
prove (i), it is sufficient to show that $T_{h}^{-1/2}q^{-1/2}N^{-1/2}||\hat{\Upsilon}_{[I_{1}]}||$

\noindent$||\mathbf{e}_{[I_{1}]}||\lesssim_{\textrm{P}}T^{-1/2}+q^{-1/2}N^{-1/2}$,
which is given by Assumption \ref{assu:3}. 

(ii) Since $||(\hat{\Upsilon}\beta)_{(1)}||\lesssim q^{1/2}N^{1/2}$,
(ii) is deduced by (i).

(iii) Left-multiplying (\ref{eq:C38}) by $(\hat{\Upsilon}e_{T})'_{[I_{1}]}$,
as $||(e_{T})'_{[I_{1}]}\beta_{[I_{1}]}||\lesssim q^{1/2}N^{1/2},$
to prove (iii) when $k=1$, it is sufficient to show that $T_{h}^{-1/2}q^{-1/2}N^{-1/2}||(\hat{\Upsilon}e_{T})'_{[I_{1}]}(\hat{\Upsilon}\mathbf{e})_{[I_{1}]}\hat{\xi}_{(1)}||\lesssim_{\textrm{P}}T^{-1}q^{1/2}N^{1/2}+q^{-1/2}N^{-1/2}$,
which is implied by Lemma \ref{lem:12}.

(iv)-(vi) are equivalent to (i)-(iii) when $k=1$ as $\tilde{D}_{(1)}=D_{(1)}=(\mathbb{I}_{N})_{[I_{1}]}.$

Next, We assume that (i) - (vi) hold for $i<k$ and proceed to prove
that (i) - (vi) hold for $k$ as well.

(i) Similar to the case of $k=1$, using (\ref{eq:C37}) and Lemma
\ref{lem:2}(iii)(iv), it is sufficient to show that $T_{h}^{-1/2}q^{-1/2}N^{-1/2}||\tilde{\mathbf{e}}_{[I_{k}]}||\lesssim_{\textrm{P}}T^{-1/2}+q^{-1/2}N^{-1/2}$
and $q^{-1/2}N^{-1/2}||(\tilde{D}_{(k)}-D_{(k)})\hat{\Upsilon}\beta||\lesssim_{\textrm{P}}T^{-1/2}+q^{-1/2}N^{-1/2}$.
The first inequality is the same as the case of $k=1$, which is implied
by Lemma \ref{lem:5}. As to the second inequality, write
\[
(\tilde{D}_{(k)}-D_{(k)})\hat{\Upsilon}\beta=\underset{i<k}{\sum}\left(\frac{(\hat{\Upsilon}\beta)_{[I_{k}]}\iota_{(i)}}{\sqrt{\lambda_{(i)}}}\varsigma'_{(i)}D_{(i)}\hat{\Upsilon}\beta-\frac{(\hat{\Upsilon}\mathbf{x})_{[I_{k}]}\hat{\xi}_{(i)}}{\sqrt{T_{h}\hat{\lambda}_{(i)}}}\hat{\varsigma}'_{(i)}\tilde{D}_{(i)}\hat{\Upsilon}\beta\right).
\]
As (v) holds for $i<k$, it is sufficient to show that
\begin{equation}
||\frac{(\hat{\Upsilon}\beta)_{[I_{k}]}\iota_{(i)}}{\sqrt{\lambda_{(i)}}}-\frac{(\hat{\Upsilon}\mathbf{x})_{[I_{k}]}\hat{\xi}_{(i)}}{\sqrt{T_{h}\hat{\lambda}_{(i)}}}||\lesssim_{\textrm{P}}T^{-1/2}+q^{-1/2}N^{-1/2}.\label{eq:C39}
\end{equation}
Plugging $(\hat{\Upsilon}\mathbf{x})_{[I_{k}]}=(\hat{\Upsilon}\beta)_{[I_{k}]}\mathbf{f}+(\hat{\Upsilon}\mathbf{e})_{[I_{k}]}$
into (\ref{eq:C39}) and using Lemma \ref{lem:2}(iii)(iv) again,
we solely need to show that $T_{h}^{-1/2}q^{-1/2}N^{-1/2}||\hat{\Upsilon}_{[I_{k}]}||\,||\mathbf{e}_{[I_{k}]}\hat{\xi}_{(i)}||\lesssim_{\textrm{P}}T^{-1/2}+q^{-1/2}N^{-1/2}$,
which holds by Assumption \ref{assu:3} and $||\hat{\xi}_{(i)}||=1$.

(ii) It is implied by (i) as $||(\hat{\Upsilon}\beta)_{[I_{k}]}||\lesssim q^{1/2}N^{1/2}.$

(iii) By (\ref{eq:C37}), we have
\[
\begin{split}(\hat{\Upsilon}e_{T})'_{[I_{k}]}\hat{\varsigma}{}_{(k)}-(\hat{\Upsilon}e_{T})'_{[I_{k}]}\varsigma'_{(k)}&=(\hat{\Upsilon}e_{T})'_{[I_{k}]}(T_{h}^{-1/2}\hat{\lambda}_{(k)}^{-1/2}\tilde{D}_{(k)}\hat{\Upsilon}\beta \mathbf{f}\hat{\xi}_{(k)}-\lambda_{(k)}^{-1/2}D_{(k)}\hat{\Upsilon}\beta\iota_{(k)})\\&\quad+T_{h}^{-1/2}\hat{\lambda}_{(k)}^{-1/2}(\hat{\Upsilon}e_{T})'_{[I_{k}]}\tilde{\mathbf{e}}_{(k)}\hat{\xi}_{(k)}.\end{split}
\]
As in the case of $k=1$, for the second term, we have $T_{h}^{-1/2}\hat{\lambda}_{(k)}^{-1/2}||(\hat{\Upsilon}e_{T})'_{[I_{k}]}\tilde{\mathbf{e}}_{(k)}\hat{\xi}_{(k)}||\lesssim_{\textrm{P}}T^{-1}q^{1/2}N^{1/2}+q^{-1/2}N^{-1/2}$,
as is deduced by Lemma \ref{lem:12}(iv). For the first term, similar
to the proof of (i), using Lemma \ref{lem:2}(iii)(iv), it is sufficient
to show that
\[
q^{-1/2}N^{-1/2}||(\hat{\Upsilon}e_{T})'_{[I_{k}]}(\tilde{D}_{(k)}-D_{(k)})\hat{\Upsilon}\beta||\lesssim_{\textrm{P}}T^{-1}q^{1/2}N^{1/2}+q^{-1/2}N^{-1/2}.
\]
Then, we can write
\[
\begin{split}(\hat{\Upsilon}e_{T})'_{[I_{k}]}(\tilde{D}_{(k)}-D_{(k)})\hat{\Upsilon}\beta=\underset{i<k}{\sum}\Big(&\frac{(\hat{\Upsilon}e_{T})'_{[I_{k}]}(\hat{\Upsilon}\beta)_{[I_{k}]}\iota_{(i)}}{\sqrt{\lambda_{(i)}}}\varsigma'_{(i)}D_{(i)}\hat{\Upsilon}\beta\\&-\frac{(\hat{\Upsilon}e_{T})'_{[I_{k}]}((\hat{\Upsilon}\beta)_{[I_{k}]}\mathbf{f}+(\hat{\Upsilon}\mathbf{e})_{[I_{k}]})\hat{\xi}_{(i)}}{\sqrt{T_{h}\hat{\lambda}_{(i)}}}\hat{\varsigma}'_{(i)}\tilde{D}_{(i)}\hat{\Upsilon}\beta\Big).\end{split}
\]
As (v) holds for $i<k$, we solely need to show that
\begin{equation}
||\frac{(\hat{\Upsilon}e_{T})'_{[I_{k}]}(\hat{\Upsilon}\beta)_{[I_{k}]}\iota_{(i)}}{\sqrt{\lambda_{(i)}}}-\frac{(\hat{\Upsilon}e_{T})'_{[I_{k}]}((\hat{\Upsilon}\beta)_{[I_{k}]}\mathbf{f}+(\hat{\Upsilon}\mathbf{e})_{[I_{k}]})\hat{\xi}_{(i)}}{\sqrt{T_{h}\hat{\lambda}_{(i)}}}||\lesssim_{\textrm{P}}T^{-1}q^{1/2}N^{1/2}+q^{-1/2}N^{-1/2}.
\end{equation}
Using Lemma \ref{lem:2}(iii)(iv) again, it is sufficient to show
\[
T_{h}^{-1/2}q^{-1/2}N^{-1/2}||(\hat{\Upsilon}e_{T})'_{[I_{k}]}(\hat{\Upsilon}\mathbf{e})_{[I_{k}]}\hat{\xi}_{(i)}||\lesssim_{\textrm{P}}T^{-1}q^{1/2}N^{1/2}+q^{-1/2}N^{-1/2},
\]
which holds by Lemma \ref{lem:12}(iii).

(iv) Using basic algebra,we have
\[
\begin{split}\hat{\varsigma}'_{(k)}\tilde{D}_{(k)}-\varsigma'_{(k)}D_{(k)}=(\hat{\varsigma}'_{(k)}-\varsigma'_{(k)})(\mathbb{I}_{N})_{[I_{k}]}+\underset{i<k}{\sum}\Big(&\frac{\varsigma'_{(k)}(\hat{\Upsilon}\beta)_{[I_{k}]}\iota_{(i)}}{\sqrt{\lambda_{(i)}}}\varsigma'_{(i)}D_{(i)}\\&-\frac{\hat{\varsigma}'_{(k)}(\hat{\Upsilon}\mathbf{x})_{[I_{k}]}\hat{\xi}_{(i)}}{\sqrt{T_{h}\hat{\lambda}_{(i)}}}\hat{\varsigma}'_{(i)}\tilde{D}_{(i)}\Big).\end{split}
\]
Using the fact that (i) holds for $k$, (iii) holds for $i<k$ and
(\ref{eq:C39}), the proof is finished.

(v) Analogously, we have
\[
\begin{split}\hat{\varsigma}'_{(k)}\tilde{D}_{(k)}\hat{\Upsilon}\beta-\varsigma'_{(k)}D_{(k)}\hat{\Upsilon}\beta=\left(\hat{\varsigma}'_{(k)}\hat{\Upsilon}\beta-\varsigma'_{(k)}\hat{\Upsilon}\beta\right)+\underset{i<k}{\sum}\Big(&\frac{\varsigma'_{(k)}(\hat{\Upsilon}\beta)_{[I_{k}]}\iota_{(i)}}{\sqrt{\lambda_{(i)}}}\varsigma'_{(i)}D_{(i)}\hat{\Upsilon}\beta\\&-\frac{\hat{\varsigma}'_{(k)}(\hat{\Upsilon}\mathbf{x})_{[I_{k}]}\hat{\xi}_{(i)}}{\sqrt{T_{h}\hat{\lambda}_{(i)}}}\hat{\varsigma}'_{(i)}\tilde{D}_{(i)}\hat{\Upsilon}\beta\Big).\end{split}
\]

Using the fact that (i) holds for $k$, (v) holds for $i<k$ and (\ref{eq:C39}),
the proof is finished.

(vi) Replace $q^{-1/2}N^{-1/2}\hat{\Upsilon}\beta$ by $\hat{\Upsilon}e_{T}$
in the proof of (v), we can get (vi).
\end{proof}
\begin{lem}
\label{lem:12}Under assumptions of Theorem \ref{thm:1}, for $k,j\leq\tilde{K},$
we have

(i) $||(\hat{\Upsilon}\beta)'_{[I_{k}]}(\hat{\Upsilon}\mathbf{e})_{[I_{k}]}\hat{\xi}_{(j)}||\lesssim_{\textrm{P}}q^{1/2}N^{1/2}+T^{1/2}.$

(ii) $||(\hat{\Upsilon}\beta)'_{[I_{k}]}\tilde{\mathbf{e}}_{(k)}\hat{\xi}_{(j)}||\lesssim_{\textrm{P}}q^{1/2}N^{1/2}+T^{1/2}.$

(iii) $||(\hat{\Upsilon}e_{T})'_{[I_{k}]}(\hat{\Upsilon}\mathbf{e})_{[I_{k}]}\hat{\xi}_{(j)}||\lesssim_{\textrm{P}}T^{-1/2}qN+T^{1/2}.$

(iv) $||(\hat{\Upsilon}e_{T})'_{[I_{k}]}\tilde{\mathbf{e}}_{(k)}\hat{\xi}_{(j)}||\lesssim_{\textrm{P}}T^{-1/2}qN+T^{1/2}.$
\end{lem}
\begin{proof}
(i) With Lemma \ref{lem:9}, Lemma \ref{lem:16}, and Assumption \ref{assu:4},
we have
\[
\begin{split}||(\hat{\Upsilon}\beta)'_{[I_{k}]}(\hat{\Upsilon}\mathbf{e})_{[I_{k}]}\hat{\xi}_{(j)}||&\leq T_{h}^{-1/2}||(\hat{\Upsilon}\beta)'_{[I_{k}]}(\hat{\Upsilon}\mathbf{e})_{[I_{k}]}\mathbf{f}'\iota_{j2}||+||(\hat{\Upsilon}\beta)'_{[I_{k}]}(\hat{\Upsilon}\mathbf{e})_{[I_{k}]}||\bigparallel T_{h}^{-1/2}\mathbf{f}'\iota_{k2}-\hat{\xi}_{(j)}\bigparallel\\&\leq T_{h}^{-1/2}||\hat{\Upsilon}||^{2}||\beta'_{[I_{k}]}(\hat{\Upsilon}\mathbf{e})_{[I_{k}]}\mathbf{f}'\iota_{j2}||+||\hat{\Upsilon}||^{2}||\beta'_{[I_{k}]}\mathbf{e}{}_{[I_{k}]}||\bigparallel T_{h}^{-1/2}\mathbf{f}'\iota_{k2}-\hat{\xi}_{(j)}\bigparallel\\&\lesssim_{\textrm{P}}q^{1/2}N^{1/2}+T^{1/2}.\end{split}
\] 

(ii) Assumptions \ref{assu:1}, \ref{assu:2}, \ref{assu:3}, and
$||\hat{\Upsilon}||\lesssim_{\textrm{P}}1$ imply that $(\hat{\Upsilon}\mathbf{x})_{[I_{k}]}\leq||\hat{\Upsilon}||\,||\beta{}_{[I_{k}]}||\,||\mathbf{f}||+||\hat{\Upsilon}||\,||\mathbf{e}_{[I_{k}]}||\lesssim_{P}q^{1/2}N^{1/2}T^{1/2}.$
Combine with (i) and Lemma \ref{lem:5}(iii), we have
\[
\begin{split}||(\hat{\Upsilon}\beta)'_{[I_{k}]}\tilde{\mathbf{e}}_{(k)}\hat{\xi}_{(j)}||&\leq||(\hat{\Upsilon}\beta)'_{[I_{k}]}(\hat{\Upsilon}\mathbf{e})_{[I_{k}]}\hat{\xi}_{(j)}||\\&\quad+\stackrel[i=1]{k-1}{\sum}||\hat{\Upsilon}||\,||\beta{}_{[I_{k}]}||\bigparallel\frac{(\hat{\Upsilon}\mathbf{x})_{[I_{k}]}\hat{\xi}_{(i)}}{\sqrt{T_{h}\hat{\lambda}_{(i)}}}\bigparallel||\hat{\varsigma}'_{(i)}\tilde{\mathbf{e}}_{(i)}\hat{\xi}_{(j)}||\lesssim_{\textrm{P}}q^{1/2}N^{1/2}+T^{1/2}.\end{split}
\]

(iii) With Lemma \ref{lem:9} and Assumption \ref{assu:4}, we have
\[
||(\hat{\Upsilon}e_{T})'_{[I_{k}]}(\hat{\Upsilon}\mathbf{e})_{[I_{k}]}\hat{\xi}_{(j)}||\leq T_{h}^{-1/2}||(\hat{\Upsilon}e_{T})'_{[I_{k}]}(\hat{\Upsilon}\mathbf{e})_{[I_{k}]}\mathbf{f}'\iota_{j2}||
\]
\[
+||(\hat{\Upsilon}e_{T})'_{[I_{k}]}(\hat{\Upsilon}\mathbf{e})_{[I_{k}]}||\bigparallel T_{h}^{-1/2}\mathbf{f}'\iota_{k2}-\hat{\xi}_{(j)}\bigparallel
\]
\[
T_{h}^{-1/2}||\hat{\Upsilon}||^{2}||(e_{T})'_{[I_{k}]}\mathbf{e}{}_{[I_{k}]}\mathbf{f}'\iota_{j2}||+||\hat{\Upsilon}||^{2}||(e_{T})'_{[I_{k}]}\mathbf{e}_{[I_{k}]}||\bigparallel T_{h}^{-1/2}\mathbf{f}'\iota_{k2}-\hat{\xi}_{(j)}\bigparallel\lesssim_{\textrm{P}}T^{-1/2}qN+T^{1/2}.
\]

(iv) Similar to (ii), with Lemma \ref{lem:5}, $||(\hat{\Upsilon}\mathbf{x})_{[I_{k}]}||\lesssim_{\textrm{P}}q^{1/2}N^{1/2}T^{1/2},$
and (iii), we have $||(\hat{\Upsilon}e_{T})'_{[I_{k}]}\tilde{\mathbf{e}}_{(k)}\hat{\xi}_{(j)}||\lesssim_{\textrm{P}}T^{-1/2}qN+T^{1/2}.$ 
\end{proof}
\begin{lem}
\label{lem:13}Under Assumptions \ref{assu:1}-\ref{assu:7}, we have

(i) $||T_{h}^{1/2}\hat{\xi}_{(k)}-\iota'_{k2}\mathbf{f}||_{\max}\lesssim_{\textrm{P}}q^{-1/2}N^{-1/2}(\log T)^{1/2}+T^{-1/2}+q^{-1}N^{-1}T^{1/2}.$

(ii) $||\hat{\lambda}_{(k)}^{1/2}\hat{\beta}_{(k)}-\hat{\Upsilon}\beta\iota_{k1}||_{\max}\lesssim_{\textrm{P}}q^{-1/2}(\log N)^{1/2}.$

(iii) $||\hat{\mathbf{e}}-\mathbf{e}||_{\max}\lesssim_{\textrm{P}}q^{-1/2}N^{-1/2}(\log T)^{1/2}+T^{-1/2}(\log NT)^{1/2}+q^{-1}N^{-1}T^{1/2}$.

(iv) $\max_{i\leq N}T_{h}^{-1/2}||\hat{\mathbf{e}}_{[i]}-\mathbf{e}_{[i]}||\lesssim_{\textrm{P}}T^{-1/2}(\log N)^{1/2}+q^{-1}N^{-1}.$
\end{lem}
\begin{proof}
(i) Recall that by (\ref{eq:C31}), (\ref{eq:B1}), and (\ref{eq:B3}),
we have $T_{h}^{1/2}\hat{\xi}_{(k)}-\iota'_{k2}\mathbf{f}=\hat{\lambda}_{(k)}^{-1/2}\hat{\varsigma}'_{(k)}\tilde{\mathbf{e}}_{(k)},$
and $\hat{\varsigma}_{(k)}=T^{-1/2}\hat{\lambda}_{(k)}^{-1/2}\tilde{\mathbf{x}}_{(k)}\hat{\xi}_{(k)}=T^{-1/2}\hat{\lambda}_{(k)}^{-1/2}(\hat{\Upsilon}\mathbf{x})_{[I_{k}]}\hat{\xi}_{(k)}.$
Thus, we have
\begin{equation}
||T_{h}^{1/2}\hat{\xi}_{(k)}-\iota'_{k2}\mathbf{f}||_{\max}\lesssim_{\textrm{P}}q^{-1}N^{-1}T^{-1/2}\left(||\hat{\xi}'_{(k)}\mathbf{f}'(\hat{\Upsilon}\beta')_{[I_{k}]}\tilde{\mathbf{e}}_{(k)}||_{\max}+||\hat{\xi}'_{(k)}(\hat{\Upsilon}\mathbf{e}')_{[I_{k}]}\tilde{\mathbf{e}}_{(k)}||_{\max}\right).\label{eq:C41}
\end{equation}
When $k=1,$ $\tilde{\mathbf{e}}_{(1)}=\hat{\Upsilon}\mathbf{e}_{[I]},$ with $||\beta'_{[I_{1}]}\tilde{\mathbf{e}}_{(1)}||_{\max}\lesssim_{\textrm{P}}q^{1/2}N^{1/2}(\log T)^{1/2}$
from Assumption \ref{assu:4}, $||\mathbf{e}_{[I_{1}]}||\lesssim_{\textrm{P}}q^{1/2}N^{1/2}+T^{1/2}$
from Assumption \ref{assu:3}, and $||\hat{\Upsilon}||\lesssim_{\textrm{P}}1$,
we have
\[
||T_{h}^{1/2}\hat{\xi}_{(1)}-\iota'_{12}\mathbf{f}||_{\max}\lesssim_{\textrm{P}}q^{-1}N^{-1}||\beta'_{[I_{1}]}\mathbf{e}_{[I_{1}]}||_{\max}+q^{-1}N^{-1}T^{-1/2}||\mathbf{e}'_{[I_{1}]}\mathbf{e}_{[I_{1}]}||
\]
\[
\lesssim_{\textrm{P}}q^{-1/2}N^{-1/2}(\log T)^{1/2}+T^{-1/2}+q^{-1}N^{-1}T^{1/2}.
\]
Now assume that this property holds for $i<k$, then for the first
term in (\ref{eq:C41}), we have
\[
||(\hat{\Upsilon}\beta')_{[I_{k}]}\tilde{\mathbf{e}}_{(k)}||_{\max}\lesssim||\beta'_{[I_{k}]}\mathbf{e}_{[I_{k}]}||_{\max}+\underset{i<k}{\sum}||\beta_{[I_{k}]}||\,||T_{h}^{-1/2}\hat{\lambda}_{(i)}^{-1/2}(\hat{\Upsilon}\mathbf{x})_{[I_{k}]}\hat{\xi}_{(i)}||\,||\hat{\varsigma}'_{(i)}\tilde{\mathbf{e}}_{(i)}||_{\max}.
\]
The assumption that (i) holds for $i<k$ hints that
\[
||\hat{\varsigma}'_{(i)}\tilde{\mathbf{e}}_{(i)}||_{\max}=\hat{\lambda}_{(k)}^{1/2}||T_{h}^{1/2}\hat{\xi}_{(k)}-\iota'_{k2}\mathbf{f}||_{\max}\lesssim_{\textrm{P}}(\log T)^{1/2}+q^{1/2}N^{1/2}T^{-1/2}+q^{-1/2}N^{-1/2}.
\]
With $||\beta_{[I_{k}]}||\lesssim q^{1/2}N^{1/2}$ and $||(\hat{\Upsilon}\mathbf{x})_{[I_{k}]}||\leq||\hat{\Upsilon}||\,||\beta_{[I_{k}]}||\,||\mathbf{f}||+||\hat{\Upsilon}||\,||\mathbf{e}_{[I_{k}]}||\lesssim_{\textrm{P}}q^{1/2}N^{1/2}T^{1/2}$
and Assumption \ref{assu:4}(ii), we have the first term in (\ref{eq:C41})
satisfies
\[
q^{-1}N^{-1}T^{-1/2}||\hat{\xi}'_{(k)}\mathbf{f}'(\hat{\Upsilon}\beta')_{[I_{k}]}\tilde{\mathbf{e}}_{(k)}||_{\max}\lesssim q^{-1}N^{-1}T^{-1/2}||\hat{\xi}'_{(k)}\mathbf{f}'||\,||(\hat{\Upsilon}\beta')_{[I_{k}]}\tilde{\mathbf{e}}_{(k)}||_{\max}
\]
\[
\lesssim_{\textrm{P}}q^{-1}N^{-1}||\beta'_{[I_{k}]}\mathbf{e}_{[I_{k}]}||_{\max}+\underset{i<k}{\sum}q^{-1/2}N^{-1/2}||\hat{\varsigma}'_{(i)}\tilde{\mathbf{e}}_{(i)}||_{\max}
\]
\[
\lesssim_{\textrm{P}}q^{-1/2}N^{-1/2}(\log T)^{1/2}+T^{-1/2}+q^{-1}N^{-1}T^{1/2}.
\]
For the second term in (\ref{eq:C41}), we have
\[
\begin{split}||\hat{\xi}'_{(k)}(\hat{\Upsilon}\mathbf{e}')_{[I_{k}]}\tilde{\mathbf{e}}_{(k)}||_{\max}&\leq||\hat{\xi}'_{(k)}(\hat{\Upsilon}\mathbf{e}')_{[I_{k}]}\tilde{\mathbf{e}}_{(k)}||\\&\leq||\hat{\Upsilon}||\,||\mathbf{e}_{[I_{k}]}||\,||\tilde{\mathbf{e}}_{(k)}||\lesssim_{\textrm{P}}q^{1/2}N^{1/2}T^{1/2}(\log T)^{1/2}+T,\end{split}
\]
where we use Assumption \ref{assu:3} and Lemma \ref{lem:5}(i) in
the last step. Therefore, (i) also holds for $k$ and this concludes
the proof by induction.

(ii) By simple algebra, $\hat{\lambda}_{(k)}^{1/2}\hat{\beta}_{(k)}=T_{h}^{-1/2}\hat{\Upsilon}\mathbf{x}\hat{\xi}_{(k)}=\hat{\Upsilon}\beta\iota_{k1}+T_{h}^{-1/2}\hat{\Upsilon}\mathbf{e}\hat{\xi}_{(k)},$
which leads to
\[
||\hat{\lambda}_{(k)}^{1/2}\hat{\beta}_{(k)}-\hat{\Upsilon}\beta\iota_{k1}||_{\max}\lesssim T_{h}^{-1}||\mathbf{ef}'\iota_{k2}||_{\max}+T_{h}^{-1}||\mathbf{e}(T_{h}^{1/2}\hat{\Upsilon}\hat{\xi}_{(k)}-\mathbf{f}'\iota_{k2})||_{\max}
\]
\[
\lesssim_{\textrm{P}}T_{h}^{-1}||\mathbf{ef}'||_{\max}+T_{h}^{-1/2}||\mathbf{e}||_{\max}||T_{h}^{1/2}\hat{\xi}_{(k)}-\mathbf{f}'\iota_{k2}||\lesssim_{P}T^{-1/2}(\log N)^{1/2},
\]
where we use Assumptions \ref{assu:3}, \ref{assu:4}, and Lemma \ref{lem:9}.

(iii) By triangle inequality, we have
\[
||\hat{\mathbf{e}}-\mathbf{e}||_{\max}\leq||\beta\left(\underset{k\leq K}{\sum}\iota_{k1}\iota'_{k2}-\mathbb{I}_{K}\right)\mathbf{f}||_{\max}+\underset{k\leq K}{\sum}||\hat{\beta}_{(k)}\hat{\mathbf{f}}_{(k)}-\beta\iota_{k1}\iota'_{k2}\mathbf{f}||_{\max}.
\]
By Assumptions \ref{assu:1}, \ref{assu:2} and Lemma \ref{lem:10},
the first term satisfies
\[
\parallel\beta\left(\underset{k\leq K}{\sum}\iota_{k1}\iota'_{k2}-\mathbb{I}_{K}\right)\mathbf{f}\parallel_{\max}\lesssim||\beta||_{\max}||\mathbf{f}||_{\max}||L_{1}L'_{2}-\mathbb{I}_{K}||\lesssim_{\textrm{P}}(\log T)^{1/2}(T^{-1}+q^{-1}N^{-1}).
\]
For the second term, observe that the triangle inequality implies
that we have
\[
||\hat{\beta}_{(k)}\hat{\mathbf{f}}_{(k)}-\beta\iota_{k1}\iota'_{k2}\mathbf{f}||_{\max}\leq||\hat{\lambda}_{(k)}^{1/2}\hat{\beta}-\beta\iota_{k1}||_{\max}||\iota'_{k2}\mathbf{f}||_{\max}
\]
\[
+||\beta\iota_{k1}||_{\max}||\iota'_{k2}\mathbf{f}-\hat{\lambda}_{(k)}^{-1/2}\hat{\mathbf{f}}_{(k)}||_{\max}+||\hat{\lambda}_{(k)}^{1/2}\hat{\beta}-\beta\iota_{k1}||_{\max}||\iota'_{k2}\mathbf{f}-\hat{\lambda}_{(k)}^{-1/2}\hat{\mathbf{f}}_{(k)}||_{\max},
\]
which, combine with (i)(ii), conclude the proof.

(iv) Since we have $T_{h}^{-1/2}||\hat{\mathbf{e}}_{[i]}-\mathbf{e}_{[i]}||=T_{h}^{-1/2}||\hat{\beta}_{[i]}\hat{\mathbf{f}}-\hat{\Upsilon}\beta_{[i]}\mathbf{f}||,$
it then follows from the triangle inequality that
\[
||\hat{\beta}_{[i]}\hat{\mathbf{f}}-\hat{\Upsilon}\beta_{[i]}\mathbf{f}||\leq||\beta_{[i]}\left(L_{1}L'_{2}-\mathbb{I}_{K}\right)\mathbf{f}||+||\beta_{[i]}L_{1}||\,||\hat{\Lambda}^{-1/2}\hat{\mathbf{f}}-L'_{2}\mathbf{f}||
\]
\[
+||\beta_{[i]}L_{1}-\hat{\beta}_{[i]}\hat{\Lambda}^{1/2}||\,||\hat{\Lambda}^{-1/2}\hat{\mathbf{f}}||.
\]
We analyze each of the three terms on the right-hand side individually.
With $T_{h}^{-1/2}||\hat{\Lambda}^{-1/2}\hat{\mathbf{f}}-L'_{2}\mathbf{f}||\lesssim_{\textrm{P}}T^{-1}+q^{-1/2}N^{-1/2}$
from Lemma \ref{lem:9}, $||\beta||_{\max}\lesssim1$, Lemma \ref{lem:10}
and (ii), we have
\[
\underset{i}{\max}T_{h}^{-1/2}||\beta_{[i]}\left(L_{1}L'_{2}-\mathbb{I}_{K}\right)\mathbf{f}||\lesssim||\beta||_{\max}||L_{1}L'_{2}-\mathbb{I}_{K}||\,||\mathbf{f}||\lesssim_{\textrm{P}}q^{-1}N^{-1}+T^{-1},
\]
\[
\underset{i}{\max}T_{h}^{-1/2}||\beta_{[i]}L_{1}||\,||\hat{\Lambda}^{-1/2}\hat{\mathbf{f}}-L'_{2}\mathbf{f}||\lesssim||\beta||_{\max}||L_{1}||\,||\hat{\Lambda}^{-1/2}\hat{\mathbf{f}}-L'_{2}\mathbf{f}||\lesssim_{\textrm{P}}q^{-1/2}N^{-1/2}+T^{-1}
\]
\[
\underset{i}{\max}T_{h}^{-1/2}||\beta_{[i]}L_{1}-\hat{\beta}_{[i]}\hat{\Lambda}^{1/2}||\,||\hat{\Lambda}^{-1/2}\hat{\mathbf{f}}||\lesssim T_{h}^{-1/2}||\beta L_{1}-\hat{\beta}\hat{\Lambda}^{1/2}||_{\max}||\mathbf{f}||\lesssim_{\textrm{P}}T^{-1/2}(\log N)^{1/2}.
\]

Therefore, we derive the desired bound.
\end{proof}
\begin{lem}
\label{lem:14}Let $\mathsection^{o}$ and $\mathcal{O}^{o}$ be 
\begin{equation}
\mathsection^{o}=\frac{\sum_{i=1}^{N_{0}}\phi_{i}^{o}\psi_{i}^{o}}{\sum_{i=1}^{N_{0}}(\phi_{i}^{o2}-\psi_{i}^{o2})}=\frac{\sum_{i=1}^{N_{0}}\phi_{i}^{3}\psi_{i}}{\sum_{i=1}^{N_{0}}\phi_{i}^{2}(\phi_{i}^{2}-\psi_{i}^{2})'},
\end{equation}
and
\begin{equation}
\mathcal{O}^{o}=\begin{cases}
\frac{1+\sqrt{1+4\mathsection^{o2}}}{\sqrt{4\mathsection^{o2}+(1+\sqrt{1+4\mathsection^{o2}})^{2}}} & \begin{array}{cc}
if & \sum_{i=1}^{N_{0}}(\phi_{i}^{o2}-\psi_{i}^{o2})>0;\end{array}\\
\frac{\sqrt{1+4\mathsection^{o2}}-1}{\sqrt{4\mathsection^{o2}+(1-\sqrt{1+4\mathsection^{o2}})^{2}}} & \begin{array}{cc}
if & \sum_{i=1}^{N_{0}}(\phi_{i}^{o2}-\psi_{i}^{o2})<0\end{array}.
\end{cases}
\end{equation}
If the first principal component of $\{\hat{\Upsilon}_{i}x_{i,t_{h}}\}$,
$i\in N_{0}=\mid I_{0}\mid$, is used for forecasting, under Assumptions
\ref{assu:1}-\ref{assu:7}, as $N_{0}\rightarrow\infty,T\rightarrow\infty,$
and $\sqrt{N_{0}}/T\rightarrow0$, then the asymptotic MSFE is 
\begin{equation}
\textrm{MSF\ensuremath{\textrm{E}_{\textrm{SsPCA}}}}=\beta_{[I_{0}]}^{2}(1-\mathcal{O}^{o^{2}})+\sigma_{\varepsilon}^{2},\quad where\quad\sigma_{\varepsilon}^{2}=\underset{T\rightarrow\infty}{\textrm{p}\lim}\frac{1}{T}\stackrel[t=1]{T}{\sum}\varepsilon_{t+h}^{2}.
\end{equation}
\end{lem}
\begin{proof}
The proof is identical to that of Lemma 5 in \cite{huang2023bond}, with the full cross-section $N$ replaced by the informative subset $N_0$ throughout. We therefore omit the details.
\end{proof}
\begin{lem}
\label{lem:15}Let $\mathsection=\sum_{i=1}^{N_{0}}\phi_{i}\psi_{i}/\sum_{i=1}^{N_{0}}(\phi_{i}^{2}-\psi_{i}^{2})$
and 
\begin{equation}
\mathcal{O}=\begin{cases}
\frac{1+\sqrt{1+4\mathsection^{2}}}{\sqrt{4\mathsection^{2}+(1+\sqrt{1+4\mathsection^{2}})^{2}}} & \begin{array}{cc}
if & \sum_{i=1}^{N_{0}}(\phi_{i}^{2}-\psi_{i}^{2})>0;\end{array}\\
\frac{\sqrt{1+4\mathsection^{2}}-1}{\sqrt{4\mathsection^{2}+(1-\sqrt{1+4\mathsection^{2}})^{2}}} & \begin{array}{cc}
if & \sum_{i=1}^{N_{0}}(\phi_{i}^{2}-\psi_{i}^{2})<0\end{array}.
\end{cases}
\end{equation}
If the first principal component of $\{x_{i,t_{h}}\}$, $i\in N_{0}=\mid I_{0}\mid$,
is used for forecasting, under Assumptions \ref{assu:1}-\ref{assu:7},
as $N_{0}\rightarrow\infty,T\rightarrow\infty,$ and $\sqrt{N_{0}}/T\rightarrow0$,
then the asymptotic MSFE is 
\begin{equation}
\textrm{MSF}\textrm{E}_{\textrm{SPCA}}=\beta_{[I_{0}]}^{2}(1-\mathcal{O}^{2})+\sigma_{\varepsilon}^{2},\quad where\quad\sigma_{\varepsilon}^{2}=\underset{T\rightarrow\infty}{\textrm{p}\lim}\frac{1}{T}\stackrel[t=1]{T}{\sum}\varepsilon_{t+h}^{2}.
\end{equation}
\end{lem}
\begin{proof}
The proof is identical to that of Lemma 6 in \cite{huang2023bond}, and it should be noticed that the principal components we extract are based on a
subset of size $N_{0}$.
\end{proof}
\begin{lem}
\label{lem:16}Under assumptions of Theorem \ref{thm:1}, we have

(i) $||\hat{\Upsilon}-\Upsilon||\lesssim_{\textrm{P}}T^{-1/2}.$

(ii) $||\hat{\Upsilon}(\hat{\theta}_{x})-\Upsilon(\theta_{x})||\lesssim_{\textrm{P}}T^{-1/2}.$

(iii) $||\hat{\kappa}-\kappa||\lesssim_{\textrm{P}}T^{-1/2}, \kappa=(\alpha'_{l},\theta')'$ with $\alpha_{l}=(\alpha,\alpha_{w})', \theta=(\theta_{f},\theta_{x})'.$
\end{lem}
\begin{proof}
(i) The scaling coefficients come from the regression of $\mathbf{y}$ on each
(standardized) predictor $\mathbf{x}_{i}$. Without loss of generality, we
assume the data are centered, which removes the intercepts from the
model. At the same frequency, the model can be written as
\[
\begin{array}{c}
\mathbf{x}=\beta\mathbf{f}+\mathbf{e},\\
\mathbf{y}=\alpha\mathbf{f}+\epsilon.
\end{array}
\]

Given the the normalization conditions $\frac{1}{T_{h}}\sum_{t=1}^{T_{h}}(\mathbf{x}_{i,t}-\overline{\mathbf{x}}_{i})^{2}=1$ and $\mathbf{\Sigma}_{f}=\mathbb{I}_{K}$, it is easy to show that
\[
\hat{\Upsilon}_{i}=\frac{\frac{1}{T_{h}}\stackrel[t=1]{T_{h}}{\sum}\mathbf{x}_{i,t}\mathbf{y}_{t+h}}{\frac{1}{T_{h}}\sum_{t=1}^{T_{h}}(\mathbf{x}_{i,t}-\overline{\mathbf{x}}_{i})^{2}}=\frac{1}{T_{h}}\stackrel[t=1]{T_{h}}{\sum}\mathbf{x}_{i,t}\mathbf{y}_{t+h}=\frac{1}{T_{h}}\stackrel[t=1]{T_{h}}{\sum}\left(\beta{}_{i}\mathbf{f}_{t}+\mathbf{e}_{i,t})(\alpha\mathbf{f}_{t}+\epsilon_{t+h}\right)
\]
\[
=\beta_{i}\underset{\mathbf{\Sigma}_{f}=\mathbb{I}_{K}}{\underbrace{\left\{ \frac{1}{T_{h}}\stackrel[t=1]{T_{h}}{\sum}\mathbf{f}_{t}\mathbf{f}_{t}'\right\} }}\alpha'+\frac{1}{T_{h}}\stackrel[t=1]{T_{h}}{\sum}\beta_{i}\mathbf{f}_{t}\epsilon'_{t+h}+\frac{1}{T_{h}}\stackrel[t=1]{T_{h}}{\sum}\alpha\mathbf{f}_{t}\mathbf{e}'_{i,t}+\frac{1}{T_{h}}\stackrel[t=1]{T_{h}}{\sum}\mathbf{e}_{i,t}\epsilon'_{t+h}
\]
\[
=\underset{\Upsilon_{i}}{\underbrace{\beta_{i}\alpha'}}+\underset{\varkappa_{i}}{\underbrace{\frac{1}{T_{h}}\stackrel[t=1]{T_{h}}{\sum}\beta_{i}\mathbf{f}_{t}\epsilon'_{t+h}+\frac{1}{T_{h}}\stackrel[t=1]{T_{h}}{\sum}\alpha\mathbf{f}_{t}\mathbf{e}'_{i,t}+\frac{1}{T_{h}}\stackrel[t=1]{T_{h}}{\sum}\mathbf{e}_{i,t}\epsilon'_{t+h}}}\equiv\Upsilon_{i}+\varkappa_{i}.
\]

Move $\Upsilon_{i}=\beta_{i}\alpha'$ to the left-hand side and taking
norms on both sides, we have
\[
||\hat{\Upsilon}_{i}-\Upsilon_{i}||\leq||\frac{1}{T_{h}}\stackrel[t=1]{T_{h}}{\sum}\beta_{i}\mathbf{f}_{t}\epsilon'_{t+h}||+||\frac{1}{T_{h}}\stackrel[t=1]{T_{h}}{\sum}\alpha\mathbf{f}_{t}\mathbf{e}'_{i,t}||+||\frac{1}{T_{h}}\stackrel[t=1]{T_{h}}{\sum}\mathbf{e}_{i,t}\epsilon'_{t+h}||
\]
\begin{equation}
\leq||\beta_{i}||\,\frac{1}{T_{h}}||\stackrel[t=1]{T_{h}}{\sum}\mathbf{f}_{t}\epsilon'_{t+h}||+||\alpha||\,\frac{1}{T_{h}}||\stackrel[t=1]{T_{h}}{\sum}\mathbf{f}_{t}\mathbf{e}'_{i,t}||+\frac{1}{T_{h}}||\stackrel[t=1]{T_{h}}{\sum}\mathbf{e}_{i,t}\epsilon'_{t+h}||.\label{eq:C47}
\end{equation}

With Assumptions \ref{assu:1}, \ref{assu:2}, \ref{assu:4}, we have
$||\mathbf{f}\epsilon'||\lesssim_{\textrm{P}}T^{1/2}$, $||\mathbf{f}\mathbf{e}_{i}'||\lesssim_{\textrm{P}}T^{1/2}$,
$||\mathbf{e}_{i}\epsilon'||\lesssim_{\textrm{P}}T^{1/2}$, and $||\beta||_{\max}\lesssim1$.
Hence, (\ref{eq:C47}) becomes
\[
||\hat{\Upsilon}_{i}-\Upsilon_{i}||\leq\frac{1}{T}O_{\textrm{p}}(T^{1/2})+\frac{1}{T}O_{\textrm{p}}(T^{1/2})+\frac{1}{T}O_{\textrm{p}}(T^{1/2})\lesssim_{\textrm{P}}T^{-1/2},
\]
which also means $||\varkappa_{i}||\lesssim_{\textrm{P}}T^{-1/2}$. 

Stack the scalars $\hat{\Upsilon}_{i}$ into the diagonal matrix $\hat{\Upsilon}=\mathrm{diag}(\hat{\Upsilon}_{i},\ldots,\hat{\Upsilon}_{N})$,
and likewise $\Upsilon=\mathrm{diag}(\Upsilon_{1},\ldots,\Upsilon_{N}).$
Accordingly, $\varkappa=\mathrm{diag}(\varkappa_{1},\ldots,\varkappa_{N})$.
By the individual bound established above, we have $|\varkappa_{i}|=O_{\textrm{P}}(T^{-1/2})$.
Since $\varkappa$ is diagonal, its operator norm equals the
largest absolute diagonal entry:
\[
||\hat{\Upsilon}-\Upsilon||=||\varkappa||=\max\{1\leq i\leq N\}|\varkappa_{i}|=\max_{i}|\varkappa_{i}|=O_{\textrm{P}}(T^{-1/2})\lesssim_{\textrm{P}}T^{-1/2},
\]
which yields the desired rate. Therefore, we have $||\hat{\Upsilon}||=||\Upsilon||+||\varkappa||=||\Upsilon||+O_{P}(T^{-1/2})$.

(ii) At the mixed-frequency, the scaling coefficient $\hat{\Upsilon}_{i}$ is obtained from the regression

\[
y_{t+h}=c_{i}+\hat{\Upsilon}_{i}X_{i,t}(\theta_{x,i})+u_{i,t+h}.
\]
This regression is actually a standard MIDAS regression with observable regressor $X_{i,t}(\theta_{x,i})$, and is of the same form as the MIDAS-NLS setup in \cite{andreou2010regression} Section 3.2. Therefore, the corresponding MIDAS-NLS estimator jointly satisfies
\[
\hat{\Upsilon}_{i}\overset{\textrm{p}}{\longrightarrow}\Upsilon_{i},\qquad\hat{\theta}_{x,i}\overset{\textrm{p}}{\longrightarrow}\theta_{x,i}.
\]
Block-diagonal stacking over i=1,...,N, we obtain
\[
\hat{\Upsilon}\overset{\textrm{p}}{\longrightarrow}\Upsilon,\qquad\hat{\theta}_{x}\overset{\textrm{p}}{\longrightarrow}\theta_{x}.
\]
Moreover, by the triangle inequality,
\[
||\hat{\Upsilon}_{i}(\hat{\theta}_{x})-\Upsilon_{i}(\theta_{x})||\leq||\hat{\Upsilon}_{i}(\hat{\theta}_{x})-\Upsilon_{i}(\hat{\theta}_{x})||+||\Upsilon_{i}(\hat{\theta}_{x})-\Upsilon_{i}(\theta_{x})||.
\]
By part (i), the first term is bounded by $\lesssim_{\textrm{P}}T^{-1/2}$. By the continuity of $\Upsilon(\theta_{x})$ and the consistency of $\hat{\theta}_{x}$, the second term is $o_{\textrm{P}}(1)$. Therefore,
\[
||\hat{\Upsilon}(\hat{\theta}_{x})-\Upsilon(\theta_{x})||\lesssim_{\textrm{P}}T^{-1/2}.
\]

(iii) Apart from the impact of estimating $\Upsilon$, another important aspect we need to consider is the effect of jointly estimating the parameters in the factor-MIDAS model. Let
\[
\kappa=(\alpha'_{l},\theta)',\qquad\alpha{}_{l}=(\alpha,\alpha_{w})',\qquad\theta=(\theta_{f},\theta_{x})'.
\]
To analyze this, we rely on Theorem $2.1$ and Lemmas $B.1-B.3$ in \cite{koh2023inference}, extending the argument to the joint estimation of $(\alpha{}_{l},\theta)$, where $\theta$ now collects both $\theta_{f}$ and $\theta_{x}$.

As the NLS estimators $\hat{\kappa}$ maximizes the objective function $\hat{Q}_{T}(\kappa)=-\frac{1}{T}\sum_{t=1}^{T}[y_{t+h}-g(\hat{F}_{t},\kappa)]^{2}$, we have
\begin{equation}
\sqrt{T}(\hat{\kappa}-\kappa)=-[\frac{1}{T}\sum_{t=1}^{T}H(\hat{F}_{t},\kappa_{T})]^{-1}\frac{1}{\sqrt{T}}\sum_{t=1}^{T}s(\hat{F}_{t},\kappa),\label{eq:C48}
\end{equation}
where $\kappa_{T}$ is the intermediate between $\hat{\kappa}$ and $\kappa$, $H(\hat{F}_{t},\kappa_{T})$ is a hessian matrix, and $s(\hat{F}_{t},\kappa)$ is a score vector. To derive the asymptotic distribution, we analyze the convergence of each term.

Let $g_{\kappa}(\cdot)=\partial g(\cdot)/\partial\kappa$. We write the term with a score vector as follows
\[
\frac{1}{\sqrt{T}}\sum_{t=1}^{T}s(\hat{F}_{t},\kappa)=2\frac{1}{\sqrt{T}}\sum_{t=1}^{T}[\varepsilon_{t+h}+\alpha'_{l}H^{-1}(HF_{t}(\theta)-\hat{F}_{t}(\theta))]g_{\kappa}(\hat{F}_{t},\kappa),
\]
\[
2\frac{1}{\sqrt{T}}\sum_{t=1}^{T}[\varepsilon_{t+h}+\alpha'_{l}H^{-1}(HF_{t}(\theta)-\hat{F}_{t}(\theta))](\Phi_{0}g_{\kappa}(F_{t},\kappa)+P_{t}),
\]
where $\Phi_{0}=\textrm{diag}(H_{0},I_{\mathcal{W}},I_{p})$ and $H_{0}=\textrm{plim}H$, $p=p_{f}+p_{x}$, and $P_{t}$ is a $(K+\mathcal{W}+p)\times1$ vector such that 
\[
P_{t}=\left[\begin{array}{c}
\hat{F}_{t}(\theta)-HF_{t}(\theta)\\
0_{\mathcal{W}\times1}\\
(\frac{\partial\hat{F}_{t}(\theta)}{\partial\theta}H^{-1}-\frac{\partial F_{t}(\theta)}{\partial\theta})'\alpha{}_{l}
\end{array}\right],
\]
with $\frac{\partial\hat{F}_{t}(\theta)'}{\partial\theta}=\textrm{diag}(\frac{\partial\hat{F}_{t}(\theta_{1})'}{\partial\theta_{1}},...,\frac{\partial\hat{F}_{t}(\theta_{K})'}{\partial\theta_{K}})$ is a $K\times K$ block-diagonal matrix. The $k$-th block is $\partial\hat{F}_{t}(\theta_{k})/\partial\theta_{k}$, which is a $p_{k}\times1$ column vector, for $k=1,...,K$. Under Assumption \ref{assu:5} and Lemma B.1 of \cite{koh2023inference}, we have $\frac{1}{\sqrt{T}}\sum_{t=1}^{T}\varepsilon_{t+h}g_{\kappa}(\hat{F}_{t},\kappa)\overset{\textrm{d}}{\rightarrow}\mathcal{N}(0,\Phi_{0}\Omega_{\kappa}\Phi'_{0})$. The remaining term generates the potential bias in the joint estimation of $\alpha$  and $\theta$. With respect to $\alpha$, the remaining term is as follows
\[
\frac{1}{\sqrt{T}}\sum_{t=1}^{T}\hat{F}_{t}(\theta)[HF_{t}(\theta)-\hat{F}_{t}(\theta)]'H^{-1}{}'\alpha
\]
\[
=-[\frac{1}{\sqrt{T}}\sum_{t=1}^{T}(\hat{F}_{t}(\theta)-HF_{t}(\theta))(\hat{F}_{t}(\theta)-HF_{t}(\theta))'+\frac{1}{\sqrt{T}}\sum_{t=1}^{T}HF_{t}(\theta)(\hat{F}_{t}(\theta)-HF_{t}(\theta))']H^{-1}{}'\alpha 
\]
\[
=-c[\Lambda^{-1}H\left\{ \sum_{j=0}^{J}b_{j}(\theta)\Gamma b_{j}(\theta)+\sum_{j=0}^{J}\sum_{l\neq j}^{J}b_{j}(\theta)\Gamma_{j-l}b_{l}(\theta)\right\} H\Lambda^{-1}
\]
\begin{equation}
+\left\{ \sum_{j=0}^{J}b_{j}(\theta)Hb_{j}(\theta)+\sum_{j=0}^{J}\sum_{l\neq j}^{J}b_{j}(\theta)Q_{j-l}b_{l}(\theta)\right\} \Gamma Q'\Lambda^{-2}]\textrm{plim}(\hat{\alpha}),\label{eq:C49}
\end{equation}
where $\textrm{plim}(\hat{\alpha})=H^{-1}{}'\alpha$, $\Lambda=\textrm{diag}(\lambda_{(1)},\lambda_{(2)}...,\lambda_{(K)})$ is the $K\times K$ diagonal matrix of the eigenvalues, $Q\equiv\textrm{plim}(\frac{1}{T_{H}}\sum_{t_{h}=l+1}^{T_{H}}\hat{f}_{t_{h}}'f_{t_{h}-l})$ is a cross-covariance matrix, and $\Gamma\equiv\lim_{N\rightarrow\infty}\textrm{Var}(\beta'e_{t_{h}}/\sqrt{N})$ denotes the limiting covariance of the idiosyncratic component. Since we assume that $\sqrt{T}/(qN)\rightarrow0$, combined with Theorem 2.1 of \cite{koh2023inference}, this implies $c_0=0$. Hence, this term is $o_{\textrm{p}}(1)$. Similarly, with respect to $\theta$, we have 
\[
\frac{1}{\sqrt{T}}\sum_{t=1}^{T}\frac{\partial\hat{F}_{t}(\theta)'}{\partial\theta}H^{-1}{}'\alpha\alpha'H^{-1}[HF_{t}(\theta)-\hat{F}_{t}(\theta)]
\]
\[
=-H^{-1}{}'\alpha\circ\frac{1}{\sqrt{T}}\sum_{t=1}^{T}\hat{F}_{t,\theta}(\theta)[\hat{F}_{t}(\theta)-HF_{t}(\theta)]'H^{-1'}\alpha 
\]
\[
=-c\textrm{plim}(\hat{\alpha})[\Lambda^{-1}H\left\{ \sum_{j=0}^{J}\frac{\partial b_{j}(\theta)}{\partial\theta}\Gamma b_{j}(\theta)+\sum_{j=0}^{J}\sum_{l\neq j}^{J}\frac{\partial b_{j}(\theta)}{\partial\theta}\Gamma_{j-l}b_{l}(\theta)\right\} H\Lambda^{-1}
\]
\begin{equation}
+\left\{ \sum_{j=0}^{J}\frac{\partial b_{j}(\theta)}{\partial\theta}Hb_{j}(\theta)+\sum_{j=0}^{J}\sum_{l\neq j}^{J}\frac{\partial b_{j}(\theta)}{\partial\theta}Q_{j-l}b_{l}(\theta)\right\} \Gamma Q'\Lambda^{-2}]\textrm{plim}(\hat{\alpha}),\label{eq:C50}
\end{equation}
where $\hat{F}_{t,\theta}(\theta)=(\frac{\partial\hat{F}_{t}(\theta_{1})'}{\partial\theta_{1}},...,\frac{\partial\hat{F}_{t}(\theta_{K})'}{\partial\theta_{K}})'$. Similarly, this term is also $o_{\textrm{p}}(1)$ as well. To apply the lemmas, we use the Hadamard product such that $(A\circ B)_{ij}=A_{ij}B_{ij}$. By applying Hadamard product, we have $\frac{\partial\hat{F}_{t}(\theta)'}{\partial\theta}H^{-1}\alpha_{l}=H^{-1}\alpha_{l}\circ\hat{F}_{t,\theta}(\theta)$ to obtain the first equality. Applying Lemma B.3 of \cite{koh2023inference}, we have the second equality. Finally, we have $\frac{1}{\sqrt{T}}\sum_{t=1}^{T}s(\hat{F}_{t},\kappa)\overset{\textrm{d}}{\rightarrow}\mathcal{N}(0,\Phi_{0}\Omega_{\kappa}\Phi'_{0})$. Next, we derive the term with Hessian matrix. First, we rewrite the first term in (\ref{eq:C48}) as follows,
\[
\frac{1}{T}\sum_{t=1}^{T}H(\hat{F}_{t},\kappa)=\frac{1}{T}\sum_{t=1}^{T}[\varepsilon_{t+h}+\alpha'_{l}H^{-1}(HF_{t}(\theta)-\hat{F}_{t}(\theta))]\frac{\partial^{2}g(\hat{F}_{t},\kappa)}{\partial\kappa\partial\kappa'}+\frac{1}{T}\sum_{t=1}^{T}\frac{\partial g(\hat{F}_{t},\kappa)}{\partial\kappa}\frac{\partial g(\hat{F}_{t},\kappa)}{\partial\kappa'}.
\]
Under Assumption \ref{assu:5} and Lemma B.1 of \cite{koh2023inference}, $\frac{1}{T}\sum_{t=1}^{T}\varepsilon_{t+h}\frac{\partial^{2}g(\hat{F}_{t},\kappa)}{\partial\kappa\partial\kappa'}=o_{\textrm{p}}(1)$. We can also show that $-\frac{1}{T}\sum_{t=1}^{T}\alpha'_{l}H^{-1}(HF_{t}(\theta)-\hat{F}_{t}(\theta))\frac{\partial^{2}g(\hat{F}_{t},\kappa)}{\partial\kappa\partial\kappa'}=o_{\textrm{p}}(1)$. Finally, for the second term, we have
\begin{equation}
\frac{1}{T}\sum_{t=1}^{T}\frac{\partial g(\hat{F}_{t},\kappa)}{\partial\kappa}\frac{\partial g(\hat{F}_{t},\kappa)}{\partial\kappa'}=\Phi_{0}\mathbf{\Sigma}_{\kappa}\Phi'_{0}+o_{\textrm{p}}(1),\label{eq:C51}
\end{equation}
where $\mathbf{\Sigma}_{\kappa}\equiv\mathbb{E}[\frac{\partial g(F_{t},\kappa)}{\partial\kappa}\frac{\partial g(F_{t},\kappa)}{\partial\kappa'}]$ by replacing $\frac{\partial g(\hat{F}_{t},\kappa)}{\partial\kappa}$ with $\Phi_{0}\frac{\partial g(F_{t},\kappa)}{\partial\kappa}+P_{t}$. Then, by Lemma B.2 of \cite{koh2023inference}, we have $\frac{1}{T}\sum_{t=1}^{T}P_{t}P_{t}'=o_{\textrm{p}}(1)$. By plugging the terms, (\ref{eq:C49}), (\ref{eq:C50}), and (\ref{eq:C51}) into (\ref{eq:C48}), we have $\sqrt{T}(\hat{\kappa}-\kappa)\overset{\textrm{d}}{\rightarrow}\mathcal{N}(0,\Phi'{}_{0}^{-1}\mathbf{\Sigma}_{\kappa}^{-1}\Omega_{\kappa}\mathbf{\Sigma}_{\kappa}^{-1}\Phi{}_{0}^{-1})$. From the asymptotic normality result, it follows that $\sqrt{T}(\hat{\kappa}-\kappa)=O_{\textrm{p}}(1)$, and hence $||\hat{\kappa}-\kappa||\lesssim_{\textrm{P}}T^{-1/2}$.
\end{proof}
\section{Supplementary Results: Tables and Figures\label{Appendix D}}

\begin{table}[H]
\begin{centering}
\caption{Finite Sample Comparison of Predictions (Bias)}
{\footnotesize\label{Table D.1}}{\footnotesize\par}
\par\end{centering}
\begin{centering}
\begin{tabular}{c|c|ccccc|cccc}
\hline
\multicolumn{11}{c}{Scenario 1: $N=200,\pi=0.5$}\tabularnewline
\hline
\hline
\multicolumn{1}{c}{} & $T_{H}$ & PCA & SPCA & sPCA & SsPCA & PLS & Bo-PCA & Bo-SPCA & Bo-sPCA & Bo-SsPCA\tabularnewline
\hline
\multirow{2}{*}{$h=1$} & $90$ & 0.013 & 0.011 & 0.010 & 0.016 & 0.020 & 0.017 & 0.007 & 0.016 & 0.015\tabularnewline
\cline{2-11}
& $180$ & 0.007 & 0.001 & 0.002 & \textbf{0.000} & 0.010 & 0.007 & 0.002 & 0.001 & 0.003\tabularnewline
\hline
\multirow{2}{*}{$h=4$} & $90$ & 0.005 & 0.005 & 0.007 & 0.009 & 0.012 & 0.018 & 0.009 & 0.011 & 0.014\tabularnewline
\cline{2-11}
& $180$ & 0.010 & 0.011 & 0.005 & \textbf{0.005} & 0.012 & 0.012 & 0.011 & 0.007 & \textbf{0.007}\tabularnewline
\hline
\multicolumn{11}{c}{Scenario 2: $N=2000,\pi=0.05$}\tabularnewline
\hline
\multicolumn{1}{c}{} & $T_{H}$ & PCA & SPCA & sPCA & SsPCA & \multicolumn{1}{c}{PLS} & Bo-PCA & Bo-SPCA & Bo-sPCA & Bo-SsPCA\tabularnewline
\hline
\multirow{2}{*}{$h=1$} & $90$ & 0.002 & 0.002 & 0.002 & \textbf{0.001} & 0.010 & 0.001 & 0.007 & 0.001 & \textbf{0.001}\tabularnewline
\cline{2-11}
& $180$ & 0.008 & 0.006 & 0.007 & 0.009 & 0.008 & 0.007 & 0.004 & 0.007 & 0.009\tabularnewline
\hline
\multirow{2}{*}{$h=4$} & $90$ & 0.005 & 0.009 & 0.006 & \textbf{0.000} & 0.012 & 0.003 & 0.011 & 0.001 & 0.002\tabularnewline
\cline{2-11}
& $180$ & 0.010 & 0.009 & 0.008 & \textbf{0.005} & 0.007 & 0.010 & 0.008 & 0.009 & \textbf{0.003}\tabularnewline
\hline
\multicolumn{11}{c}{Scenario 3: $N=2000,\pi=0.05$}\tabularnewline
\hline
\multicolumn{1}{c}{} & $T_{H}$ & PCA & SPCA & sPCA & SsPCA & PLS & Bo-PCA & Bo-SPCA & Bo-sPCA & Bo-SsPCA\tabularnewline
\hline
\multirow{2}{*}{$h=1$} & $90$ & 0.002 & 0.011 & 0.008 & \textbf{0.000} & 0.002 & 0.001 & 0.001 & 0.005 & \textbf{0.000}\tabularnewline
\cline{2-11}
& $180$ & 0.003 & 0.002 & 0.003 & 0.008 & 0.003 & 0.003 & 0.001 & 0.002 & 0.006\tabularnewline
\hline
\multirow{2}{*}{$h=4$} & $90$ & 0.011 & 0.011 & 0.003 & \textbf{0.000} & 0.000 & 0.004 & 0.003 & 0.004 & \textbf{0.002}\tabularnewline
\cline{2-11}
& $180$ & 0.006 & 0.005 & 0.003 & 0.005 & 0.006 & 0.006 & 0.007 & 0.006 & \textbf{0.005}\tabularnewline
\hline
\end{tabular}
\par\end{centering}
$\,$

{\footnotesize\textbf{Notes: }}{\footnotesize In this table, we evaluate
the performance of PCA, SPCA, sPCA, SsPCA, PLS, and boosting with
factors extracted through PCA, SPCA, sPCA, and SsPCA (referred to
as Bo-PCA, Bo-SPCA, Bo-sPCA, and Bo-SsPCA, respectively), in terms
of Bias when $h=1,4$. The number of boosting iterations is fixed at $M=50$. All reported values are based on averages over 1,000 Monte Carlo repetitions.}{\footnotesize\par}
\end{table}

\begin{table}[H]
\caption{Additional Finite Sample Comparison of Predictions (MSFE)}

\begin{centering}
\begin{tabular}{c|c|cccc}
\hline 
\multicolumn{6}{c}{Scenario 1: $N=200,\pi=0.5$}\tabularnewline
\hline 
\hline 
\multicolumn{1}{c}{} & $T_{H}$ & Bo-PCA & Bo-SPCA & Bo-sPCA & Bo-SsPCA\tabularnewline
\hline 
\multirow{2}{*}{$h=1$} & $90$ & 1.019 & 0.867 & 0.972 & \textbf{0.838}\tabularnewline
\cline{2-6}
 & $180$ & 0.949 & 0.859 & 0.900 & \textbf{0.846}\tabularnewline
\hline 
\multirow{2}{*}{$h=4$} & $90$ & 1.093 & 0.891 & 1.025 & \textbf{0.888}\tabularnewline
\cline{2-6}
 & $180$ & 1.008 & 0.896 & 0.949 & \textbf{0.892}\tabularnewline
\hline 
\multicolumn{6}{c}{Scenario 2: $N=2000,\pi=0.05$}\tabularnewline
\hline 
\multicolumn{1}{c}{} & $T_{H}$ & Bo-PCA & Bo-SPCA & Bo-sPCA & Bo-SsPCA\tabularnewline
\hline 
\multirow{2}{*}{$h=1$} & $90$ & 1.002 & 0.904 & 0.925 & \textbf{0.858}\tabularnewline
\cline{2-6}
 & $180$ & 0.943 & 0.889 & 0.891 & \textbf{0.871}\tabularnewline
\hline 
\multirow{2}{*}{$h=4$} & $90$ & 1.072 & 0.929 & 0.977 & \textbf{0.892}\tabularnewline
\cline{2-6}
 & $180$ & 1.010 & 0.935 & 0.938 & \textbf{0.908}\tabularnewline
\hline 
\multicolumn{6}{c}{Scenario 3: $N=2000,\pi=0.05$}\tabularnewline
\hline 
\multicolumn{1}{c}{} & $T_{H}$ & Bo-PCA & Bo-SPCA & Bo-sPCA & Bo-SsPCA\tabularnewline
\hline 
\multirow{2}{*}{$h=1$} & $90$ & 1.052 & 0.931 & 0.960 & \textbf{0.862}\tabularnewline
\cline{2-6}
 & $180$ & 0.998 & 0.944 & 0.938 & \textbf{0.904}\tabularnewline
\hline 
\multirow{2}{*}{$h=4$} & $90$ & 1.100 & 0.974 & 1.005 & \textbf{0.902}\tabularnewline
\cline{2-6}
 & $180$ & 1.050 & 0.995 & 0.987 & \textbf{0.933}\tabularnewline
\hline 
\end{tabular}
\par\end{centering}
$\,$

\label{Table D.2}{\footnotesize\textbf{Notes: }}{\footnotesize In this
table, we evaluate the performance of boosting with factors extracted through PCA, SPCA, sPCA, and SsPCA
(referred to as Bo-PCA, Bo-SPCA, Bo-sPCA, and Bo-SsPCA, respectively)
in terms of MSFE when $h=1,4$. The number of boosting iterations is fixed at $M=100$. All reported values are based on averages over 1,000 Monte Carlo repetitions.}{\footnotesize\par}
\end{table}

\begin{table}[h]
\caption{Additional Finite Sample Comparison of Predictions (Bias)}

\begin{centering}
\begin{tabular}{c|c|cccc}
\hline 
\multicolumn{6}{c}{Scenario 1: $N=200,\pi=0.5$}\tabularnewline
\hline 
\hline 
\multicolumn{1}{c}{} & $T_{H}$ & Bo-PCA & Bo-SPCA & Bo-sPCA & Bo-SsPCA\tabularnewline
\hline 
\multirow{2}{*}{$h=1$} & $90$ & 0.019 & 0.009 & 0.014 & 0.017\tabularnewline
\cline{2-6}
 & $180$ & 0.006 & 0.001 & 0.001 & 0.003\tabularnewline
\hline 
\multirow{2}{*}{$h=4$} & $90$ & 0.017 & 0.014 & 0.012 & 0.014\tabularnewline
\cline{2-6}
 & $180$ & 0.010 & 0.010 & 0.008 & \textbf{0.008}\tabularnewline
\hline 
\multicolumn{6}{c}{Scenario 2: $N=2000,\pi=0.05$}\tabularnewline
\hline 
\multicolumn{1}{c}{} & $T_{H}$ & Bo-PCA & Bo-SPCA & Bo-sPCA & Bo-SsPCA\tabularnewline
\hline 
\multirow{2}{*}{$h=1$} & $90$ & 0.005 & 0.007 & 0.001 & \textbf{0.001}\tabularnewline
\cline{2-6}
 & $180$ & 0.006 & 0.005 & 0.006 & 0.008\tabularnewline
\hline 
\multirow{2}{*}{$h=4$} & $90$ & 0.004 & 0.007 & 0.001 & 0.004\tabularnewline
\cline{2-6}
 & $180$ & 0.010 & 0.008 & 0.008 & \textbf{0.003}\tabularnewline
\hline 
\multicolumn{6}{c}{Scenario 3: $N=2000,\pi=0.05$}\tabularnewline
\hline 
\multicolumn{1}{c}{} & $T_{H}$ & Bo-PCA & Bo-SPCA & Bo-sPCA & Bo-SsPCA\tabularnewline
\hline 
\multirow{2}{*}{$h=1$} & $90$ & 0.002 & 0.000 & 0.006 & 0.001\tabularnewline
\cline{2-6}
 & $180$ & 0.003 & 0.001 & 0.002 & 0.006\tabularnewline
\hline 
\multirow{2}{*}{$h=4$} & $90$ & 0.000 & 0.002 & 0.004 & 0.001\tabularnewline
\cline{2-6}
 & $180$ & 0.007 & 0.008 & 0.005 & \textbf{0.004}\tabularnewline
\hline 
\end{tabular}
\par\end{centering}
$\,$

\label{Table D.3}{\footnotesize\textbf{Notes: }}{\footnotesize In this table, we evaluate
the performance of boosting with
factors extracted through PCA, SPCA, sPCA, and SsPCA (referred to
as Bo-PCA, Bo-SPCA, Bo-sPCA, and Bo-SsPCA, respectively), in terms
of Bias when $h=1,4$. The number of boosting iterations is fixed at $M=100$. All reported values are based on averages over 1,000 Monte Carlo repetitions.}{\footnotesize\par}
\end{table}

\begin{table}[H]
\small
\caption{Finite Sample Comparison of Factor Distances $(d(\hat{\mathbf{f}}_H,\mathbf{f}_H))$}

\begin{centering}
\setlength{\tabcolsep}{3pt}
\begin{tabular}{c|c|ccccc|cccc}
\hline 
\multicolumn{11}{c}{{\small Scenario 1: $N=200,\pi=0.5$}}\tabularnewline
\hline 
\hline 
\multicolumn{1}{c}{} & {\small$T_{H}$} & {\small PCA} & {\small SPCA} & {\small sPCA} & {\small SsPCA} & {\small PLS} & {\small Bo-PCA} & {\small Bo-SPCA} & {\small Bo-sPCA} & {\small Bo-SsPCA}\tabularnewline
\hline 
\multirow{2}{*}{{\small$h=1$}} & {\small$90$} & 1.244 & 1.188 & 1.038 & 1.102 & 1.190 & 1.243 & 1.191 & 1.047 & 1.106\tabularnewline
\cline{2-11}
 & {\small$180$} & 1.278 & 1.171 & 0.959 & 1.019 & 1.160 & 1.274 & 1.164 & 0.962 & 1.029\tabularnewline
\hline 
\multirow{2}{*}{{\small$h=4$}} & $90$ & 1.240 & 1.185 & 1.051 & 1.101 & 1.197 & 1.236 & 1.193 & 1.054 & 1.104\tabularnewline
\cline{2-11}
 & $180$ & 1.273 & 1.171 & 0.957 & 1.018 & 1.161 & 1.269 & 1.172 & 0.961 & 1.017\tabularnewline
\hline 
\multicolumn{11}{c}{{\small Scenario 2: $N=2000,\pi=0.05$}}\tabularnewline
\hline 
\multicolumn{1}{c}{} & {\small$T_{H}$} & {\small PCA} & {\small SPCA} & {\small sPCA} & {\small SsPCA} & \multicolumn{1}{c}{{\small PLS}} & {\small Bo-PCA} & {\small Bo-SPCA} & {\small Bo-sPCA} & {\small Bo-SsPCA}\tabularnewline
\hline 
\multirow{2}{*}{{\small$h=1$}} & {\small$90$} & 1.197 & 1.227 & 1.138 & 1.156 & 1.230 & 1.209 & 1.232 & 1.147 & 1.166\tabularnewline
\cline{2-11}
 & {\small$180$} & 1.213 & 1.217 & 1.115 & 1.121 & 1.220 & 1.221 & 1.215 & 1.117 & 1.121\tabularnewline
\hline 
\multirow{2}{*}{{\small$h=4$}} & {\small$90$} & 1.206 & 1.226 & 1.147 & 1.169 & 1.224 & 1.215 & 1.231 & 1.154 & 1.171\tabularnewline
\cline{2-11}
 & {\small$180$} & 1.220 & 1.217 & 1.117 & 1.130 & 1.225 & 1.230 & 1.216 & 1.121 & 1.130\tabularnewline
\hline 
\multicolumn{11}{c}{{\small Scenario 3: $N=2000,\pi=0.05$}}\tabularnewline
\hline 
\multicolumn{1}{c}{} & {\small$T_{H}$} & {\small PCA} & {\small SPCA} & {\small sPCA} & {\small SsPCA} & {\small PLS} & {\small Bo-PCA} & {\small Bo-SPCA} & {\small Bo-sPCA} & {\small Bo-SsPCA}\tabularnewline
\hline 
\multirow{2}{*}{{\small$h=1$}} & {\small$90$} & {\small 1.403} & 1.398 & 1.397 & \textbf{1.385} & 1.385 & 1.402 & 1.394 & 1.396 & {\small\textbf{1.384}}\tabularnewline
\cline{2-11}
 & {\small$180$} & 1.409 & 1.403 & 1.402 & {\small\textbf{1.382}} & 1.383 & 1.409 & 1.404 & 1.403 & {\small\textbf{1.382}}\tabularnewline
\hline 
\multirow{2}{*}{{\small$h=4$}} & {\small$90$} & 1.400 & 1.393 & 1.394 & {\small\textbf{1.386}} & 1.385 & 1.400 & 1.392 & 1.394 & {\small\textbf{1.385}}\tabularnewline
\cline{2-11}
 & {\small$180$} & 1.409 & 1.402 & 1.402 & {\small\textbf{1.381}} & 1.381 & 1.409 & 1.402 & 1.402 & {\small\textbf{1.381}}\tabularnewline
\hline 
\end{tabular}
\par\end{centering}
$\,$

\label{Table D.4}{\footnotesize\textbf{Notes: }}{\footnotesize In this
table, we evaluate the performance of PCA, SPCA, sPCA, SsPCA, PLS,
and boosting with factors extracted through PCA, SPCA, sPCA, and SsPCA
(referred to as Bo-PCA, Bo-SPCA, Bo-sPCA, and Bo-SsPCA, respectively)
in terms of the distance between the estimated factor space and the true
factor space $d(\hat{\mathbf{f}}_H,\mathbf{f}_H)=||\mathbb{P}_{\hat{\mathbf{f}}_H'}-\mathbb{P}_{\mathbf{f}_H'}||$
when $h=1,4$. The number of boosting iterations is fixed at $M=50$. All reported values are based on averages over 1,000 Monte Carlo repetitions.}{\footnotesize\par}
\end{table}

\begin{table}[H]
\caption{Additional Finite Sample Comparison of Factor Distances $(d(\hat{\mathbf{f}}_H,\mathbf{f}_H))$}

\begin{centering}
\begin{tabular}{c|c|cccc}
\hline 
\multicolumn{6}{c}{Scenario 1: $N=200,\pi=0.5$}\tabularnewline
\hline 
\hline 
\multicolumn{1}{c}{} & $T_{H}$ & Bo-PCA & Bo-SPCA & Bo-sPCA & Bo-SsPCA\tabularnewline
\hline 
\multirow{2}{*}{$h=1$} & $90$ & 1.243 & 1.189 & 1.043 & 1.103\tabularnewline
\cline{2-6}
 & $180$ & 1.277 & 1.169 & 0.962 & 1.026\tabularnewline
\hline 
\multirow{2}{*}{$h=4$} & $90$ & 1.238 & 1.187 & 1.052 & 1.102\tabularnewline
\cline{2-6}
 & $180$ & 1.271 & 1.168 & 0.961 & 1.019\tabularnewline
\hline 
\multicolumn{6}{c}{Scenario 2: $N=2000,\pi=0.05$}\tabularnewline
\hline 
\multicolumn{1}{c}{} & $T_{H}$ & Bo-PCA & Bo-SPCA & Bo-sPCA & Bo-SsPCA\tabularnewline
\hline 
\multirow{2}{*}{$h=1$} & $90$ & 1.203 & 1.233 & 1.145 & 1.160\tabularnewline
\cline{2-6}
 & $180$ & 1.218 & 1.214 & 1.117 & 1.121\tabularnewline
\hline 
\multirow{2}{*}{$h=4$} & $90$ & 1.211 & 1.233 & 1.152 & 1.170\tabularnewline
\cline{2-6}
 & $180$ & 1.226 & 1.219 & 1.122 & 1.130\tabularnewline
\hline 
\multicolumn{6}{c}{Scenario 3: $N=2000,\pi=0.05$}\tabularnewline
\hline 
\multicolumn{1}{c}{} & $T_{H}$ & Bo-PCA & Bo-SPCA & Bo-sPCA & Bo-SsPCA\tabularnewline
\hline 
\multirow{2}{*}{$h=1$} & $90$ & 1.402 & 1.395 & 1.396 & \textbf{1.383}\tabularnewline
\cline{2-6}
 & $180$ & 1.409 & 1.404 & 1.402 & \textbf{1.381}\tabularnewline
\hline 
\multirow{2}{*}{$h=4$} & $90$ & 1.400 & 1.393 & 1.394 & \textbf{1.384}\tabularnewline
\cline{2-6}
 & $180$ & 1.409 & 1.401 & 1.401 & \textbf{1.382}\tabularnewline
\hline 
\end{tabular}
\par\end{centering}
$\,$

\label{Table D.5}{\footnotesize\textbf{Notes: }}{\footnotesize In this
table, we evaluate the performance of boosting with factors extracted through PCA, SPCA, sPCA, and SsPCA
(referred to as Bo-PCA, Bo-SPCA, Bo-sPCA, and Bo-SsPCA, respectively)
in terms of the distance between the estimated factor space and the true
factor space $d(\hat{\mathbf{f}}_H,\mathbf{f}_H)=||\mathbb{P}_{\hat{\mathbf{f}}_H'}-\mathbb{P}_{\mathbf{f}_H'}||$
when $h=1,4$. The number of boosting iterations is fixed at $M=100$. All reported values are based on averages over 1,000 Monte Carlo repetitions.}{\footnotesize\par}
\end{table}

\begin{table}[H]
\caption{Tuning-Parameter Robustness Check for SsPCA}
\begin{center}
\begin{tabular}{cccccc}
\hline
Offset from selected $\lfloor qN \rfloor$ & $-200$ & $-100$ & $0$ & $+100$ & $+200$ \\
\hline
Avg.\ MSFE & 0.927 & 0.898 & 0.864 & 0.869 & 0.879 \\
\hline
\end{tabular}
\end{center}
\label{Table D.6}{\footnotesize \textbf{Notes: }{\footnotesize For each of the 1,000 Monte Carlo replications under Scenario 2 ($N=2000$, $\pi=0.05$, $h=4$, $T_H=180$), we perturb $\lfloor qN \rfloor$ around its cross-validated value by $\pm 100$ and $\pm 200$, holding $K$ at its replication-specific CV value. The offset $0$ corresponds to the CV optimum. Perturbations that fall outside the grid $[100, 2000]$ are excluded. The reported values are averages of the OOS MSFE across valid replications.}{\footnotesize\par}}
\end{table}

\begin{table}[H]
\setlength{\tabcolsep}{3pt} 

\caption{Core Countries For Inclusion}

\begin{centering}
{\footnotesize{}%
\begin{tabular}{lp{11.5cm}}
\hline
{\footnotesize\textbf{Country }} & {\footnotesize\textbf{Rationale for Inclusion}}\tabularnewline
\hline 
\hline 
{\footnotesize United States (USA)} & {\scriptsize Primary source of the target variables; own historical
dynamics provide essential predictive information.}\tabularnewline
\hline 
{\footnotesize China (CHN)} & {\scriptsize Significant influence on U.S. inflation, supply chains,
and import demand; largest external driver.}\tabularnewline
\hline 
{\footnotesize Canada (CAN)} & {\scriptsize Largest trading partner; strong linkages through energy
and housing markets.}\tabularnewline
\hline 
{\footnotesize Mexico (MEX)} & {\scriptsize Integrated manufacturing supply chains; co-movement with
U.S. unemployment cycles.}\tabularnewline
\hline 
{\footnotesize Germany (DEU)} & {\scriptsize Eurozone manufacturing hub; high correlation with U.S.
manufacturing PMI.}\tabularnewline
\hline 
{\footnotesize United Kingdom (GBR)} & {\scriptsize Major global financial center; transmission channel for
monetary policy.}\tabularnewline
\hline 
{\footnotesize Japan (JPN)} & {\scriptsize Large capital market; yen exchange rate and liquidity
spillovers to the U.S.}\tabularnewline
\hline 
{\footnotesize Korea (KOR)} & {\scriptsize Technology and semiconductor cycle closely tied to NASDAQ
performance and U.S. CPI.}\tabularnewline
\hline 
{\footnotesize Taiwan (TWN)} & {\scriptsize Global semiconductor hub; directly linked to U.S. technology
profits and investment cycle.}\tabularnewline
\hline 
{\footnotesize France (FRA)} & {\scriptsize Core Eurozone economy with diversified industry and consumption;
relatively strong predictive power.}\tabularnewline
\hline 
{\footnotesize Australia (AUS)} & {\scriptsize Major commodity exporter; business cycle synchronized
with the U.S. economy.}\tabularnewline
\hline 
\end{tabular}}{\footnotesize\par}
\par\end{centering}
\label{Table D.7}
\end{table}

\noindent\textbf{Data coverage and missing values.} Because the data coverage
varies across countries or regions, with
China having relatively limited availability, we nevertheless include
it due to its indispensable role in influencing the U.S. economy. To ensure a balanced panel, we select monthly data from October 1996
to December 2024. This corresponds to the quarterly target variables
spanning 1996 Q3 to 2024 Q4.
Variables with five or more consecutive
months of missing data are excluded from the analysis. To handle the remaining
missing values in the dataset, we apply the Expectation\textendash Maximization
(EM) algorithm.

\begin{sidewaystable}[H]
\setlength{\tabcolsep}{0.75pt} 

\caption{Predictors Selected by SsPCA (Macro)}

\begin{centering}
\resizebox{\linewidth}{!}{
{\scriptsize{}%
\begin{tabular}{c|cccc|cccc}
\cline{2-9}
\multicolumn{1}{c}{} & \multicolumn{4}{c}{{\scriptsize$h=1$}} & \multicolumn{4}{c}{{\scriptsize$h=4$}}\tabularnewline
\cline{2-9}
\multicolumn{1}{c}{} & \multicolumn{2}{c}{{\scriptsize Pre-COVID}} & \multicolumn{2}{c}{{\scriptsize Post-COVID}} & \multicolumn{2}{c}{{\scriptsize Pre-COVID}} & \multicolumn{2}{c}{{\scriptsize Post-COVID}}\tabularnewline
\cline{2-9}
\multicolumn{1}{c}{} & {\tiny Predictors} & {\tiny Corr} & {\tiny Predictors} & \multicolumn{1}{c}{{\tiny Corr}} & {\tiny Predictors} & {\tiny Corr} & {\tiny Predictors} & {\tiny Corr}\tabularnewline
\hline 
\hline 
\multirow{10}{*}{{\tiny GDP Growth}} & {\tiny 3-Month Treasury C Minus FEDFUNDS} & {\tiny 0.45} & {\tiny USA-Years 11-15 lagged returns, nonannual} & {\tiny 0.30} & {\tiny CHN-Net debt-to-price} & {\tiny 0.45} & {\tiny JPN-Inventory growth} & {\tiny 0.32}\tabularnewline
 & {\tiny 6-Month Treasury C Minus FEDFUNDS} & {\tiny 0.43} & {\tiny GBR-Earnings persistence} & {\tiny 0.30} & {\tiny Moody\textquoteright s Aaa Corporate Bond Minus FEDFUNDS} & {\tiny 0.45} & {\tiny TWN-Frazzini-Pedersen market beta} & {\tiny 0.31}\tabularnewline
 & {\tiny 1-Year Treasury C Minus FEDFUNDS} & {\tiny 0.41} & {\tiny DEU-Short-term reversal} & {\tiny 0.40} & {\tiny CHN-Book-to-market enterprise value} & {\tiny 0.38} & {\tiny JPN-CAPEX growth (1 year)} & {\tiny 0.34}\tabularnewline
 & {\tiny AUS-Current Price to price over last year} & {\tiny 0.36} & {\tiny FRA-Short-term reversal} & {\tiny 0.35} & {\tiny CHN-Debt-to-market} & {\tiny 0.38} & {\tiny JPN-Change in net operating assets} & {\tiny 0.31}\tabularnewline
 & {\tiny FRA-Operating profits-to-lagged book assets} & {\tiny 0.35} & {\tiny AUS-Change sales minus change receivables} & {\tiny 0.34} & {\tiny AUS-Earnings variability} & {\tiny 0.39} & {\tiny GBR-Sales growth (1 quarter)} & {\tiny 0.31}\tabularnewline
 & {\tiny AUS-Book-to-market enterprise value} & {\tiny 0.38} & {\tiny AUS-Intrinsic value-to-market} & {\tiny 0.32} & {\tiny 5-Year Treasury C Minus FEDFUNDS} & {\tiny 0.43} & {\tiny FRA-Hiring rate} & {\tiny 0.35}\tabularnewline
 & {\tiny USA-Coefficient of variation for share turnover} & {\tiny 0.34} & {\tiny 3-Month Treasury C Minus FEDFUNDS} & {\tiny 0.37} & {\tiny Moody\textquoteright s Baa Corporate Bond Minus FEDFUNDS} & {\tiny 0.44} & {\tiny FRA-Sales Growth (1 year)} & {\tiny 0.32}\tabularnewline
 & {\tiny FRA-Operating profits-to-book assets} & {\tiny 0.36} & {\tiny DEU-Price momentum t-3 to t-1} & {\tiny 0.28} & {\tiny 10-Year Treasury C Minus FEDFUNDS} & {\tiny 0.44} & {\tiny FRA-Sales Growth (3 years)} & {\tiny 0.31}\tabularnewline
 & {\tiny AUS-Asset tangibility} & {\tiny 0.39} & {\tiny 6-Month Treasury C Minus FEDFUNDS} & {\tiny 0.35} & {\tiny CHN-Change in operating cash flow to assets} & {\tiny 0.47} & {\tiny FRA-R\&D-to-sales} & {\tiny 0.35}\tabularnewline
 & {\tiny AUS-Kaplan-Zingales index} & {\tiny 0.38} & {\tiny 1-Year Treasury C Minus FEDFUNDS} & {\tiny 0.34} & {\tiny CHN-Idiosyncratic skewness from the CAPM} & {\tiny 0.39} & {\tiny FRA-Years 6-10 lagged returns, nonannual} & {\tiny 0.37}\tabularnewline
\hline 
\multirow{10}{*}{{\tiny Inflation}} & {\tiny CAN-Change in net noncurrent operating assets} & {\tiny 0.34} & {\tiny CAN-Change PPE and Inventory} & {\tiny 0.42} & {\tiny DEU-Capital turnover} & {\tiny 0.27} & {\tiny FRA-Change sales minus change SG\&A} & {\tiny 0.32}\tabularnewline
 & {\tiny CAN-Change in long-term net operating assets} & {\tiny 0.36} & {\tiny CAN-Change in noncurrent operating assets} & {\tiny 0.38} & {\tiny CHN-Sales Growth (3 years)} & {\tiny 0.33} & {\tiny JPN-Ohlson O-score} & {\tiny 0.34}\tabularnewline
 & {\tiny CAN-Change in noncurrent operating assets} & {\tiny 0.34} & {\tiny CAN-Change in net noncurrent operati assets} & {\tiny 0.37} & {\tiny CHN-Gross profits-to-assets} & {\tiny 0.38} & {\tiny JPN-Cash-base operat profits-to-lag book assets} & {\tiny 0.32}\tabularnewline
 & {\tiny GBR-Assets turnover} & {\tiny 0.32} & {\tiny USA-Net operating assets} & {\tiny 0.40} & {\tiny CHN-Assets turnover} & {\tiny 0.32} & {\tiny KOR-Percent operating accruals} & {\tiny 0.28}\tabularnewline
 & {\tiny CAN-Change in net operating assets} & {\tiny 0.34} & {\tiny CAN-Change in net operating assets} & {\tiny 0.36} & {\tiny MEX-Dimson beta} & {\tiny 0.28} & {\tiny CHN-Operating leverage} & {\tiny 0.34}\tabularnewline
 & {\tiny USA-Abnormal corporate investment} & {\tiny 0.30} & {\tiny DEU-R\&D-to-sales} & {\tiny 0.41} & {\tiny KOR-Market correlation} & {\tiny 0.27} & {\tiny JPN-Cash-to-assets} & {\tiny 0.33}\tabularnewline
 & {\tiny CAN-Frazzini-Pedersen market beta} & {\tiny 0.35} & {\tiny TWN-Change in current operati working capital} & {\tiny 0.38} & {\tiny Personal Cons. Exp: Services} & {\tiny 0.31} & {\tiny KOR-Percent operating accruals} & {\tiny 0.34}\tabularnewline
 & {\tiny CAN-Num of consec qtrs with earning increases} & {\tiny 0.38} & {\tiny TWN-Capital turnover} & {\tiny 0.43} & {\tiny CHN-Operating leverage} & {\tiny 0.28} & {\tiny USA-Change in current operating assets} & {\tiny 0.31}\tabularnewline
 & {\tiny CAN-Change PPE and Inventory} & {\tiny 0.39} & {\tiny CAN-Change in long-term net operating assets} & {\tiny 0.39} & {\tiny CHN-Return on net operating assets} & {\tiny 0.34} & {\tiny CHN-Change sales minus change Inventory} & {\tiny 0.34}\tabularnewline
 & {\tiny GBR-Years 6-10 lagged returns, nonannual} & {\tiny 0.35} & {\tiny CAN-Mispricing factor: Management} & {\tiny 0.35} & {\tiny CHN-Cash-based operat profits-to-book assets} & {\tiny 0.35} & {\tiny JPN-Altman Z-score} & {\tiny 0.32}\tabularnewline
\hline 
\multirow{10}{*}{{\tiny IP Growth}} & {\tiny AUS-Change PPE and Inventory} & {\tiny 0.51} & {\tiny USA-Abnormal corporate investment} & {\tiny 0.35} & {\tiny Moody\textquoteright s Aaa Corporate Bond Minus FEDFUNDS} & {\tiny 0.53} & {\tiny FRA-Hiring rate} & {\tiny 0.35}\tabularnewline
 & {\tiny AUS-Years 6-10 lagged returns, annual} & {\tiny 0.48} & {\tiny DEU-R\&D-to-market} & {\tiny 0.34} & {\tiny Moody\textquoteright s Baa Corporate Bond Minus FEDFUNDS} & {\tiny 0.53} & {\tiny TWN-Cash-to-assets} & {\tiny 0.32}\tabularnewline
 & {\tiny GBR-Operating leverage} & {\tiny 0.44} & {\tiny GBR-Amihud Measure} & {\tiny 0.34} & {\tiny 5-Year Treasury C Minus FEDFUNDS} & {\tiny 0.49} & {\tiny FRA-Percent total accruals} & {\tiny 0.35}\tabularnewline
 & {\tiny AUS-Liquidity of book assets} & {\tiny 0.43} & {\tiny AUS-Intrinsic value-to-market} & {\tiny 0.39} & {\tiny 10-Year Treasury C Minus FEDFUNDS} & {\tiny 0.50} & {\tiny FRA-Sales Growth (3 years)} & {\tiny 0.30}\tabularnewline
 & {\tiny 3-Month Treasury C Minus FEDFUNDS} & {\tiny 0.43} & {\tiny DEU-Price momentum t-3 to t-1} & {\tiny 0.31} & {\tiny AUS-Change in current operating liabilities} & {\tiny 0.45} & {\tiny FRA-Years 6-10 lagged returns, nonannual} & {\tiny 0.35}\tabularnewline
 & {\tiny AUS-Change in operating cash flow to assets} & {\tiny 0.45} & {\tiny DEU-Price per share} & {\tiny 0.34} & {\tiny KOR-Long-term reversal} & {\tiny 0.53} & {\tiny TWN-Short-term reversal} & {\tiny 0.42}\tabularnewline
 & {\tiny 6-Month Treasury C Minus FEDFUNDS} & {\tiny 0.46} & {\tiny GBR-Dollar trading volume} & {\tiny 0.36} & {\tiny KOR-Earnings-to-price} & {\tiny 0.41} & {\tiny FRA-R\&D-to-sales} & {\tiny 0.37}\tabularnewline
 & {\tiny Real Estate Loans at All Commercial Banks} & {\tiny 0.42} & {\tiny GBR-Market Equity} & {\tiny 0.37} & {\tiny KOR-Highest 5 days of return scaled by volatility} & {\tiny 0.48} & {\tiny TWN-Liquidity of book assets} & {\tiny 0.31}\tabularnewline
 & {\tiny CAN-Cash-base operat profit-to-lag book assets} & {\tiny 0.47} & {\tiny CAN-Change sales minus change SG\&A} & {\tiny 0.37} & {\tiny KOR-Intrinsic value-to-market} & {\tiny 0.44} & {\tiny CPI : Apparel} & {\tiny 0.31}\tabularnewline
 & {\tiny FRA-Labor force efficiency} & {\tiny 0.42} & {\tiny 3-Month Treasury C Minus FEDFUNDS} & {\tiny 0.41} & {\tiny KOR-Sales Growth (3 years)} & {\tiny 0.43} & {\tiny USA-Years 16-20 lagged returns, nonannual} & {\tiny 0.32}\tabularnewline
\hline 
\multirow{10}{*}{{\tiny Unemployment}} & {\tiny 3-Month Treasury C Minus FEDFUNDS} & {\tiny 0.47} & {\tiny KOR-Change sales minus change SG\&A} & {\tiny 0.32} & {\tiny 5-Year Treasury C Minus FEDFUNDS} & {\tiny 0.59} & {\tiny FRA-Hiring rate} & {\tiny 0.41}\tabularnewline
 & {\tiny 6-Month Treasury C Minus FEDFUNDS} & {\tiny 0.45} & {\tiny CAN-Idiosyncratic skew from the CAPM} & {\tiny 0.33} & {\tiny 10-Year Treasury C Minus FEDFUNDS} & {\tiny 0.62} & {\tiny JPN-Change in net operating assets} & {\tiny 0.42}\tabularnewline
 & {\tiny 1-Year Treasury C Minus FEDFUNDS} & {\tiny 0.41} & {\tiny 1-Year Treasury C Minus FEDFUNDS} & {\tiny 0.38} & {\tiny Moody\textquoteright s Aaa Corporate Bond Minus FEDFUNDS} & {\tiny 0.65} & {\tiny GBR-Gross profits-to-lagged assets} & {\tiny 0.37}\tabularnewline
 & {\tiny USA-Quarterly return on assets} & {\tiny 0.41} & {\tiny USA-Quarterly return on equity} & {\tiny 0.30} & {\tiny Moody\textquoteright s Baa Corporate Bond Minus FEDFUNDS} & {\tiny 0.65} & {\tiny TWN-Cash-to-assets} & {\tiny 0.39}\tabularnewline
 & {\tiny AUS-Change PPE and Inventory} & {\tiny 0.39} & {\tiny USA-Coeffi of variation for share turnover} & {\tiny 0.37} & {\tiny 3-Month Treasury C Minus FEDFUNDS} & {\tiny 0.58} & {\tiny JPN-Cash-to-assets} & {\tiny 0.36}\tabularnewline
 & {\tiny AUS-Change in quarterly return on assets} & {\tiny 0.38} & {\tiny MEX-Price momentum t-6 to t-1} & {\tiny 0.33} & {\tiny 6-Month Treasury C Minus FEDFUNDS} & {\tiny 0.57} & {\tiny TWN-Liquidity of book assets} & {\tiny 0.43}\tabularnewline
 & {\tiny USA-Change gross margin - change sales} & {\tiny 0.40} & {\tiny DEU-The high-low bid-ask spread} & {\tiny 0.32} & {\tiny CAN-Liquidity of book assets} & {\tiny 0.45} & {\tiny FRA-Years 6-10 lagged returns, nonannual} & {\tiny 0.38}\tabularnewline
 & {\tiny USA-Mispricing factor: Performance} & {\tiny 0.36} & {\tiny AUS-Intrinsic value-to-marke} & {\tiny 0.35} & {\tiny AUS-Liquidity of book assets} & {\tiny 0.47} & {\tiny JPN-Inventory growth} & {\tiny 0.41}\tabularnewline
 & {\tiny AUS-Price momentum t-6 to t-1} & {\tiny 0.37} & {\tiny USA-Return on equity} & {\tiny 0.30} & {\tiny 1-Year Treasury C Minus FEDFUNDS} & {\tiny 0.52} & {\tiny FRA-R\&D-to-sales} & {\tiny 0.45}\tabularnewline
 & {\tiny AUS-Change sales - change Inventory} & {\tiny 0.39} & {\tiny GBR-Market Equity} & {\tiny 0.37} & {\tiny FRA-Labor force efficiency} & {\tiny 0.46} & {\tiny CHN-Sales Growth (1 year)} & {\tiny 0.44}\tabularnewline
\hline 
\end{tabular}}}{\scriptsize\par}
\par\end{centering}
$\,$

{\footnotesize\textbf{\label{Table D.8}Notes: }}{\footnotesize In
this table, we report the top 10 predictors selected by SsPCA for
forecasting macro targets at horizons $h=1$ quarter and $h=4$ quarters,
separately for the pre-COVID and post-COVID periods. The forecasting
sample covers 2011Q4\textendash 2024Q4, and the COVID breakpoint is
set at March 2020. At each rolling window, the top 50 predictors are
selected, and the top 10 shown here are ranked separately within each
subsample by the number of times they were selected. The reported
correlations are the average correlations of the selected predictors
with the target variable. For readability, we present descriptive
labels of predictors (with minor abbreviations for formatting) instead
of the original variable codes. Detailed variable definitions can
be found in \cite{mccracken2016fred} and \cite{jensen2023there}
online documents.}{\footnotesize\par}
\end{sidewaystable}

\begin{sidewaystable}[H]
\setlength{\tabcolsep}{0.75pt} 

\caption{Predictors Selected by SsPCA (Finance)}

\begin{centering}
\resizebox{\linewidth}{!}{
{\scriptsize{}%
\begin{tabular}{c|cccc|cccc}
\cline{2-9}
\multicolumn{1}{c}{} & \multicolumn{4}{c}{{\scriptsize$h=1$}} & \multicolumn{4}{c}{{\scriptsize$h=4$}}\tabularnewline
\cline{2-9}
\multicolumn{1}{c}{} & \multicolumn{2}{c}{{\scriptsize Pre-COVID}} & \multicolumn{2}{c}{{\scriptsize Post-COVID}} & \multicolumn{2}{c}{{\scriptsize Pre-COVID}} & \multicolumn{2}{c}{{\scriptsize Post-COVID}}\tabularnewline
\cline{2-9}
\multicolumn{1}{c}{} & {\tiny Predictors} & {\tiny Corr} & {\tiny Predictors} & \multicolumn{1}{c}{{\tiny Corr}} & {\tiny Predictors} & {\tiny Corr} & {\tiny Predictors} & {\tiny Corr}\tabularnewline
\hline 
\hline 
\multirow{10}{*}{{\tiny S\&P500}} & {\tiny 6-Month Treasury C Minus FEDFUNDS} & {\tiny 0.45} & {\tiny TWN-Capital turnover} & {\tiny 0.44} & {\tiny CHN-Net debt-to-price} & {\tiny 0.44} & {\tiny KOR-Percent operating accruals} & {\tiny 0.41}\tabularnewline
 & {\tiny 1-Year Treasury C Minus FEDFUNDS} & {\tiny 0.42} & {\tiny TWN-Assets-to-market} & {\tiny 0.45} & {\tiny KOR-Long-term reversal} & {\tiny 0.44} & {\tiny KOR-Altman Z-score} & {\tiny 0.49}\tabularnewline
 & {\tiny 3-Month Treasury C Minus FEDFUNDS} & {\tiny 0.43} & {\tiny TWN-Assets turnover} & {\tiny 0.44} & {\tiny KOR-Years 2-5 lagged returns, nonannual} & {\tiny 0.49} & {\tiny CHN-Liquidity of market assets} & {\tiny 0.41}\tabularnewline
 & {\tiny AUS-Profit margin} & {\tiny 0.44} & {\tiny AUS-Change in operating cash flow to assets} & {\tiny 0.40} & {\tiny Moody\textquoteright s Aaa Corporate Bond Minus FEDFUNDS} & {\tiny 0.40} & {\tiny KOR-Gross profits-to-assets} & {\tiny 0.45}\tabularnewline
 & {\tiny AUS-Quality minus Junk: Profitability} & {\tiny 0.42} & {\tiny TWN-Cash-to-assets} & {\tiny 0.42} & {\tiny CHN-Change in net financial assets} & {\tiny 0.40} & {\tiny CHN-Net debt-to-price} & {\tiny 0.44}\tabularnewline
 & {\tiny AUS-Change PPE and Inventory} & {\tiny 0.49} & {\tiny TWN-Quality minus Junk: Profitability} & {\tiny 0.42} & {\tiny Moody\textquoteright s Baa Corporate Bond Minus FEDFUNDS} & {\tiny 0.40} & {\tiny KOR-Return on equity} & {\tiny 0.38}\tabularnewline
 & {\tiny AUS-Idiosyncratic volatility from CAPM (252d)} & {\tiny 0.42} & {\tiny 3-Month Treasury C Minus FEDFUNDS} & {\tiny 0.46} & {\tiny CHN-Percent total accruals} & {\tiny 0.37} & {\tiny CHN-Percent total accruals} & {\tiny 0.38}\tabularnewline
 & {\tiny CAN-Net operating assets} & {\tiny 0.40} & {\tiny 6-Month Treasury C Minus FEDFUNDS} & {\tiny 0.51} & {\tiny CHN-Quality minus Junk: Safety} & {\tiny 0.40} & {\tiny CHN-Altman Z-score} & {\tiny 0.38}\tabularnewline
 & {\tiny TWN-Capital turnover} & {\tiny 0.40} & {\tiny 1-Year Treasury C Minus FEDFUNDS} & {\tiny 0.50} & {\tiny CHN-Liquidity of market assets} & {\tiny 0.41} & {\tiny CHN-Change in operating cash flow to assets} & {\tiny 0.52}\tabularnewline
 & {\tiny AUS-Asset tangibility} & {\tiny 0.41} & {\tiny USA-Coeffi of variation for share turnover} & {\tiny 0.36} & {\tiny KOR-Assets turnover} & {\tiny 0.40} & {\tiny TWN-Operating accruals} & {\tiny 0.45}\tabularnewline
\hline 
\multirow{10}{*}{{\tiny VIX}} & {\tiny 3-Month AA Financial Commercial Paper Rate} & {\tiny 0.42} & {\tiny CAN-Change in noncurrent operating assets} & {\tiny 0.35} & {\tiny 5-Year Treasury C Minus FEDFUNDS} & {\tiny 0.53} & {\tiny FRA-R\&D-to-sales} & {\tiny 0.48}\tabularnewline
 & {\tiny 3-Month Treasury Bill} & {\tiny 0.39} & {\tiny AUS-Operating cash flow-to-market} & {\tiny 0.35} & {\tiny 10-Year Treasury C Minus FEDFUNDS} & {\tiny 0.53} & {\tiny FRA-R\&D capital-to-book assets} & {\tiny 0.47}\tabularnewline
 & {\tiny 6-Month Treasury Bill} & {\tiny 0.41} & {\tiny 3-Month AA Financial Commercial Paper Rate} & {\tiny 0.44} & {\tiny Moody\textquoteright s Aaa Corporate Bond Minus FEDFUNDS} & {\tiny 0.53} & {\tiny KOR-Firm age} & {\tiny 0.40}\tabularnewline
 & {\tiny 1-Year Treasury Rate} & {\tiny 0.40} & {\tiny 3-Month Treasury Bill} & {\tiny 0.40} & {\tiny Moody\textquoteright s Baa Corporate Bond Minus FEDFUNDS} & {\tiny 0.48} & {\tiny DEU-Total skewness} & {\tiny 0.40}\tabularnewline
 & {\tiny 3-Month Treasury C Minus FEDFUNDS} & {\tiny 0.45} & {\tiny 6-Month Treasury Bill} & {\tiny 0.43} & {\tiny CHN-Net debt-to-price} & {\tiny 0.45} & {\tiny KOR-Sales-to-market} & {\tiny 0.39}\tabularnewline
 & {\tiny 6-Month Treasury C Minus FEDFUNDS} & {\tiny 0.47} & {\tiny 1-Year Treasury Rate} & {\tiny 0.43} & {\tiny 6-Month Treasury C Minus FEDFUNDS} & {\tiny 0.54} & {\tiny CHN-Net debt-to-price} & {\tiny 0.43}\tabularnewline
 & {\tiny 1-Year Treasury C Minus FEDFUNDS} & {\tiny 0.44} & {\tiny 3-Month Treasury C Minus FEDFUNDS} & {\tiny 0.47} & {\tiny 3-Month Treasury C Minus FEDFUNDS} & {\tiny 0.52} & {\tiny DEU-R\&D capital-to-book assets} & {\tiny 0.48}\tabularnewline
 & {\tiny Effective Federal Funds Rate} & {\tiny 0.38} & {\tiny 6-Month Treasury C Minus FEDFUNDS} & {\tiny 0.53} & {\tiny 1-Year Treasury C Minus FEDFUNDS} & {\tiny 0.51} & {\tiny Altman Z-score} & {\tiny 0.52}\tabularnewline
 & {\tiny AUS-Quality minus Junk: Profitability} & {\tiny 0.35} & {\tiny 1-Year Treasury C Minus FEDFUNDS} & {\tiny 0.52} & {\tiny MEX-Ebitda-to-market enterprise value} & {\tiny 0.47} & {\tiny KOR-Taxable income-to-book income} & {\tiny 0.49}\tabularnewline
 & {\tiny TWN-Capital turnover} & {\tiny 0.37} & {\tiny TWN-Book-to-market equity} & {\tiny 0.37} & {\tiny KOR-Years 2-5 lagged returns, nonannual} & {\tiny 0.48} & {\tiny 5-Year Treasury C Minus FEDFUNDS} & {\tiny 0.43}\tabularnewline
\hline 
\multirow{10}{*}{{\tiny Oil Price}} & {\tiny CAN-Net operating assets} & {\tiny 0.39} & {\tiny CAN-Change in noncurrent operat assets} & {\tiny 0.34} & {\tiny JPN-Change in noncurrent operati liabilities} & {\tiny 0.37} & {\tiny CHN-Operating leverage} & {\tiny 0.37}\tabularnewline
 & {\tiny CAN-Change in long-term net operat assets} & {\tiny 0.36} & {\tiny JPN-Short-term reversal} & {\tiny 0.37} & {\tiny KOR-Highest 5d of return scaled by volatilit} & {\tiny 0.38} & {\tiny CHN-Change sales minus change Inventory} & {\tiny 0.34}\tabularnewline
 & {\tiny CAN-Sales Growth (3 years)} & {\tiny 0.36} & {\tiny JPN-Idiosyncratic skew from the CAPM} & {\tiny 0.36} & {\tiny CHN-Change sales minus change Inventory} & {\tiny 0.33} & {\tiny CHN-Cash-base operat profit-to-book assets} & {\tiny 0.34}\tabularnewline
 & {\tiny CAN-Change in noncurrent operat assets} & {\tiny 0.34} & {\tiny JPN-Idiosyncratic skew from the FF3 model} & {\tiny 0.36} & {\tiny KOR-Idiosyncratic volatilit from CAPM (21d)} & {\tiny 0.39} & {\tiny DEU-Short-term reversal} & {\tiny 0.39}\tabularnewline
 & {\tiny GBR-Ohlson O-score} & {\tiny 0.33} & {\tiny AUS-Highest 5d of return scaled by volatility} & {\tiny 0.31} & {\tiny KOR-Highest 5 days of return} & {\tiny 0.36} & {\tiny Moody\textquoteright s Seasoned Aaa Corporate Bond Yield} & {\tiny 0.32}\tabularnewline
 & {\tiny AUS-Price momentum t-6 to t-1} & {\tiny 0.34} & {\tiny CAN-Change in net noncurrent operat assets} & {\tiny 0.33} & {\tiny DEU-Inventory change} & {\tiny 0.34} & {\tiny AUS-Share turnover} & {\tiny 0.29}\tabularnewline
 & {\tiny GBR-Inventory change} & {\tiny 0.34} & {\tiny AUS-Change in operating cash flow to assets} & {\tiny 0.37} & {\tiny KOR-Years 2-5 lagged returns, nonannual} & {\tiny 0.35} & {\tiny FRA-Years 6-10 lagged returns, nonannual} & {\tiny 0.35}\tabularnewline
 & {\tiny USA-Coefficient of variation for share turnover} & {\tiny 0.35} & {\tiny CAN-Change in net operating assets} & {\tiny 0.31} & {\tiny KOR-Maximum daily return} & {\tiny 0.35} & {\tiny CHN-Operating profits-to-book assets} & {\tiny 0.31}\tabularnewline
 & {\tiny CAN-Change in net operating assets} & {\tiny 0.34} & {\tiny CAN-Asset Growth} & {\tiny 0.31} & {\tiny KOR-Sales Growth (3 years)} & {\tiny 0.33} & {\tiny CHN-Operating cash flow to assets} & {\tiny 0.31}\tabularnewline
 & {\tiny 1-Year Treasury C Minus FEDFUNDS} & {\tiny 0.33} & {\tiny GBR-Earnings persistence} & {\tiny 0.34} & {\tiny KOR-Return volatility} & {\tiny 0.35} & {\tiny Unfilled Orders for Durable Goods} & {\tiny 0.34}\tabularnewline
\hline 
\multirow{10}{*}{{\tiny Housing Price}} & {\tiny AUS-Cash-base operat profits-tolag book assets} & {\tiny 0.40} & {\tiny DEU-Years 2-5 lagged returns, nonannual} & {\tiny 0.40} & {\tiny Moody\textquoteright s Aaa Corporate Bond Minus FEDFUNDS} & {\tiny 0.43} & {\tiny KOR-Liquidity of book assets} & {\tiny 0.44}\tabularnewline
 & {\tiny CHN-Change PPE and Inventory} & {\tiny 0.36} & {\tiny JPN-Years 2-5 lagged returns, nonannual} & {\tiny 0.41} & {\tiny Moody\textquoteright s Baa Corporate Bond Minus FEDFUNDS} & {\tiny 0.43} & {\tiny KOR-Cash-to-assets} & {\tiny 0.46}\tabularnewline
 & {\tiny AUS-Operating profits-to-lagged book equity} & {\tiny 0.42} & {\tiny USA-Change in current operating assets} & {\tiny 0.40} & {\tiny Nonrevolvin consumer credit to Personal Income} & {\tiny 0.39} & {\tiny KOR-Sales Growth (3 years)} & {\tiny 0.41}\tabularnewline
 & {\tiny GBR-Coeffi of variation for share turnover} & {\tiny 0.37} & {\tiny AUS-Ohlson O-score} & {\tiny 0.39} & {\tiny JPN-Coeffic of variation for dollar trading volume} & {\tiny 0.37} & {\tiny Help-Wanted Index for United States} & {\tiny 0.43}\tabularnewline
 & {\tiny CAN-Market Equity} & {\tiny 0.32} & {\tiny AUS-Cash-based operat profits-to-book assets} & {\tiny 0.43} & {\tiny CAN-Year 1-lagged return, annual} & {\tiny 0.36} & {\tiny Real M2 Money Stock} & {\tiny 0.44}\tabularnewline
 & {\tiny AUS-Cash-based operat profits-to-book assets} & {\tiny 0.38} & {\tiny JPN-Long-term reversal} & {\tiny 0.39} & {\tiny 10-Year Treasury C Minus FEDFUNDS} & {\tiny 0.37} & {\tiny CHN-Inventory growth} & {\tiny 0.40}\tabularnewline
 & {\tiny CAN-Dollar trading volume} & {\tiny 0.32} & {\tiny KOR-Net debt issuance} & {\tiny 0.39} & {\tiny JPN-Pitroski F-score} & {\tiny 0.36} & {\tiny KOR-Idiosyncratic volatility from CAPM (252d)} & {\tiny 0.40}\tabularnewline
 & {\tiny MEX-Net equity issuance} & {\tiny 0.34} & {\tiny TWN-Dividend yield} & {\tiny 0.37} & {\tiny JPN-Coefficient of variation for share turnover} & {\tiny 0.34} & {\tiny FRA-Hiring rate} & {\tiny 0.42}\tabularnewline
 & {\tiny GBR-Coeffi of variation for dollar trad volume} & {\tiny 0.32} & {\tiny AUS-Free cash flow-to-price} & {\tiny 0.43} & {\tiny KOR-Growth in book debt (3 years)} & {\tiny 0.38} & {\tiny DEU-Firm age} & {\tiny 0.43}\tabularnewline
 & {\tiny Initial Claims} & {\tiny 0.34} & {\tiny KOR-Change in financial liabilities} & {\tiny 0.37} & {\tiny VIX} & {\tiny 0.47} & {\tiny KOR-Num of 0 trade w turnover as tiebreaker (12m)} & {\tiny 0.41}\tabularnewline
\hline 
\end{tabular}}}{\scriptsize\par}
\par\end{centering}
$\,$

{\footnotesize\textbf{\label{Table D.9}Notes: }}{\footnotesize In
this table, we report the top 10 predictors selected by SsPCA for
forecasting financial targets at horizons $h=1$ quarter and $h=4$
quarters, separately for the pre-COVID and post-COVID periods. The
forecasting sample covers 2011Q4\textendash 2024Q4, and the COVID
breakpoint is set at March 2020. At each rolling window, the top 50
predictors are selected, and the top 10 shown here are ranked separately
within each subsample by the number of times they were selected. The
reported correlations are the average correlations of the selected
predictors with the target variable. For readability, we present descriptive
labels of predictors (with minor abbreviations for formatting) instead
of the original variable codes. Detailed variable definitions can
be found in \cite{mccracken2016fred} and \cite{jensen2023there}
online documents.}{\footnotesize\par}
\end{sidewaystable}

\begin{figure}[H]
\caption{Top 50 Predictors for GDP Growth Rate Forecast Selected by SsPCA}

\begin{centering}
\includegraphics[scale=0.36]{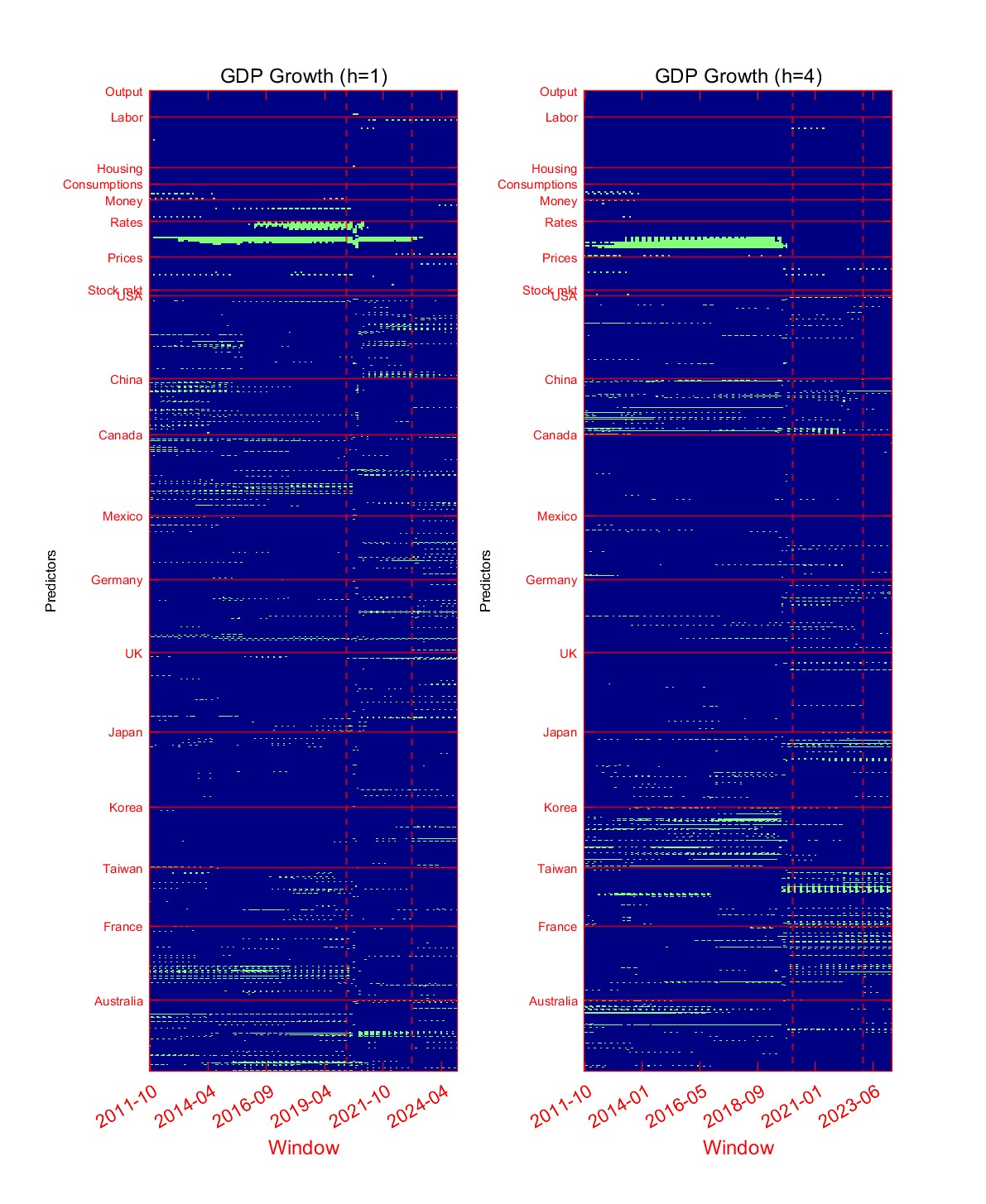}
\par\end{centering}
$\,$

{\footnotesize\textbf{\label{Figure D.1}Notes: }}{\footnotesize In this
figure, each panel visualizes the top 50 predictors selected by SsPCA
throughout the forecasting sample from 2011Q4 to 2024Q4 (159 rolling
windows for one-quarter-ahead forecasts, $h=1$; 150 windows for one-year-ahead
forecasts, $h=4$.) of GDP growth rate. The predictor set is constructed
from FRED-MD and Global Factor Data. For clarity, the FRED-MD predictor
pane is expanded threefold. Only variables that ever ranked among
the top 50 based on their correlation with the target are shown in
green, while blue indicates predictors not selected. Red horizontal
lines separate predictor categories. The vertical dashed red lines
mark the onset of COVID-19 in March 2020 and the recovery phase beginning
in January 2023.}{\footnotesize\par}
\end{figure}

\begin{figure}[H]

\caption{Top 50 Predictors for Inflation Rate Forecast Selected by SsPCA}

\begin{centering}
\includegraphics[scale=0.36]{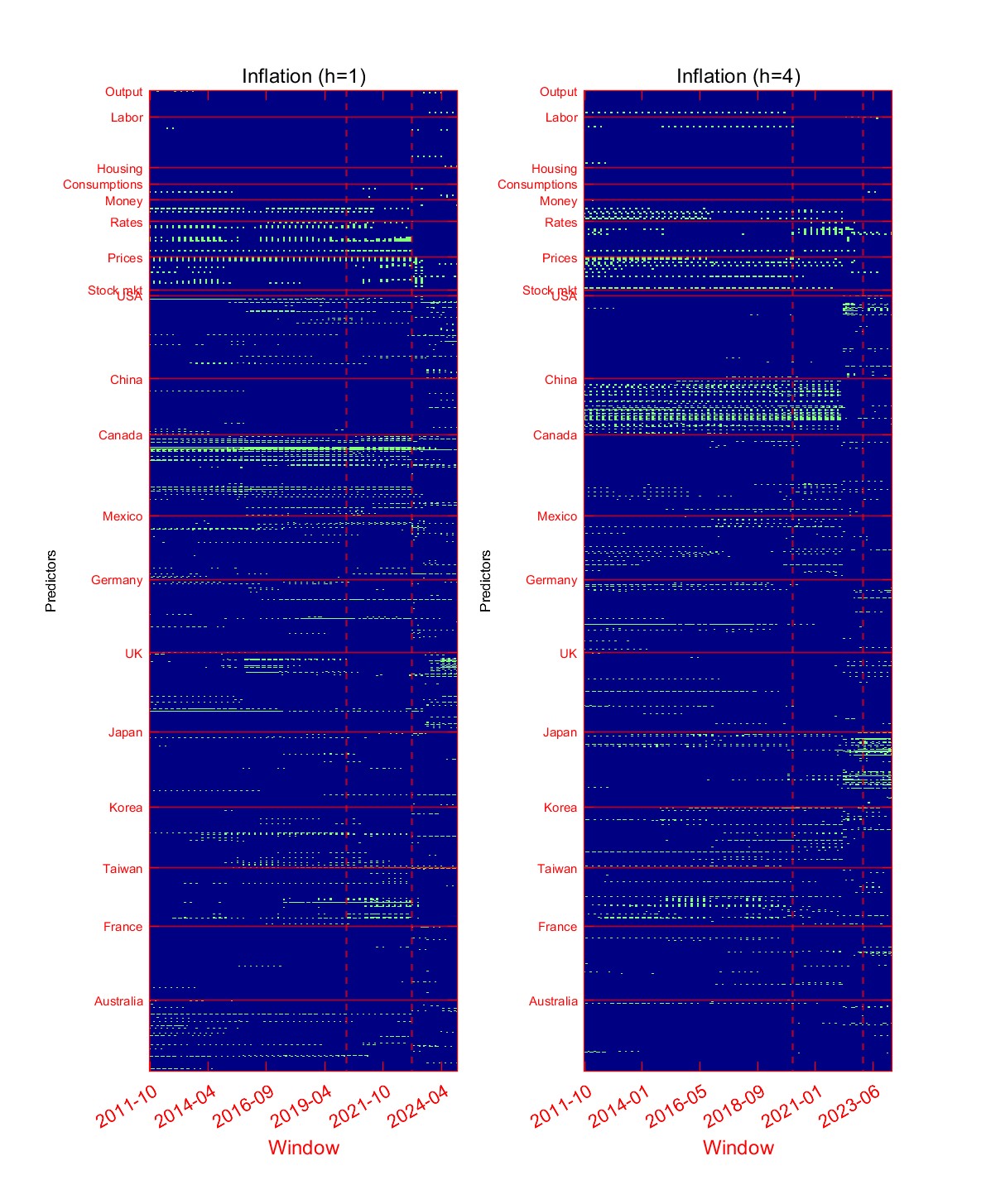}
\par\end{centering}
$\,$

{\footnotesize\textbf{\label{Figure D.2}Notes: }}{\footnotesize In
this figure, each panel visualizes the top 50 predictors selected
by SsPCA throughout the forecasting sample from 2011Q4 to 2024Q4 (159
rolling windows for one-quarter-ahead forecasts, $h=1$; 150 windows
for one-year-ahead forecasts, $h=4$.) of inflation rate. The predictor
set is constructed from FRED-MD and Global Factor Data. For clarity,
the FRED-MD predictor pane is expanded threefold. Only variables that
ever ranked among the top 50 based on their correlation with the target
are shown in green, while blue indicates predictors not selected.
Red horizontal lines separate predictor categories. The vertical dashed
red lines mark the onset of COVID-19 in March 2020 and the recovery
phase beginning in January 2023.}{\footnotesize\par}
\end{figure}

\begin{figure}[H]
\caption{Top 50 Predictors for IP Growth Forecast Selected by SsPCA}

\begin{centering}
\includegraphics[scale=0.36]{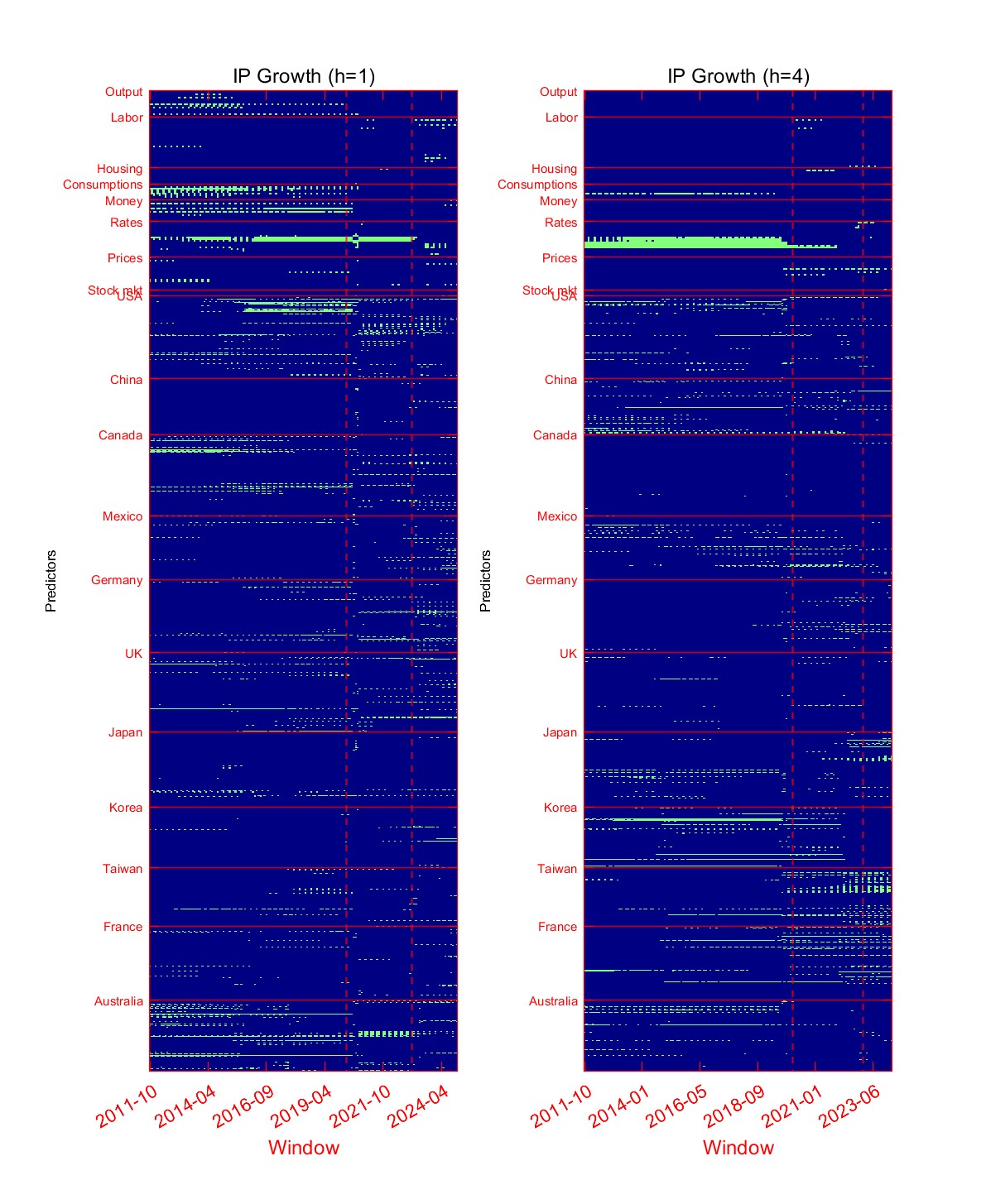}
\par\end{centering}
$\,$

{\footnotesize\textbf{\label{Figure D.3}Notes: }}{\footnotesize In
this figure, each panel visualizes the top 50 predictors selected
by SsPCA throughout the forecasting sample from 2011Q4 to 2024Q4 (159
rolling windows for one-quarter-ahead forecasts, $h=1$; 150 windows
for one-year-ahead forecasts, $h=4$.) of IP growth. The predictor
set is constructed from FRED-MD and Global Factor Data. For clarity,
the FRED-MD predictor pane is expanded threefold. Only variables that
ever ranked among the top 50 based on their correlation with the target
are shown in green, while blue indicates predictors not selected.
Red horizontal lines separate predictor categories. The vertical dashed
red lines mark the onset of COVID-19 in March 2020 and the recovery
phase beginning in January 2023.}{\footnotesize\par}
\end{figure}

\begin{figure}[H]
\caption{Top 50 Predictors for Unemployment Rate Forecast Selected by SsPCA}

\begin{centering}
\includegraphics[scale=0.36]{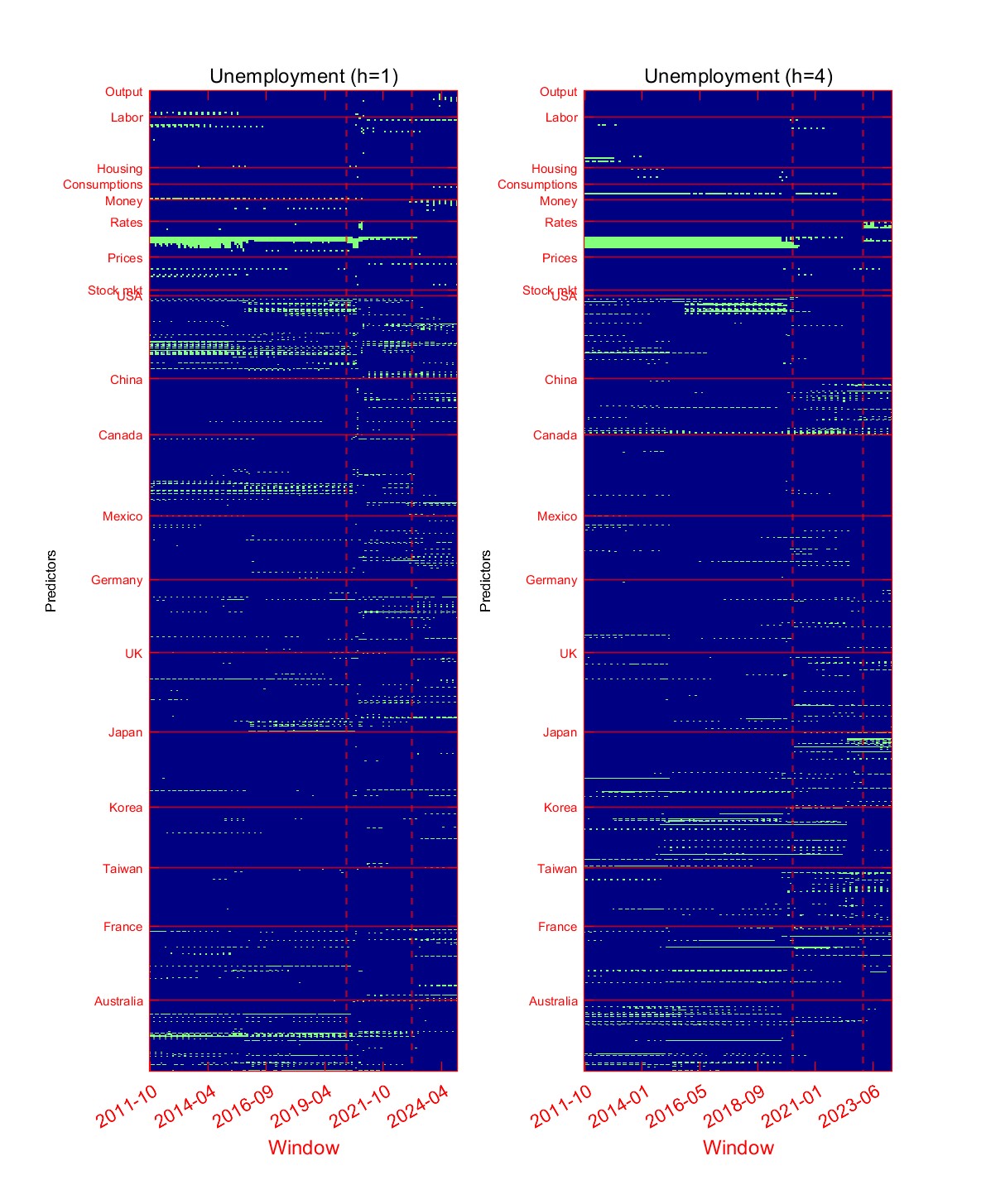}
\par\end{centering}
$\,$

{\footnotesize\textbf{\label{Figure D.4}Notes: }}{\footnotesize In
this figure, each panel visualizes the top 50 predictors selected
by SsPCA throughout the forecasting sample from 2011Q4 to 2024Q4 (159
rolling windows for one-quarter-ahead forecasts, $h=1$; 150 windows
for one-year-ahead forecasts, $h=4$.) of the unemployment rate. The predictor
set is constructed from FRED-MD and Global Factor Data. For clarity,
the FRED-MD predictor pane is expanded threefold. Only variables that
ever ranked among the top 50 based on their correlation with the target
are shown in green, while blue indicates predictors not selected.
Red horizontal lines separate predictor categories. The vertical dashed
red lines mark the onset of COVID-19 in March 2020 and the recovery
phase beginning in January 2023.}{\footnotesize\par}
\end{figure}

\begin{figure}[H]
\caption{Top 50 Predictors for S\&P 500 Index Forecast Selected by SsPCA}

\begin{centering}
\includegraphics[scale=0.36]{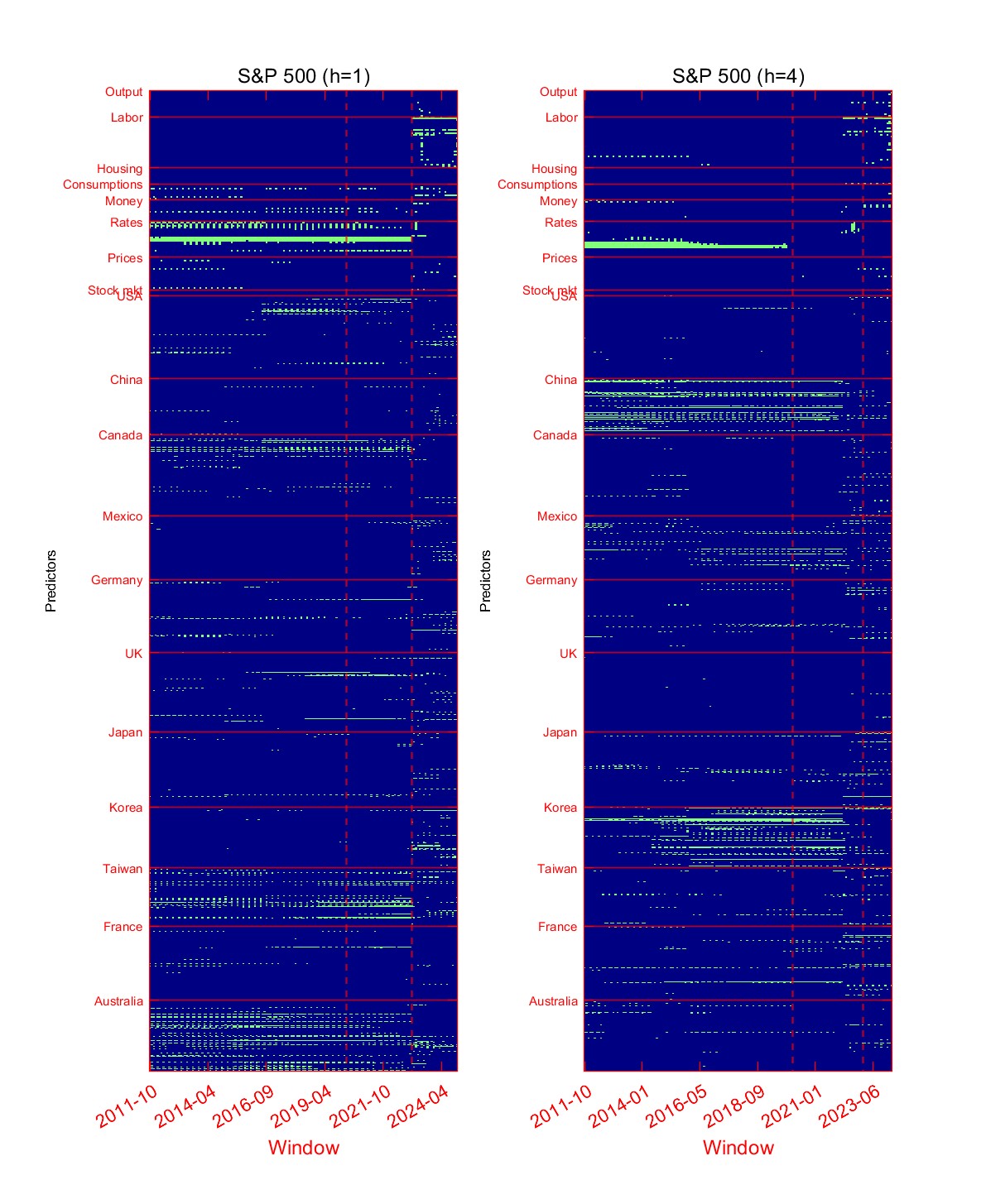}
\par\end{centering}
$\,$

{\footnotesize\textbf{\label{Figure D.5}Notes: }}{\footnotesize In
this figure, each panel visualizes the top 50 predictors selected
by SsPCA throughout the forecasting sample from 2011Q4 to 2024Q4 (159
rolling windows for one-quarter-ahead forecasts, $h=1$; 150 windows
for one-year-ahead forecasts, $h=4$.) of S\&P 500 index. The predictor
set is constructed from FRED-MD and Global Factor Data. For clarity,
the FRED-MD predictor pane is expanded threefold. Only variables that
ever ranked among the top 50 based on their correlation with the target
are shown in green, while blue indicates predictors not selected.
Red horizontal lines separate predictor categories. The vertical dashed
red lines mark the onset of COVID-19 in March 2020 and the recovery
phase beginning in January 2023.}{\footnotesize\par}
\end{figure}

\begin{figure}[H]
\caption{Top 50 Predictors for VIX Index Forecast Selected by SsPCA}

\begin{centering}
\includegraphics[scale=0.36]{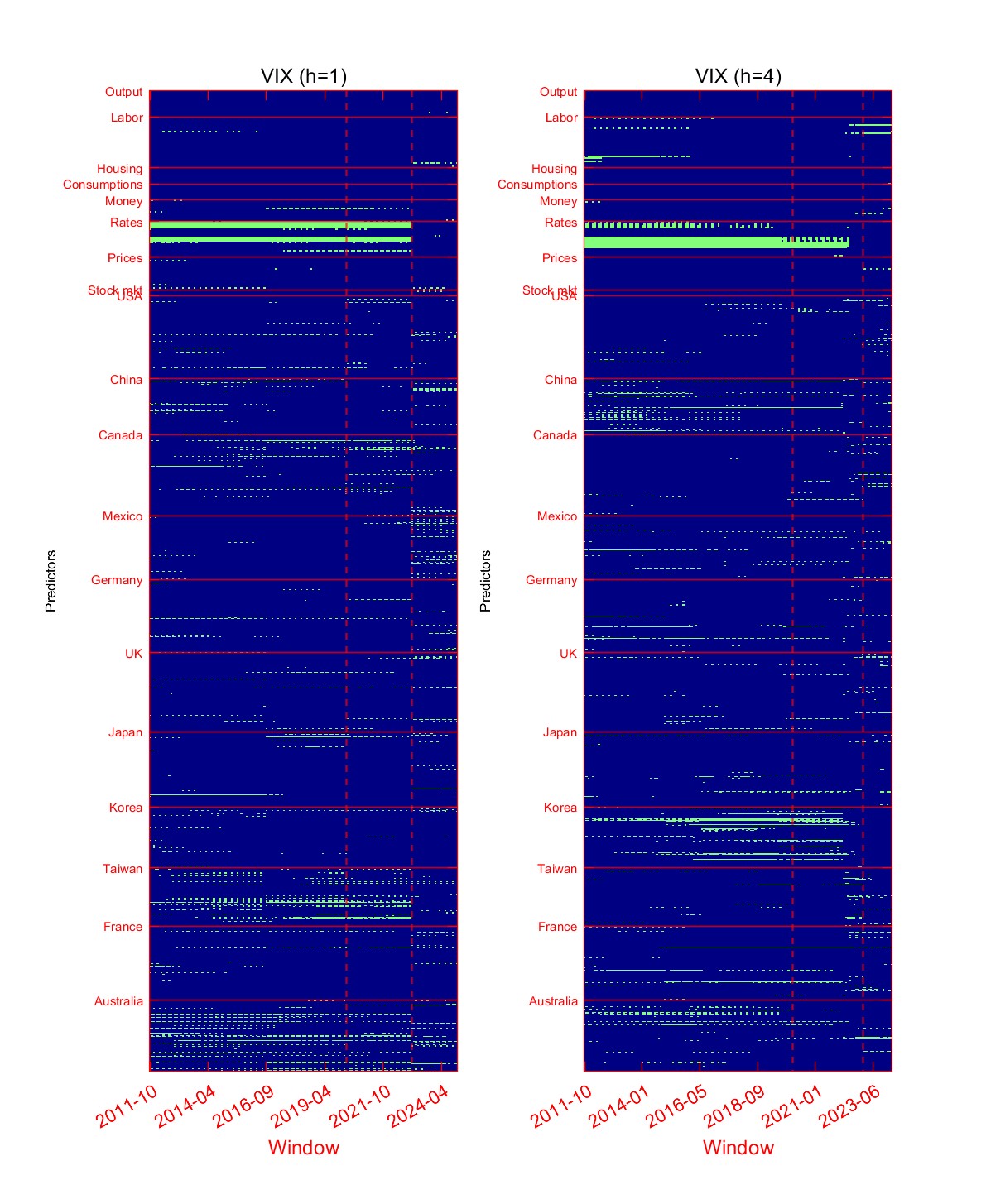}
\par\end{centering}
{\footnotesize\textbf{\label{Figure D.6}Notes: }}{\footnotesize In
this figure, each panel visualizes the top 50 predictors selected
by SsPCA throughout the forecasting sample from 2011Q4 to 2024Q4 (159
rolling windows for one-quarter-ahead forecasts, $h=1$; 150 windows
for one-year-ahead forecasts, $h=4$.) of VIX index. The predictor
set is constructed from FRED-MD and Global Factor Data. For clarity,
the FRED-MD predictor pane is expanded threefold. Only variables that
ever ranked among the top 50 based on their correlation with the target
are shown in green, while blue indicates predictors not selected.
Red horizontal lines separate predictor categories. The vertical dashed
red lines mark the onset of COVID-19 in March 2020 and the recovery
phase beginning in January 2023.}{\footnotesize\par}
\end{figure}

\begin{figure}[H]
\caption{Top 50 Predictors for Crude Oil Price Forecast Selected by SsPCA}

\begin{centering}
\includegraphics[scale=0.36]{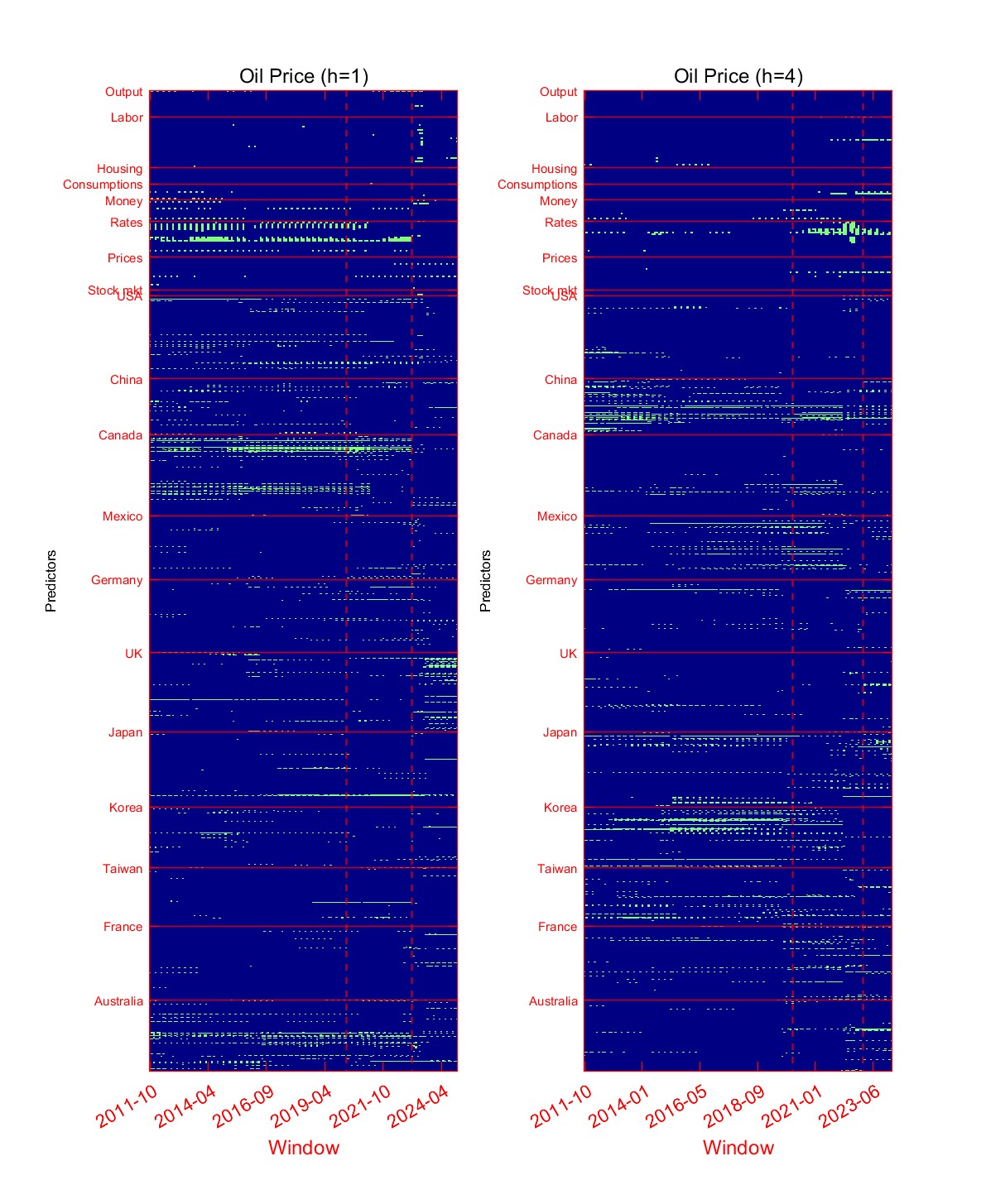}
\par\end{centering}
$\,$

{\footnotesize\textbf{\label{Figure D.7}Notes: }}{\footnotesize In
this figure, each panel visualizes the top 50 predictors selected
by SsPCA throughout the forecasting sample from 2011Q4 to 2024Q4 (159
rolling windows for one-quarter-ahead forecasts, $h=1$; 150 windows
for one-year-ahead forecasts, $h=4$.) of crude oil price. The predictor
set is constructed from FRED-MD and Global Factor Data. For clarity,
the FRED-MD predictor pane is expanded threefold. Only variables that
ever ranked among the top 50 based on their correlation with the target
are shown in green, while blue indicates predictors not selected.
Red horizontal lines separate predictor categories. The vertical dashed
red lines mark the onset of COVID-19 in March 2020 and the recovery
phase beginning in January 2023.}{\footnotesize\par}
\end{figure}

\begin{figure}[H]
\caption{Top 50 Predictors for Housing Price Forecast Selected by SsPCA}

\begin{centering}
\includegraphics[scale=0.36]{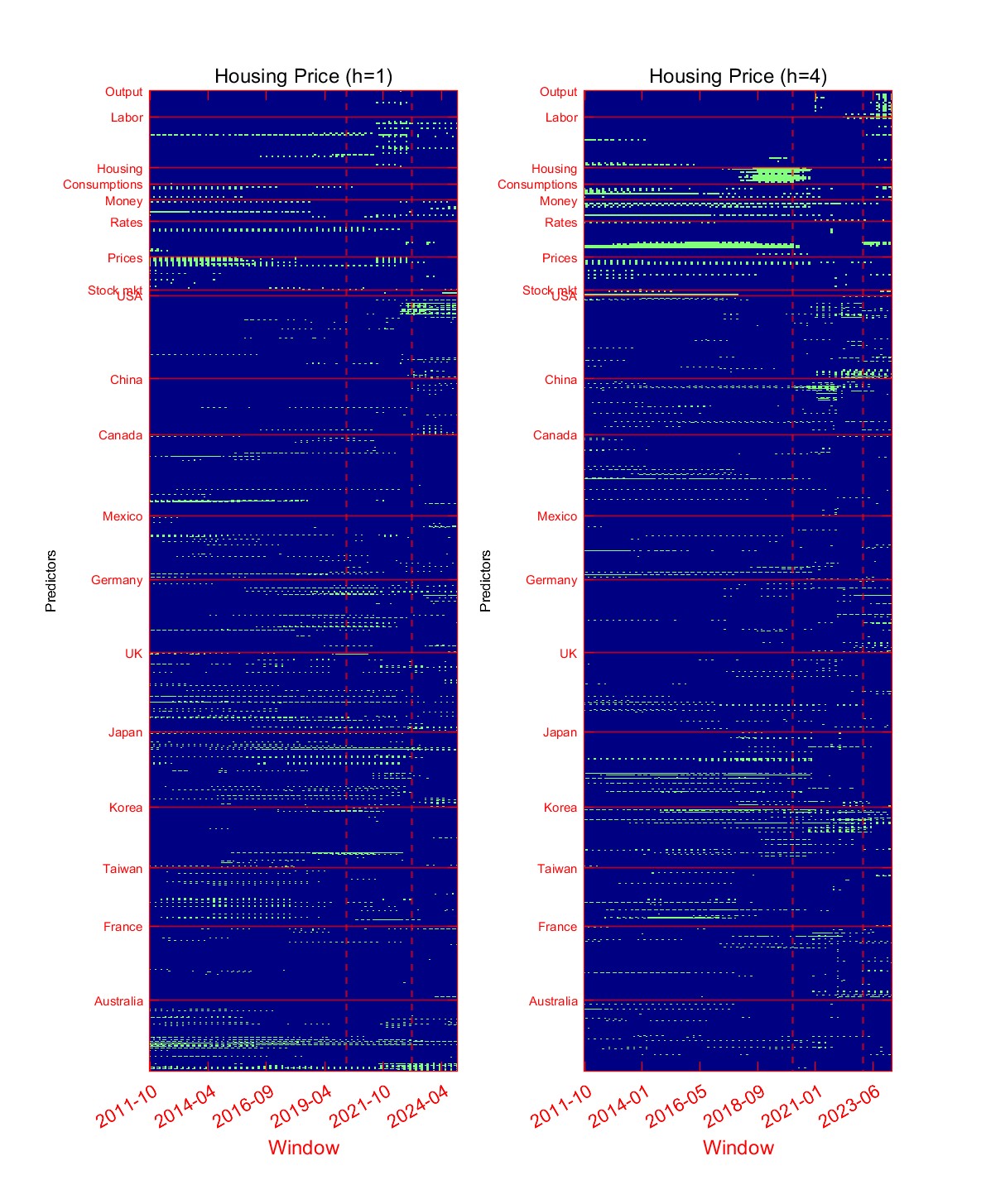}
\par\end{centering}
$\,$

\label{Figure D.8}{\footnotesize\textbf{Notes: }}{\footnotesize In
this figure, each panel visualizes the top 50 predictors selected
by SsPCA throughout the forecasting sample from 2011Q4 to 2024Q4 (159
rolling windows for one-quarter-ahead forecasts, $h=1$; 150 windows
for one-year-ahead forecasts, $h=4$.) of housing prices. The predictor
set is constructed from FRED-MD and Global Factor Data. For clarity,
the FRED-MD predictor pane is expanded threefold. Only variables that
ever ranked among the top 50 based on their correlation with the target
are shown in green, while blue indicates predictors not selected.
Red horizontal lines separate predictor categories. The vertical dashed
red lines mark the onset of COVID-19 in March 2020 and the recovery
phase beginning in January 2023.}{\footnotesize\par}
\end{figure}

\begin{figure}[H]
\caption{Histograms of the Standardized Prediction Errors $(h=4)$}
\begin{centering}
\scalebox{1.25}[1.00]{\includegraphics[scale=0.125]{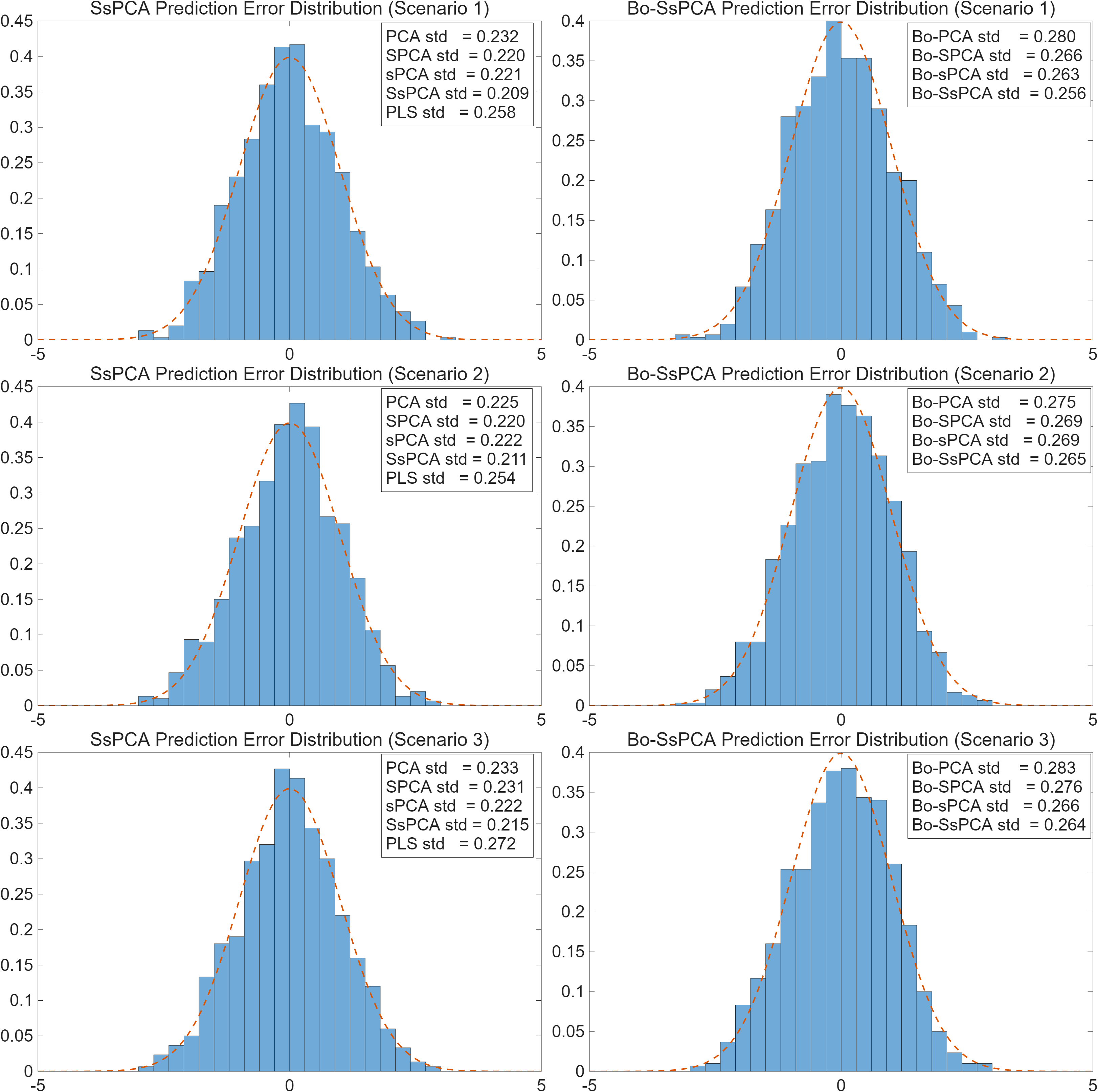}}$,$$\,$
\par\end{centering}
\label{Figure D.9}{\footnotesize\textbf{Notes: }}{\footnotesize The figure
reports histograms of the standardized prediction errors for each scenario
using SsPCA and boosting with SsPCA factors (Bo-SsPCA), for $h=4$ and
$T_H=180$, with the number of boosting iterations fixed at $M=50$, based on
1,000 Monte Carlo repetitions. The dashed curve is the standard normal
density, and standard deviations of all comparable methods are also reported
for reference. The histograms are generated under the same scenario settings
as in Table \ref{Table 1} of the main text, and provide the finite-sample
validation of the central limit theorem in Theorem \ref{thm:4}.}{\footnotesize\par}
\end{figure}

\subsection{Detailed Discussion of the Predictors Selected by SsPCA\label{Appendix D.1}}

This subsection expands on Section \ref{Section 5} by describing in detail the
predictors reported in Tables \ref{Table D.8} and \ref{Table D.9} and in
Figures \ref{Figure D.1} to \ref{Figure D.8}.

A closer look at the selected predictors reveals several economically intuitive patterns. Many of the top predictors are interest rate and credit spread measures, such as Treasury rates and corporate bond rates relative to the federal funds rate, together with short-term money market rates such as the 3-month Treasury bill and commercial paper rates. This highlights the central role of monetary and credit market conditions in shaping macro-financial outcomes \citep{Ang2006,ludvigson2009macro}. Balance sheet indicators, such as leverage, operating accruals, and cash-to-assets ratios, also feature prominently, consistent with the literature emphasizing financial frictions and firm fundamentals. In addition, measures of sales growth, profitability, and inventory adjustment appear frequently, pointing to the relevance of real activity indicators for forecasting GDP, inflation, IP growth, and unemployment rate.

For financial targets such as the S\&P 500, VIX, oil, and housing
prices, similar regularities emerge. Equity returns and volatility
are again linked to monetary policy variables and spreads, while global
predictors, including idiosyncratic volatility from CAPM, turnover
measures, and exchange rate factors, are also selected, reflecting
the cross-border dimension of financial conditions. Oil and housing
prices are often associated with sales growth and
leverage indicators, aligning with their sensitivity to both demand
and financing conditions.

The relevance of these predictors is not static, and a comparison of the pre-COVID and post-COVID columns of Tables \ref{Table D.8} and \ref{Table D.9} reveals a clear shift that mirrors the structural break visible in the heatmaps around March 2020. Before COVID, selection is concentrated in the rates block: term spreads and credit spreads dominate the top ranks for GDP growth, IP growth, unemployment, the S\&P 500, and the VIX. This reflects an environment in which the yield curve and credit conditions were the primary leading indicators, since the term spread embeds market expectations of future growth \citep{estrella1991term,Ang2006} and credit spreads carry forward-looking information on default risk and the risk-bearing capacity of the financial sector \citep{gilchrist2012credit}.

After COVID, the information set tilts toward real-activity and supply-side measures. For GDP growth and unemployment, the pre-COVID dominance of rates gives way to inventory, fixed-investment, and hiring measures, though rates remain relevant throughout, for example Japanese inventory and CAPEX growth and the French hiring rate at the one-year horizon. Inflation is supply-oriented in both subsamples, with operating-asset and inventory measures present before COVID and intensifying afterward, as change in PPE and inventory and capital turnover move to the top ranks. This realignment is consistent with the supply-driven nature of the pandemic shock, in which global supply chain disruptions and a compositional shift from services to goods became leading drivers of inflation and real activity \citep{di2022global}. A second post-COVID shift is geographic. Trade- and technology-linked Asian economies become more visible after COVID: Japanese real-activity measures rise for macroeconomic targets, Taiwanese turnover and profitability measures rise for financial targets, and Korean measures remain prominent at the longer horizon throughout the sample. This pattern is in line with the role these economies play in global supply chains and the semiconductor cycle. A third and more subtle change is the greater appearance of reversal and momentum measures after COVID, for example for GDP growth, oil, and housing prices, consistent with the elevated volatility and faster information turnover that characterized the post-pandemic period. Taken together, these shifts show that SsPCA does not rely on a fixed set of predictors but reallocates attention from interest-rate signals toward supply-side and global production signals as the dominant source of macro-financial fluctuations changed.

Finally, SsPCA does not pick predictors at random. Despite the large pool of $153$ global factors from eleven countries, the method often selects the same type of factor across different countries and targets: turnover ratios, R\&D-to-sales measures, profitability indicators, and debt measures recur across multiple countries. At the same time, many factors are consistently excluded, indicating that SsPCA filters out irrelevant predictors and concentrates on those with genuine forecasting value. Beyond well-established predictors such as rates and spreads, it also identifies less conventional factors, such as idiosyncratic skewness and liquidity-based measures, that may carry useful forecasting signal rather than noise.

\section{Competing Procedures and Their Algorithms\label{Appendix E}}

This appendix collects the algorithms of the procedures against which SsPCA-MIDAS
is compared in Sections \ref{Section 4} and \ref{Section 5}.

\textbf{Principal Component Analysis (PCA)} is a widely used unsupervised method for dimensionality reduction, extracting mutually orthogonal components ranked by their explained variance. We present the PCA algorithm in the MIDAS regression framework below:
\setcounter{lyxalgorithm}{1}
\begin{lyxalgorithm}
\textbf{\textup{(Unsurpervised Learning: PCA-MIDAS)\label{Algorithm 2}}}

\noindent Inputs: $\overline{Y},\bold{x}_{H}=(\underline{\bold{x}}_{H},\bold{x}_{T_{H}-mh+1}^{T_{H}}),\underline{W},\textrm{ and }w_{T}$.

S1. Apply singular value decomposition (SVD) on predictors
$\underline{\bold{x}}_{H}$, and get the first $K$ left singular
vectors written as $\widehat{\varsigma}$. The estimated high-frequency
factors are given by $\underline{\widehat{\bold{f}}}_{H,\textrm{PCA}}=\widehat{\varsigma}'\underline{\bold{x}}_{H}$.

S2. a. Estimate the coefficients $\alpha_{l}=(\alpha,\alpha_{w})',$ and weighting
parameter  $\theta_{f}$  jointly using NLS based on $\overline{Y},\underline{W}$,
and estimated aggregated factors, which involve solving the optimization
problem $\left(\hat{\alpha}_{l},\hat{\theta}_{f}\right)=\underset{\alpha_{l},\theta_{f}}{\textrm{argmin}}\parallel\overline{Y}-\alpha\underline{\hat{F}}_{\textrm{PCA}}({\theta}_{f})-\alpha_{w}\underline{W}\parallel^{2},$
thereby obtaining $\hat{\alpha}_{l},$ and $\hat{\theta}_{f}$.

$\qquad$b. The resulting prediction for $\ensuremath{y_{T+h}}$ is
given by $\widehat{y}_{T+h}^{\textrm{PCA}}=\hat{\alpha}\hat{F}_{T,\textrm{PCA}}(\hat{\theta}_{f})+\widehat{\alpha}_{w}w_{T}\equiv\widehat{\alpha}\sum_{j=0}^{J}b_{j}^{f}(\hat{\theta}_{f})\widehat{\varsigma}'\bold{x}_{T-j/m}$
$+\widehat{\alpha}_{w}w_{T}.$

\noindent outputs: $\widehat{y}_{T+h}^{\textrm{PCA}}$, $\widehat{\underline{F}}_{\textrm{PCA}}(\hat{\theta}_{f})$, $\widehat{\varsigma},$ $\widehat{\alpha}_{l}$,
and $\hat{\theta}_{f}$. 
\end{lyxalgorithm}

\textbf{Scaled Principal Component Analysis (sPCA)}. Being unsupervised, PCA ignores target information and, when factors are weak, may fail to separate signal from noise, leading to biased forecasts.
To mitigate this issue, it is prudent to use supervised learning methods
that incorporate the target variable $y$ to guide the factor extraction. One approach is through supervised data weighting: increasing
weight to relevant predictors while shrinking weight to irrelevant
ones, e.g., see \cite{huang2022scaled} who proposed sPCA. Next, we present the sPCA algorithm in the MIDAS framework.
\begin{lyxalgorithm}
\textbf{\textup{(Supervised Learning: sPCA-MIDAS)\label{Algorithm 3}}}

\noindent Inputs: $\overline{Y},\bold{x}_{H}=(\underline{\bold{x}}_{H},\bold{x}_{T_{H}-mh+1}^{T_{H}}),\underline{W},\textrm{ and }w_{T}.$

S1. a. Obtain the scaling coefficients $\hat{\Upsilon}_{i}$ by regressing
$\overline{Y}$ on each aggregated predictor $\underline{X}_{i}(\theta_{x,i})$.
Stack them to form the diagonal scaling matrix $\hat{\Upsilon}(\theta_{x})$,
and construct the scaled high-frequency predictors $\bold{x}_{H,\textrm{scaled}}(\theta_{x})=\hat{\Upsilon}(\theta_{x})\bold{x}_{H}$.

$\qquad$b. Apply the standard PCA (as in \textbf{Algorithm} \ref{Algorithm 2}
S1) to $\underline{\bold{x}}_{H,\textrm{scaled}}(\theta_{x})$, obtaining the
first $K$ left singular vectors $\widehat{\varsigma}_{\textrm{scaled}}(\theta_{x})$.
This in turn yields the estimated high-frequency factors $\underline{\widehat{\bold{f}}}_{H,\textrm{sPCA}}({\theta}_{x})=\widehat{\varsigma}_{\textrm{scaled}}({\theta}_{x})'\underline{\bold{x}}_{H,\textrm{scaled}}$.

S2. Apply \textbf{Algorithm} \ref{Algorithm 2} S2 to obtain $\hat{\alpha}_l=(\hat{\alpha},\hat{\alpha}_w)'$ and  $\hat{\theta}=(\hat{\theta}_f,\hat{\theta}_x)'$ jointly by solving $\left(\hat{\alpha}_l,\hat{\theta}\right)=\underset{\hat{\alpha}_l,\hat{\theta}}{\textrm{argmin}}\parallel\overline{Y}-\alpha\underline{\hat{F}}_{\textrm{sPCA}}(\theta)-\alpha_{w}\underline{W}\parallel^{2}$
and the resulting prediction for $\ensuremath{y_{T+h}}$ is given
by $\widehat{y}_{T+h}^{\textrm{sPCA}}=\hat{\alpha}\hat{F}_{T,\textrm{sPCA}}(\hat{\theta})+\widehat{\alpha}_{w}w_{T}\equiv\widehat{\alpha}\sum_{j=0}^{J}b_{j}^f(\hat{\theta}_f)\widehat{\varsigma}_{\textrm{scaled}}(\hat{\theta}_x)'\bold{x}_{T-j/m,\textrm{scaled}}+\widehat{\alpha}_{w}w_{T}$.

outputs: $\widehat{y}_{T+h}^{\textrm{sPCA}}$, $\widehat{\underline{F}}_{\textrm{sPCA}}(\hat{\theta})$, $\widehat{\varsigma}_{\textrm{scaled}}(\hat{\theta}_x)$,
$\widehat{\alpha}_{l}$, and $\hat{\theta}$.
\end{lyxalgorithm}

\textbf{Supervised Principal Component Analysis (SPCA)}. Another supervised direction is supervised selection. \cite{giglio2023prediction,giglio2021test} introduce SPCA, which iteratively applies selection, PCA, and projection. We describe the
SPCA algorithm within the framework of MIDAS as follows:
\begin{lyxalgorithm}
\textbf{\textup{(Supervised Learning: SPCA-MIDAS)\label{Algorithm 4}}}

\noindent Inputs: $\overline{Y},\bold{x}_{H}=(\underline{\bold{x}}_{H},\bold{x}_{T_{H}-mh+1}^{T_{H}}),\underline{W},\textrm{ and }w_{T}.$
Initialization:  $Y_{(1)}:=\overline{Y}\mathbb{M}_{\underline{W}'}$,   $\bold{x}_{(1)}:=\underline{\bold{x}}_{H}$.

S1. a. For given values of the parameter ${\theta}_{x}$,
initialize the aggregated predictors as $X_{(1)}({\theta}_{x})=\underline{X}({\theta}_{x})$.

For $k=1,2,...,K$ iterate the following steps using $X_{(k)}({\theta}_{x})$,
$Y_{(k)}$:

$\qquad$b. Supervised selection: 

$\qquad$Select a subset $\widehat{I}_{(k)}({\theta}_{x}):$
\[
\widehat{I}_{(k)}(\theta_x):=\left\{ i\mid T^{-1}\Vert(X_{(k)}(\theta_{x}))_{[i]}Y'_{(k)}\Vert_{\max}\,\ge\,\widehat{c}_{qN}^{(k)}\right\} \subset[N],
\]
where $\widehat{c}_{qN}^{(k)}$ is the $(1-q)$th-quantile of $\left\{ T^{-1}\mid(X_{(k)}(\theta_x))_{[i]}Y'_{(k)}\mid\right\} _{i\in[N]}$, and $q$ is a scaling factor between 0 and 1. This hints selecting predictors with sufficiently high absolute
values of correlation or covariance.

$\qquad$c. Apply the \textbf{Algorithm \ref{Algorithm 2}} S1, PCA,
to this subset $\bold{x}_{(k),[\widehat{I}_{(k)}(\theta_x)]}$
but only getting the first left singular vector  $\widehat{\varsigma}_{(k)}(\theta_{x,[\widehat{I}_{(k)}(\theta_x)]})$.
Hence, we can estimate the $k$-th latent factor as $\widehat{\underline{\bold{f}}}_{H,(k)}(\theta_x)=\widehat{\varsigma}{}_{(k)}(\theta_{x,[\widehat{I}_{(k)}(\theta_x)]})'\bold{x}_{(k),[\widehat{I}_{(k)}(\theta_x)]}$.
$\widehat{\underline{\bold{f}}}_{H,(k)}(\theta_x)$ can also be rewritten
as  $\widehat{\underline{\bold{f}}}_{H,(k)}(\theta_x) =\widehat{\zeta}{}_{(k)}'(\theta_x)\bold{x}_{(k)}$,
where  $\widehat{\zeta}_{(k)}(\theta_x)=(\mathbb{I}_{N}\sideset{-}{_{i=1}^{k-1}}\sum\widehat{\beta}_{(i)}(\theta_x)$
$\widehat{\zeta}{'}_{(i)}(\theta_x))_{[\widehat{I}_{(k)}(\theta_x)]}^{'}$
$\widehat{\varsigma}_{(k)}(\theta_{x,[\widehat{I}_{(k)}(\theta_x)]})$ is constructed recursively using $\widehat{\beta}_{(k-1)}(\theta_x)$
(defined in d). 

$\qquad$d. Projection step: Project $\bold{x}_{(k)}$ onto this factor $\widehat{\underline{\bold{f}}}_{H,(k)}(\theta_x)$ and project
$Y_{(k)}$ onto the corresponding aggregated factor $\widehat{\underline{F}}_{(k)}(\theta)$, where $\theta=(\theta_f,\theta_x)'$,
separately. Then estimate coefficients $\widehat{\beta}_{(k)}(\theta_x)=\bold{x}_{(k)}\widehat{\underline{\bold{f}}}_{H,(k)}'(\theta_x)(\widehat{\underline{\bold{f}}}_{H,(k)}(\theta_x)$ $\widehat{\underline{\bold{f}}}_{H,(k)}'(\theta_x))^{-1}$
and $\widehat{\alpha}_{(k)}(\theta)=Y_{(k)}\widehat{\underline{F}}_{(k)}'(\theta)(\widehat{\underline{F}}_{(k)}(\theta)\widehat{\underline{F}}_{(k)}'(\theta))^{-1}$.
Then compute residuals of  $\bold{x}_{(k)}$  and $Y_{(k)}$,
denoted by $Y_{(k+1)}=Y_{(k)}-\widehat{\alpha}_{(k)}(\theta)\widehat{\underline{F}}_{(k)}(\theta)$
and $\bold{x}_{(k+1)}=\bold{x}_{(k)}-\widehat{\beta}_{(k)}(\theta_x)\widehat{\underline{\bold{f}}}_{H,(k)}(\theta_x)$.

$\qquad$e. Iterate K times in this stage: compute the (univariate)
correlation of the $Y_{(k+1)}$ and $\bold{x}_{(k+1)}$; select
predictors with sufficiently high correlation; extract another one
latent factor; projection, etc. Stop at $k=\widehat{K}$, where $\widehat{K}$
is chosen based on a proper stopping rule.

$\qquad$$\blacktriangleright$ the algorithm terminates as soon as the
stopping rule (\ref{equation 7}) is met.

S2. Once all factors are collected $\widehat{\underline{F}}_{\textrm{SPCA}}(\theta)=(\widehat{\underline{F}}{}_{(1)}'(\theta),\widehat{\underline{F}}{}_{(2)}'(\theta),...,\widehat{\underline{F}}{}_{(\widehat{K})}'(\theta))'$.
Then, following Algorithm \ref{Algorithm 2} S2, obtain $\hat{\alpha}_l=(\hat{\alpha},\hat{\alpha}_{w})'$
and $\hat{\theta}$ by solving $\left(\hat{\alpha}_{l},\hat{\theta}\right)=\underset{\alpha_l,\theta}{\textrm{argmin}}\parallel\overline{Y}-\alpha\underline{\hat{F}}_{\textrm{SPCA}}(\theta)-\alpha_{w}\underline{W}\parallel^{2}$.
The resulting prediction for $\ensuremath{y_{T+h}}$ is given by $\widehat{y}_{T+h}^{\textrm{SPCA}}=\hat{\alpha}\hat{F}_{T,\textrm{SPCA}}(\hat{\theta})+\widehat{\alpha}_{w}w_{T}\equiv\widehat{\alpha}\sum_{j=0}^{J}$
$b^f_{j}(\hat{\theta}_f)\widehat{\zeta}_{SPCA}'(\hat\theta_x)\bold{x}_{T-j/m}+\widehat{\alpha}_{w}w_{T}$, where $\widehat{\zeta}_{SPCA}(\hat\theta_x):=(\widehat{\zeta}_{(1)}(\hat\theta_x),\widehat{\zeta}_{(2)}(\hat\theta_x),...,\widehat{\zeta}_{(\widehat{K})}(\hat\theta_x)$.

outputs: $\widehat{y}_{T+h}^{\textrm{SPCA}}$, $\widehat{\underline{F}}_{\textrm{SPCA}}(\hat{\theta})$,
$\widehat{\zeta}_{\textrm{SPCA}}(\hat{\theta}_x)$,
$\widehat{\alpha}_l,$  $\hat{\theta},$ and factor loadings $\widehat{\beta}(\hat\theta_x)$.
\end{lyxalgorithm}

\textbf{Boosting} is a supervised learning method that builds a strong learner by iteratively fitting weak learners to the residuals of the previous fit, and it can also select the most relevant predictors in a
high-dimensional dataset. Applied to latent factors, it selects the factors most relevant to the target variable $y$, whether in their raw form or as orthonormal transformations, thereby improving the predictive
capability of PCA. In particular, \cite{bai2009boosting} propose a “boosting with latent factors” method, which applies boosting to PCA factors in a linear factor model and establishes its consistency. Following the introduction of the factor-MIDAS model by \cite{marcellino2010factor}, research on boosting in this context has grown, and \cite{lahiri2022boosting} provide empirical evidence that
boosting with factors generally yields the best results in the MIDAS context.

Boosting typically handles lagged variables in two ways: a component-wise approach, which treats each lag as a separate predictor, and a block-wise approach, which groups the lags of the same variable together. In either case, boosting typically
involves two tuning parameters: the shrinkage parameter or learning
rate $\nu\in(0,1]$, and the number of iterations $M$. In practice,
$\nu$ is usually fixed while $M$ is tuned. A common stopping rule is to choose the number of boosting iterations $\widehat{M}$ by minimizing a cross-validated prediction error criterion, although information criteria such as AIC or BIC can also be used as alternatives.

To show that cleaner factors can improve the learning and predictive
abilities of boosting, we adopt a block-wise algorithm adapted from
\cite{bai2009boosting} to our mixed-frequency setup. Whereas their
block-wise procedure determines the number of lags by minimizing an
information criterion such as AIC or BIC, in our setup the number of
lags is determined by the MIDAS weighting scheme. The resulting
procedure is detailed in the algorithm below.

\begin{lyxalgorithm}
\textbf{\textup{(Supervised Learning: Block-wise $L_{2}$ Boosting-MIDAS)}}\label{Algorithm 5}

\noindent Inputs: $\overline{Y},w_{T},\bold{x}_{T_{H}-mh+1}^{T_{H}}$
or scaled $\bold{x}_{T_{H}-mh+1}^{T_{H}}$, estimated weighting parameters
$\left\{ \hat\theta_x,\hat\theta_f\right\}$, any packages of $\left\{ \widehat{\bold{f}}_H,\widehat{\alpha}_{w},\widehat{\zeta}\,or\,\hat{\varsigma}\right\} $
derived from PCA-based methods\footnote{The specific contents of each package depend on the outputs of the corresponding method.}. Initialization: $\hat{\Phi}_{(0)}=mean(\overline{Y})$,
$\widehat{\alpha}_{(0)}=(0_{1},0_{2}...,0_{K})$.

S1. For $m=1,2,...,M$ iterate the following steps using $\hat{\Phi}_{(m)}$:

$\qquad$a. Compute the residual $\hat{u}_{(m)}=\overline{Y}-\hat{\Phi}_{(m-1)}$.

$\qquad$b. Define the base learner as a simple linear regression
model. Choose $\hat{i}_{m}=\arg\min_{i\leq\hat{K}}\sum(\hat{u}_{(m)}-\hat{b}_{i}^{(m)}\widehat{\underline{F}}_{i}(\hat{\theta}_f)-\hat{a}_{i}^{(m)})^{2},$
where $i$ refers to the index of the estimated factor, $\hat{a}_{i}^{(m)}$
and $\hat{b}_{i}^{(m)}$ are the OLS estimators for the linear equation
$\hat{u}_{(m)}=a+b\widehat{\underline{F}}_{i}(\hat{\theta}_f)+\epsilon,$
such that $a$ and $b$ are the intercept and slope parameters and
$\epsilon$ is the error term.

$\qquad$c. Denote $\hat{\vartheta}_{m}=\hat{b}_{\hat{i}_{m}}^{(m)}\widehat{\underline{F}}_{\hat{i}_{m}}(\hat{\theta}_f)+\hat{a}_{\hat{i}_{m}}^{(m)}$,
update $\widehat{\alpha}_{(m)}=\widehat{\alpha}_{(m-1)}+\nu\hat{b}_{\hat{i}_{m}}^{(m)}$
and the forecasts $\hat{\Phi}_{(m)}=\hat{\Phi}_{(m-1)}+\nu\hat{\vartheta}_{m}$,
where $\nu\in(0,1]$.

$\qquad$d. Iterate $M$ times in this stage. Stop at $m=\widehat{M}$,
then get the $\hat{\Phi}_{(\hat{M})}=\hat{\Phi}_{(0)}+\sideset{\nu}{_{m=1}^{\hat{M}}}\sum\hat{\vartheta}_{m}$
and $\widehat{\alpha}_{(\widehat{M})}=\widehat{\alpha}_{(0)}+\sideset{\nu}{_{m=1}^{\hat{M}}}\sum\hat{b}_{\hat{i}_{m}}^{(m)}$.

S2. Get the prediction
\[
\begin{split}\widehat{y}_{T+h}^{\textrm{Boosting}}&=\hat{\Phi}_{(0)}+\widehat{\alpha}_{(\widehat{M})}\widehat{F}_{T}(\hat{\theta}_f)+\widehat{\alpha}_{w}w_{T}\\&\equiv\hat{\Phi}_{(0)}+\widehat{\alpha}_{(\widehat{M})}\sum_{j=0}^{J}b_{j}^f(\hat{\theta}_f)\widehat{\zeta}'(\hat{\theta}_x)\bold{x}_{T-j/m}+\widehat{\alpha}_{w}w_{T}.\end{split}
\]

outputs: $\widehat{y}_{T+h}^{\textrm{Boosting}}$, $\hat{\Phi}$,  and $\widehat{\alpha}_{(\widehat{M})}$. 
\end{lyxalgorithm}

\textbf{Partial Least Squares (PLS)}. We also compare our method with another commonly used supervised
learning technique in financial analysis, PLS.
PLS has been demonstrated to outperform PCA in certain contexts, as
shown by \cite{kelly2013market}. To the best of our knowledge, there is
currently no literature exploring the application of PLS in the MIDAS
framework, which presents an interesting avenue for future research.
The PLS-MIDAS procedure is described in the algorithm below.

\begingroup\renewcommand{\baselinestretch}{1}\small\normalsize 
\refstepcounter{lyxalgorithm}\noindent{\normalsize\textbf{Algorithm \thelyxalgorithm. (Alternative Supervised Learning: PLS-MIDAS)}}\label{Algorithm 6}\par

\textit{Inputs: $\overline{Y},\bold{x}_{H}=(\underline{\bold{x}}_{H},\bold{x}_{T_{H}-mh+1}^{T_{H}}),\underline{W},\textrm{ and }w_{T}.$
Initialization:  $Y_{(1)}:=\overline{Y}\mathbb{M}_{\underline{W}'}$,   $\bold{x}_{(1)}:=\underline{\bold{x}}_{H}$.}

\textit{S1. a. For given values of the parameter ${\theta}_{x}$ and ${\theta}_{f}$,
initialize the aggregated predictors as $X_{(1)}({\theta}_{x})=\underline{X}({\theta}_{x})$.}

\textit{For $k=1,2,...,K$ iterate the following steps using $X_{(k)}({\theta}_{x})$,
$Y_{(k)}$:}

\textit{$\qquad$b. Obtain the weight vector $\widehat{\varsigma}_{(k)}(\theta_x)$
from the largest left singular vector of $X_{(k)}(\theta_x)Y'_{(k)}$.}

\textit{$\qquad$c. Estimate the $k$-th factor as $\widehat{\underline{F}}_{(k)}(\theta)=\widehat{\varsigma}{}_{(k)}(\theta_x)'X_{(k)}(\theta_x)$.}

\textit{$\qquad$d. Estimate coefficients $\widehat{\beta}_{(k)}(\theta)=X_{(k)}(\theta)\widehat{\underline{F}}{}_{(k)}'(\theta)(\widehat{\underline{F}}_{(k)}(\theta)\widehat{\underline{F}}{}_{(k)}'(\theta))^{-1}$
and $\widehat{\alpha}_{(k)}(\theta)=Y_{(k)}\widehat{\underline{F}}_{(k)}'(\theta)\left(\widehat{\underline{F}}_{(k)}(\theta)\widehat{\underline{F}}_{(k)}'(\theta)\right)^{-1}$.}

\textit{$\qquad$e. Remove $\widehat{\underline{F}}_{(k)}$ to obtain
residuals for the next step: $Y_{(k+1)}=Y_{(k)}-\widehat{\alpha}_{(k)}(\theta)\widehat{\underline{F}}_{(k)}(\theta)$
and $X_{(k+1)}(\theta_x)=X_{(k)}(\theta_x)-\widehat{\beta}_{(k)}(\theta)\widehat{\underline{F}}_{(k)}(\theta)$.}

\textit{S2. Collect all the necessary estimators and perform NLS regression
like }\textbf{\textit{Algorithm \ref{Algorithm 4}}}\textit{ S2, we
can derive $\widehat{y}_{T+h}^{\textrm{PLS}}=\hat{\alpha}\hat{F}_{T,\textrm{PLS}}(\hat{\theta})+\widehat{\alpha}_{w}w_{T}\equiv\widehat{\alpha}\widehat{\varsigma}(\hat{\theta}_x)'\sum_{j=0}^{J}b^x_{j}(\hat{\theta}_x)\bold{x}_{T-j/m}+\widehat{\alpha}_{w}w_{T}$.}

\textit{outputs: $\widehat{y}_{T+h}^{\textrm{PLS}},\widehat{\underline{F}}_{\textrm{PLS}}(\hat{\theta}),\hat{\varsigma}(\hat{\theta}_x),\widehat{\beta}(\hat{\theta}),\widehat{\alpha}_l$, 
and $\hat{\theta}$.}
\endgroup

\section{Tuning Parameter Selection}\label{Appendix F}

Tuning parameters (hyperparameters) is crucial in machine learning,
as it affects model performance, generalization, accuracy, and robustness.
Proper tuning helps avoid overfitting or underfitting and ensures
the development of a reliable model. 

For PCA, sPCA, and PLS, we only need to tune one parameter, which
is the number of extracted latent factors, denoted as $K$. For SPCA and our
proposed SsPCA, an additional parameter $\lfloor qN\rfloor$
must be tuned. This parameter represents the number of predictors
in an informative subset of the dataset. This additional tuning cost
is essential for supervised selection and contributes to improved
predictive outcomes.

Although boosting typically involves a trade-off between the learning rate $\nu$ and the number of boosting iterations $M$, we fix both parameters in this paper to avoid additional costs, since our purpose is to explore whether boosting with cleaner factors can enhance learning and prediction capabilities.

In Monte Carlo simulations, we use a standard 3-fold cross-validation
(CV), where the entire sample is split into three sequential folds
due to the time series dependence. Each fold is used once as a validation
set, while the remaining two serve as the training set. A parsimonious
tuning method like 3-fold CV can effectively validate the generalizability
of our predictive model. We chose 3-fold CV over higher-fold (e.g.,
5 or 10) alternatives to preserve the temporal order inherent in time
series data, thereby preventing data leakage and enabling valid out-of-sample
predictions. Although using more folds could reduce variance and provide
finer validation, they may disrupt sequential dependencies and introduce
bias. Moreover, 3-fold CV offers a good balance between computational
efficiency and generalization, which is important given the large
datasets and extensive tuning in our experiments.

In empirical practices, we adopt a rolling-window CV, allocating five-sixths
of the time periods in each window to training and the remaining one-sixth
to validation, with the split chosen in light of the available sample
length to balance estimation stability and validation reliability.
This approach rolls the in-sample windows forward over time so that
validation data always occur later than training data, thereby avoiding
data leakage and maintaining a sufficient sample size for stable parameter
estimation. This differs from the simulation setting, where the data-generating
process is stationary and controlled, making standard CV sufficient
for our purposes.

To illustrate how the windows roll forward in the empirical application,
where the rolling window is 180 months ($T_{H}$) for the predictors and
60 quarters ($T$) for the target variables, consider forecasting 2011 Q4
with $h=1$. The model is first
estimated using monthly predictor data from October 1996 to September
2011. Within the quarter, the monthly window rolls forward month by
month, producing three successive forecasts of the same quarterly
target, while the in-sample quarterly period remains fixed at 1996 Q4
to 2011 Q3. When moving to the next quarter, 2012 Q1, the quarterly
window rolls forward by one quarter, with both window sizes maintained
throughout the evaluation. The same scheme applies to the other horizons, with only the forecast origin adjusted according to $h$.

Since our goal is prediction accuracy, the optimal tuning parameters
are selected based on the average MSE across validation folds. Specifically,
for methods where only $K$ is tuned, the optimal $\widehat{K}$ is
selected based on the minimum average MSE. For our SsPCA, we jointly
tune $\widehat{K}$ and $\widehat{\left\lfloor qN\right\rfloor }$
using this criterion. We also checked parameter selection based on
average $R^{2}$, which yielded very similar results, suggesting no
loss of explanatory power. After
determining these parameters, we refit the model on the entire dataset
before performing predictions. 

A natural concern is the additional computational cost of supervision relative to unsupervised PCA. This cost arises from two sources: the joint estimation of the MIDAS weighting parameters by nonlinear least squares, and the iterative selection-extraction-projection procedure. Both grow with the cross-sectional dimension $N$ and the number of extracted factors, and tuning the additional parameter $\lfloor qN\rfloor$ further enlarges the search. From an economic standpoint, however, this cost is a worthwhile investment rather than a deadweight loss. In high-dimensional macro-financial panels, the predictive signal is concentrated in a small subset of series, while the vast majority of predictors contribute mostly noise. The extra computation buys a sharper separation of signal from noise, and the resulting gain in forecast accuracy, documented in the simulation and empirical application sections of the main text, more than compensates for it. In other words, the marginal computational cost of supervision is small relative to the marginal predictive benefit it delivers, which is the relevant trade-off for a forecaster. The cost also remains manageable in practice: the procedure scales to our empirical application with over 1,500 predictors across multiple targets and horizons without difficulty.


\spacingset{1}
\renewcommand{\bibfont}{\normalsize}
\setlength{\bibsep}{0pt plus 0.5pt}

\bibliography{Mixed_frequency}

@InProceedings{ghysels2009multi,
  author    = {Ghysels, Eric and Valkanov, Rossen I and Serrano, Antonio Rubia},
  booktitle = {EFA 2009 Bergen Meetings Paper},
  title     = {Multi-period forecasts of volatility: Direct, iterated, and mixed-data approaches},
  year      = {2009},
}

@Article{ghysels2007midas,
  author    = {Ghysels, Eric and Sinko, Arthur and Valkanov, Rossen},
  journal   = {Econometric reviews},
  title     = {MIDAS regressions: Further results and new directions},
  year      = {2007},
  number    = {1},
  pages     = {53--90},
  volume    = {26},
  publisher = {Taylor \& Francis},
}

@Article{foroni2015unrestricted,
  author    = {Foroni, Claudia and Marcellino, Massimiliano and Schumacher, Christian},
  journal   = {Journal of the Royal Statistical Society Series A: Statistics in Society},
  title     = {Unrestricted mixed data sampling (MIDAS): MIDAS regressions with unrestricted lag polynomials},
  year      = {2015},
  number    = {1},
  pages     = {57--82},
  volume    = {178},
  publisher = {Oxford University Press},
}

@Article{giglio2023prediction,
  author  = {Giglio, Stefano and Xiu, Dacheng and Zhang, Dake},
  journal = {University of Chicago, Becker Friedman Institute for Economics Working Paper},
  title   = {Prediction when factors are weak},
  year    = {2023},
  number  = {2023-47},
}

@Article{stock2002forecasting,
  author    = {Stock, James H and Watson, Mark W},
  journal   = {Journal of the American statistical association},
  title     = {Forecasting using principal components from a large number of predictors},
  year      = {2002},
  number    = {460},
  pages     = {1167--1179},
  volume    = {97},
  publisher = {Taylor \& Francis},
}

@Article{huang2022scaled,
  author    = {Huang, Dashan and Jiang, Fuwei and Li, Kunpeng and Tong, Guoshi and Zhou, Guofu},
  journal   = {Management Science},
  title     = {Scaled PCA: A new approach to dimension reduction},
  year      = {2022},
  number    = {3},
  pages     = {1678--1695},
  volume    = {68},
  publisher = {INFORMS},
}

@Article{giglio2021test,
  author    = {Giglio, Stefano and Xiu, Dacheng and Zhang, Dake},
  journal   = {The Journal of Finance},
  title     = {Test assets and weak factors},
  year      = {2025},
  number    = {1},
  pages     = {259--319},
  volume    = {80},
  publisher = {Wiley Online Library},
}

@Article{bai2009boosting,
  author    = {Bai, Jushan and Ng, Serena},
  journal   = {Journal of Applied Econometrics},
  title     = {Boosting diffusion indices},
  year      = {2009},
  number    = {4},
  pages     = {607--629},
  volume    = {24},
  publisher = {Wiley Online Library},
}

@Article{lahiri2022boosting,
  author    = {Lahiri, Kajal and Yang, Cheng},
  journal   = {international Journal of forecasting},
  title     = {Boosting tax revenues with mixed-frequency data in the aftermath of COVID-19: The case of New York},
  year      = {2022},
  number    = {2},
  pages     = {545--566},
  volume    = {38},
  publisher = {Elsevier},
}

@Article{marcellino2010factor,
  author    = {Marcellino, Massimiliano and Schumacher, Christian},
  journal   = {Oxford Bulletin of Economics and Statistics},
  title     = {Factor MIDAS for nowcasting and forecasting with ragged-edge data: A model comparison for German GDP},
  year      = {2010},
  number    = {4},
  pages     = {518--550},
  volume    = {72},
  publisher = {Wiley Online Library},
}

@Article{kelly2013market,
  author    = {Kelly, Bryan and Pruitt, Seth},
  journal   = {The Journal of Finance},
  title     = {Market expectations in the cross-section of present values},
  year      = {2013},
  number    = {5},
  pages     = {1721--1756},
  volume    = {68},
  publisher = {Wiley Online Library},
}

@Article{hounyo2023forecasting,
  author    = {Hounyo, Ulrich and Li, Zhendong},
  journal   = {International Journal of Forecasting},
  title     = {Forecasting economic time series in the presence of weak factors: Multiple supervised learning-based approach},
  year = {2026},
  number = {2},
  pages = {414-433},
  volume = {42},
  publisher = {Elsevier},
}

@Article{fan2011high,
  author    = {Fan, Jianqing and Liao, Yuan and Mincheva, Martina},
  journal   = {Annals of statistics},
  title     = {High dimensional covariance matrix estimation in approximate factor models},
  year      = {2011},
  number    = {6},
  pages     = {3320},
  volume    = {39},
  publisher = {NIH Public Access},
}

@Article{bai2003inferential,
  author    = {Bai, Jushan},
  journal   = {Econometrica},
  title     = {Inferential theory for factor models of large dimensions},
  year      = {2003},
  number    = {1},
  pages     = {135--171},
  volume    = {71},
  publisher = {Wiley Online Library},
}

@Article{wang2017asymptotics,
  author    = {Wang, Weichen and Fan, Jianqing},
  journal   = {Annals of statistics},
  title     = {Asymptotics of empirical eigenstructure for high dimensional spiked covariance},
  year      = {2017},
  number    = {3},
  pages     = {1342--1374},
  volume    = {45},
  publisher = {NIH Public Access},
}

@Article{wohlrabe2014assessing,
  author    = {Wohlrabe, Klaus and Buchen, Teresa},
  journal   = {Journal of Forecasting},
  title     = {Assessing the macroeconomic forecasting performance of boosting: evidence for the United States, the Euro area and Germany},
  year      = {2014},
  number    = {4},
  pages     = {231--242},
  volume    = {33},
  publisher = {Wiley Online Library},
}

@Article{bai2002determining,
  author    = {Bai, Jushan and Ng, Serena},
  journal   = {Econometrica},
  title     = {Determining the number of factors in approximate factor models},
  year      = {2002},
  number    = {1},
  pages     = {191--221},
  volume    = {70},
  publisher = {Wiley Online Library},
}

@Article{onatski2009testing,
  author    = {Onatski, Alexei},
  journal   = {Econometrica},
  title     = {Testing hypotheses about the number of factors in large factor models},
  year      = {2009},
  number    = {5},
  pages     = {1447--1479},
  volume    = {77},
  publisher = {Wiley Online Library},
}

@Article{onatski2010determining,
  author    = {Onatski, Alexei},
  journal   = {The Review of Economics and Statistics},
  title     = {Determining the number of factors from empirical distribution of eigenvalues},
  year      = {2010},
  number    = {4},
  pages     = {1004--1016},
  volume    = {92},
  publisher = {The MIT Press},
}

@Article{onatski2012asymptotics,
  author    = {Onatski, Alexei},
  journal   = {Journal of Econometrics},
  title     = {Asymptotics of the principal components estimator of large factor models with weakly influential factors},
  year      = {2012},
  number    = {2},
  pages     = {244--258},
  volume    = {168},
  publisher = {Elsevier},
}

@Article{bai2023approximate,
  author    = {Bai, Jushan and Ng, Serena},
  journal   = {Journal of Econometrics},
  title     = {Approximate factor models with weaker loadings},
  year      = {2023},
  number    = {2},
  pages     = {1893--1916},
  volume    = {235},
  publisher = {Elsevier},
}

@Article{pelger2022interpretable,
  author    = {Pelger, Markus and Xiong, Ruoxuan},
  journal   = {Journal of Business \& Economic Statistics},
  title     = {Interpretable sparse proximate factors for large dimensions},
  year      = {2022},
  number    = {4},
  pages     = {1642--1664},
  volume    = {40},
  publisher = {Taylor \& Francis},
}

@Article{huang2023bond,
  author    = {Huang, Dashan and Jiang, Fuwei and Li, Kunpeng and Tong, Guoshi and Zhou, Guofu},
  journal   = {Journal of Econometrics},
  title     = {Are bond returns predictable with real-time macro data?},
  year      = {2023},
  number    = {2},
  pages     = {105438},
  volume    = {237},
  publisher = {Elsevier},
}

@Article{bai2004estimating,
  author    = {Bai, Jushan},
  journal   = {Journal of Econometrics},
  title     = {Estimating cross-section common stochastic trends in nonstationary panel data},
  year      = {2004},
  number    = {1},
  pages     = {137--183},
  volume    = {122},
  publisher = {Elsevier},
}

@Article{bai2004panic,
  author    = {Bai, Jushan and Ng, Serena},
  journal   = {Econometrica},
  title     = {A PANIC attack on unit roots and cointegration},
  year      = {2004},
  number    = {4},
  pages     = {1127--1177},
  volume    = {72},
  publisher = {Wiley Online Library},
}

@Misc{ghysels2004midas,
  author = {Ghysels, E and Santa-Clara, P and Valkanov, R},
  title  = {The midas touch: Mixed data sampling regression models. UCLA: Finance},
  year   = {2004},
}

@Article{Clements2008,
  author    = {Clements, Michael P and Galv{\~a}o, Ana Beatriz},
  journal   = {Journal of Business \& Economic Statistics},
  title     = {Macroeconomic forecasting with mixed-frequency data: Forecasting output growth in the United States},
  year      = {2008},
  number    = {4},
  pages     = {546--554},
  volume    = {26},
  publisher = {Taylor \& Francis},
}

@Article{kim2018methods,
  author    = {Kim, Hyun Hak and Swanson, Norman R},
  journal   = {Journal of Forecasting},
  title     = {Methods for backcasting, nowcasting and forecasting using factor-MIDAS: With an application to Korean GDP},
  year      = {2018},
  number    = {3},
  pages     = {281--302},
  volume    = {37},
  publisher = {Wiley Online Library},
}

@Article{ferrara2019nowcasting,
  author    = {Ferrara, Laurent and Marsilli, Cl{\'e}ment},
  journal   = {The World Economy},
  title     = {Nowcasting global economic growth: A factor-augmented mixed-frequency approach},
  year      = {2019},
  number    = {3},
  pages     = {846--875},
  volume    = {42},
  publisher = {Wiley Online Library},
}

@Article{johnstone2009consistency,
  author    = {Johnstone, Iain M and Lu, Arthur Yu},
  journal   = {Journal of the American Statistical Association},
  title     = {On consistency and sparsity for principal components analysis in high dimensions},
  year      = {2009},
  number    = {486},
  pages     = {682--693},
  volume    = {104},
  publisher = {Taylor \& Francis},
}

@Article{koh2023inference,
  author  = {Koh, Julia},
  journal = {Available at SSRN 5096277},
  title   = {Inference for Factor-MIDAS Regression Models},
  year    = {2023},
}

@Article{lettau2020estimating,
  author    = {Lettau, Martin and Pelger, Markus},
  journal   = {Journal of Econometrics},
  title     = {Estimating latent asset-pricing factors},
  year      = {2020},
  number    = {1},
  pages     = {1--31},
  volume    = {218},
  publisher = {Elsevier},
}

@Article{uematsu2022estimation,
  author    = {Uematsu, Yoshimasa and Yamagata, Takashi},
  journal   = {Journal of Business \& Economic Statistics},
  title     = {Estimation of sparsity-induced weak factor models},
  year      = {2022},
  number    = {1},
  pages     = {213--227},
  volume    = {41},
  publisher = {Taylor \& Francis},
}

@Article{freyaldenhoven2022factor,
  author    = {Freyaldenhoven, Simon},
  journal   = {Journal of Econometrics},
  title     = {Factor models with local factors determining the number of relevant factors},
  year      = {2022},
  number    = {1},
  pages     = {80--102},
  volume    = {229},
  publisher = {Elsevier},
}

@Article{kan1999gmm,
  author    = {Kan, Raymond and Zhang, Chu},
  journal   = {Journal of Financial Economics},
  title     = {GMM tests of stochastic discount factor models with useless factors},
  year      = {1999},
  number    = {1},
  pages     = {103--127},
  volume    = {54},
  publisher = {Elsevier},
}

@Article{kan1999two,
  author    = {Kan, Raymond and Zhang, Chu},
  journal   = {the Journal of Finance},
  title     = {Two-pass tests of asset pricing models with useless factors},
  year      = {1999},
  number    = {1},
  pages     = {203--235},
  volume    = {54},
  publisher = {Wiley Online Library},
}

@Article{kleibergen2009tests,
  author    = {Kleibergen, Frank},
  journal   = {Journal of econometrics},
  title     = {Tests of risk premia in linear factor models},
  year      = {2009},
  number    = {2},
  pages     = {149--173},
  volume    = {149},
  publisher = {Elsevier},
}

@Article{bair2006prediction,
  author    = {Bair, Eric and Hastie, Trevor and Paul, Debashis and Tibshirani, Robert},
  journal   = {Journal of the American Statistical Association},
  title     = {Prediction by supervised principal components},
  year      = {2006},
  number    = {473},
  pages     = {119--137},
  volume    = {101},
  publisher = {Taylor \& Francis},
}

@Article{bai2008forecasting,
  author    = {Bai, Jushan and Ng, Serena},
  journal   = {Journal of Econometrics},
  title     = {Forecasting economic time series using targeted predictors},
  year      = {2008},
  number    = {2},
  pages     = {304--317},
  volume    = {146},
  publisher = {Elsevier},
}

@Article{gao2024supervised,
  author  = {Gao, Zhaoxing and Tsay, Ruey S},
  journal = {Journal of the American Statistical Association},
  title   = {Supervised dynamic {PCA}: Linear dynamic forecasting with many predictors},
  year    = {2025},
  volume  = {120},
  number  = {550},
  pages   = {869--883},
}

@Article{kelly2023financial,
  author    = {Kelly, Bryan and Xiu, Dacheng and others},
  journal   = {Foundations and Trends{\textregistered} in Finance},
  title     = {Financial machine learning},
  year      = {2023},
  number    = {3-4},
  pages     = {205--363},
  volume    = {13},
  publisher = {Now Publishers, Inc.},
}

@Article{babii2022machine,
  author    = {Babii, Andrii and Ghysels, Eric and Striaukas, Jonas},
  journal   = {Journal of Business \& Economic Statistics},
  title     = {Machine learning time series regressions with an application to nowcasting},
  year      = {2022},
  number    = {3},
  pages     = {1094--1106},
  volume    = {40},
  publisher = {Taylor \& Francis},
}

@Article{chao2022selecting,
  author  = {Chao, John C and Swanson, Norman R},
  title   = {Consistent estimation, variable selection, and forecasting in {FAVAR} models},
  journal = {Available at SSRN 4010249},
  year    = {2022},
}

@Article{andreou2010regression,
  author    = {Andreou, Elena and Ghysels, Eric and Kourtellos, Andros},
  journal   = {Journal of Econometrics},
  title     = {Regression models with mixed sampling frequencies},
  year      = {2010},
  number    = {2},
  pages     = {246--261},
  volume    = {158},
  publisher = {Elsevier},
}

@Article{mccracken2016fred,
  author    = {McCracken, Michael W and Ng, Serena},
  journal   = {Journal of Business \& Economic Statistics},
  title     = {FRED-MD: A monthly database for macroeconomic research},
  year      = {2016},
  number    = {4},
  pages     = {574--589},
  volume    = {34},
  publisher = {Taylor \& Francis},
}

@TechReport{mccracken2020fred,
  author      = {McCracken, Michael and Ng, Serena},
  institution = {National Bureau of Economic Research},
  title       = {FRED-QD: A quarterly database for macroeconomic research},
  year        = {2020},
}

@Article{jensen2023there,
  author    = {Jensen, Theis Ingerslev and Kelly, Bryan and Pedersen, Lasse Heje},
  journal   = {The Journal of Finance},
  title     = {Is there a replication crisis in finance?},
  year      = {2023},
  number    = {5},
  pages     = {2465--2518},
  volume    = {78},
  publisher = {Wiley Online Library},
}

@TechReport{rey2015dilemma,
  author      = {Rey, H{\'e}l{\`e}ne},
  institution = {National Bureau of Economic Research},
  title       = {Dilemma not trilemma: the global financial cycle and monetary policy independence},
  year        = {2015},
}

@Article{miranda2020us,
  author    = {Miranda-Agrippino, Silvia and Rey, H{\'e}lene},
  journal   = {The Review of Economic Studies},
  title     = {US monetary policy and the global financial cycle},
  year      = {2020},
  number    = {6},
  pages     = {2754--2776},
  volume    = {87},
  publisher = {Oxford University Press},
}

@Article{Ang2006,
  author    = {Ang, Andrew and Piazzesi, Monika and Wei, Min},
  journal   = {Journal of econometrics},
  title     = {What does the yield curve tell us about GDP growth?},
  year      = {2006},
  number    = {1-2},
  pages     = {359--403},
  volume    = {131},
  publisher = {Elsevier},
}

@Article{ludvigson2009macro,
  author    = {Ludvigson, Sydney C and Ng, Serena},
  journal   = {The Review of Financial Studies},
  title     = {Macro factors in bond risk premia},
  year      = {2009},
  number    = {12},
  pages     = {5027--5067},
  volume    = {22},
  publisher = {Oxford University Press},
}

@Article{gonccalves2014bootstrapping,
  author    = {Gon{\c{c}}alves, S{\'\i}lvia and Perron, Benoit},
  journal   = {Journal of Econometrics},
  title     = {Bootstrapping factor-augmented regression models},
  year      = {2014},
  number    = {1},
  pages     = {156--173},
  volume    = {182},
  publisher = {Elsevier},
}

@Article{beyhum2024factor,
  author  = {Beyhum, Jad and Striaukas, Jonas},
  title   = {Factor-augmented sparse {MIDAS} regressions with an application to nowcasting},
  journal = {Journal of Business \& Economic Statistics},
  year    = {2026},
  note    = {Forthcoming},
}

@article{estrella1991term,
  title={The term structure as a predictor of real economic activity},
  author={Estrella, Arturo and Hardouvelis, Gikas A},
  journal={The journal of Finance},
  volume={46},
  number={2},
  pages={555--576},
  year={1991},
  publisher={Wiley Online Library}
}

@article{gilchrist2012credit,
  title={Credit spreads and business cycle fluctuations},
  author={Gilchrist, Simon and Zakraj{\v{s}}ek, Egon},
  journal={American economic review},
  volume={102},
  number={4},
  pages={1692--1720},
  year={2012},
  publisher={American Economic Association}
}

@techreport{di2022global,
  title={Global supply chain pressures, international trade, and inflation},
  author={Di Giovanni, Julian and Kalemli-{\"O}zcan, {\d S}ebnem and Silva, Alvaro and Yildirim, Muhammed A},
  year={2022},
  institution={National Bureau of Economic Research}
}

@article{fan2013large,
  title={Large covariance estimation by thresholding principal orthogonal complements},
  author={Fan, Jianqing and Liao, Yuan and Mincheva, Martina},
  journal={Journal of the Royal Statistical Society Series B: Statistical Methodology},
  volume={75},
  number={4},
  pages={603--680},
  year={2013},
  publisher={Oxford University Press}
}

@Article{bai2006confidence,
  author    = {Bai, Jushan and Ng, Serena},
  journal   = {Econometrica},
  title     = {Confidence intervals for diffusion index forecasts and inference for factor-augmented regressions},
  year      = {2006},
  number    = {4},
  pages     = {1133--1150},
  volume    = {74},
  publisher = {Wiley Online Library},
}

@Article{davis1970rotation,
  author    = {Davis, Chandler and Kahan, W. M.},
  journal   = {SIAM Journal on Numerical Analysis},
  title     = {The rotation of eigenvectors by a perturbation. {III}},
  year      = {1970},
  number    = {1},
  pages     = {1--46},
  volume    = {7},
  publisher = {SIAM},
}

@Article{giglio2021asset,
  author    = {Giglio, Stefano and Xiu, Dacheng},
  journal   = {Journal of Political Economy},
  title     = {Asset pricing with omitted factors},
  year      = {2021},
  number    = {7},
  pages     = {1947--1990},
  volume    = {129},
  publisher = {University of Chicago Press},
}

@Article{wedin1972perturbation,
  author    = {Wedin, Per-{\AA}ke},
  journal   = {BIT Numerical Mathematics},
  title     = {Perturbation bounds in connection with singular value decomposition},
  year      = {1972},
  number    = {1},
  pages     = {99--111},
  volume    = {12},
  publisher = {Springer},
}

\end{document}